\documentclass[a4paper,12pt]{article}
\usepackage[latin1]{inputenc}
\usepackage[british]{babel}
\usepackage{natbib}
\usepackage{graphicx}
\usepackage[pdftex]{lscape}
\usepackage{verbatim}
\usepackage{amsmath}
\usepackage{amsthm}
\usepackage{amssymb}
\usepackage{amsfonts}
\usepackage{setspace}
\usepackage{tabularx,rotate,multicol,caption}
\usepackage{arydshln}
\usepackage{latexsym}
\usepackage{a4}
\usepackage{bbm}
\usepackage{mathrsfs}
\usepackage{natbib}
\usepackage{fixmath}
\usepackage{tikz}
\usepackage{makeidx}
\usepackage{ushort}
\usepackage{enumitem}
\usepackage{lscape}
\usepackage{float}
\usepackage{appendix}
\usepackage{soul}
\usepackage{amsmath,booktabs,longtable}
\usepackage{url}

\usepackage[top=2.1cm,bottom=2.1cm,left=2.1cm,right=2.1cm]{geometry} 
 
\newcommand{\rk}{\mathrm{rank}\,}

\newcommand{\ud}{\mathrm{d}}

\newcommand{\pift}{\tilde{\pi}_\phi}
\newcommand{\psf}{\theta=(A_0,A_+,\Lambda_1,\Lambda_2)}

\newcommand{\Chol}{\Omega_{1,tr}}
\newcommand{\Choli}{\Omega_{1,tr}^{-1}}
\newcommand{\opf}{\omega_{p,1}^{2}}
\newcommand{\ops}{\omega_{p,2}^{2}}
\newcommand{\oqf}{\omega_{q,1}^{2}}
\newcommand{\oqs}{\omega_{q,2}^{2}}
\newcommand{\opqf}{\omega_{pq,1}}
\newcommand{\opqs}{\omega_{pq,2}}

\newcommand{\Ai}{A_{0}^{-1}}
\newcommand{\Ait}{A_{0}^{-1\prime}}
\newcommand{\e}{\varepsilon}

\newcommand{\Qa}{\mathcal{Q}(\phi)}

\newcommand{\Qr}{\ensuremath\mathcal{Q}(\phi\,|F,S)}
\newcommand{\IS}{IS(\phi\,|F,S)}
\newcommand{\ja}{j^\ast}

\newcommand{\Qli}{Q(\lambda^i)}
\newcommand{\Qlip}{Q^{\perp}(\lambda^i)}
\newcommand{\qja}{q_{j^\ast}^{i}}
\newcommand{\Pn}{\mathcal{P}(n)}
\newcommand{\Dn}{\mathcal{D}(n)}

\newcommand{\Eji}{\mathcal{E}_{\ja}^i(\phi)}
\newcommand{\qt}{\tilde{q}_1^i}
\newcommand{\qtj}{\tilde{q}_{\ja}^i}
\newcommand{\qj}{q_{\ja}^i}
\newcommand{\Fr}{\textbf{F}(\phi,Q)}
\newcommand{\Sr}{\textbf{S}(\phi,Q)}

\ifdefined\O
    \renewcommand{\O}{\ensuremath\mathcal{O}\left(2\right)}
\else
    \newcommand{\O}{\ensuremath\mathcal{O}\left(2\right)}
\fi
\ifdefined\On
    \renewcommand{\On}{\ensuremath\mathcal{O}\left(n\right)}
\else
    \newcommand{\On}{\ensuremath\mathcal{O}\left(n\right)}
\fi
\ifdefined\Re
    \renewcommand{\Re}{\ensuremath\mathbb{R}}
\else
    \newcommand{\Re}{\ensuremath\mathbb{R}}
\fi
\ifdefined\Ar
    \renewcommand{\Ar}{\ensuremath\mathcal{A}^r(\phi)}
\else
    \newcommand{\Ar}{\ensuremath\mathcal{A}^r(\phi)}
\fi
\ifdefined\An
    \renewcommand{\An}{\ensuremath\mathcal{A}_0^n(\phi)}
\else
    \newcommand{\An}{\ensuremath\mathcal{A}_0^n(\phi)}
\fi

\def\thesection{\Roman{section}}
\makeatletter
\renewcommand{\thefigure}{\arabic{figure}}

\newtheorem{theorem}{Theorem}
\newtheorem{prop}{Proposition}
\newtheorem{lemma}{Lemma}
\newtheorem{algo}{Algorithm}
\newtheorem{corol}{Corollary}
\theoremstyle{definition}
\newtheorem{assump}{Assumption}
\newtheorem{defin}{Definition}

\newtheorem{remark}{Remark}

\date{}
\makeindex

\begin{document}

\title{\textbf{Partially identified heteroskedastic SVARs}\thanks{Emanuele Bacchiocchi gratefully acknowledges financial support from
Ministero dell'Universit\`a e della Ricerca (MUR), PRIN 20229PFAX5, and the University of Bologna, RFO grants.}}

\author{Emanuele Bacchiocchi\thanks{University of Bologna, Department of Economics. Email: e.bacchiocchi@unibo.it}\hspace{2cm}
Andrea Bastianin\thanks{University of Milan, Department of Economics, Management and Quantitative Methods and Fondazione Eni Enrico Mattei. Email: andrea.bastianin@unimi.it}\\\\
Toru Kitagawa\thanks{Brown University, Department of Economics. Email: toru\_kitagawa@brown.edu}\hspace{3.5cm}
Elisabetta Mirto\thanks{Study Center Gerzensee, Foundation of the Swiss National Bank. Email: elisabetta.mirto@szgerzensee.ch}\\
}

\maketitle

\begin{abstract}
\noindent This paper presents new results on the identification of heteroskedastic structural vector autoregressive (HSVAR) models. Point identification through heteroskedasticity alone fails when two or more structural shocks have identical population relative variance shifts across regimes.
This paper proposes a new strategy that allows researchers to continue using HSVAR models in this empirically relevant case. We show that a combination of heteroskedasticity and zero restrictions can recover point identification in HSVAR models even in the absence of heterogeneous variance shifts. We derive the identified sets for impulse responses and show how to compute them. We perform inference on the impulse response functions, building on the robust Bayesian approach developed for set-identified SVARs. 
To illustrate our proposal, we present an empirical example based on the literature on the global crude oil market, where standard identification is expected to fail under heteroskedasticity. 
A set of Monte Carlo experiments complements the analysis and investigates the properties of our methodology.

\end{abstract}


\begin{flushleft}
\textit{Keywords}: Heteroskedastic SVAR, point and set identification, robust Bayesian approach.\newline
\bigskip
\end{flushleft}

\newpage

\defcitealias{LMNS20}{L\"{u}tkepohl et al. (2021)} 
\defcitealias{GKV22}{Giacomini et al. (2022)} 
\defcitealias{LMYJoE20}{L\"{u}tkepohl et al. (2020)} 
\defcitealias{CMT23WP}{Carriero et al. (2024)} 
\defcitealias{RWZ10RES}{Rubio-Ram{\'{\i}}rez et al. (2010)} 
\defcitealias{GMMO18}{Gafarov et al. (2018)} 
\defcitealias{GMS18}{Granziera et al. (2018)} 
\defcitealias{ARW18}{Arias et al. (2018)} 
\defcitealias{HWZ07}{Hamilton et al. (2007)} 
\defcitealias{CKY24}{Chan et al. (2021)} 
\defcitealias{BPSS21}{Brunnermeier et al. (2021)}

\doublespacing
\section{Introduction}
\label{sec:intro}

In recent years, there has been growing interest in the use of external information to identify the parameters of econometric models. This external information, derived from features of the data or from proxies designed to capture specific structural shocks, is often combined with restrictions on the parameter space suggested by economic theory. A case in point is the literature that exploits heteroskedasticity to identify simultaneous equations systems \citep{Rigobon03} and structural vector autoregressive (SVARs) models \citep{LanneLutkepohl08JMCB}. When the data exhibit volatility clusters that can be attributed to shifts in the variance of structural innovations -- while holding the parameters of the conditional mean constant -- there are gains in terms of identification of the structural parameters (see, for example, \citeauthor{KLbook}, \citeyear{KLbook}, Chapter 14, or the recent empirical work of Brunnermeier et al., 2021). 

While heteroskedastic SVAR models have become a standard tool in macroeconometrics, it is important to recognize that two caveats apply when using them. First, identification through heteroskedasticity is a purely statistical identification strategy so that shocks can have structural interpretations only if their corresponding impulse response functions are accompanied with credible economic interpretation. Second, the size of the shifts in the variance of different shocks must be sufficiently heterogeneous to attain identification.

If some of the shifts in the variances are not distinct, heteroskedasticity cannot deliver identification of all the structural shocks. In this case, the HSVAR model becomes partially identified, but the current literature does not study how to construct the impulse response identified set and how to perform estimation and inference for it. It is also of interest how combining with other types of identification assumptions involving zero and sign restrictions can tighten the impulse response identified set. 

As the main contribution of this paper, we show a novel identification analysis for HSVARs in the absence of fully heterogeneous variance shifts of the structural shocks between the regimes. We show that when multiple structural shocks share the same proportionality constant for their variance shifts between the regimes, the impulse response functions to these shocks are not point-identified and they are only set-identified. We characterize analytically the impulse response identified set and clarify how it depends on heterogeneity of the structural shock variance shifts. 
We then combine the presence of heteroskedasticity with other types of identifying assumptions like zero and/or sign restrictions, and characterize how they can tighten the identified set.
In particular, we characterize the conditions under which a combination of heteroskedasticity and zero restrictions restores point identification in HSVARs in the absence of heterogeneous variance shifts. 

As another contribution of this paper, we propose an estimation and inference procedure for the identified set. Whether HSVARs can be point- or set-identified can be learned from the reduced-form error covariance matrices. Given that the reduced-form error covariance matrices indicate set-identified impulse responses, we adapt the robust Bayesian approach of \cite{GK18} to our setting and propose a procedure to estimate and infer the impulse response identified set. 
In practice, we face estimation uncertainty in the reduced-form error covariance matrices which translates into uncertainty between point- and set-identified models. 
We handle this uncertainty in our proposal in a (robust) Bayesian manner, either by Bayesian model selection or model averaging. 

The former runs the following two-step procedure. 
In the first step, we perform a Bayesian model selection to judge whether evidence supports a set-identified model or not based on Savage-Dickey Density Ratio (SDDR). Our proposal in this step extends the proposal of \cite{LW20} by allowing joint equality restrictions involving more than two variance shifts. If the first step is in support of the set-identified scenario, the second step performs the robust-Bayes inference for set-identified impulse responses conditional on that the support of the reduced-form parameters is constrained to the set-identified model. Conversely, if the first step supports the point-identified models against any of the set-identified models, the second step implements the standard Bayesian inference with point-identified impulse responses. 
The latter approach of model averaging aggregates the robust Bayesian posterior in the set-identified models and the standard posterior in the point-identified models by the robust-Bayesian model averaging procedure of \cite{GKV22} with the weights given by the posterior model probabilities.
Our proposal thus offers a new tool for empirical analyses with which the practitioners can perform a robust inference even when volatility shifts provide useful but insufficient source of information for the identification of structural shocks.

A set of simulation experiments and an empirical example on the identification of structural shocks driving the real price of crude oil illustrate our methodology. In this example, we show how to set or point identify structural shocks when the standard HSVAR approach fails. 

The rest of the paper is organized as follows. The next section briefly surveys the literature; Section \ref{sec:HSVAR} introduces the econometric framework and provides results on point and set identification of HSVARs. Section \ref{sec:TestInference} focuses on the inferential analysis of identified sets through a robust Bayesian approach. Section \ref{sec:MC} reports a set of Monte Carlo experiments while Section \ref{sec:empirics} presents the empirical example. Section \ref{sec:conclusion} provides some concluding remarks. An appendix with the proofs and many other results completes the paper.

\subsection{Related literature}
\label{sec:literature}

This paper builds on the literature of identification through heteroskedasticity both for simultaneous equations systems \citep{Rigobon03,KV2010,Lewbel12,LMYJoE20} and SVARs \citep[see e.g.][]{LanneLutkepohl08JMCB,Bacchiocchi17,KLbook,LN2017,SF01JoE}. 
For estimation and inference, our proposal of the two-step procedure makes use of the recent advances on how to empirically assess validity of heteroskedasticity identification including the frequentist statistical tests for heterogeneous variance shifts by \citetalias{LMNS20} and \cite{Lewis22}, and the Bayesian assessment of identification through heteroskedasticity in \cite{LW20}. Different approaches exploiting heteroskedasticity for the identification of structural shocks include \cite{Lewis21}, which studies identification through time-varying volatility of unknown form. \citet{lutkepohl2021heteroscedastic} analyzes identification through heteroskedasticity in the context of proxy SVAR models. 

A paper that closely relates to ours is \citetalias{CMT23WP}, where the authors combine different identification strategies, each of which contributes to complement the identification challenges of the others. Specifically, they propose to combine sign restrictions, heteroskedasticity, and external instruments. In our approach we focus on the settings where identification fails due to insufficient heterogeneity in the breaks of the variances, and study how adding zero and sign restrictions can restore identification. In their approach, instead, heteroskedasticity is imposed to strengthen identification in the settings where point-identification of the impulse responses is available in the presence of sign, narrative, and external instrument restrictions.

While in our contribution the main assumption is that only the variances of the shocks are subject to breaks, other authors find evidence of structural shifts among the structural parameters of the model, too \citep[see, among others][]{SimsZha06AER,InoueRossi11RESTATS,BG06RESTATS,BacchiocchiFanelli15,BK20SVARWB}. This literature allows for impulse responses to be different across the different regimes, while they are the same up to scaling when only the variances do change. In a related framework, \cite{MM14ECTA} exploit instability in the moments to strengthen identification of structural parameters that are assumed to remain stable over time.

This paper contributes also to the literature on point and set identified SVARs. For point identification, we exploit the general criteria in \citetalias{RWZ10RES}, with subsequent modifications proposed by \cite{BKglob20}, for global identification on SVARs. As for set identification, our approach can combine the sign restrictions to set identify the structural impulse response functions of interest (Rubio-Ram{\'{\i}}rez et al., 2010; \citeauthor{Uhlig05JME}, \citeyear{Uhlig05JME}). 

Strictly connected to this last point is the literature on how to do inference on set identified models. As stated in the introduction, our approach builds on \cite{GK18}, but other approaches have been proposed in the literature to perform inference. See, among others \cite{GMMO18,GMS18,BH15} and \cite{ARW18}.

\section{SVARs and HSVARs: definitions and identification issues}
\label{sec:HSVAR}

\subsection{Econometric framework}
\label{sec:model}

Consider the following Structural Vector Autoregressive (SVAR) model
\begin{equation}
	\label{eq:HSVAR}
	A_{0}y_{t}=\mu+\sum_{i=1}^{l}A_{i}y_{t-i}+\e_{t}
\end{equation}
where $y_{t}$ is a $n$-dimensional vector of observable variables, $\e_{t}$ is a vector of mutually orthogonal white noise processes, normally distributed with mean zero and time-varying covariance matrix. Specifically, let the covariance matrix of the structural shocks $\e_{t}$ be as follows

\begin{equation}
	\label{eq:CovarHSVAR}
	E(\e_{t}\e_{t}^\prime)=\left\{\begin{array}{lcl}
    	\Lambda_1 & & \text{if }\quad 1 \leq t\leq T_{B}\\
        \Lambda_2 & & \text{if }\quad T_{B}< t\leq T
	\end{array}\right.
\end{equation}
where $1 < T_B < T$ is the break date, while $\Lambda_1$ and $\Lambda_2$ are $(n\times n)$ diagonal matrices made of strictly positive numbers.\footnote{Throughout the paper we focus on the setting with two regimes, being this case the most problematic one in terms of lack of information coming from the heteroskedasticity in the data to identify the model. For the generalization to more regimes one can refer to \citetalias{CMT23WP} and \cite{LW20}.} The $n\times 1$ vector $\mu$ contains the intercepts and the $n\times n$ matrices $A_{i}$, with $i=0,\ldots,l$, collect the structural parameters. Moreover, let $A_0$ be invertible and characterized by unit elements on the main diagonal, such that the \textit{i}-th row in Eq. (\ref{eq:HSVAR}) reads as an equation for the \textit{i}-th endogenous variable. The structural parameters can be indicated as $\theta=\left(A_0, A_+, \Lambda_1, \Lambda_2\right)\in \Theta \subset \Re^{\left(n+m\right)n+n}$, with $m=nl+1$, and where the $n\times m$ matrix $A_+=\left(\mu, A_{1},\ldots, A_{l}\right)$. We denote the open dense set of all structural parameters by $\mathbb{P}^S\subset \Re^{\left(n+m\right)n+n}$. The model in Eq.s (\ref{eq:HSVAR})-(\ref{eq:CovarHSVAR}) is a standard SVAR model with structural shocks characterized by different volatility regimes. 
Hereafter, this model is referred to as heteroskedastic SVAR (HSVAR).

The reduced-form VAR model can be written as  
\begin{equation}
	\label{eq:HVAR}
	y_{t}=b+\sum_{i=1}^{l}B_{i}y_{t-i}+u_{t}
\end{equation}
where $b=\Ai \mu$, $B_{i}=\Ai A_{i}$. Furthermore, for both the regimes, the vector of error terms is defined as $u_t=\Ai\e_t$, with 
\begin{equation}
	\label{eq:CovarHVAR}
	E(u_tu_t^\prime)=\left\{\begin{array}{lcl}
	\Omega_1 = \Ai \Lambda_1 \Ait & & \text{if }\quad 1\leq t \leq T_{B}\\
  \Omega_2 = \Ai \Lambda_2 \Ait& & \text{if }\quad T_{B}<t\leq 	T.
	\end{array}\right.
\end{equation}
The VAR model, thus, presents different covariance matrices of the error terms $\Omega_1$ and $\Omega_2$, and thus heteroskedasticity. The reduced-form parameters are denoted by $\phi=(B, \Omega_1, \Omega_2)\in \Phi \subset \Re^{nm}\times \Omega_n \times \Omega_n$, where the $n\times m$ matrix $B=\left(b,B_{1},\ldots,B_{l}\right)$ and $\Omega_n$ is the space of positive-semidefinite matrices of dimension $n\times n$. 
The set of all reduced-form parameters is denoted by $\mathbb{P}^R\subset \Re^{nm+n\left(n+1\right)}$. The reduced form will be also referred to as HVAR.

Conditional on the restrictions of the domain $\Phi$ such that all the roots of the characteristic polynomial lie outside the unit circle, the equivalent VMA$(\infty)$ representation for the HVAR is $y_t = c+\sum_{j=0}^{\infty}C_{j}\Ai\e_{t-j}$, where $C_{j}$ is the \textit{j}-th coefficient matrix of $\bigg(I_n-\sum_{i=1}^{l}B_{i}L^i\bigg)^{-1}$. Based on the VMA representation, a structural innovation $\delta\in\mathbb{R}^n$ at time $t$ produces the response $C_hA_0^{-1}\delta$ at horizon $h$. The normalization of $\delta$ used to define the impulse responses of interest is specified in Definition \ref{def:IdSet}.\footnote{The HVAR in Eq.s (\ref{eq:HVAR})-(\ref{eq:CovarHVAR}) is characterized by the same parameters for the conditional mean over the two regimes, therefore breaks are confined to second moments parameters. As a consequence, the absence of unit roots is a characteristic of the model in the whole sample, and not within each regime. Furthermore, the VMA representation is unique and not regime-specific. It follows that, if shocks have the same magnitude, impulse response functions in the two regimes are equivalent.} 

\subsection{Preliminary results on the identification of HSVARs}
\label{sec:identification}

As is well known in the SVAR literature, without any restriction it is impossible to uniquely pin down the structural parameters based on the reduced form of the model. If, instead, we suppose the parameters of the conditional mean in the HSVAR in Eq. (\ref{eq:HSVAR}) to remain stable across the two regimes, then \cite{Rigobon03}, for a bivariate case, and \cite{LanneLutkepohl08JMCB}, for the general case, proved there is some gain in terms of identification. In this section we introduce some general theoretical results for the HSVAR in Eq. (\ref{eq:HSVAR}). All the results will be formally presented in Appendix \ref{app:Preliminary}. 
 
Consider an $n$-variable HSVAR model as in Eq.s (\ref{eq:HSVAR})-(\ref{eq:CovarHSVAR}). Following the parametrization and notations we have been using so far, let $C\equiv A_0^{-1}\Lambda_1^{1/2}$, and $\Lambda\equiv \Lambda_1^{-1}\Lambda_2$, where this latter contains the shifts in the variance of the structural shocks between the first and second regime. Remembering the unit diagonal elements of $A_0$, the diagonal matrix $\Lambda_1$ can be obtained by observing that $\Lambda_1^{-1/2}=diag\;\big(C^{-1}\big)$. Given the reduced-form covariance matrix at regime 1 and 2, denoted by $\Omega_1$ and $\Omega_2$, respectively, $C$ and $\Lambda$ solve 
\begin{equation}
    \Omega_1 = CC' \hspace{1cm}\text{and}\hspace{1cm} \Omega_2 = C \Lambda C'. 
    \label{eqSys}
\end{equation}

It is important to remark that the identification of the HSVAR is related to the set of solutions, in terms of $(C,\Lambda)$, of the system of equations in (\ref{eqSys}). As largely discussed in the literature (\citeauthor{Rigobon03}, \citeyear{Rigobon03} and  \citeauthor{LanneLutkepohl08JMCB}, \citeyear{LanneLutkepohl08JMCB}, among others) point identification is obtained only when the elements in $\Lambda$ are all distinct. In this case, fixing a specific ordering of the structural shocks, as well as the standard sign normalization, leads the solution $(C,\Lambda)$ to be unique. 
More precisely, the definition of point identification we rely on is exact identification defined in  \citetalias{RWZ10RES}. In the current setting, it can be stated as follows:

\begin{defin}[Exact identification]
\label{def:exactId}
Consider an HSVAR as in Eq.s (\ref{eq:HSVAR})-(\ref{eq:CovarHSVAR}). The HSVAR is exactly identified if and only if, for almost every reduced-form parameter point $(B,\Omega_1,\Omega_2) \in \Phi$,
there exists a unique corresponding structural parameter point $(A_0,A_+,\Lambda_1,\Lambda_2)$.\footnote{For each maintained model $\mathcal{M}$, Definition \ref{def:exactId} is applied on the reduced-form domain $\Phi_{\mathcal M}\subseteq\Phi$ generated by its admissible structural parameters, including any imposed eigenvalue-multiplicity pattern. \textit{Almost every} is understood relative to a reference measure $\mu_{\mathcal M}$ equivalent to Lebesgue measure in regular local coordinates on $\Phi_{\mathcal M}$.}
\end{defin}

Moreover, another important result concerns the representation of the system in (\ref{eqSys}) as an eigen-decomposition problem. To see this, let $\Omega_{1,tr}$ be a lower triangular Cholesky decomposition of $\Omega_1$. Following Proposition A.1 of \citet{Uhlig05JME}, the set of non-singular matrices solving the system of equations $\Omega_1 = CC'$ can be expressed as $C=\Omega_{1,tr} Q$, $Q \in \mathcal{O}(n)$, where $\mathcal{O}(n)$ is the set of $n \times n$ orthogonal matrices. Plugging this representation of $C$ into Eq. (\ref{eqSys}) leads to
\begin{equation}
C = \Omega_{1,tr} Q \hspace{1cm}\text{and}\hspace{1cm} \Omega_{1,tr}^{-1} \Omega_2 \Omega_{1,tr}^{-1\prime} = Q \Lambda Q^{\prime}.\label{eigen decomposition}
\end{equation}
Symmetry of $\Omega_{1,tr}^{-1} \Omega_2 \Omega_{1,tr}^{-1\prime}$ and orthogonality of $Q$ imply that solving Eq. (\ref{eigen decomposition}) is precisely the eigen-decomposition problem. Identification of $(C,\Lambda)$ can be therefore cast as uniqueness of the eigen-decomposition of $\Omega_{1,tr}^{-1} \Omega_2 \Omega_{1,tr}^{-1\prime}$ into the diagonal matrix of eigenvalues and the corresponding eigenvectors collected in $Q$. Accordingly, an HSVAR can be identified up to permutations and sign changes if the eigenvalues of $\Omega_{1,tr}^{-1} \Omega_2 \Omega_{1,tr}^{-1\prime}$ are all distinct.\footnote{A geometric interpretation for this result, for the simple bivariate case, is discussed in Appendix \ref{sec:SetIdentBivariate}.} 

The identification up to permutations and sign changes highlights that exact identification of HSVAR requires the order and sign restrictions of the structural equations or shocks. If the available assumptions fail to pin down the order and sign of the shocks, we have indeterminacy of the eigen-decomposition: the ordering and sign of the columns of $C$ and ranking of the elements in $\Lambda$ cannot be uniquely determined even if the eigenvalues have no multiplicity. 
This case corresponds to local-identification of HSVAR. See \cite{BK20} for an example of locally- but not globally-identified HSVAR.

 To focus on the identification issue due to eigenvalue multiplicity, we assume that the assumptions that can order the structural equations according to the ranking of eigenvalues in $\Lambda$ are available, i.e., we can rank the magnitudes of the volatility breaks among the structural shocks. 
 We hence impose  
\begin{equation}
\label{eq:NormOrd}
    \lambda_1\geq\,\lambda_2\,\geq\ldots\geq\,\lambda_n,
\end{equation}
and use this ordering in the theoretical identification analysis. In Sections \ref{sec:TestInference} and \ref{sec:empirics}, no ordering of the shifts is imposed in the prior or in the posterior sampler; when applying the identification results and computing the log-SDDRs, the shifts are ordered as in Eq. \eqref{eq:NormOrd} in each posterior draw. This ranking fixes the order of the eigenspaces associated with distinct eigenvalues, but does not determine the columns of $C$ or, equivalently, the ordering of the equations given by $(A_0,\,A_+)$, within a repeated-eigenvalue block.
Moreover, imposing the diagonal elements of $A_0$ to be equal to one helps fixing the sign of the columns of $C$ as well as $Q$, resolving the second form of indeterminacy due to the signs of the structural equations. 
We denote by $N\subset\mathbb{P}^S$ the set of structural parameters satisfying these ranking and sign restrictions, and $N_Q(\phi) \subset \mathcal{O}(n)$ is the set of orthogonal matrices satisfying the sign restrictions given $\phi \in \Phi$.

Given the order and sign restrictions $N$, the identification issue of HSVARs reduces to the absence of eigenvalue multiplicity in $\Lambda=(\lambda_1,\,\lambda_2,\,\ldots,\,\lambda_n)$. If two (or more) eigenvalues are equal, say $\lambda_j=\lambda_{j+1}$, then the corresponding columns of $Q$ -- i.e. the eigenvectors associated to $\lambda_j$ and $\lambda_{j+1}$ -- denoted by $q_j$ and $q_{j+1}$, are not unique even under (\ref{eq:NormOrd}) and the sign normalizations. They represent a basis for the two-dimensional vector space in $\Re^{n}$, but any other couple of orthogonal unitary vectors belonging to such a space are valid column vectors to enter in $Q$. Matrix $Q$ is hence no longer a singleton in $\mathcal{O}(n)$, while we can construct the set of admissible orthogonal matrices solving the eigen-decomposition problem in Eq. (\ref{eigen decomposition}). 

We introduce the following notations to represent multiplicity of eigenvalues and the corresponding eigenspace.

\begin{defin}[Eigenspace of multiple eigenvalues]
\label{def:eigenspace}
	Given the reduced-form parameter $\phi \in \Phi$, let the eigen-decomposition problem in Eq. (\ref{eigen decomposition}) subject to the ranking restriction (\ref{eq:NormOrd}) give eigenvalues $\lambda_1 \geq \lambda_2 \geq \dots \geq \lambda_n$. We denote the set of distinct eigenvalues among $(\lambda_1, \dots, \lambda_n)$ by $ \lambda^1 > \lambda^2 > \dots > \lambda^{k}$, 
	where the generic \textit{i}-th distinct eigenvalue has algebraic multiplicity equal to $m_i$, i.e. $g(\lambda^i)=m_i$, $i=1,\ldots,k$, with 
	$\sum_{i=1}^{k}m_i=n$.
	Let $Q(\lambda^i)$ be the unit-vector section of the eigenspace associated with the \textit{i}-th distinct eigenvalue $\lambda^i$, i.e.
	\begin{equation}
	\label{eq:eigenspace}
		Q(\lambda^i)=\bigg(\text{span}\left(q_1^i,\,\ldots,q_{m_i}^i\right)\:\cap\:\mathcal{S}^{n-1} \bigg)\subset \Re^{n}
	\end{equation}	
	where $Q^i=\big(q_1^i,\,\ldots,q_{m_i}^i\big)$ are linearly independent (not unique) eigenvectors associated to the eigenvalue of $\lambda^i$ with $m_i$ multiplicity and
	$\mathcal{S}^{n-1}$ is the unit sphere in $\Re^{n}$.
	Moreover, given the result in Lemma \ref{lemma:multiplicity} in Appendix \ref{app:Proofs}, $\text{dim}\,\text{span}\big(Q(\lambda^i)\big)=m_i$.
	\qed
\end{defin}

Being explicit about the dependence of the eigenvalues and eigenspace on the reduced-form parameters $\phi$ in (\ref{eqSys}), we denote the set of admissible matrices $Q$ solving the eigen-decomposition (\ref{eqSys}) with the eigenvalues sorted by (\ref{eq:NormOrd}) and the sign normalization $N_Q(\phi)$ by
\begin{equation}
	\label{eq:AdmSpace}
	\Qa = \Big\{Q=\big(Q^1,\ldots,Q^k\big) \in \mathcal{O}(n)\cap N_Q(\phi): q_j^i\in Q(\lambda^i) ,i=1,\ldots,k,\hspace{0.2cm}j=1,\ldots,m_i\Big\}.
\end{equation}
We formally state the set identification of HSVAR under eigenvalue multiplicity in the next theorem.

\begin{theorem}[Eigenvalue multiplicity and set identification]
		\label{theo:setIdHSVAR}
        Consider an HSVAR subject to the ranking and sign normalization restrictions $N$. If $m_i>1$ for some $i \in \{1, \dots, k\}$, $Q \in \mathcal{O}(n)$ are set-identified and their identified set given $\phi$ is the set of eigenvectors in the eigen-decomposition problem (\ref{eigen decomposition}) subject to $N_{Q}(\phi)$. 
        Accordingly, $\Qa$ spans the identified set of $C = A_0^{-1} \Lambda_1^{1/2} = \Omega_{1,tr} Q$ by
        \begin{equation} \label{eq:IS_for_C}
        \Big\{C = \Omega_{1,tr} Q: Q \in \Qa \Big\}.
        \end{equation}
\end{theorem}

\medskip

Based on Theorem \ref{theo:setIdHSVAR}, the identified set of the subvector parameters of interest in $C$ can be obtained by projecting the identified set of $C$ in Eq. (\ref{eq:IS_for_C}) onto the parameters of interest. If $m_i > 1$, the subset column vectors of $Q$ from $\big(\sum_{j < i} m_j +1\big)$-th to $\big(\sum_{j < i} m_j +m_i\big)$-th columns cannot be uniquely determined by the ranking and sign normalization restrictions only, since $Q(\lambda^i)$ is a unit sphere of positive dimension $(m_i-1)$. This then translates to set-identification of the corresponding column vectors in $C$. In contrast, if $m_i = 1$, $Q(\lambda^i)$ consists of only two unit-length vectors of the opposite sign, and the sign normalization restriction selects one, delivering point-identification of the corresponding column vector of $Q$ and $C$.   

\vspace{0.4cm}
\noindent
\textit{Equality and sign restrictions}
\vspace{0.3cm}

\noindent

Theorem \ref{theo:setIdHSVAR} gives the starting point of our analysis on partially identified HSVARs. In what follows, we extend Theorem \ref{theo:setIdHSVAR} by additionally imposing the sign and zero restrictions.

\cite{GK18} and \citetalias{ARW18}, in the context of SVARs, exploit that imposing constraints on the structural parameters, or on suitable functions of them, corresponds to restricting the columns of the orthogonal matrix $Q\in \Qa$. This generalizes to the current context of HSVARs. 

As in \citetalias{CMT23WP} and \cite{LW20}, we allow for zero restrictions in the $A_0$ matrix. Specifically, according to \cite{GK18}, a zero restriction on the elements of the \textit{i}-th equation, i.e., the \textit{i}-th row of $A_0$, corresponds to an equality restriction on the \textit{i}-th column of $Q$. The set of zero restrictions can be compacted as follows:
\begin{equation}
		\label{eq:GenFormRest}
        \Fr \equiv 
        \Big(
                \big(F_{1}(\phi)q_1\big)^\prime\:,\:
                \big(F_{2}(\phi)q_2\big)^\prime\:,\:
                \ldots \:,\:
                \big(F_{n}(\phi)q_n\big)^\prime
        \Big)^\prime 
				= \bf{0}
\end{equation}
where $F_{i}(\phi)$, of dimension $f_{i}\times n$, depends on the reduced-form parameters $\phi=(B, \Omega_1, \Omega_2)$ only, while $q_{i}$ is the \textit{i}-th column of $Q$. The total number of restrictions characterizing the HSVAR is given by $f=f_1+\cdots+f_n$.

Focusing our attention on the \textit{i}-th eigenspace as in the eigen-decomposition in Definition \ref{def:eigenspace}, let the set of restrictions on the vectors $\big(q_1^i,\,\ldots\,,q_{m_i}^i\big)\in Q(\lambda_i)$ be contained in the $f^i\times n$ matrix $\textbf{F}^i(\phi,Q)$, with $f^i$ denoting the total number of restrictions on the vectors $\big(q_1^i,\,\ldots,\,,q_{m_i}^i\big)$. Moreover, let the $f_j^i$ restrictions on the \textit{j}-th vector $q_j^i$ be defined as $F_j^i(\phi) q_j^i=0$.
This allows us to introduce the following definition: 

\begin{defin}[Non redundant restrictions]
\label{def:RedRes}
	Given reduced-form parameter $\phi=(B,\Omega_1,\Omega_2)$, let the HSVAR be characterized by the eigen-decomposition in Definition \ref{def:eigenspace}. Let the $m_i$ vectors $\big(q_1^i,\,\ldots,\,,q_{m_i}^i\big)\in Q(\lambda_i)$, $i=1,\ldots,k$, be characterized by zero restrictions of the form $F_j^i(\phi)q_j^i=0$. Moreover, for each $i$, let $V_i(\phi)\in\mathbb{R}^{n\times m_i}$ have orthonormal columns spanning the eigenspace associated with $\lambda^i$, and write $x_j^i=V_i(\phi)'q_j^i$ for $j=1,\ldots,m_i$. The zero restrictions $F_j^i(\phi)q_j^i=0$ are non-redundant at $(\phi,Q)$ if, for every $i$ and $j$,
    \begin{equation*}
    \rk
    \begin{pmatrix}
    F_j^i(\phi)V_i(\phi) \\
    (x_1^i)' \\
    \vdots \\
    (x_{j-1}^i)'
    \end{pmatrix}
    = f_j^i+j-1,
    \end{equation*}
    where the last $j-1$ rows are omitted when $j=1$.\qed
\end{defin}

Finally, like the standard SVARs, we can combine zero and sign restrictions to help the identification of the structural parameters. Sign restrictions can be seen as constraints on the columns of the $Q$ matrix. As in \cite{GK18}, we notate all the sign restrictions as 
\begin{equation}
\label{eq:SignRest}
		\Sr\geq \textbf{0}. 
\end{equation}

Based on the parametrization of the model and the set of all possible restrictions considered above, we define the set of admissible structural parameters:

\begin{defin}[Admissible parameters]
        \label{def:admissible}
        A structural parameter point $\left(A_0,A_+,\Lambda_1,\Lambda_2\right)$ is said admissible if it satisfies the normalization restrictions $N$, the equality restrictions in Eq. (\ref{eq:GenFormRest}) and the sign restrictions in Eq. (\ref{eq:SignRest}). Given reduced-form parameter $\phi\in \Phi$, the set of admissible parameters can be defined as
				\begin{eqnarray}
						\label{eq:Arestr}
						\Ar & \equiv & \Big\{\big(A_0,A_+,\Lambda_1,\Lambda_2\big)=\big(D^{-1} Q^\prime\Choli,\,D^{-1}Q^\prime\Choli B, \,D^{-2},\,D^{-1}\Lambda D^{-1}\big)\in N\,\Big|\nonumber\\
                        & & \hspace{2.5cm}\,Q\in\Qa,\,\:\textbf{F}(\phi,Q)=\textbf{0},\,\:\textbf{S}(\phi,Q)\geq\textbf{0}\Big\}\nonumber,
				\end{eqnarray}
        where $D$ is a diagonal matrix defined by $D\equiv diag(Q^\prime \Choli)$ and $\Lambda$ is a diagonal matrix of eigenvalues defined in Eq. (\ref{eigen decomposition}). \qed
\end{defin}
 
We also introduce the set of admissible matrices $Q$.

\begin{defin}[Admissible $Q$ matrices]
				\label{def:AdmQ}
				An orthogonal matrix $Q$ is said admissible if, given reduced-form parameters $\phi \in \Phi$, it satisfies the normalization restrictions $N$, the equality restrictions in Eq. (\ref{eq:GenFormRest}) and the sign restrictions in Eq. (\ref{eq:SignRest}). The set of all admissible Q matrices is defined as $\Qr\equiv\Big\{Q\in\Qa\:\Big|\: \big(A_0,A_+,\Lambda_1,\Lambda_2\big)\in \Ar\Big\}.\hspace{2cm}\qed$
\end{defin}

Finally, given that the attention could not be limited to the structural parameters but on transformations of them, like impulse response functions, it is also important to define the so called \textit{identified set}.

\begin{defin}[Identified set]
                \label{def:IdSet}
				Given the set of admissible $Q$ matrices $\Qr$ according to Definition \ref{def:AdmQ}, an identified set is defined as $\IS\equiv\Big\{\eta\big(\phi,Q\big)\:\Big|\:Q\in\Qr\Big\}$,
				with $\eta(\phi,Q)$ being the transformation of the structural parameters of interest such as the impulse response of $g$-th variable to $j$-th shock in horizon $h$, $\eta(\phi,Q) = IR_{gj}^h=e_g^\prime C_h(B)\Chol Q e_j \equiv c_{gh}^\prime(\phi)\,q_j$,\:\:
				where $IR_{gj}^h$ is the (\textit{g,j})-th element of $IR^h$ and $c_{gh}^\prime(\phi)$ is the \textit{g}-th row of $C_h(B)\Chol$. Thus, the impulse responses defined here correspond to one-standard-deviation structural shocks in the first regime, since $C=\Omega_{1,\mathrm{tr}}Q=A_0^{-1}\Lambda_1^{1/2}$.\qed
\end{defin}

\subsection{Identification in HSVARs with proportional volatility shifts}
\label{sec:PointIdHSVAR}

As discussed in the previous sections, point identification in HSVAR can be achieved only if the eigen-decomposition problem in Eq. (\ref{eigen decomposition}) is characterized by $n$ distinct eigenvalues, corresponding to non proportional shifts in the variances of the structural shocks across the two regimes. If this is not the case, point identification can still be reached by combining heteroskedasticity with zero restrictions, formalized by the next theorem. Throughout this subsection, we fix the multiplicity pattern and the ordering of shocks within each eigenspace, and consider specifications
with $f_j^i\leq m_i-j$ for every $i$ and $j$. Identification statements are understood to hold for almost every reduced-form parameter on the support of the corresponding multiplicity model. We assume that the sign normalization resolves column-sign ambiguity and that additional sign restrictions do not isolate admissible solutions at their boundaries.

\begin{theorem}[Eigenvalue multiplicity and exact identification]
		\label{theo:PointIdHSVAR}
		Consider an HSVAR characterized by the eigenvalues and eigenspaces as in Definition \ref{def:eigenspace} and by the admissible parameters
		as in Definition \ref{def:admissible}. The structural parameter $\left(A_0,A_+,\Lambda_1,\Lambda_2\right)\in \Ar$ is exactly identified if and only if, for each eigenvalue $\lambda^i$, $i=1,\ldots,k$, and the corresponding eigenspace $Q(\lambda^i)=\left(q_1^i,\,\ldots,q_{m_i}^i\right)$, the unit-length vector $q_j^i$
		is subject to $f_j^i=m_i-j$ non-redundant zero restrictions, for $j=1,\ldots,m_i$.
\end{theorem}

\medskip

\begin{proof} 
	See Appendix \ref{app:Proofs}.
\end{proof} 

\medskip

\begin{remark}
	\label{rem:HSVARandSVAR}
	The previous theorem generalizes two important results in the literature of SVAR models. First, when all the eigenvalues are distinct, 
    then $k=n$ and the algebraic multiplicity $m_i=1$, $i=1,\ldots,n$. As a consequence, $Q(\lambda_i)=\left(q^i\right)$, and no 
	zero restriction is needed for exact identification, being $f^i=1-1=0$. 
	Second, if one has no credible belief about structural breaks on the second moments of the observable variables, or the shift 
	produces perfectly proportional covariance matrices in the two regimes, i.e. $\Omega_2=\lambda \Omega_1$ for any positive scalar $\lambda$ 
	(specifically, $\lambda=1$ in the case of no breaks), then $\Omega=\Omega_{1,tr}^{-1}\Omega_2 \Omega_{1,tr}^{-1\prime}=
	\Omega_{1,tr}^{-1}\sqrt{\lambda}\Omega_{1,tr} \sqrt{\lambda}\Omega_{1,tr}^{\prime}\Omega_{1,tr}^{-1\prime}=\lambda I_n$; there is just one 
	eigenvalue with associated eigenspace $\mathcal{S}^{n-1}$, i.e, the unit sphere in $\Re^{n}$. The condition in Theorem \ref{theo:PointIdHSVAR} reduces to the identification condition for global identification in \citetalias{RWZ10RES} (Theorem 7).
\end{remark}	

\smallskip

Theorem \ref{theo:PointIdHSVAR} provides a necessary and sufficient condition for exact identification of HSVAR in the presence of multiplicity of eigenvalues but augmented by zero restrictions. Setting Theorem \ref{theo:PointIdHSVAR} as a benchmark, we consider scenarios that leave the $(q_1^i,\ldots,q_{m_i}^i)$ columns of $Q$ set-identified. Specifically,
\begin{equation}
\label{eq:RestSet}
	f_j^i\leq m_i-j, \hspace{1cm}j=1,\ldots,m_i
\end{equation}
with strict inequality for at least one $j=\{1,\ldots,m_i\}$. Theorem \ref{theo:setIdHSVAR} and the departures from the conditions of Theorem \ref{theo:PointIdHSVAR} described in Eq. (\ref{eq:RestSet}) provide the formal basis for our methodology.
 
As in many empirical applications, suppose we are interested in one single shock, i.e. $q_{j^{\ast}}$, $j^{\ast}$-th column vector of $Q$. With the order and sign normalizations imposed, if $q_{j^{\ast}}$ is the eigenvector whose corresponding eigenvalue has no multiplicity, we can achieve its point identification. 
If this is not the case and belongs to the eigenspace generated by $\lambda_i$ (with multiplicity $m_i \geq 2$), the set of admissible $Q$ matrices of Definition \ref{def:AdmQ} projected onto $j^{\ast}$-th column of $Q$ gives its identified set. 
Similar to the identified set constructed in \cite{GK18} for the standard SVARs, the identified set for $q_{j^{\ast}}$ in the current HSVARs is given by the surface of an $m_i$-dimensional hypersphere in $\mathbb{R}^n$ intersected with the half-spaces or hyperplanes representing the sign normalization and the sign and zero restrictions, retaining only vectors that admit an orthogonal completion satisfying all restrictions on the remaining columns.  
Thus constructed identified set for $q_j^{\ast}$ then spans the identified set of the impulse response functions to $j^{\ast}$-th shock. 

These steps of obtaining the identified set for impulse responses are similar to the set-identified SVARs analyzed in \cite{GK18}. Appendix \ref{sec:SetIdentBivariate} provides geometric illustration of the construction of identified set and Appendix \ref{app:Preliminary} studies analytical properties of the identified set in HSVAR including its convexity. 
Section \ref{sec:algo} below presents computation for the impulse response identified set given reduced-form parameter $\phi$.

\section{Inference in Set-identified HSVARs: a Robust Bayes Approach}
\label{sec:TestInference}

In this section we discuss how to conduct inference on set-identified HSVARs, where the set identification comes from that not all the shifts in the variances of the shocks are different. We first introduce Bayesian estimation for the structural parameters of the HSVAR. In our analysis, the Bayesian estimation for the structural parameters is useful for two purposes. First, the Bayesian posterior can be used to quantify uncertainty for eigenvalue multiplicity based on Savage-Dickey Density Ratios (SDDR). If SDDR supports eigenvalue multiplicity, we make use of the posterior for the structural parameters to induce the posterior for the reduced-form parameters conditional on eigenvalue multiplicity. Inputting such posterior marginalized into the reduced-form parameters conditional on eigenvalue multiplicity and replacing the conditional prior for the structural parameters given the reduced-form parameters with ambiguous belief, we propose a Robust Bayes approach to do inference on the identified set of interest given eigenvalue multiplicity.

\subsection{Estimating the structural parameters of the HSVAR}
\label{sec:Estimation}

Bayesian estimation of the structural parameters of the HSVAR in Eq. (\ref{eq:HSVAR}) can be performed as in \citetalias{CMT23WP} and \cite{LW20}. As mentioned in Section \ref{sec:HSVAR}, the sign normalization is done by fixing the diagonal elements of $A_0$ to be equal to one. This allows the diagonal covariance matrices of the shocks, collected in $\Lambda_1$ and $\Lambda_2$, to be free to vary.\footnote{Alternatively, one could leave the diagonal elements of $A_0$ free to vary, and fix to the identity matrix one of the covariance matrices $\Lambda_1$ or $\Lambda_2$.} Moreover, to facilitate both the estimation procedure and the interpretation in terms of the eigen-decomposition problem in Eq. (\ref{eigen decomposition}), the elements in $\Lambda_1$ are directly estimated, while those in $\Lambda_2$ are obtained indirectly through the shifts in the volatility with respect to the first regime, i.e. $\Lambda = \Lambda_2\,\Lambda_1^{-1}$. The elements in the $\Lambda$ matrix represent the eigenvalues in the eigen-decomposition problem in Eq. (\ref{eigen decomposition}). The details of the estimation procedure are provided in Appendix \ref{app:BayesSVAR}; in this section, instead, we overview a prior construction and Gibbs sampler tailored to our methodology.

Let $A_{0,i}$ and $A_{+,i}$ denote the parameters of the \textit{i}-th equation, $\lambda_{1,i}$ the variance of the \textit{i}-th shock in the first regime, and $\lambda_{i}$ be the shift in the variance between the second and the first regime. Following  \citetalias{CMT23WP} and \cite{LW20}, we express the restricted structural parameters as $A_{0,i}'=W_i\alpha_i+w_i$, where $W_i$ is a known $n\times(n-1-f_i)$ matrix of full column rank, $w_i$ is a known $n\times1$ vector, and $\alpha_i\in\mathbb{R}^{n-1-f_i}$ collects the free coefficients. Here, $f_i$ denotes the number of independent off-diagonal zero restrictions in the $i$-th equation, in addition to its unit-diagonal normalization. 

This parametrization allows to impose independent informative prior distributions on $\alpha_i$, $p(\alpha_i)$ as in the spirit of \cite{BH15}, or noninformative ones centered at zero and with quite disperse variance, as in \cite{LW20}. 
The prior distributions for the autoregressive parameters in $A_{+,i}$, are set to follow a normal distribution. 
The prior distributions for both $\lambda_{1,i}$ and $\lambda_i$ are set to follow independent inverse gamma $IG2(\underline{a}_i,\underline{b}_i)$ distributions, where $IG2(a_i,b_i)\equiv IG(a_i/2,b_i/2)$, that, when combined with the normality of the error terms as in Eq.~(1), allow the conditional posterior distributions to be conjugate. We set the inverse gamma prior distribution of the variance shifts $\lambda_i$ to $IG2(0.1,2.1)$, which has the mode at 1 and is diffuse, letting the prior be rather noninformative about the presence of heteroskedasticity.         

To obtain an analytically tractable posterior with computationally feasible posterior sampling algorithms, we specify certain parametrization of the HSVAR and a parametric prior distribution as described in Appendix \ref{app:BayesSVAR}. 
If some of the eigenvalues obtained by the eigen-decomposition in Eq. (\ref{eigen decomposition}) present multiplicity, one complication is that conditional on multiplicity, the  covariance matrices of the reduced-form errors are constrained through Eq. (\ref{eigen decomposition}). This implies that the posterior distribution for the reduced-form parameter conditional on the eigenvalue multiplicity is constrained in a nontrivial manner, invalidating the stylized Bayesian inference for reduced-form VARs with a conjugate prior such as a normal-Wishart prior.   

The robust Bayesian inference procedure we propose below requires a posterior distribution for the reduced-form parameters conditional on the eigenvalue multiplicity as an input. To obtain a posterior for the reduced form parameters constrained by the eigenvalue multiplicity, we leverage a posterior distribution for the structural parameters conditional on the eigenvalue multiplicity to induce a posterior for constrained reduced form parameters. For this purpose, we build on the Bayesian estimation approach of \citetalias{CMT23WP} and \cite{LW20} sketched above and modify it to obtain a posterior distribution for the structural parameters conditional on equal variance shifts, as shown in Appendix \ref{app:BayesSVAR}. 

In the presence of uncertainty of eigenvalue multiplicity, a fully Bayesian approach would proceed to combine the posterior distributions conditional on eigenvalue multiplicity or no-multiplicity averaged with respect to the posterior probability over each case. 
Viewing no-multiplicity and multiplicity of eigenvalues as different models, the fully Bayesian approach can be seen as a Bayesian model averaging. Its robust Bayes analogue allowing for set-identification and posterior ambiguity in some models is proposed by \cite{GKV22} and it can be applied in the current context.

\subsection{Assessing the identifying conditions}
\label{sec:SDDR}

To quantify the degree of empirical support for eigenvalue multiplicity, we use Savage-Dickey Density Ratio (SDDR). 
SDDR has a simple interpretation in terms of the ratio of the posterior and prior probabilities for the model (hypothesis) of interest and small computational cost since it can be computed with estimation of the only unrestricted model.
Given the order of the structural equations\footnote{As already introduced in Section \ref{sec:identification}, no ordering restriction is imposed on these shifts when specifying their priors; when computing the SDDRs, the shifts are ordered as in Eq.\eqref{eq:NormOrd} in each posterior draw.}, consider assessing model $M_{ij}$ defined by the equality of two eigenvalues $ \lambda_i=\lambda_j$, or equivalently, $\frac{\lambda_i}{\lambda_j}=1$. Model $M_{ij}$ corresponds to the case that \textit{i}-th and \textit{j}-th structural shocks cannot be distinguished and the \textit{i}-th and \textit{j}-th rows of $A_0$ are not point identified. 
The SDDR for $M_{ij}$ can be written as
\begin{equation}
    \label{eq:SDDR}
    SDDR_{M_{ij}} = \frac{p\big(\frac{\lambda_i}{\lambda_j}=1|Y\big)}{p\big(\frac{\lambda_i}{\lambda_j}=1\big)}
\end{equation}
where $Y$ denotes the data. 
The SDDR for $M_{ij}$ is the ratio of the posterior distribution to the prior distribution evaluated at the constrained eigenvalues of $\lambda_i/\lambda_j = 1$. 
With the model specification in Eq. (\ref{eq:HSVAR}) and a parametric class of prior distributions introduced by \cite{LW20}, SDDR can be easily evaluated. 
They propose an inverse gamma ratio $IG2R(a_1,b_1,a_2,b_2)$ distribution constructed as the ratio of two inverse gamma distributions, $IG2(a_1,b_1)$ and $IG2(a_2,b_2)$. 
As shown in the previous section, both the conditional prior and posterior distributions of the generic $\lambda_i$ follow an inverse gamma distribution and, as a consequence, the ratio of the eigenvalues follows an inverse gamma ratio distribution. Computing the SDDR, thus, reduces to evaluate the posterior and prior inverse gamma ratio distributions at $1$. 
The logarithm of SDDR being positive suggests a shift of the belief towards model $M_{ij}$ in light of the data.
On the contrary, a negative value of the logarithm of SDDR for $M_{ij}$ is evidence against model $M_{ij}$. 

SDDR for a hypothesis of pairwise equivalence of two eigenvalues can be generalized to a subset of eigenvalues to be jointly equivalent. The following proposition defines SDDR for equivalence and suggests a way to compute it, given a prior specification based on the ratio of the inverse gamma distributions.

\begin{prop}[Joint SDDR for equality of a subset of relative variance shifts]
		\label{prop:SDDRgroup}

Consider an $n$-dimensional structural VAR with $M=2$ volatility regimes, and let $\lambda_i>0$, $i=1,\ldots,n$, denote the relative variances in the second regime with respect to the first regime. Let $\mathcal{I} = \{i_1,\ldots,i_K\} \subseteq \{1,\ldots,n\},\: 2\leq K\leq n$, and suppose that, under the unrestricted model, $\lambda_{i_k}$ are independent $IG2(\underline{a}_{i_k},\underline{b}_{i_k}),\: k=1,\ldots,K$.

Consider the model $M_{i_1 i_2 \cdots i_K}$ satisfying joint equality restrictions $M_{i_1 i_2 \cdots i_K}: \frac{\lambda_{i_1}}{\lambda_{i_2}}= \ldots = \frac{\lambda_{i_1}}{\lambda_{i_K}} = 1$ and define the SDDR for $M_{i_1 i_2 \cdots i_K}$ as $SDDR_{M_{i_1 i_2 \cdots i_K}}  = \frac{p\left(\frac{\lambda_{i_1}}{\lambda_{i_2}} = \ldots = \frac{\lambda_{i_1}}{\lambda_{i_K}} = 1 \Big|Y\right)}
    {p\left(\frac{\lambda_{i_1}}{\lambda_{i_2}} = \ldots = \frac{\lambda_{i_1}}{\lambda_{i_K}} = 1 \right)}$.    
The joint prior probability associated with $H_{0}$ is
\begin{equation}
\label{eq:joint_prior_ordinate}
    p \left(\frac{\lambda_{i_1}}{\lambda_{i_2}}
    = \ldots =
    \frac{\lambda_{i_1}}{\lambda_{i_K}} = 1\right) =
    \frac{\Gamma(A_{\mathcal{I}})}{\displaystyle\prod_{k=1}^{K}\Gamma(\underline{a}_{i_k}/2)}
    \frac{\displaystyle\prod_{k=1}^{K}\underline{b}_{i_k}^{\underline{a}_{i_k}/2}}{\displaystyle\left(\sum_{k=1}^{K}\underline{b}_{i_k}\right)^{A_{\mathcal{I}}}},
\end{equation}
where $A_{\mathcal{I}} = \frac{1}{2}\sum_{k=1}^{K}\underline{a}_{i_k}$. 

Let $\psi$ collect all the remaining model parameters. Conditional on $Y$ and $\psi$, the relative variances are independently distributed as $\lambda_{i_k}\mid Y,\psi
\sim IG2\left(a_{i_k}^{\star},b_{i_k}(\psi)\right),\,k=1,\ldots,K,$
where
\begin{align}
    a_{i_k}^{\star} & = \underline{a}_{i_k}+(T-T_B),\label{eq:posterior_shape_parameter}\\
    b_{i_k}(\psi)   & = \underline{b}_{i_k} + \lambda_{1,i_k}^{-1}\sum_{t=T_B+1}^{T}\Big(
        A_{0,i_k}y_t-\mu_{i_k} -A_{1,i_k}y_{t-1} - \cdots - A_{l,i_k}y_{t-l}\Big)^2.
    \label{eq:posterior_scale_parameter}
\end{align}
Let $A_{\mathcal I}^{\star} = \frac{1}{2} \sum_{k=1}^{K}a_{i_k}^{\star}$. The conditional posterior density of the joint restriction is
\begin{equation}
\label{eq:joint_conditional_ordinate}
    d_{\mathcal I}(\psi) \equiv p \left(\frac{\lambda_{i_1}}{\lambda_{i_2}} = \ldots = \frac{\lambda_{i_1}}{\lambda_{i_K}} = 1 \Big| Y,\,\psi\right) =
    \frac{\Gamma\left(A_{\mathcal I}^{\star}\right)}{\displaystyle\prod_{k=1}^{K}\Gamma\left(a_{i_k}^{\star}/2\right)}
    \frac{\displaystyle\prod_{k=1}^{K}\left[b_{i_k}(\psi)\right]^{a_{i_k}^{\star}/2}}
        {\displaystyle\left(\sum_{k=1}^{K}b_{i_k}(\psi)\right)^{A_{\mathcal I}^{\star}}}.
\end{equation}
The exact marginal posterior density is therefore
\begin{equation}
\label{eq:joint_marginal_ordinate}
    \begin{split}
        p \left(\frac{\lambda_{i_1}}{\lambda_{i_2}} = \ldots = \frac{\lambda_{i_1}}{\lambda_{i_K}} = 1 \Big| Y\right) 
            & = \int d_{\mathcal I}(\psi) \,p(\psi\mid Y)\,d\psi\\
            & = E_{\psi\mid Y}\left[d_{\mathcal I}(\psi)\right].
    \end{split}
\end{equation}

\end{prop}

\medskip

\begin{proof} 
	See Appendix \ref{app:Proofs}.
\end{proof} 

Proposition \ref{prop:SDDRgroup} generalizes the SDDR for the hypothesis of pairwise identity of eigenvalues proposed by \cite{LW20} to the hypothesis of joint equalities among any subset of $K\leq n$ eigenvalues. Their result corresponds to the case of $K=2$, whereas for $K>2$ our formulation accounts for the dependence among the $K-1$ ratios induced by their common numerator.

The denominator of SDDR given by Eq. (\ref{eq:joint_prior_ordinate}) can be computed by evaluating it directly. 
The numerator of SDDR given by Eq. (\ref{eq:joint_marginal_ordinate}) is approximated by evaluating the integration with MCMC draws. Let 
$\{\psi^{(s)}\}_{s=1}^{S}$ be posterior draws and $d_{\mathcal I}^{(s)} = d_{\mathcal I}\left(\psi^{(s)}\right)$. An estimator for \eqref{eq:joint_marginal_ordinate} can be constructed by
\begin{equation}
\label{eq:joint_posterior_ordinate}
    \overline{d}_{\mathcal I,S} = \frac{1}{S} \sum_{s=1}^{S}d_{\mathcal I}^{(s)}.
\end{equation}

Let $M_0$ be the model with no multiplicity of eigenvalues, and $\mathcal{H}$ be the set collecting the models with eigenvalue multiplicity:\footnote{The candidate family $\mathcal H$ imposes one repeated-eigenvalue block at a time. Together with $M_0$, it covers all multiplicity patterns for $n\le3$; for $n\ge4$, patterns with multiple distinct repeated-eigenvalue blocks have zero prior probability, and model selection and averaging are conditional on this restricted model space.} 
\begin{equation}
\mathcal{H} = \big\{M_{i_1 \cdots i_K}: \{i_1, \dots, i_K \} \subset \{1, \dots, n \}, 2 \leq K \leq n \big\}. 
\end{equation}
The joint equality restrictions defining $M_{i_1\cdots i_K}$ have Lebesgue measure zero under $M_0$. If the prior under $M_{i_1\cdots i_K}$ is the prior under $M_0$ conditional on these restrictions, the SDDR in Proposition \ref{prop:SDDRgroup} equals the Bayes factor for $M_{i_1\cdots i_K}$ against $M_0$.
Since $M_0$ is the reference model, its Bayes factor against itself is identically one. Thus, assigning prior model probabilities $\big\{p(M):M \in \mathcal{H} \cup \{M_0\}\big\}$, we obtain the posterior model probabilities according to
\begin{align}
\label{eq:ModProb}
p(M \mid Y) &= \frac{\operatorname{SDDR}_M\,p(M)} {\sum_{M' \in \mathcal{H} }\operatorname{SDDR}_{M'}\,p(M') + p(M_0)}, \mspace{10mu} M \in \mathcal{H}  \\
p(M_0 \mid Y) &= \frac{p(M_0)} {\sum_{M' \in \mathcal{H} }\operatorname{SDDR}_{M'}\,p(M') + p(M_0)}. 
\end{align}
These posterior model probabilities are used to summarize the posterior belief for eigenvalue multiplicity and can be used for Bayesian model averaging to incorporate uncertainty over set- and point-identified models into posterior inference as in \cite{GKV22}.


\subsection{Inference on the identified set}
\label{sec:Inference}

When SDDR or the posterior model probability strongly supports a certain combination of eigenvalue multiplicity in light of data, we propose to perform a robust Bayesian approach to infer the identified set for impulse responses given the supported multiplicity of the eigenvalues. Combined with the posterior model probabilities, such robust Bayesian posterior can be used as an input for Bayesian model averaging.

Let $\pift$ be a probability measure on the space $\Phi$ of reduced-form parameters such that it satisfies the sign, normalization and equality restrictions, and the eigenvalue multiplicity as in Eq. (\ref{eq:AdmSpace}) with probability one. 
The Bayes rule updates the prior for $\phi$ to yield its posterior. 
As we discussed in Section \ref{sec:Estimation}, our proposal obtains the posterior for $\phi$ conditional on eigenvalue multiplicity by marginalizing the posterior for the structural parameters with respect to $\phi$.  
Since the structural parameters are a function of $\big(\phi,Q\big)\in \Phi\times \On$, in general one defines a joint prior for the two sets of parameters $\big(\phi,Q\big)$ as $\pi_{\phi,Q}=\pi_{Q|\phi}\:\pi_\phi$, where $\pi_{Q|\phi}$ is supported on $\Qr\in \On$.

For the case of eigenvalue multiplicity, a conditional prior $\pi_{Q|\phi}$ needs to be specified in order to obtain a posterior distribution for the structural parameters and impulse responses, whereas the choice of such a prior can have non-diminishing impact on the posterior inference since it is never updated by the data. The robust Bayes inference of \cite{GK18} fixes a single prior $\pi_\phi$ for the reduced-form parameters, but introduces a set of priors for $\pi_{Q|\phi}$, leading to the class of posteriors for $(\phi,Q)$ and the impulse response of interest $r=r(\phi,Q)$.

However, unlike the set-identified SVAR analysis in \cite{GK18}, a complication in the set-identified HSVAR is that the prior and posterior for $\phi$ given the eigenvalue multiplicity is difficult to obtain since its support is constrained to a thin set in the space of $\Phi$ through Eq. (\ref{eigen decomposition}).
To overcome this challenge, we obtain posterior draws for $\phi$ conditional on eigenvalue multiplicity in the following steps. We first obtain draws from the posterior distribution of the structural parameters $\theta$, with the eigenvalue multiplicity imposed on $\Lambda$ following the procedure shown in Section \ref{sec:Estimation} and Appendix \ref{app:PostL}. We then transform them into the reduced form parameters of the two covariance matrices $\Omega_1$ and $\Omega_2$ through Eq. (\ref{eqSys}). 
With these draws from the posterior for $\phi$ given multiplicity, we apply the robust Bayes analysis of \cite{GK18}. Note that the draws from the posterior distribution of $\phi$ are sampled from a single posterior distribution for $\theta$, while the posterior for $\theta$ is merely a device to obtain the draws of $\phi$. 
For inference on $\theta$, we only rely on the posterior for $\phi$ and let $\pi_{Q|\phi}$ be ambiguous. 

Following \cite{GK18}, we summarize the results of this procedure by reporting the posterior mean bounds interval, that can be seen as an estimator for the identified set, and an associated robustified credible region measuring the uncertainty related to the former. In particular, if we define $\ell(\phi)=\mathrm{inf}\big\{r(\phi,Q):\,Q\in\Qr\big\}$ and $u(\phi)=\mathrm{sup}\big\{r(\phi,Q):\,Q\in\Qr\big\}$, the posterior mean bounds interval can be written as $\Big[\int_\Phi\,\ell(\phi)\ud \tilde{\pi}_{\phi|Y}\:;\:\int_\Phi\,u(\phi)\ud \tilde{\pi}_{\phi|Y}\Big]$. The robustified credible region, instead, consists in an interval $C_\alpha$ for which the posterior probability is greater than or equal to $\alpha$ uniformly, i.e. $\pi_{r|Y}(C_\alpha)\geq \alpha$ over all the posteriors for $r$ in the class of posteriors of $r$

\subsection{Computing posterior bounds}
\label{sec:algo}

In what follows, we present an algorithm to implement our proposal in the case of two volatility regimes with known break date. Assuming knowledge of the break date is rather standard in the literature, being the break dates associated to well documented changes in the policy conduct or to financial crises.\footnote{See, among many others, \cite{LanneLutkepohl08JMCB}, \cite{BG06RESTATS}, \cite{ABCF19}, \cite{Rigobon03}, \cite{Bacchiocchi17}, \cite{CMT23WP}. Moreover, \cite{Rigobon03} also shows consistency of the estimated parameters in the case of break date miss-specification. Alternatively, \citetalias{CMT23WP} and \cite{LW20}, determine endogenously the break dates through a Markov Switching process, while \cite{LN2017JEDC} model the shifts in the variances by a smooth transition mechanism. Our approach can be generalized in these directions to endogenize the definition of the break date.} 

\begin{algo}
\label{algo:PostBounds}
	Let $y_{-l+1},\ldots,y_0,y_1,\ldots,y_T$ be a sample of observations characterized by a break in the volatility occurring at time $T_B$, that is known or exogenously determined. Fix the normalization rule $N$, the equality restrictions $\Fr$, and the sign restrictions $\Sr$.
	\begin{itemize}[left=35pt, parsep=0pt]
	\item[(Step 1)] Use the techniques proposed in Sections \ref{sec:Estimation} and \ref{sec:SDDR} to estimate the HSVAR, with equality restrictions in $A_0$, and 
        assess for possible eigenvalue multiplicity by calculating the SDDR for all combinations of $\lambda$'s. Select the most restricted HSVAR with a log SDDR greater than $-5$. 
        Move to step 2.
	\item[(Step 2)] For each draw from the posterior of the constrained structural parameters \linebreak[4]
        $\psf$, obtain the corresponding parameters of the reduced form $\phi=(B,\Omega_1,\Omega_2)$. Using the constrained covariance matrices $\Omega_1$ and $\Omega_2$ solve the eigen-decomposition problem in Eq. (\ref{eigen decomposition}) and obtain the orthogonal matrix $Q$.  
    \item[(Step 3)] Extract from $Q$ the basis of the space $\Qli$, whose columns are associated with the possibly multiple eigenvalues, and define the matrix $\bar{Q}_{\lambda^i}$ containing the $n-m_i$ eigenvectors orthogonal to $\Qli$. If the zero restrictions meet the rank condition in Theorem \ref{theo:PointIdHSVAR}, then the point identified columns of $\Qli$, for the draw of $\phi$, can be easily determined through Algorithm 1 in \citetalias{RWZ10RES}; then, if the draw satisfies the sign restriction, move to Step 4. If, instead, the zero restrictions do not meet the rank condition in Theorem \ref{theo:PointIdHSVAR}, then the model is only set identified and, given the draw of $\phi$, search numerically for an admissible rotation by following the sub-routines below:
		\begin{itemize}[left=27pt, parsep=0pt]
		\item[(Step 3.1)] Let $z_1 \sim N(0,I_n)$ be a draw of an \textit{n}-variate standard normal random variable. 
        Let $\tilde q_1^i=M_1z_1$ be the $n\times 1$ residual vector from projecting $z_1$ onto the columns of the $n\times(f_1^i+n-m_i)$ matrix $\bigl[F_1^i(\phi)',\,\bar Q_{\lambda^i}\bigr]$. For $k = 2,\ldots m_i$, run the following procedure sequentially: draw $z_k\sim N(0,I_n)$ and compute $\tilde{q}_k^i = M_k z_k$, where $M_k z_k$ is the residual vector in the linear projection of $z_k$ onto the $n\times(f_k^i+n-m_i+k-1)$ matrix $\big(F_k^i(\phi)^\prime,\bar{Q}_{\lambda^i},\tilde{q}_1^i,\ldots,\tilde{q}_{k-1}^i\big)$. The vectors $\tilde{q}_1^i,\ldots,\tilde{q}_{m_i}^i$ are mutually orthogonal, orthogonal to $\bar{Q}_{\lambda^i}$, and satisfy the equality restrictions.
		\item[(Step 3.2)] Given $\tilde{q}_1^i,\ldots,\tilde{q}_{m_i}^i$ obtained in the previous step, define
			\begin{equation}
			\label{eq:algoSignSt}
				Q_{\lambda_i} = \bigg[\pm\frac{\tilde{q}_1^i}{\|\tilde{q}_1^i\|},\ldots,
				\pm\frac{\tilde{q}_{m_i}^i}{\|\tilde{q}_{m_i}^i\|}\bigg],
				\nonumber
			\end{equation}
			where $\|\cdot\|$ is the Euclidean metric in $\Re$, then arrange the sign of each column of $Q_{\lambda_i}$ according to the sign normalization as defined by $S\in \Dn$.
			Based on the obtained $Q_{\lambda_i}$ with appropriate sign normalization, form the $Q$ matrix by collecting the columns in $\bar{Q}_{\lambda^i}$ and $Q_{\lambda_i}$ according to the correct ordering determined by the permutation matrix $P\in \Pn$.
		\item[(Step 3.3)] Check whether $Q$ obtained in (Step 3.2) is such that\\
		    $(A_0,A_+)=\big(PS D^{-1}Q^\prime\Omega_{1,tr}^{-1},\,PSD^{-1}Q^\prime\Omega_{1,tr}^{-1}B\big)$, for appropriate $S\in \Dn$ and $P\in \Pn$, satisfies the sign restrictions $S(\phi,Q) \geq 0$. If so, retain this $Q$ and proceed to (Step 4). Otherwise, repeat (Step 3.1) and (Step 3.2) a maximum of $L$ times (e.g. L = 3000) or until $Q$ does satisfy $S(\phi,Q) \geq 0$. If none of the $L$ draws of $Q$ satisfies $S(\phi,Q) \geq 0$, approximate $\Qr$ as being empty and return to (Step 2) for the following draw of $\phi$.
		\end{itemize}
	\item[(Step 4)] Given $\phi$ and $Q=\big(\bar{Q}_{\lambda^i},Q_{\lambda_i}\big)$ obtained in (Step 3), 
        with the correct ordering determined by $P\in \Pn$ and correct sign normalization determined by $S\in \Dn$, compute the lower and upper bounds of $\IS$ by solving the following constrained nonlinear optimization problem:
		\begin{eqnarray}
		\label{eq:algoOpt}
			& & \ell(\phi) = \mathrm{arg}\: \underset{Q_{\lambda_i}}{\mathrm{min}}\: c_{gh}^\prime(\phi)\qja,\nonumber\\
			\mathrm{s.t.}	& & Q^\prime Q=I_n,\hspace{0.2cm} \textbf{F}(\phi,Q)=0\nonumber\\
			& & \big(PSD^{-1}Q^\prime\Omega_{1,tr}^{-1},\,PSD^{-1}Q^\prime\Omega_{1,tr}^{-1}B\big)\in N, \hspace{0.2cm}\mathrm{ and }\hspace{0.2cm} \textbf{S}(\phi,Q)\geq 0\nonumber
		\end{eqnarray}
		and $u(\phi) = \mathrm{arg\:} \underset{Q_{\lambda_i}}{\mathrm{max}}\: c_{gh}^\prime(\phi)\qja$ under the same set of constraints. 
	\item[(Step 5)] Repeat (Step 2) - (Step 4) $M$ times to obtain $\big[\ell(\phi_m),\: u(\phi_m)\big]$, $m=1,\ldots,M$. Approximate the 
		set of posterior means by the sample averages of $\Big(\ell(\phi_m),\:m=1,\ldots,M\Big)$ and $\Big(u(\phi_m),\:m=1,\ldots,M\Big)$.
    \item[(Step 6)] To obtain an approximation of the smallest robust credible 
        region with credibility $\alpha \in (0,1)$, define $d(\eta, \phi) = \mathrm{max}\,\{|\eta - \ell(\phi)| , |\eta - u(\phi)|\}$, and let $\tilde{z}_\alpha(\eta)$ be the sample $\alpha$-th quantile of $\big(d(\eta, \phi_m):\:m = 1,\ldots M\big)$. An approximated smallest robust credible region for $\eta$ is an interval centered at $\mathrm{arg}\: \mathrm{min}_\eta \,\tilde{z}_\alpha(\eta)$ with radius $\mathrm{min}_\eta\, \tilde{z}_\alpha(\eta)$.
	\end{itemize}
\end{algo}

The posterior bounds and robust credible regions computed in Steps 2-6 are conditional on the eigenvalue-multiplicity model selected in Step 1. They do not incorporate posterior uncertainty over alternative multiplicity models.
Some remarks about the algorithm are in order. The first one is about the selection of the best model, as done in Step 1. As already highlighted in Section \ref{sec:SDDR}, the SDDR represents a formal approach for comparisons among eigenvalues. 
A valid alternative, instead, is represented by the frequentist test by \citetalias{LMNS20}, that allows for testing the equivalence of groups of eigenvalues of cardinality larger than two. The authors propose a procedure consisting in a set of sequential tests of hypothesis that, given the ordering of the $\lambda$'s from largest to smallest, allows to test for clusters of equal variance shifts. Specifically, the procedure starts by testing the null of equal $\lambda$'s  for all the $n$ structural shocks. If rejected, it moves on to test the two clusters of $n-1$ $\lambda$'s. If these two null hypotheses are not supported by the data, they can test all null hypotheses involving $n-2$ consecutive $\lambda$'s, etc. The same sequential test strategy can be adopted in the Step 1 of our Algorithm \ref{algo:PostBounds}, substituting the SDDR calculations with the frequentist test statistics of \citetalias{LMNS20}.

The second remark is about the computation of the two covariance matrices $\Omega_1$ and $\Omega_2$. Once the best model is selected, the estimation of the structural parameters $\theta$, that imposes eigenvalue multiplicity, as discussed in Section \ref{sec:Estimation} and Appendix \ref{app:BayesSVAR}, allows to implement the robust Bayesian approach starting from constrained covariance matrices, in particular, $\Omega_2$. This would be rather complicated if one opted to start from the orthogonal-reduced-form approach, as largely diffuse in the literature (see \citetalias{RWZ10RES}, \citetalias{ARW18}, \cite{GK18}, just to mention few of them). In this respect, in fact, while it is rather simple to obtain draws with restrictions of different kinds on the orthogonal $Q$ matrix, we do not know any algorithms allowing to generate restricted covariance matrices of the reduced-form residuals, where the constraints come from equal eigenvalues as in the eigen-decomposition problem in Eq. (\ref{eigen decomposition}).

Third, Step 3 uses Gaussian projections and sign checks to search numerically for an admissible rotation. QR-based constructions, as used by \citetalias{ARW18}, could alternatively be adapted to generate candidate rotations subject to the zero and eigenspace restrictions. Finding an admissible rotation establishes non-emptiness, whereas failure to obtain one within $L$ attempts does not establish that the admissible set is empty.

Fourth, the constrained nonlinear optimization problem in Step 4 is less demanding than the one in Algorithm 1 by \cite{GK18}, as the argument is not the entire matrix $Q$ but just a subset of its columns. Even if the HVAR model is relatively large, the number of eigenvectors generating the subspace $\Qli$, in many applications, is relatively small and we do not expect concerns about the convergence properties of the numerical optimization step. On the contrary, we could replace Step 4 by a new algorithm, as in the spirit of Algorithm 2 in \cite{GK18}, where the constrained nonlinear optimization problem is substituted by iterating many times Step 3.1-Step 3.3 and approximate the interval $\big[\ell(\phi_m),\: u(\phi_m)\big]$ with the minimum and maximum values obtained in such iterations. If the number of iterations goes to infinity, such alternative bounds still provide a consistent estimator of the identified set.

Fifth, the algorithm works even in the case the zero restrictions allow to point identify the matrix $Q$. In this case, if the draw satisfies the sign restriction, the constrained nonlinear optimization problem simply returns $\ell(\phi)=u(\phi)$. The inference, then, becomes standard.

From a theoretical point of view, \cite{GK18} discuss the importance of convexity, continuity and differentiability of the identified set $\IS$ for the posterior means to have a valid frequentist interpretation. If this is the case, they show that the set of posterior means is a consistent estimator of the true identified set and the robust credible region is an asymptotically valid confidence set for the true identified set. 
Regarding convexity, Theorems \ref{theo:SetIdZero} and \ref{theo:SetIdZeroSign} in Appendix \ref{app:Preliminary} provide sufficient conditions under zero restrictions and under zero and sign restrictions, respectively. Theorem \ref{theo:SetIdZeroSign} requires sign restrictions to be imposed exclusively on the shock of interest and therefore does not establish convexity for our empirical application, where two shocks are sign restricted.

Continuity of the identified set and differentiability of its bounds require further verification in our HSVAR setting, in particular when equality restrictions on variance shifts constrain the reduced-form parameter space. Establishing the frequentist properties of our procedure would therefore require verifying these regularity conditions, together with the asymptotic behavior of the constrained reduced-form posterior. Since we do not establish them here, we interpret the reported robust credible regions in Bayesian terms and do not claim frequentist coverage.


\section{Simulation experiments}
\label{sec:MC}
We propose two different experiments: the former aims at evaluating the frequentist properties in recovering the true identified set in small and large samples under lack of point identification; the latter, instead, investigates about the consequences of potential miss-specification.

\subsection{Eigenvalue multiplicity and frequentist properties}
\label{sec:MCcoverage}

We use a data generating process (DGP) based on a trivariate HSVAR(1) with two equal variance shifts and a distinct one. In line with the empirical example discussed in the following section, and to better emphasize the different variance shifts, we prefer a parametrization where the $A_0$ matrix is fully unrestricted, the variance matrix in the first regime is the identity matrix, and the variance shifts are the eigenvalues in the eigen-decomposition problem in Eq. (\ref{eigen decomposition}). Details on the DGP are reported in Appendix \ref{app:MC}. Concerning the variance shifts, we fix the three eigenvalues as follows: $\Lambda^{(1)}=diag(4.622,\: 0.250,\: 0.250)$. The values in $\Lambda^{(1)}$, as well as those in $A_0$, are inspired by the empirical analysis performed in the next section, where the first structural shock is point identified by heteroskedasticity, while the other two are only set identified. Further experiments are reported in Appendix \ref{app:MC}. 

We repeat the experiments for $T = \{100,\:200,\:400,\:1000,\:10000\}$. The break date has been fixed according to the empirical example in the next section, $T_B = 0.435\: T$. The number of Monte Carlo replications is fixed to $10000$. The prior distributions are those described in Section \ref{sec:Estimation}, and in the implementation of Algorithm \ref{algo:PostBounds}, the sign normalization is obtained by imposing the diagonal elements of $A_0^{-1}$ to be positive. No other sign restrictions have been included. 

The results of this first experiment are reported in Table \ref{tab:CoverageApp} (Panel A) and in Figure \ref{fig:CoverageMed} in Appendix \ref{app:MC}. First of all, Step 1 in Algorithm \ref{algo:PostBounds} aims at selecting the favorite model through the SDDR. In Table \ref{tab:CoverageApp}, for each value of $T$ we report the percentage rate of correct model selection. 
It clearly emerges that, even for very small samples, with a specification that is similar to the one obtained in the following empirical application, the percentage of correct model selection is very high, approaching hundred percent for all values of $T$. 

Secondly, being impulse responses the object of interest in almost all empirical applications, for each variable of interest, in the table we also report the coverage rate of the robust credible set at $90\%$ credibility, after a one standard deviation (first regime) shock in one of the two set identified structural shocks, specifically $\varepsilon_{2t}$. We report the coverage for different horizons ($h=1$, $h=6$ and $h=12$). The true identified set has been obtained analytically. Coverage varies across variables and horizons. For the first variable, small-sample overcoverage declines as the sample size increases, with coverage between $87\%$ and $89\%$ at $T = 10000$. For the second variable, coverage at $h = 1$ rises from $73\%$ to $90\%$, whereas the longer-horizon intervals remain conservative. For the third variable, coverage improves with the sample size but remains below the nominal level, reaching $87\%$ at the reported horizons for $T = 10000$. The simulations, therefore, reveal substantial variation in finite-sample coverage, including residual undercoverage in some cases.

\subsection{Miss-specification of the true model}
\label{sec:MCmiss}

A second Monte Carlo experiment deals with potential mistakes in selecting the true model. We consider the same DGP as in the previous experiment, but with distinct eigenvalues. Moreover, to make the simulation exercise as realistic as possible, we fix the eigenvalues to values of the same order of magnitude as the empirical application presented in the next section: $\Lambda^{(5)}=diag\left(4.62,\: 0.48,\: 0.26\right)$. The sample size is fixed to $T = \{400,\,1000,\,10000\}$, and the break date is $T_B = 0.435 T$. Finally, no other sign restrictions than normalization (positive diagonal elements of $C=A_0^{-1}$), as before, are included. Further experiments are shown in Appendix \ref{app:MCmiss}.

In the Monte Carlo experiment we force our procedure to impose a form of miss-specification provided by eigenvalue multiplicity for the last two eigenvalues. Thus, we erroneously treat the models to be just set identified, and not as they are, point identified. 
The comparison therefore isolates the consequences of imposing the false equality $\lambda_2=\lambda_3$ without introducing additional sign miss-specification.
The results are reported in Table \ref{tab:Miss}, Panel B, in Appendix \ref{app:MCmiss}.

In the table, we report, for the three variables and for different horizons ($h=1$, $h=6$, $h=12$), the true IRFs, the coverage of both the correctly and miss-specified models, and the ratio between mean interval widths in the two specifications, i.e. $90\%$ robust credible interval over $90\%$ posterior credible interval ($R_W$), following a one standard deviation variation in the structural shock $\varepsilon_{2,t}$. The robust credible intervals are obtained by using Algorithm \ref{algo:PostBounds} when forcing the second and third eigenvalues to be equal. In this design, imposing equality preserves high coverage of the true point IRF, but at a cost in precision. For $T = 400$, coverage ranges from $95.1\%$ to $99.5\%$, while the ratio of mean interval widths ranges from $1.97$ to $4.40$. As the sample size increases, coverage approaches $100\%$, but the width ratio becomes substantially larger, reaching $28.68$ at $T = 10000$. Thus, the favorable coverage observed in this experiment is accompanied by increasingly wide intervals relative to the correctly specified model.


\section{Empirical Application}
\label{sec:empirics}
\noindent\textit{Data and model.} We apply our methodology to the SVAR model for the global crude oil market of \citet{kilian200AER} that includes three variables: the percent change in global crude oil production ($\Delta prod_t$), an index of global real economic activity ($rea_t$) and the logarithm of the real price of crude oil ($rpo_t$). Data are monthly and the sample period runs from February 1973 through December 2006, as in \citetalias{LMNS20}. While \cite{kilian200AER} identifies three structural shocks that drive the real price of crude oil using a simple recursive scheme, \citetalias{LMNS20}, \citet{LN2014JaE} and \citet{LukAdvEcnm} show that the same structural innovations can also be recovered by exploiting the existence of distinct volatility regimes. We set the lag order of the VAR equal to 12 and follow \citetalias{LMNS20} that distinguish between two volatility regimes with -- an exogenously determined -- change point in October 1987. 

\begin{table}[t]
\centering
\newcolumntype{L}[1]{>{\footnotesize\hsize=#1\hsize\raggedright\arraybackslash}X}
\newcolumntype{C}[1]{>{\footnotesize\hsize=#1\hsize\centering\arraybackslash}X}
\renewcommand{\arraystretch}{1.15}
\setlength{\tabcolsep}{3pt}

\caption{Estimated eigenvalues, log--SDDR, and posterior model probabilities}

\begin{tabularx}{\textwidth}{C{0.55}C{1.70}C{0.85}L{1.50}C{0.62}C{0.78}}\hline

\multicolumn{2}{c}{\footnotesize{\textit{Estimated eigenvalues}}} &
\multicolumn{1}{c}{\footnotesize{\shortstack{\textit{log--SDDR against}\\\textit{homoskedasticity}}}} &
\multicolumn{3}{c}{\footnotesize{\shortstack{\textit{log--SDDR against proportional changes}\\\textit{in variances of shocks}}}}\\
\hline

\footnotesize{Eigenvalue} &
\footnotesize{\mbox{Median [68\% HPD]}} &
\footnotesize{$\lambda_i=1$} &
\footnotesize{Model} &
\footnotesize{log--SDDR} &
\footnotesize{\mbox{$P(M\mid Y)$}}\\
\hline

$\hat{\lambda}_1$ &
\mbox{4.7297 [3.9374, 5.3646]} &
-69.7860 &
\mbox{$M_0$: all distinct} &
\textit{} &
0.3955 \\

$\hat{\lambda}_2$ &
\mbox{0.4944 [0.4154, 0.5662]} &
-5.4203 &
\mbox{$M_{12}$: $\lambda_1=\lambda_2$} &
-67.4908 &
0.000 \\

$\hat{\lambda}_3$ &
\mbox{0.2606 [0.2185, 0.2932]} &
-37.1916 &
\mbox{$M_{23}$: $\lambda_2=\lambda_3$} &
0.4241 &
0.6045 \\

&&&
\mbox{$M_{13}$: $\lambda_1=\lambda_3$} &
-122.6481 &
0.000 \\

&&&
\mbox{$M_{123}$: $\lambda_1=\lambda_2=\lambda_3$} &
-159.8445 &
0.000 \\
\hline

\end{tabularx}

\label{tab:TabLam}

\caption*{\scriptsize{\textit{Notes}: Eigenvalues are posterior medians with 68\% HPD intervals, sorted from largest to smallest.
Homoskedasticity log--SDDR for $H_0\!:\lambda_i=1$ reports per-shock values: large negative values indicate strong evidence of heteroskedasticity.
Joint \textit{log--SDDR} (log Bayes factors) for the proportionality restriction on variance changes across regimes: large negative values indicate strong evidence against proportionality
(Kass--Raftery rule of thumb: $|\text{log--SDDR}|>5$ $\approx$ strong evidence). Log--SDDRs are computed from 50,000 posterior draws after a burn-in of 5,000 draws.
Posterior model probabilities are computed over $M_0,M_{12},M_{23},M_{13},M_{123}$ assuming equal prior model probabilities; $M_0$ is the unrestricted reference model.}}
\end{table}

As shown in Section \ref{sec:SDDR}, the SDDR provides a convenient tool to assess potential eigenvalue multiplicity. Table \ref{tab:TabLam} shows that, conditional on the 1987 break date, there is strong evidence against homoskedastic shocks. The table further assesses whether volatility changes are proportional across shocks. The pairs $(\lambda_1,\lambda_2)$ and $(\lambda_1,\lambda_3)$ yield extremely large negative log SDDRs, providing very strong evidence against proportional variance changes. By contrast, the $(\lambda_2,\lambda_3)$ pair produces a log SDDR of 0.4241, which constitutes no strong evidence against proportionality. This finding implies that identification based solely on heteroskedasticity may not be sufficient, as only one structural shock can be statistically distinguished by volatility shifts. When the identified shock cannot be readily interpreted in economic terms, or when additional shocks are of interest, supplementary sign or exclusion restrictions are required. In such cases, Theorem \ref{theo:PointIdHSVAR} and its implementation in Algorithm \ref{algo:PostBounds} provide a systematic approach to achieve point and/or exact identification.\footnote{To assess the robustness of our results with respect to the choice of the break date, we conducted two complementary exercises. First, we recomputed the SDDR while continuously shifting the break date within a $\pm 12$-month window around October 1987. The results, reported in Appendix \ref{app:Empirics}, indicate that the conclusions regarding the proportionality of variances remain qualitatively unchanged. Second, we re-estimated the HSVAR model with sign restrictions using alternative break dates -- October 1986, October 1987, and October 1988 -- and obtained virtually identical results. We thank an anonymous reviewer for this comment.}

\bigskip

\noindent\textit{Identification.} To illustrate the methodology in this context, we estimate an HSVAR specification that imposes eigenvalue multiplicity and exploits both static and dynamic sign restrictions to identify the structural shocks driving the real price of crude oil. Further details and comparisons with alternative identification schemes are provided in Appendix \ref{app:Empirics}.

Consistent with the evidence in Table \ref{tab:TabLam}, we impose $\lambda_2 = \lambda_3$. With two non-distinct eigenvalues, only one structural shock can be point identified, while the remaining shocks are set-identified through a combination of static and dynamic sign restrictions. As summarized in Table \ref{tab:Sign}, we assume that a negative oil supply shock raises the real price of crude oil and reduces global real economic activity on impact, whereas a positive aggregate demand shock increases both oil prices and production contemporaneously. We also impose sign normalizations on the main diagonal of $C \equiv A^{-1}_0$.
Moreover, following \citet{kilian2014role}, we constrain the response of the real price of crude oil to a negative supply shock to remain positive for twelve months after impact. These additional restrictions rule out implausible dynamics in which the real price of oil falls below its initial level following a supply disruption.

\begin{table}[t]
\centering
\newcolumntype{L}[1]{>{\footnotesize\hsize=#1\hsize\raggedright\arraybackslash}X}
\newcolumntype{C}[1]{>{\footnotesize\hsize=#1\hsize\centering\arraybackslash}X}
\renewcommand{\arraystretch}{1.2}
\caption{Sign Restrictions on impact responses ($C\equiv A^{-1}_0$) in model $\mathcal{M}_2$}
\begin{tabularx}{\textwidth}{L{1.6}C{.8}C{.8}C{.8}}\hline
 & Oil supply disruption & Positive aggregate demand shock & Positive oil-specific demand shock \\\hline
$\Delta prod_t$ & (-) & + & *  \\
$rea_t$ &  - & (+) & *\\
$rpo_t$ & + & + & (+)\\\hline
\end{tabularx}
\label{tab:Sign}
\caption*{\scriptsize{\textit{Notes}: `` * '' denotes that the sign of the impact response is unrestricted. Signs along the main diagonal are in brackets to highlight that these are not actual sign restrictions, but sign normalizations placed on $C\equiv A^{-1}_0$.}}
\end{table}

\bigskip

\noindent\textit{Results.} All specifications are estimated with Algorithm \ref{algo:PostBounds} until we obtain 1000 realizations of the non-empty identified set.\footnote{In what follows, as for the Monte Carlo experiments, Step 4 of Algorithm \ref{algo:PostBounds} is implemented using a simulation-based approach similar to that in \citet{GK18}. Specifically, we repeat Steps 3.1-3.3 1000 times and approximate the interval $\big[\ell(\phi_m),: u(\phi_m)\big]$ by the minimum and maximum values obtained across these iterations. On a Windows workstation -- equipped with an Intel Core i9-12900 processor and 64 GB RAM -- generating 1,000 draws for a given shock using MATLAB without parallelization takes about 35 minutes. Results based on the exact solution of the constrained nonlinear optimization problem described in Step 4 are reported in Appendix \ref{app:Empirics}.} Compared with the representation of the HSVAR model in Eq.s \eqref{eq:HSVAR} and \eqref{eq:CovarHSVAR}, to ensure comparability with the literature, we normalize the structural shocks to have unit variance in the first regime, so that $\Lambda_1 = I_n$, while the elements of $\Lambda_2$ capture the volatility changes across regimes. We estimate the whole set of structural impulse response functions over an horizon of 24 months and focus on one-standard deviation shocks that are expected to raise the real price of crude oil. In the case of supply shocks, we plot the responses to a negative shock representing a disruption of crude oil supply. Notice that we report the implied response of world crude oil production obtained by cumulating $\Delta prod_t$. 

\begin{figure}[ht!]
\caption{Impulse Response Functions for the Global Market for Crude Oil: HSVAR with Eigenvalue Multiplicity Identified with Sign Restrictions}

  \includegraphics[width=.33\linewidth]{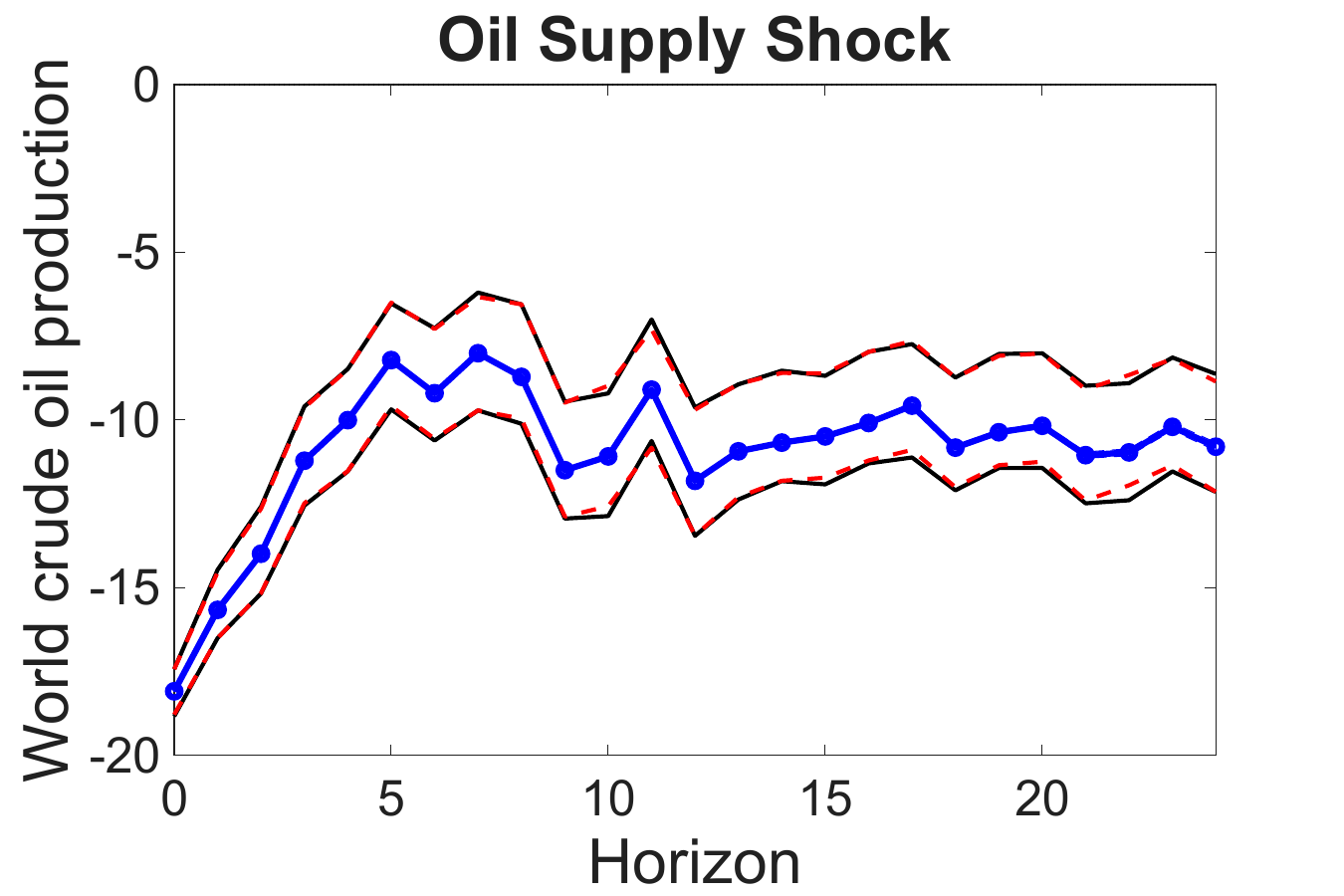}%
  \includegraphics[width=.33\linewidth]{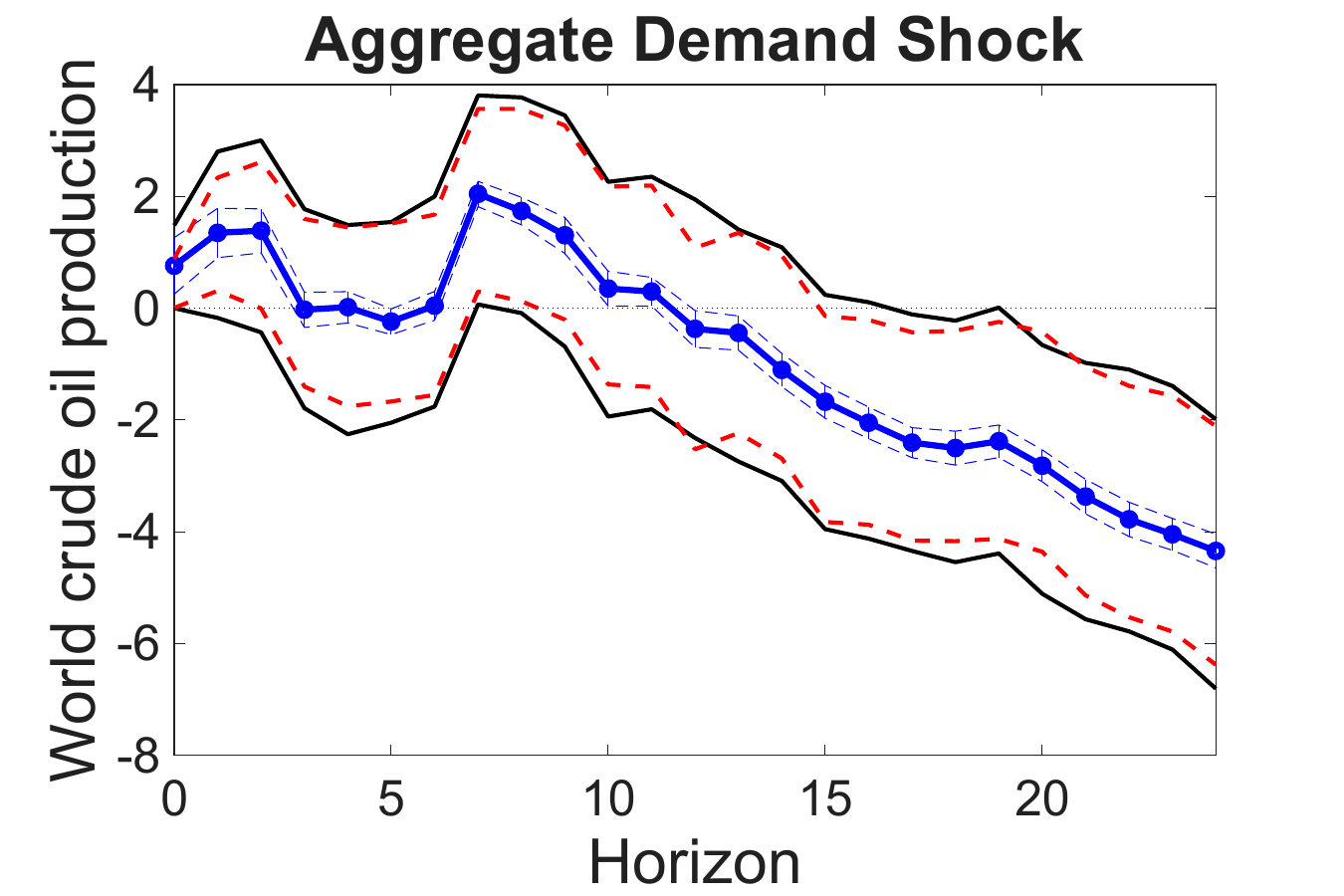}
  \includegraphics[width=.33\linewidth]{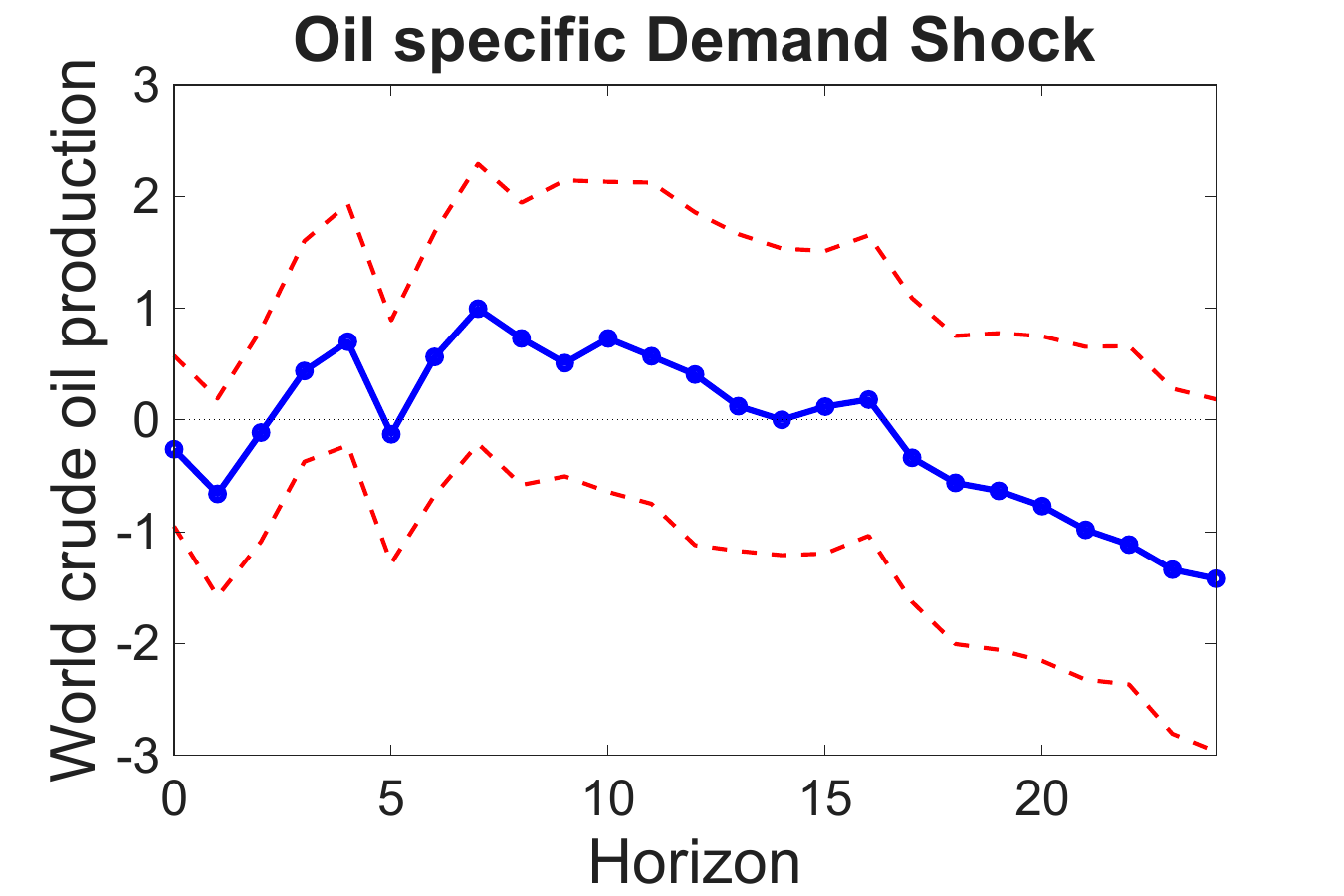}\\
  \includegraphics[width=.33\linewidth]{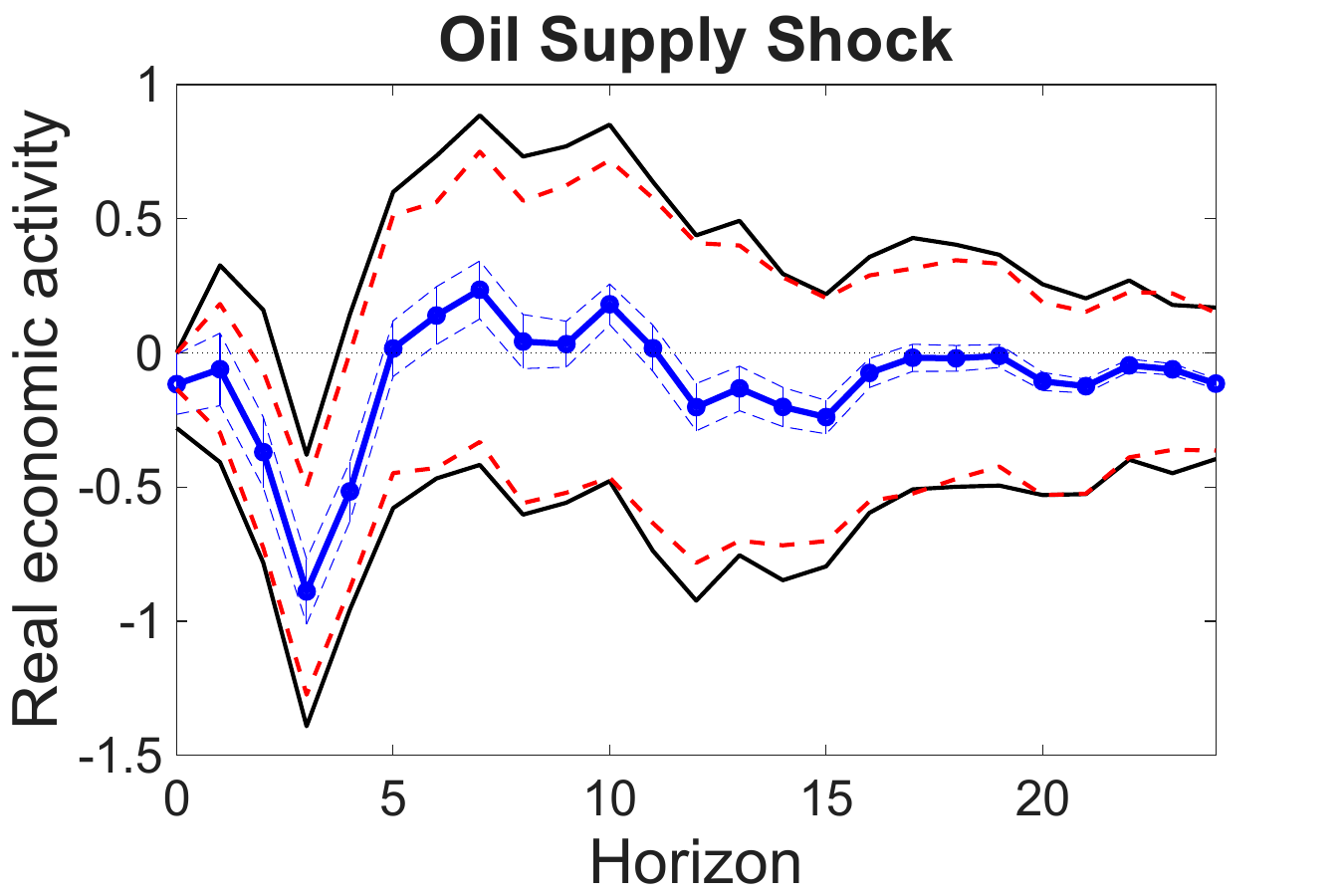}%
  \includegraphics[width=.33\linewidth]{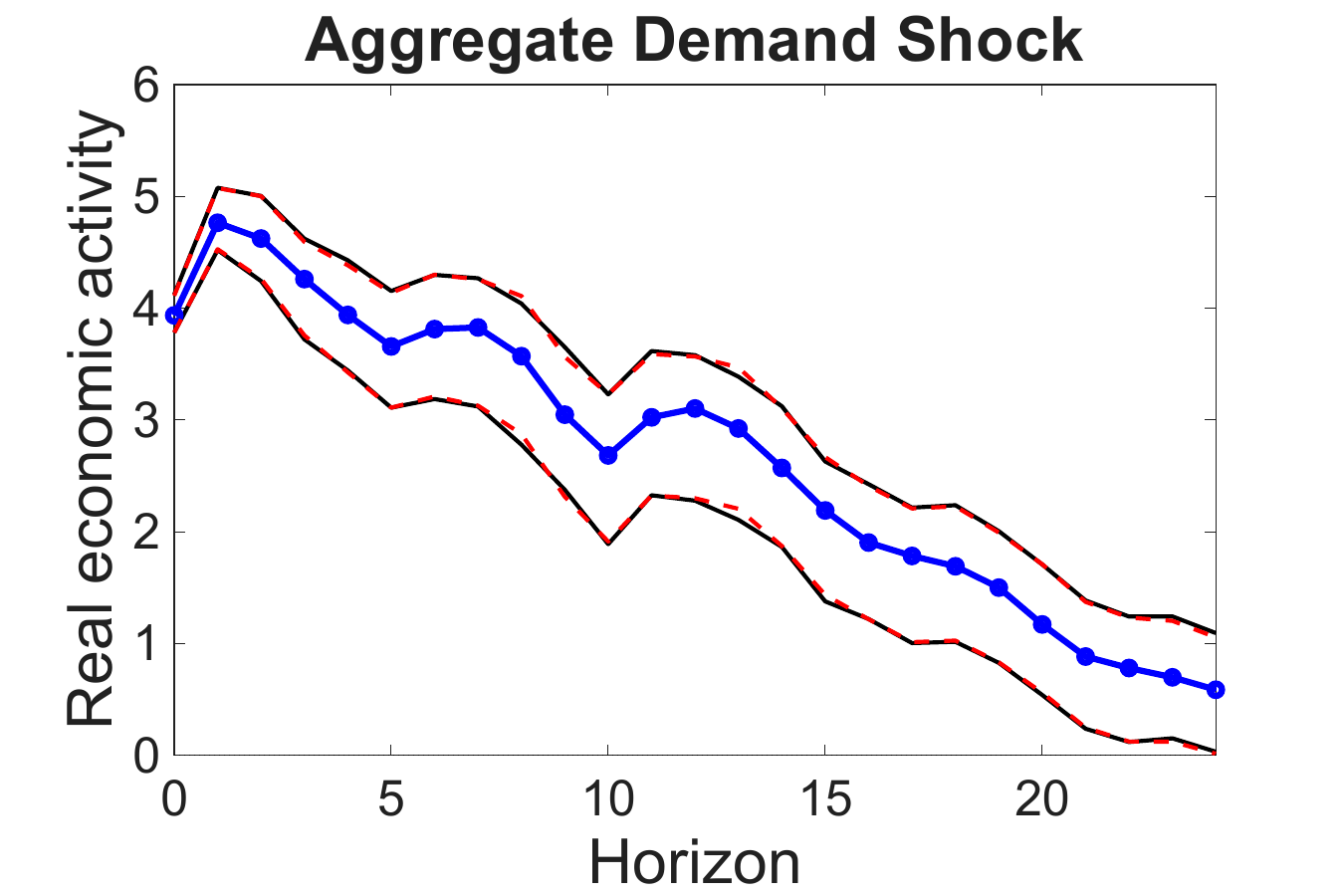}%
  \includegraphics[width=.33\linewidth]{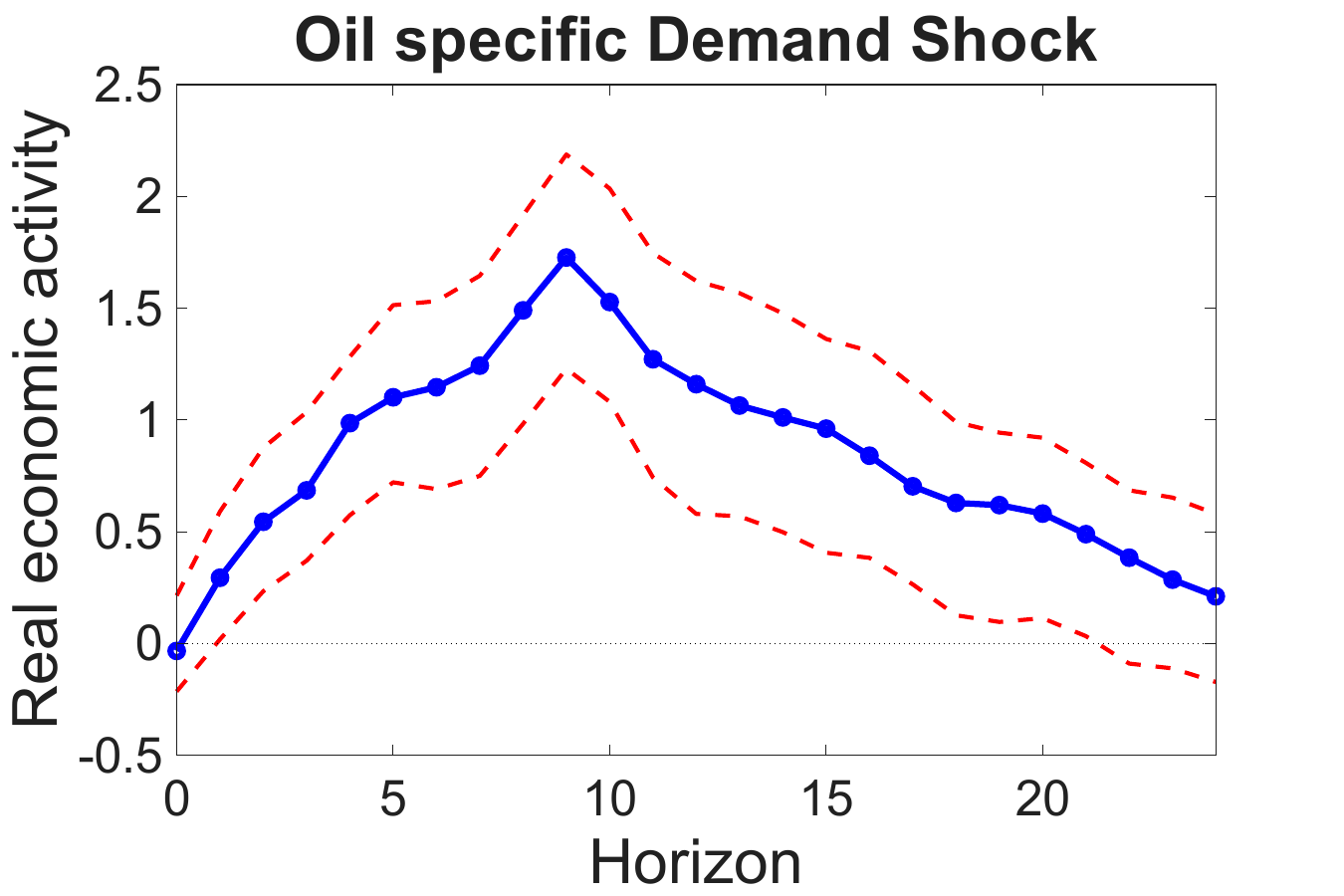}\\
  \includegraphics[width=.33\linewidth]{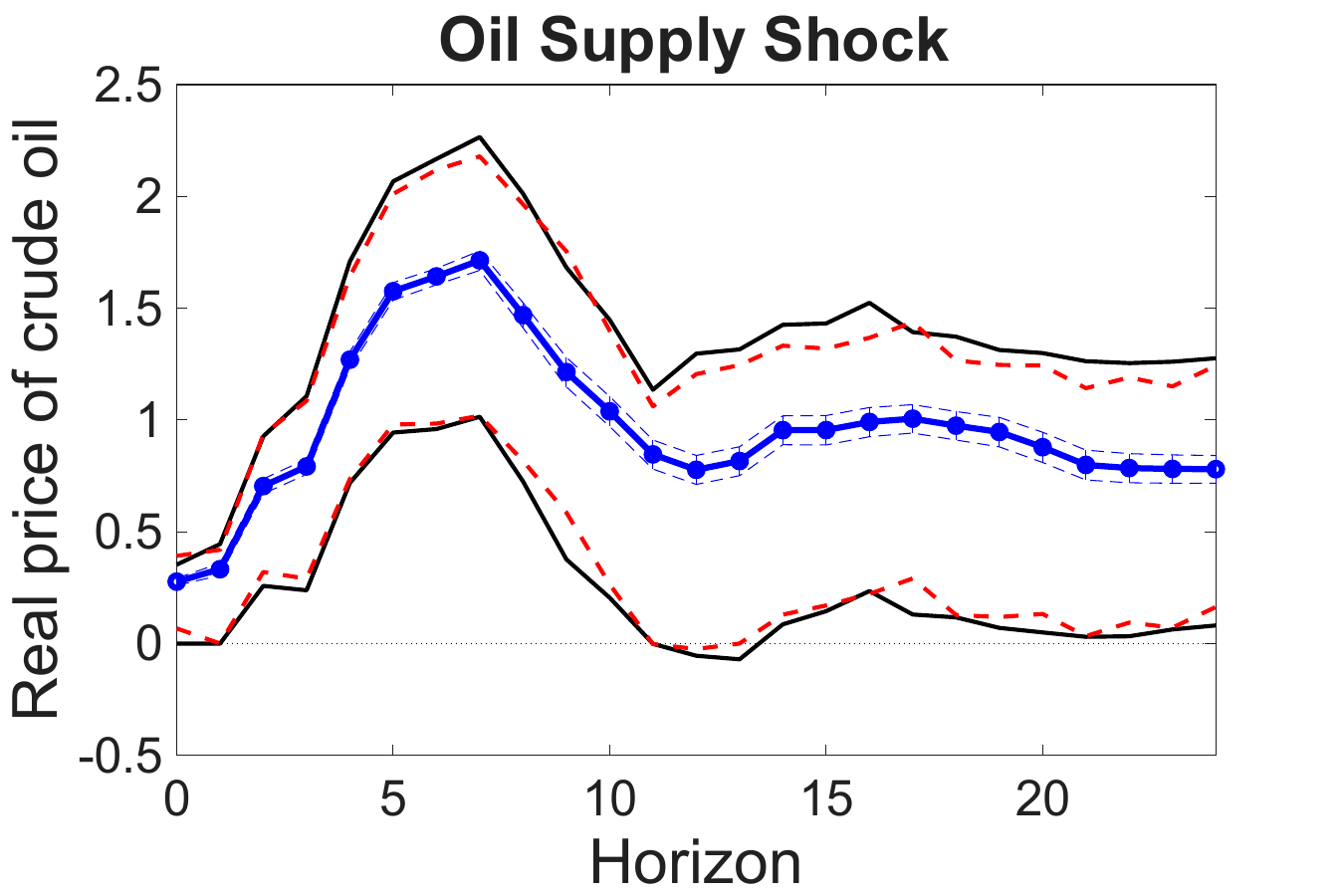}%
  \includegraphics[width=.33\linewidth]{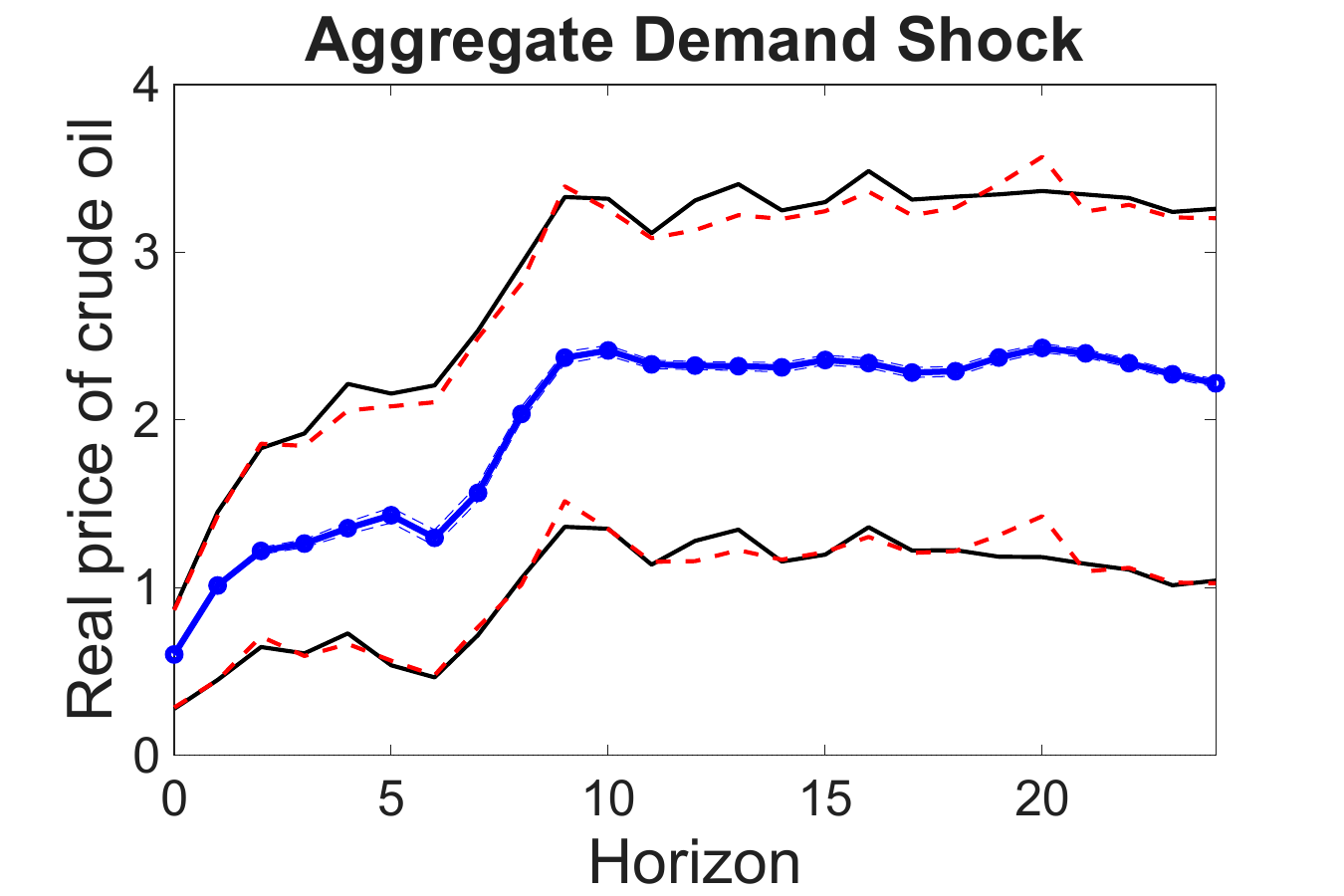}%
  \includegraphics[width=.33\linewidth]{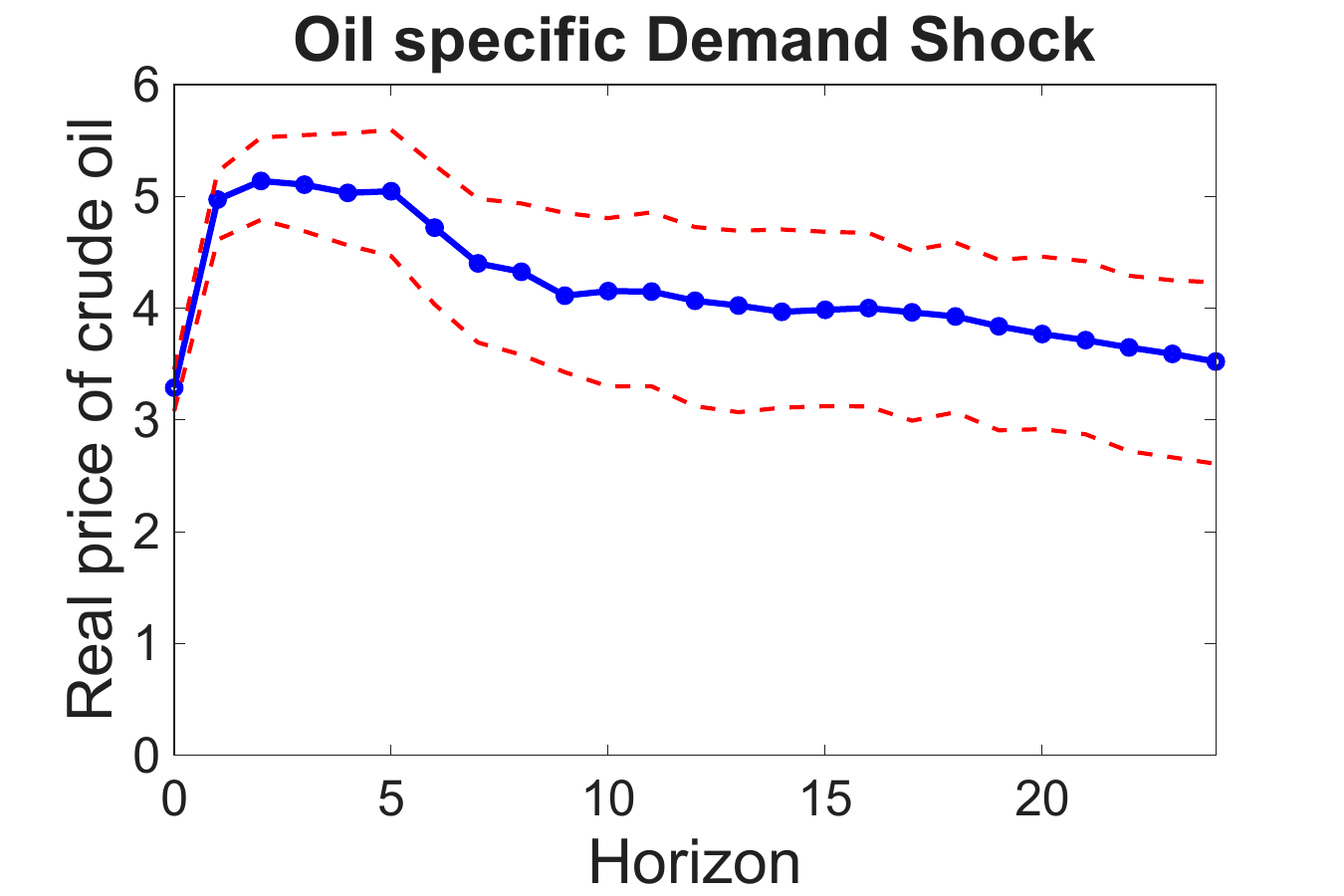}
\caption*{\scriptsize\textit{Notes}: the blue line with dots represents the standard Bayesian posterior mean response, the dashed red lines identify upper and lower bounds of the highest posterior density region with credibility 68\%. Plots in first and second columns of the figure also report the set of posterior means (blue vertical bars) and the bounds of the robust credible region with credibility 68\% (solid black curves). Identification via heteroskedasticity with multiple eigenvalues (i.e. only one shock is point identified), static and dynamic sign restrictions.}
\label{fig:IRF_M2_dynsignnew}
\end{figure}

Impulse responses from HSVAR model are displayed in Figure \ref{fig:IRF_M2_dynsignnew}. Since the only distinct eigenvalue is associated with an oil-specific demand shock, sign restrictions are imposed on the remaining columns of $Q$. Static and dynamic sign restrictions allow to set identify supply and demand shocks that are expected to drive the real price of crude oil. The impulse responses appearing in the first two columns of Figure \ref{fig:IRF_M2_dynsignnew} specifically refer to the columns of $Q$ associated with eigenvalue multiplicity. Here, we also present the set of posterior means (blue vertical bars) and the bounds of the robust credible region with credibility 68\% obtained with Algorithm \ref{algo:PostBounds}.

Figure \ref{fig:IRF_M2_dynsignnew} shows that by imposing sign restrictions on the two columns of $Q$ associated with the non-distinct eigenvalues we can recover oil supply and aggregate demand shocks whose effects on the endogenous variables of the VAR are consistent with expectations from economic theory and previous analyses \citep[see e.g.][]{kilian2012agnostic}. The shape of impulse responses in the first two columns of Figure \ref{fig:IRF_M2_dynsignnew} are in line with those implied by the recursively identified model of \citet{kilian200AER}. Also notice that the combination of  volatility changes and sign restrictions delivers reasonable impulse responses even in the presence of multiple eigenvalues and using less sign restrictions than what is usually done in the literature. Indeed, we leave one of the columns of $Q$ completely unrestricted and exploit changes in volatility to point identify the corresponding shock. Contrary to the standard HSVAR (Figure \ref{fig:IRF_M1} in Appendix \ref{app:Empirics}), identified through heteroskedasticity only, the oil supply shocks implied by the model in Figure \ref{fig:IRF_M2_dynsignnew} are associated with impulse responses that reasonably summarize the expected effects of such shocks on real economic activity and the real price of crude oil. Here, we observe a positive response of the price of crude oil with a peak after six months. 
For the conventional prior used here, the plotted HPD regions are visually close to the robust credible regions. This numerical agreement, however, does not imply that the HPD regions satisfy the uniform posterior credibility requirement over admissible conditional priors on $Q$ that defines the robust regions.


\section{Conclusion}
\label{sec:conclusion}
This paper deals with SVAR models with structural breaks on the second moments of the structural shocks, offering some new contributions. We first study the identification theory and propose a set of results for easily checking whether the model is globally identified. Second, we study the consequences on the impulse response functions of HSVARs that do not satisfy such identifying conditions. We deal with a SVAR model with heteroskedasticity, where non distinct changes in the variance shifts raise an identification issue. We solve the identification problem by imposing equality and sign restrictions and provide a methodology that helps giving a structural economic interpretation to the set or point identified shocks, by requiring fewer restrictions to be imposed. A way to do inference on the model both in case of point and set identification is also proposed, as well as an empirical application of our approach to the global crude oil market model. Some issues remain to be addressed by future research, such as extending the model to more than two volatility regimes and analyzing the consequences of having an unknown break date.

\vspace{0.6cm}

{\large \textbf{Data Availability Statement}}\\[.2cm]
The data used in this study were obtained from the replication package accompanying \citet{LMNS20}, ``Testing identification via heteroskedasticity in structural vector autoregressive models'', published in The Econometrics Journal. The data are publicly available as supplementary materials at \url{https://doi.org/10.1093/ectj/utaa008} and were used without modification.

\singlespacing
\bibliographystyle{econometrica}
\bibliography{acompat,HSVAR}
\newpage

\appendixtitleon
\appendixtitletocon
\begin{appendices}

\renewcommand{\thesection}{\Alph{section}}\setcounter{section}{0}
\setlength{\baselineskip}{12pt}
\renewcommand{\baselinestretch}{1.1}
\renewcommand\thefigure{\thesection.\arabic{figure}}    

\doublespacing
\section{Further results on the identification of HSVARs}
\label{app:Preliminary}

Some of the theoretical results we present in this section are not completely new in the literature, although they have been derived independently from other authors. However, we have decided to report them as prerequisites for a better understanding of the main results provided in the paper. Detailed references will be reported accordingly.  

Consider an $n$-variable SVAR model with two regimes in structural shock variances, but maintaining the homogeneity of the structural coefficients. We normalize the covariance matrix of the structural shocks to $n \times n$ identity matrix in the first regime. Following the parametrization and notations we have been using so far, we analyze identification of the $n \times n$ matrix $C$, defined as $C\equiv A_0^{-1}\Lambda_1^{1/2}$, and $\Lambda\equiv \Lambda_1^{-1}\Lambda_2$. Given the reduced-form covariance matrix at regime 1 and 2, denoted by $\Omega_1$ and $\Omega_2$, respectively, $C$ and $\Lambda$ solve
\begin{eqnarray}
\Omega_1 & = & CC', \label{eq.1st} \\
\Omega_2 & = & C \Lambda C'. \label{eq.2nd}
\end{eqnarray}

The next theorem characterizes the set of $(C,\Lambda)$ solving this equation system. To state it, we define $\Pn$ as the set of $n \times n$ permutation matrices, such that pre-multiplying $P \in \Pn$ to any matrix $M$ performs a row-permutation of $M$, and post-multiplying it performs a column-permutation. Moreover, in the case of a diagonal matrix $D$ of size $n$, $P^{\prime} D P$ performs a permutation of the diagonal elements of $D$. Let $\Dn$ be the set of $n \times n$ diagonal matrices whose diagonal entries are either $+1$ or $-1$. That is, if the $i$-th diagonal entry of $S \in \Dn$ is $-1$, pre-multiplying (post-multiplying) $S$ to any matrix $M$ flips the sign of the $i$-th row (resp. column) vector of $M$. 

\begin{theorem}[Sign normalization and column-permutation]
\label{theo:SignPerm}
Assume $\Omega_1$ and $\Omega_2$ are non-singular. Suppose $(C^{\ast},\Lambda^{\ast})$ is a solution of the equation system (\ref{eq.1st})-(\ref{eq.2nd}). Then, if the diagonal elements of $\Lambda^{\ast}$ are all distinct, the set of solutions solving (\ref{eq.1st})-(\ref{eq.2nd}) is represented as
\begin{equation}
	\label{obs equiv sol}
	\big\{ (C,\Lambda) = (C^{\ast} SP', P \Lambda^{\ast}P^{\prime}): P\in \Pn, S \in \Dn \big\}. 
\end{equation} 
\end{theorem}

\medskip
 
\begin{proof} 
See the appendix \ref{app:Proofs}.
\end{proof} 

This theorem clarifies the fundamental indeterminacy of the solutions in the equation system (\ref{eq.1st})-(\ref{eq.2nd}). 
The observational equivalence with respect to $S \in \mathcal{D}(n)$ corresponds to the indeterminacy of the signs of structural shocks common in any SVAR modeling (see \citeauthor{LanneLutkepohlMaciejowska}, \citeyear{LanneLutkepohlMaciejowska}, for an equivalent result on HSVARs). We often control such sign indeterminacy by imposing the sign normalization restrictions that pin down $S$, e.g., restricting the diagonal elements of $A_0=\Lambda_1^{1/2}C^{-1}$ to be non-negative. The observational equivalence with respect to the permutations corresponds to the indeterminacy of the structural parameters with respect to the reordering of the structural equations. \citet{Rigobon03} noted this indeterminacy of the ordering of the structural equations in bivariate HSVAR models and argued that sign restrictions placed on the off-diagonal elements of $A_0=\Lambda_1^{1/2}C^{-1}$ resolve such indeterminacy.  

\bigskip

Theorem \ref{theo:SignPerm} implies that with sign normalization restrictions imposed, exact-identification of $(C,\Lambda)$ requires an assumption that pins down the ordering of the equations (i.e., permutation matrix $P$). One way to constrain the ordering of the equations is to exploit available knowledge on the ratios of the structural shock variances of regime 1 to regime 2. In particular, assuming a complete ordering of the structural shocks according to their variance ratios can fix the order of the structural equations based on the diagonal entries of the true $\Lambda$. Hence, if a solution of $\Lambda$ is such that all of its diagonal elements are distinct, a complete ordering of such elements reduces the set of solutions in Eq. (\ref{obs equiv sol}) to a singleton. In light of Theorem \ref{theo:SignPerm}, we formalize the following normalization rule, whose uniqueness requirement applies only when the variance shifts are distinct. With repeated shifts, the ranking orders the distinct eigenspaces, while admissible bases within each repeated-eigenvalue block remain subject to the normalization restrictions.

\begin{defin}(Normalization) 
				\label{def:norm} 
				A normalization rule can be characterized by a set $N\subset\mathbb{P}^S$ such that for any structural parameter point $\theta=\left(A_0, A_+, \Lambda\right)\in \mathbb{P}^S$, with distinct elements on the diagonal of $\Lambda$, there exists a unique permutation matrix $P\in \Pn$ and a unique diagonal matrix $S\in \Dn$, with +1 and -1 along the diagonal, such that $(PSA_0,PSA_+,P\Lambda P^\prime) \in N$.\qed
\end{defin}

It is important to stress that all the results developed in the paper can be rephrased in terms of different normalization rules coherent with
Definition \ref{def:norm}.

The following theorem, instead, follows as a corollary of Theorem \ref{theo:SignPerm}. 

\begin{theorem}[Exact identification]
	\label{theo:HSVAR_Ident} 
	Under the same assumptions of Theorem \ref{theo:SignPerm}, then, with sign normalization restrictions and complete ordering of the structural shocks according to the variance ratios imposed, $(C,\Lambda)$ is exactly-identified.
\end{theorem}

\medskip
 
\begin{proof} 
    The theorem, apart from the potential indeterminacy due to the column permutation, corresponds to Proposition 1 in \cite{LanneLutkepohlMaciejowska}, and can be proved in exactly the same way. However, according to Theorem \ref{theo:SignPerm} here before, fixing the ordering of the shocks is also necessary in order to have exactly identified $(C,\Lambda)$.
\end{proof} 

If a solution of $\Lambda$ has some of the diagonal elements identical, then invariance of $\Lambda$ with respect to a permutation that permutes only these elements fails to uniquely pin down $C$. 
Partial identification of $C$ matrix in this case is to be considered below.

The identification result of Theorem \ref{theo:HSVAR_Ident} is not constructive and it does not provide an explicit analytical expression of $(C,\Lambda)$ as a function of $(\Omega_1,\Omega_2)$. A more constructive identification result for $(C,\Lambda)$ can be obtained by representing the equation systems (\ref{eq.1st})-(\ref{eq.2nd}) as a certain eigen-decomposition problem, as already presented in Eq. (\ref{eigen decomposition}). This perspective yields the following succinct analytical characterization of the solutions of the equation system.    

\begin{theorem}[Identification and eigen-decomposition] 
	\label{theo:HSVAReigen}
	The set of solutions solving system (\ref{eq.1st})-(\ref{eq.2nd}) can be represented by $ (\Omega_{1,tr}Q,\Lambda)$, 
	where $\Lambda$ is a diagonal matrix of eigenvalues of $\Omega_{1,tr}^{-1} \Omega_2 \Omega_{1,tr}^{-1\prime}$ 
	and $Q$ is an orthogonal matrix of the corresponding eigenvectors. 
	Hence, if the eigenvalues of $\Omega_{1,tr}^{-1} \Omega_2 \Omega_{1,tr}^{-1\prime}$ are all distinct, 
	$(C,\Lambda)$ is identified up to permutations and sign changes of the structural equations.
\end{theorem}

\medskip
 
\begin{proof} 
    The result immediately follows from Theorem A9.9, and the related proof, in \cite{Muirhead}.
\end{proof} 

The claim of this theorem simplifies computation of an estimator of $(C,\Lambda)$; the maximum likelihood estimator for $(C,\Lambda)$ can be computed by performing an eigen-decomposition on the maximum likelihood estimator of $\Omega_{1,tr}^{-1} \Omega_2 \Omega_{1,tr}^{-1\prime}$ subject to the sign normalization. If a complete ordering assumption on $\Lambda$ is available (e.g., the diagonal elements of $\Lambda$ is decreasing), we can obtain a point-estimator for $(C,\Lambda)$ by ordering the eigenvalues accordingly through the decomposition. If the ordering assumption is not available, then permutations of the diagonal elements in $\Lambda$ and the corresponding eigenvectors in $Q$ span the identified set of $(C,\Lambda)$. The idea of treating identification and estimation of HSVARs as an eigen-decomposition issue has been also pursued by \citetalias{LMNS20}, that developed their test for identification via heteroskedasticity as a test on equivalent eigenvalues.

An alternative way to see and address the identification problem of $(C,\Lambda)$ is to look at system (\ref{eq.1st})-(\ref{eq.2nd}) in a slightly different way. In fact, given that $\Lambda$ is made of positive elements, it is possible to rewrite Eq. (\ref{eq.2nd}) as 
$\Omega_2=C\Lambda^{1/2}\Lambda^{1/2}C^\prime$. The quantity $C\Lambda^{1/2}$ could not be unique because of the presence of an orthogonal matrix $Q_2$ such that $\Omega_2=C\Lambda^{1/2}Q_2Q_2^\prime\Lambda^{1/2}C^\prime$. Using the result in Proposition A.1 of \citet{Uhlig05JME} for the decomposition of $\Omega_1$ and $\Omega_2$ yields the following system
\begin{eqnarray}
C & = & \Omega_{1,tr}Q_1\nonumber\\
\Omega_{2,tr} & = & C\Lambda^{1/2}Q_2\nonumber
\end{eqnarray}
and plugging the definition of $C$ into the second equation we obtain
\begin{equation}
\label{eq:SVDrel} 
\Omega_{1,tr}^{-1}\Omega_{2,tr}=Q_1\Lambda^{1/2}Q_2.
\end{equation}
The next theorem discusses the identification issue of the structural parameters $(C,\Lambda)$ in terms of the uniqueness of $Q_1$ and $Q_2$.

\begin{theorem}[Identification and Singlular Value Decomposition]
	\label{theo:SVD} 
	The set of solutions of system (\ref{eq.1st})-(\ref{eq.2nd}) can be represented by $(\Omega_{1,tr}Q_1,\Lambda)$, 
	where $\Lambda$ is a diagonal matrix made of positive elements and $Q_1$ is an orthogonal matrix solving the Singular Value Decomposition 
	of $\Omega_{1,tr}^{-1}\Omega_{2,tr}=Q_1\Lambda^{1/2}Q_2$. 
	If the entries in $\Lambda$ are all distinct, the decomposition is unique up to matched sign changes and permutations of the columns of $Q_1$ and rows of $Q_2$, with the diagonal entries of $\Lambda$ reordered accordingly. 
	Hence, $(C,\Lambda)$ is identified up to permutations and sign changes of the structural equations.
\end{theorem}

\medskip
 
\begin{proof} 
See the Appendix \ref{app:Proofs}.
\end{proof} 

The following theorems, based on Proposition 3 in \cite{GK18}, provide sufficient conditions for the impulse response identified set $\IS$ to 
be convex. Precisely, we first consider the case of zero restrictions only, and then extend to the case of zero and sign restrictions. First of all, however, we introduce a specific ordering of the shocks according to the identifying restrictions in Eq. (\ref{eq:RestSet}).

\begin{defin}[Ordering of variables]
	\label{def:ordering}
	The variables associated with the eigenvalue $\lambda^i$, with algebraic multiplicity $m_i$, are ordered according to the number of zero
	restrictions on $(q_1^i,\ldots,q_{m_i}^i)$, and specifically, such that they follow the relation
	\begin{equation}
		\label{eq:ordering}
		f_1^i\geq f_2^i \geq \ldots \geq f_{m_i}^i\geq 0.
	\end{equation}
	In case of ties, the shock of interest, represented by the $\ja$-th column of $(q_1^i,\ldots,q_{m_i}^i)$, is ordered first.
	In other words, let $\ja=1$ if no other column has a larger number of restrictions than $q_{\ja}^1$. If $\ja\geq 2$, then let 
	the variables be ordered such that $f_{\ja-1}^i>f_{\ja}^i$.
\end{defin}
\medskip

\begin{theorem}[Convexity of identified set under zero restrictions] 
	\label{theo:SetIdZero}
	Consider an HSVAR characterized by the eigenvaules and eigenspaces as in Definition \ref{def:eigenspace} and by the admissible parameters
	as in Definition \ref{def:admissible}. Let $\lambda^i$ be an eigenvalue of algebraic multiplicity $g(\lambda^i)=m_i$, with associated eigenspace 
	$Q(\lambda^i)$ as in Eq. (\ref{eq:eigenspace}), containing $\qja$, the column of $Q$ associated with the $\ja$-th structural shock (shock of
	interest). 
    Moreover, let $r=\eta(\phi,Q)=c_{lh}(\phi)'q_{j^*}^i$ denote the impulse response of interest.
	Finally, let the variables be ordered as in Definition \ref{def:ordering} and assume that $\mathcal Q(\phi\mid F,\varnothing)\neq\varnothing$ for $\phi$-almost every reduced-form parameter in the support of the maintained multiplicity model. No additional sign restrictions are imposed.

    Then the identified set $IS(\phi\mid F,\varnothing)$ is nonempty and bounded for every $l\in\{1,\ldots,n\}$ and $h=0,1,\ldots$, $\phi$-almost surely. It is also convex if one of the following conditions holds and the corresponding completion condition stated below is satisfied:
	\begin{enumerate}
		\item $\ja=1$ and $f_1^i<m_i-1$;
		\item $\ja\geq 2$ and $f_j^i<m_i-j$, for $j=1,\ldots,(\ja-1)$;
		\item $\ja\geq 2$ and there exists $1\leq k < (\ja-1)$ such that $(q_1^i,\ldots,q_k^i)$ is exactly identified as in Theorem 
		\ref{theo:PointIdHSVAR} and $f_j^i<m_i-j$, for $j=k+1,\ldots,\ja$.
	\end{enumerate}
    For $\phi$-almost every point in the maintained support, impose the following completion condition under the applicable case. Every unit vector $q\in\mathcal Q(\lambda^i)$ satisfying $F_{j^*}^i(\phi)q=0$ and $\sigma_{j^*}(\phi)'q>0$ can be completed to a matrix $Q\in\mathcal Q(\phi\mid F,\varnothing)$ with $q_{j^*}^i=q$. In case (3), this condition applies to vectors $q$ that are also orthogonal to the first $k$ exactly identified columns.
\end{theorem}

\begin{proof}
		See Appendix \ref{app:Proofs}.
\end{proof}

The previous theorem just consider zero restrictions on the vectors of $Q(\lambda^i)$. The following one, instead, also allows for sign restrictions, 
although these last can be imposed on the vector $\qja$ associated with the shock of interest. 

\begin{theorem}[Convexity of identified set under zero and sign restrictions] 
	\label{theo:SetIdZeroSign}
    Retain the notation and the three zero-restriction configurations of Theorem \ref{theo:SetIdZero}, where, as before, $\qja$ is the column vector corresponding to the shock of interest, and let $\qja\in Q(\lambda^i))$, the eigenspace associated with the eigenvalue $\lambda^i$, of algebraic multiplicity $g(\lambda^i)=m_i$. Moreover, let the sign restrictions be imposed on the shock of interest, only. Assume, in addition, that for $\phi$-almost every reduced-form parameter in the support of the multiplicity-constrained model, every unit vector $q$ in the eigenspace associated with $\lambda^i$ that satisfies the zero restrictions, the sign restrictions $S_{j^*}(\phi)q\geq 0$, and the strict normalization $\sigma_{j^*}(\phi)'q>0$ can be completed to a matrix $Q\in\mathcal Q(\phi\mid F,S)$ with $q_{j^*}^i=q$. In case (2), $q$ must also be orthogonal to the first $k$ exactly identified columns.
	\begin{enumerate}
		\item Let the zero restrictions $F^i(\phi,Q)=0$ satisfy one of the conditions (1) and (2) of Theorem \ref{theo:SetIdZero}.
			If there exists a unit length vector $q\in Q(\lambda^i)$ such that
			\begin{equation}
				\label{eq:Cond1ConvexSign}
				F_{\ja}^i(\phi)\,q=0 \hspace{0.5cm}\text{and}\hspace{0.5cm}
				\left(\begin{array}{c}S_{\ja}(\phi)\\\sigma^{\ja\prime}\end{array}\right)\,q>0
			\end{equation}
			then the identified set is non empty and convex for every $l\in\{1,\ldots,n\}$ and $h=0,1,2,\ldots$
		\item Let the zero restrictions $F^i(\phi,Q)=0$ satisfy condition (3) of Theorem \ref{theo:SetIdZero}, and let 
			$\big(q_1^i(\phi),\ldots,q_k^i(\phi)\big)$ be the first $k$ vectors that are exactly identified. If there exists a unit vector $q\in \Re^n$ such that
			\begin{equation}
				\label{eq:Cond2ConvexSign}
				\Big(F_{\ja}^i(\phi)^{\prime}\:,\:v_1^{i}\:,\:\ldots\:,\:v_{(n-m_i)}^{i}\:,\:
			     q_1^{i}(\phi)\:,\:\ldots\:,\:q_k^{i}(\phi)\Big)^\prime\,q=0 
				\hspace{0.5cm}\text{and}\hspace{0.5cm}\left(\begin{array}{c}S_{\ja}(\phi)\\\sigma^{\ja\prime}\end{array}\right)\,q>0
			\end{equation}
			where $(v_1^{i},\ldots,v_{(n-m_i)}^i)$ is a basis for the space $Q^{\perp}(\lambda^i)$, then the identified set is non empty and convex for every $l\in\{1,\ldots,n\}$ and $h=0,1,2,\ldots$
	\end{enumerate}
\end{theorem}
 
\begin{proof}
		See Appendix \ref{app:Proofs}.
\end{proof}

Taken jointly, the results of Theorems \ref{theo:SetIdZero} and \ref{theo:SetIdZeroSign} generalize Proposition 3 in \cite{GK18} to the case 
of a structural break on the variances of the shocks, with potential eigenvalue multiplicity. 

On the other side, the two theorems provide important insights on the possibility to apply the standard identification-through-heteroskedasticity approach to the case in which some of the switches in the variances are the same. These new results represent the foundations for developing the estimator for the bounds of the identified set and produce the related inference, as discussed in Section \ref{sec:TestInference}.

\doublespacing
\section{Bayesian estimation of HSVAR}
\label{app:BayesSVAR}

Building on Appendix A in \citetalias{CMT23WP} and Appendix B in \cite{LW20}, this appendix shows Bayesian estimation for HSVAR including the Savage-Dickey density ratio for eigenvalue multiplicity and the posterior distribution for structural parameters conditional on eigenvalue multiplicity.

\subsection{Prior distributions}
\label{app:priors}

In Section \ref{sec:Estimation} we have introduced the prior distributions used in the estimation procedure of the HSVAR. Here below we provide some more details. 

The first point refers to the simultaneous matrix $A_0$, whose i-th row can be written as $A_{0,i}'=W_i\alpha_i+w_i$, where $W_i$ is a known $n\times(n-1-f_i)$ matrix of full column rank, $w_i$ is a known $n\times1$ vector, and $\alpha_i\in\mathbb{R}^{n-1-f_i}$ collects the free coefficients. Let $p(\alpha_i)$ be the prior distribution for the free parameters of $A_{0,i}$, for $i=1,\ldots,n$. This can be an informative prior distribution as advocated by \cite{BH15}, or a diffuse prior centered at zero with large prior variances, as in \cite{LW20}. 

Following \citetalias{CMT23WP}, the prior distribution for the autoregressive parameters $A_{+,i}$ is set to follow a multivariate normal distribution $p(A_{+,i})\sim N(\mu_{A_{+,i}},\,V_{A_{+,i}})$, for $i=1,\ldots,n$. Alternatively, \cite{LW20} and \cite{WZ03JEDC} suggest a normal prior distribution conditional on $\alpha_i$, of the form $p(A_{+,i}|\alpha_i)\sim N(d_i+D_i\alpha_i,\,H_i)$, where $d_i$ is a $m\times 1$ known vector, $D_i$ is a known $m \times (n-1-f_i)$ matrix and $H_i$ is a known $m \times m$ positive definite matrix such that the associated prior on reduced-form parameters is in the spirit of Litterman's original random walk prior.

Finally, the prior distributions for the variances in the first regime $\lambda_{1,i}$ and the variance shifts $\lambda_i$ are set to follow independent inverse gamma distributions, $IG2(\underline{a}_{1,i},\underline{b}_{1,i})$ and $IG2(\underline{a}_i,\underline{b}_i)$, respectively.

\subsection{Likelihood function}
\label{app:lf}

Given the specification of the HSVAR in Eq. (\ref{eq:HSVAR}) and the normality assumption on the distribution of the structural shocks, the likelihood function can be written as 
\begin{eqnarray}
p\Big(\textbf{Y}|A_0,A_+,\Lambda_1,\Lambda\Big) & = & (2\pi)^{-\frac{Tn}{2}}\left|\text{det}(A_0)\right|^T\left(\prod_{i=1}^n \lambda_{1,i}^{-T/2}\right)\left(\prod_{i=1}^n\lambda_i^{-(T-T_b)/2}\right)\times\nonumber\\
& & \text{exp}\left\{-\frac{1}{2}\sum_{i=1}^{n} \lambda_{1,i}^{-1}\left[\sum_{t=1}^{T_B}\Big(A_{0,i}y_t-A_{+,i}x_t\Big)\Big(A_{0,i}y_t-A_{+,i}x_t\Big)^\prime + \right.\right.\nonumber\\
& & \hspace{1.5cm}\left.\left.\sum_{t=T_B+1}^{T} \lambda_i^{-1}\Big(A_{0,i}y_t-A_{+,i}x_t\Big)\Big(A_{0,i}y_t-A_{+,i}x_t\Big)^\prime\right]\right\}\label{eq:lf},
\end{eqnarray}
where the vector $x_t$ contains the deterministic components of the VAR and the lagged dependent variables. This form of the likelihood function, combined with appropriate priors, allows treating the model equation-by-equation, and, importantly, allows including the restrictions on eigenvalue multiplicity in the estimation and inference procedures. 

\subsection{Conditional posterior distribution of $(A_0,A_+)$}
\label{app:PostA}

According to the likelihood function in Eq. (\ref{eq:lf}) and the independent prior distributions in Section \ref{app:priors}, the conditional posterior distribution of $(A_{0,i},A_{+,i})$ is
\begin{equation}
    p(A_{0,i},A_{+,i}|Y,\Lambda_1,\Lambda) \propto p(Y|A_0,A_+,\Lambda_1,\Lambda)\times p(A_{+,i})\times p(\alpha_i)\nonumber
\end{equation}
\begin{eqnarray}
    & \propto & \left|\text{det}(A_0)\right|^T \text{exp}\left\{-\frac{1}{2}\sum_{i=1}^{n} \lambda_{1,i}^{-1}\left[\sum_{t=1}^{T_B}\Big((W_i \alpha_i + w_i)y_t-A_{+,i}x_t\Big)\Big((W_i \alpha_i + w_i)y_t-A_{+,i}x_t\Big)^\prime + \right.\right.\nonumber\\
    & & \hspace{1cm}\left.\left.\sum_{t=T_B+1}^{T} \lambda_i^{-1}\Big((W_i \alpha_i + w_i)y_t-A_{+,i}x_t\Big)\Big((W_i \alpha_i + w_i)y_t-A_{+,i}x_t\Big)^\prime\right]\right\}\times\nonumber\\
    & & \hspace{2.5cm} \text{exp}\left[-\frac{1}{2}\Big(A_{+,i}-\mu_{A_{+,i}}\Big) V_{A_{+,i}}^{-1}\Big(A_{+,i}-\mu_{A_{+,i}}\Big)^\prime\right]\times p(\alpha_i).\nonumber
\end{eqnarray}
Using simple algebra, the previous relation can be written as
\begin{eqnarray}
    p(A_{0,i},A_{+,i}|Y,\Lambda_1,\Lambda) & \propto & \left|\text{det}(A_0)\right|^T \text{exp}\left\{-\frac{1}{2}\big(\alpha_i-\bar{\alpha}_i\big)^\prime\bar{V}_{\alpha_i}^{-1}\big(\alpha_i-\bar{\alpha}_i\big)\right\}\times p(\alpha_i)\times\nonumber\\
    & & \hspace{1cm} \text{exp}\left\{-\frac{1}{2}\big(A_{+,i}-\bar{\mu}_{A_{+,i}}\big)^\prime\bar{V}_{A_{+,i}}^{-1}\big(A_{+,i}-\bar{\mu}_{A_{+,i}}\big)\right\}.\label{eq:postA}
\end{eqnarray}
The first line of the previous equation can be defined as $p(\alpha_i|Y,\Lambda_1,\Lambda)$, and represents the posterior density of $\alpha_i$ conditional on the variances of the first regime, $\Lambda_1$, and the variance shifts, $\Lambda$, where 
\begin{eqnarray} 
    \bar{\alpha}_{i} & = & \bar{V}_{\alpha_i}W_i^\prime\left[\Upsilon \bar{V}_{A_{+,i}}\left(V_{A_{+,i}}^{-1}\mu_{A_{+,i}}+\lambda_{1,i}^{-1}\sum_{t=1}^{T_B}x_t\Big(w_i^\prime y_t\Big)^\prime +\lambda_{1,i}^{-1}\sum_{t=T_B+1}^{T}\lambda_i^{-1} x_t\Big(w_i^\prime y_t\Big)^\prime\right)+\right.\nonumber\\
    & & \hspace{2cm} -\left.\left(\lambda_{1,i}^{-1} \sum_{t=1}^{T_B}y_t y_t^\prime w_i +\lambda_{1,i}^{-1}\sum_{t=T_B+1}^{T}\lambda_i^{-1} y_t y_t^\prime w_i\right)\right] \label{eq:muA0post} 
\end{eqnarray}
and
\begin{equation}
    \bar{V}_{\alpha_i}^{-1} = W_i^\prime\left(\lambda_{1,i}^{-1}\sum_{t=1}^{T_B} y_t y_t^\prime + \lambda_{1,i}^{-1} \sum_{t=T_B+1}^{T}\lambda_i^{-1} y_t y_t^\prime - \Upsilon\bar{V}_{A_{+,i}}\Upsilon^\prime\right)W_i\label{eq:VA0post}
\end{equation}
where $\Upsilon = \lambda_{1,i}^{-1}\sum_{t=1}^{T_B}y_t x_t^\prime + \lambda_{1,i}^{-1}\sum_{t=T_B+1}^{T}\lambda_i^{-1} y_t x_t^\prime$.

In Section \ref{app:priors} we mentioned the possibility of imposing independent informative prior distributions on $\alpha_i$ \textit{\'{a} la} \cite{BH15}, or noninformative ones as in \cite{LW20}. In this respect, for uniform priors $p(\alpha_i)$, or if they belong to the generalized normal family, a modified version of the original \cite{WZ03JEDC} algorithm, proposed by \cite{V09} and \citetalias{CKY24}, allows drawing from $p(\alpha_i|Y,\Lambda_1,\Lambda)$. 
Under other prior specifications, draws from a suitable reference conditional distribution based on the \cite{WZ03JEDC} sampler may be used as proposals in an independence Metropolis-Hastings update, with the acceptance probability computed from the target and proposal densities.
Finally, for Jeffreys flat priors $p(A_0,A_+)\propto \left|\text{det}A_0\right|^{-n}$ without any unit normalizations on $A_0$, the posterior moments can be obtained from Eq.s (\ref{eq:muA0post})-(\ref{eq:VA0post}) by substituting $W_i=I_n$, $w_i=0$ and $V_{A_{+,i}}^{-1}=0$.

The second line of Eq. (\ref{eq:postA}), instead, can be indicated as $p(A_{+,i}|Y,A_0,\Lambda_1,\Lambda)$, and is the posterior distribution of $A_{+,i}$ conditional on $\Lambda_1$, $\Lambda$ and $A_0$. This last conditional distribution, given the likelihood function and the normal prior distribution of $A_{+,i}$, follows a Gaussian distribution $p(A_{+,i}|Y,A_0,\Lambda_1,\Lambda)\sim N(\bar{\mu}_{A_{+,i}},\,V_{A_{+,i}})$, where the two parameters can be defined as      
\begin{eqnarray}
    \bar{\mu}_{A_{+,i}}^\prime & = & \left[\mu_{A_{+,i}}^\prime V_{A_{+,i}}^{-1}+\lambda_{1,i}^{-1}\sum_{t=1}^{T_B}A_{0,i}y_t x_t^\prime + \lambda_{1,i}^{-1}\sum_{t=T_B+1}^{T}\lambda_i^{-1} A_{0,i}y_t x_t^\prime \right]\bar{V}_{A_{+,i}}\nonumber\\
    V_{A_{+,i}} & = & \left[V_{A_{+,i}}^{-1}+\lambda_{1,i}^{-1}\sum_{t=1}^{T_B}x_t x_t^\prime + \lambda_{1,i}^{-1}\sum_{t=T_B+1}^{T}\lambda_i^{-1} x_t x_t^\prime\right]^{-1}\nonumber
\end{eqnarray}

\subsection{Conditional posterior distribution of $\lambda_{1,i}$}
\label{app:PostL1}

As already introduced in Section \ref{sec:Estimation}, when the prior is an inverse gamma $p(\lambda_{1,i})\sim IG2(\underline{a}_{1,i},\underline{b}_{1,i})$, the conditional posterior distribution of the variance of each structural shock in the first regime becomes
\begin{equation}
    p(\lambda_{1,i}|Y,A_0,A_+,\Lambda) \propto p(Y|A_0,A_+,\Lambda_1,\Lambda)\times p(\lambda_{1,i})\nonumber
\end{equation}
\begin{eqnarray}
    & \propto & \left(\lambda_{1,i}\right)^{-\frac{T+\underline{a}_{1,i}+2}{2}}\text{exp}\left\{-\frac{1}{2\lambda_{1,i}}\left[\underline{b}_{1,i} + \sum_{t=1}^{T_B}\left(A_{0,i}y_t - A_{+,i}x_t\right)\left(A_{0,i}y_t - A_{+,i}x_t\right)^\prime + \right.\right.\nonumber\\
    & & \hspace{2cm} \left.\left.\sum_{t=T_B+1}^{T}\lambda_i^{-1}\left(A_{0,i}y_t - A_{+,i}x_t\right)\left(A_{0,i}y_t - A_{+,i}x_t\right)^\prime      \right]\right\} \nonumber
\end{eqnarray}
which is the kernel of an inverse gamma \textit{IG2}, with shape parameter $\bar{a}_{1,i} = (T+\underline{a}_{1,i})$ and scale parameter 
\begin{equation}
\bar{b}_{1,i} = \left[\underline{b}_{1,i} + \sum_{t=1}^{T_B}\left(A_{0,i}y_t - A_{+,i}x_t\right)\left(A_{0,i}y_t - A_{+,i}x_t\right)^\prime + \sum_{t=T_B+1}^{T}\lambda_i^{-1}\left(A_{0,i}y_t - A_{+,i}x_t\right)\left(A_{0,i}y_t - A_{+,i}x_t\right)^\prime\right].\nonumber
\end{equation}

\subsection{Conditional posterior distribution of $\lambda_{i}$}
\label{app:PostL}

The variance shifts play a crucial role on the identification through heteroskedasticity literature. In \citetalias{CMT23WP} and \cite{LW20} the posterior distribution of these objects is obtained independently for the different shocks, starting from an inverse gamma prior distribution $p(\lambda_i)\sim IG2(\underline{a}_i,\underline{b}_i)$, in general common to all shocks, but specific information coming from the data through the likelihood function. Specifically, in both papers they obtain a posterior distribution $p(\lambda_i|Y,A_0,A_+,\Lambda_1)$ of the form
\begin{eqnarray}
    p(\lambda_{i}|Y,A_0,A_+,\Lambda_1) & \propto & p(Y|A_0,A_+,\Lambda_1,\Lambda)\times p(\lambda_{i})\nonumber\\
                                       & \propto & \left(\lambda_{i}\right)^{-\frac{\underline{a}_i+N+2}{2}}\text{exp}\left\{-\frac{\underline{b}_i+S_i}{2\lambda_{i}}\right\} \nonumber
\end{eqnarray}
which is the kernel of an inverse gamma distribution with $IG2(\underline a_i+N,\;\underline b_i+S_i)$, where $N=T-T_B$ and $S_i=\lambda_{1,i}^{-1}\sum_{t=T_B+1}^{T}\Big(A_{0,i}y_t - A_{+,i}x_t\Big)\Big(A_{0,i}y_t - A_{+,i}x_t\Big)^\prime$.

In our contribution, instead, once there is statistical evidence of eigenvalue multiplicity, we ought to impose such a constraint on the estimation procedure. Suppose, for example, that there is evidence that two eigenvalues, say $\lambda_i$ and $\lambda_j$, are potentially equal, i.e. $\lambda_i=\lambda_j=\bar{\lambda}$. The definition of the likelihood function provided in Section \ref{app:lf}, combined with a common inverse gamma prior distribution $p(\bar{\lambda})\sim IG2(\underline{a},\underline{b})$ for the two $\lambda$'s, with $\underline{a}=\underline{a}_i+\underline{a}_j$ and $\underline{b}=\underline{b}_i+\underline{b}_j$, allow to easily derive the posterior distribution of the common eigenvalue, which will be shared by the two \textit{i}-th and \textit{j}-th structural shocks. Specifically, if we define $N=T-T_B$ and $S_k=\lambda_{1,k}^{-1}\sum_{t=T_B+1}^{T}\Big(A_{0,k}y_t - A_{+,k}x_t\Big)\Big(A_{0,k}y_t - A_{+,k}x_t\Big)^\prime$, $k=\{i,\,j\}$, we obtain
\begin{eqnarray}
    p(\bar{\lambda}|Y,A_0,A_+,\Lambda_1) & \propto & p(Y|A_0,A_+,\Lambda_1,\Lambda)\times p(\bar{\lambda})\nonumber\\
                                         & \propto & \left(\bar{\lambda}\right)^{-\frac{\underline{a}+2N+2}{2}}\text{exp}\left\{-\frac{\underline{b}+S_i+S_j}{2\bar{\lambda}}\right\}\nonumber
\end{eqnarray}
which, again, is the kernel of an inverse gamma distribution $IG2(\underline{a}+2N,\;\underline{b}+S_i+S_j)$.
The generalization to $m$ multiple eigenvalues is straightforward. This approach allows to obtain, given the eigen-decomposition in Eq. (\ref{eigen decomposition}), a constrained covariance matrix for the error terms in the VAR model in the second regime.   

\subsection{From structural- to reduced-form draws}
\label{app:StoR}

Conditional on the maintained eigenvalue-multiplicity specification, each posterior draw of the structural parameters $\theta^{(s)}=\big(A_0^{(s)},A_+^{(s)},\Lambda_1^{(s)},\Lambda^{(s)}\big)$ is transformed into a reduced-form draw $\phi^{(s)}=\big(B^{(s)},\Omega_1^{(s)},\Omega_2^{(s)}\big)$ through
\[
\begin{aligned}
B^{(s)}&=\big(A_0^{(s)}\big)^{-1}A_+^{(s)},\\
\Omega_1^{(s)}&=\big(A_0^{(s)}\big)^{-1}\Lambda_1^{(s)}
\big[\big(A_0^{(s)}\big)^{-1}\big]',\\
\Omega_2^{(s)}&=\big(A_0^{(s)}\big)^{-1}\Lambda^{(s)}\Lambda_1^{(s)}
\big[\big(A_0^{(s)}\big)^{-1}\big]'.
\end{aligned}
\]
Here, $\Lambda^{(s)}$ contains the relative variance shifts, so that the second-regime structural covariance matrix is recovered as $\Lambda_2^{(s)}=\Lambda^{(s)}\Lambda_1^{(s)}$. Since the equality restrictions on the diagonal elements of $\Lambda$ are imposed during posterior sampling, the transformed draws satisfy the maintained multiplicity restrictions. These draws represent the marginal posterior of the reduced-form parameters induced by the chosen structural prior. The robust Bayesian procedure in Section III.3 retains this reduced-form posterior and allows the conditional prior on $Q$ to vary over the class of distributions supported on $\mathcal Q(\phi\mid F,S)$. For each reduced-form draw, impulse-response bounds are computed over this admissible set and subsequently used to construct posterior mean bounds and robust credible regions.

\subsection{On the SDDR for equality of a subset of relative variance shifts}
\label{app:SDDRgroup}

In this section we report the proof of Proposition \ref{prop:SDDRgroup}, where we use all the ingredients provided before to derive a closed form for the SDDR for the equality of a general subset of relative variance shifts and the related way to numerically compute it.

\vspace{1cm}

\begin{proof}[Proof of Proposition \ref{prop:SDDRgroup}] 

For $k=1,\ldots,K$, define
\[
\alpha_k=\frac{\underline{a}_{i_k}}{2},
\qquad
\beta_k=\frac{\underline{b}_{i_k}}{2},
\]
so that
\[
A_{\mathcal I}
=
\sum_{k=1}^{K}\alpha_k.
\]
Moreover, let
\[
\tau=\lambda_{i_1},
\qquad
r_k=\frac{\lambda_{i_1}}{\lambda_{i_k}},
\quad k=2,\ldots,K,
\]
and collect the ratios in
$\boldsymbol{r}=(r_2,\ldots,r_K)'$. The inverse transformation, thus, is defined as
\[
\lambda_{i_1}=\tau,
\qquad
\lambda_{i_k}=\frac{\tau}{r_k},
\quad k=2,\ldots,K,
\]
and the related Jacobian is
\[
\left|
\frac{
\partial(\lambda_{i_1},\ldots,\lambda_{i_K})
}{
\partial(\tau,r_2,\ldots,r_K)
}
\right|
=
\tau^{K-1}
\prod_{k=2}^{K}r_k^{-2}.
\]

Given that the prior distributions of the relative variances are independent, the joint density of $\tau$ and $\boldsymbol{r}$ is
\[
p(\tau,\boldsymbol{r})
=
\frac{
\displaystyle\prod_{k=1}^{K}\beta_k^{\alpha_k}
}{
\displaystyle\prod_{k=1}^{K}\Gamma(\alpha_k)
}
\tau^{-A_{\mathcal I}-1}
\left(
\prod_{k=2}^{K}r_k^{\alpha_k-1}
\right)
\exp\left\{
-\frac{
\beta_1+\sum_{k=2}^{K}\beta_k r_k
}{\tau}
\right\}.
\]

Integrating out $\tau$, and using standard results for the integration of the Gamma function gives
\[
p(\boldsymbol{r})
=
\frac{
\Gamma(A_{\mathcal I})
}{
\displaystyle\prod_{k=1}^{K}\Gamma(\alpha_k)
}
\frac{
\displaystyle
\left(\prod_{k=1}^{K}\beta_k^{\alpha_k}\right)
\left(\prod_{k=2}^{K}r_k^{\alpha_k-1}\right)
}{
\displaystyle
\left(
\beta_1+\sum_{k=2}^{K}\beta_k r_k
\right)^{A_{\mathcal I}}
}.
\]
Now, substituting back $\alpha_k=\underline{a}_{i_k}/2$ and $\beta_k=\underline{b}_{i_k}/2$, and evaluating the resulting density at $r_2=\cdots=r_K=1$, gives
\[
p\left(
\frac{\lambda_{i_1}}{\lambda_{i_2}}
=\ldots=
\frac{\lambda_{i_1}}{\lambda_{i_K}}
=1
\right)
=
\frac{
\Gamma(A_{\mathcal I})
}{
\displaystyle\prod_{k=1}^{K}
\Gamma(\underline{a}_{i_k}/2)
}
\frac{
\displaystyle\prod_{k=1}^{K}
\underline{b}_{i_k}^{\underline{a}_{i_k}/2}
}{
\displaystyle
\left(
\sum_{k=1}^{K}\underline{b}_{i_k}
\right)^{A_{\mathcal I}}
},
\]
which proves \eqref{eq:joint_prior_ordinate}.

Similarly, conditional on $Y$ and $\psi$, where here $\psi$ denotes all the other parameters, the relative variances are independent $IG2$ random variables with parameters
$a_{i_k}^{\star}$ and $b_{i_k}(\psi)$. Therefore, applying the same change of variables conditionally on $(Y,\psi)$, with $\underline{a}_{i_k}$ and $\underline{b}_{i_k}$ replaced by $a_{i_k}^{\star}$ and $b_{i_k}(\psi)$, respectively, and evaluating the resulting density at $r_2=\cdots=r_K=1$, gives the result in Eq. \eqref{eq:joint_conditional_ordinate}.

Finally, integrating the conditional posterior probability with respect to the posterior distribution of the remaining model parameters yields
\[
\begin{split}
p\left(
\frac{\lambda_{i_1}}{\lambda_{i_2}}
=\ldots=
\frac{\lambda_{i_1}}{\lambda_{i_K}}
=1
\Bigm|Y
\right)
&=
\int
p\left(
\frac{\lambda_{i_1}}{\lambda_{i_2}}
=\ldots=
\frac{\lambda_{i_1}}{\lambda_{i_K}}
=1
\Bigm|Y,\psi
\right)
p(\psi\mid Y)\,d\psi
\\
&=
\int
d_{\mathcal I}(\psi)\,
p(\psi\mid Y)\,d\psi
\\
&=
E_{\psi\mid Y}
\left[d_{\mathcal I}(\psi)\right],
\end{split}
\]
which proves Eq. \eqref{eq:joint_marginal_ordinate}.

\end{proof}

\section{Geometry of identification in bivariate SVARs and HSVARs}
\label{sec:SetIdentBivariate}

This appendix is dedicated to the identification issue in bivariate SVAR models. We first introduce the notion of set identification in standard SVARs as in \cite{GK18}, and derive point identification as a particular case. Second, we move to the core of the paper and extend the set and point identification notions to SVARs characterized by structural breaks, that, as shown in \cite{BacchiocchiFanelli15}, is more general than the separate analysis of each single regime. 

Consider the following bivariate model, where, for simplicity, the dynamics is omitted, as not directly involved in the identification issue:
\begin{equation}
\label{eq:SVAR}
\left(
\begin{array}{cc}1 & -\beta\\-\alpha & 1\end{array}
\right)
\left(
\begin{array}{c} p_t \\ q_t \end{array}
\right)=
\left(
\begin{array}{c} \varepsilon_t \\ \eta_t \end{array}
\right)
\end{equation}
or, more compactly,
\begin{equation}
\label{eq:SVARcomp}
AY_t = \epsilon_t.
\end{equation}
In order to ease the explanation, we introduce a theoretical foundation to the model and interpret the first equation as a demand equation while the latter as a supply equation.
The vector $Y_t=(p_t\,,\,q_t)^\prime$ collects the two observable variables and $\epsilon_t=(\varepsilon_t\,,\,\eta_t)^\prime$ the two structural shocks. Furthermore, let the structural shocks be characterized by null expected values and by the following covariance matrix:
\begin{equation}
\label{eq:Sigma}
\Sigma\equiv \text{Cov}\left(\begin{array}{c} \varepsilon_t \\ \eta_t \end{array}\right)
= \left(\begin{array}{cc} \sigma_\varepsilon^2 & 0 \\ 0& \sigma_\eta^2 \end{array}\right).\nonumber
\end{equation}
The $A$ matrix, containing the parameters of the simultaneous relationships among the observable variables, $\alpha$ and $\beta$, can also be rescaled by dividing for the standard deviations of the shocks and obtain
\begin{equation}
\label{eq:Azero}
A_0\equiv \Sigma^{-1/2} A =
\left(\begin{array}{cc} 1/\sigma_\varepsilon & -\beta/\sigma_\varepsilon \\ -\alpha/\sigma_\eta & 1/\sigma_\eta \end{array}\right).\nonumber
\end{equation}
Actually, from the observation of $p_t$ and $q_t$, the amount of information is contained in the estimable covariance matrix
\begin{equation}
\label{eq:Omega}
\Omega\equiv \text{Cov}\left(\begin{array}{c} p_t \\ q_t \end{array}\right)
= \left(\begin{array}{cc} \omega_p^2 & \omega_{pq} \\ \omega_{pq} & \omega_q^2 \end{array}\right)\nonumber
\end{equation}
that is connected to the structural parameters through the non-linear system of equations
\begin{eqnarray}
\label{eq:sys}
\Omega & = & A^{-1}\Sigma A^{-1\prime}\nonumber\\ 
			 & = & A_0^{-1}A_0^{-1\prime} 
\end{eqnarray}
for which, when the solution with respect to the structural parameters is unique, the identification problem is clearly solved. As is well known, however, without imposing any restriction on the structural parameters, the solution cannot be unique as the three (estimable) empirical moments contained in $\Omega$ are not sufficient to consistently estimate the four structural parameters in $A_0$. In fact, the amount of information contained in $\Omega$ is the same as the one contained in the three elements of its lower triangular Cholesky factorization, that can be given by
\begin{equation}
\label{eq:Chol}
\Chol = \left(
\begin{array}{cc} \omega_p & 0 \\ \omega_{pq}/\omega_p & \left(\omega_q^2-\omega_{pq}^2/\omega_p^2\right)^{1/2}\end{array}
\right)\nonumber
\end{equation}
whose inverse is given by
\begin{equation}
\label{eq:CholInv}
\Choli = \left(
\begin{array}{cc} \frac{1}{\omega_p} & 0 \\ -\frac{\omega_{pq}}{\omega_p^2}\left(\omega_q^2-\frac{\omega_{pq}^2}{\omega_p^2}\right)^{-1/2} & \left(\omega_q^2-\frac{\omega_{pq}^2}{\omega_p^2}\right)^{-1/2}\end{array}
\right)
= \left(\omega_1 \,,\,\omega_2\right)
\end{equation}
where the two $(2\times 1)$ vectors $\omega_1$ and $\omega_2$ are the two columns of $\Choli$.

According to \cite{Uhlig05JME}, the identification issue can be seen in terms of the non-uniqueness of an orthogonal matrix $Q \in \O$, where $\O$ is the set of $(2\times 2)$ orthonormal matrices, such that $A_0=Q^\prime\Choli$ for which $\Omega=(A_0^\prime A_0)^{-1}$. In other words, a unique $Q$ guarantees a unique $A_0$ through a suitable rotation of $\Choli$, containing all the information coming from the reduced form parameters (or, better, the information coming from the data).

Denoting with $Q\equiv(q_1\,,\,q_2)$, with $q_1$ and $q_2$ being the columns of $Q$, the $A_0$ matrix can be given by
\begin{equation}
\label{eq:AzeroFact} 
A_0 = Q^\prime \Choli = \left(\begin{array}{c}q_1^\prime \\ q_2^\prime\end{array}\right) \left(\omega_1\,,\,\omega_2\right).
\end{equation}
Moreover, we consider the following assumptions:

\begin{assump}
\label{ass:SignNorm}
(Sign normalization) Coherently with the sign normalization presented in Section \ref{sec:identification}, among all the possible rotations of $\Choli$ through the orthogonal matrix $Q \in \O$, we consider only those guaranteeing the elements on the main diagonal of $A_0=Q^\prime\Choli$ to be positive. In different words, we consider only $q_1$ and $q_2$ such that $q_1^\prime \omega_1>0$ and $q_2^\prime \omega_2>0$.
\end{assump}

\begin{assump}
\label{ass:SignRestr}
(Sign restriction on $\alpha$ and $\beta$) Coherently with the demand and supply curves in Eq. (\ref{eq:SVAR}), we assume $\alpha\geq 0$ and $\beta\leq 0$.
\end{assump}

The first assumption, that is standard in the SVAR literature, asserts that, as the product of $A_0^{-1} A_0^{-1\prime}$ is invariant to sign changes on the columns of $A_0^{-1}$, we select only $A_0=Q^\prime\Choli$ such that the elements on the main diagonal are strictly positive. Positive own-impact responses additionally require $1-\alpha\beta>0$, equivalently $\det Q=1$,
which is guaranteed by Assumption \ref{ass:SignRestr}. The second assumption, instead, refers to the economic interpretation of the two equations of the bivariate model in Eq. (\ref{eq:SVAR}) as a demand and a supply equation, respectively.

\subsection{Set identification in bivariate SVARs}
\label{sec:SetIdentSVAR}

Given the bivariate SVAR model discussed before, the following proposition provides the identification set for the two structural parameters $\alpha$ and $\beta$, according to the two potential cases of $\omega_{pq}\geq0$ or $\omega_{pq}<0$.

\begin{theorem}[Set identification in bivariate SVARs]
\label{theo:IdSVAR}
Given the bivariate model described above, let
\[
\Omega=
\begin{pmatrix}
\omega_p^2 & \omega_{pq}\\
\omega_{pq} & \omega_q^2
\end{pmatrix}
\]
be positive definite. Under Assumption 1 and the additional restriction
\[
Q\in SO(2)
=\{Q\in\mathcal O(2):\det Q=1\},
\]
the marginal identified sets for $\alpha$ and $\beta$ are:
\begin{itemize}
\item[(Case I):] if $\omega_{pq}>0$, then
\[
\alpha\in\left(-\infty,\frac{\omega_q^2}{\omega_{pq}}\right),
\qquad
\beta\in\left(-\infty,\frac{\omega_p^2}{\omega_{pq}}\right);
\]
\item[(Case II):] if $\omega_{pq}<0$, then
\[
\alpha\in\left(\frac{\omega_q^2}{\omega_{pq}},+\infty\right),
\qquad
\beta\in\left(\frac{\omega_p^2}{\omega_{pq}},+\infty\right);
\]
\item[(Case III):] if $\omega_{pq}=0$, then
\[
\alpha\in\mathbb R,
\qquad
\beta\in\mathbb R.
\]
\end{itemize}
If Assumption 2 also holds, the marginal identified sets become:
\begin{itemize}
\item[(Case I):] if $\omega_{pq}>0$, then
\[
\alpha\in
\left[\frac{\omega_{pq}}{\omega_p^2},
      \frac{\omega_q^2}{\omega_{pq}}\right),
\qquad
\beta\in(-\infty,0];
\]
\item[(Case II):] if $\omega_{pq}<0$, then
\[
\alpha\in[0,+\infty),
\qquad
\beta\in
\left(\frac{\omega_p^2}{\omega_{pq}},
      \frac{\omega_{pq}}{\omega_q^2}\right];
\]
\item[(Case III):] if $\omega_{pq}=0$, then
\[
\alpha\in[0,+\infty),
\qquad
\beta\in(-\infty,0].
\]
\end{itemize}
In all cases, admissible pairs satisfy
\[
\omega_{pq}-\alpha\omega_p^2-\beta\omega_q^2
+\alpha\beta\omega_{pq}=0.
\]
In particular, when $\omega_{pq}=0$,
$\beta=-\alpha\omega_p^2/\omega_q^2$.
\end{theorem}

\medskip
 
\begin{proof} 
See the Appendix \ref{sec:ProofsBivariate}.
\end{proof} 

As is well known, unless we impose at least one equality restriction on one of the structural parameters, the bivariate model cannot be point identified. 
Fixing either $\alpha$ or $\beta$ at a known value, including a nonzero one, globally identifies this bivariate model on its admissible reduced-form domain. Such coefficient restrictions differ from nonzero restrictions imposed directly on entries of the standardized matrix $A_0$, which may yield local but not global identification; see \cite{BK20}.
If no point restriction is imposed, the structural model remains unidentified, but focusing on the sign restrictions coherent with the theoretical interpretation of the model, together with the sign and magnitude of the elements in the reduced form covariance matrix among the observable variables, it is possible to obtain an identified set for $\alpha$ and $\beta$. This is what Theorem \ref{theo:IdSVAR} reports.
 
\begin{figure}
\caption{Identification of $\alpha$ and $\beta$. Case I: $\omega_{pq}\geq 0$}
\label{fig:SetIdentCaseI}
\begin{center}
\begin{tikzpicture}[scale=2.5]
\filldraw[color=red!10](0.7,-0.87) arc(-120:60:1);
\draw [->] (0,0) -- (2.4,0);
\draw [->] (1.2,-1.2) -- (1.2,1.2);
\draw [black, ultra thick] [->] (1.2,0) -- (1.2,0.8);
\node [right] at (1.2,0.8) {$\omega_2$};
\draw [black, ultra thick][->] (1.2,0) -- ++(330:0.9);
\node [right] at (1.8,-0.55) {$\omega_1$};
\draw [black, ultra thick][->] (1.2,0) -- ++(60:1);
\node at (1.9,1) {$\beta=-\infty$};
\node [right] at (1.5,0.5) {$q_1$};
\draw [black, dashed][->] (1.2,0) -- ++(240:1);
\node at (0.8,-0.5) {$q_1$};
\node at (0.5,-1) {$\beta=+\infty$};
\draw [black, ultra thick][->] (1.2,0) -- ++(0:1);
\node [above right] at (2.2,0) {$\beta=0$};
\node at (1.8,-0.1) {$q_1$};
\draw (1.2,0) circle [radius=1];
\draw [red,ultra thick,domain=0:60] plot ({1.2+cos(\x)}, {sin(\x)});
\node at (1.2,-1.5) {Sign restriction: $\beta\leq 0$};
\end{tikzpicture}
\hspace{0.5cm}
\begin{tikzpicture}[scale=2.5]
\filldraw[color=green!10](2.2,0) arc(0:180:1);
\draw [->] (0,0) -- (2.4,0);
\draw [->] (1.2,-1.2) -- (1.2,1.2);
\draw [black, ultra thick] [->] (1.2,0) -- (1.2,0.8);
\node [right] at (1.2,0.8) {$\omega_2$};
\draw [black, ultra thick][->] (1.2,0) -- ++(330:0.9);
\node [right] at (1.8,-0.55) {$\omega_1$};
\draw [black, ultra thick][->] (1.2,0) -- ++(60:1);
\node at (1.9,1) {$\alpha=0$};
\draw [black, ultra thick][->] (1.2,0) -- ++(150:1);
\node at (0.6,-0.1) {$q_2$};
\node [right] at (0.6,0.4) {$q_2$};
\draw [black, ultra thick][<->] (0.2,0) -- ++(0:2);
\node [above right] at (0.2,0) {$\alpha=+\infty$};
\node [above right] at (2.2,0) {$\alpha=-\infty$};
\node at (1.8,-0.1) {$q_2$};
\draw (1.2,0) circle [radius=1];
\draw [red,ultra thick,domain=0:60][<->] plot ({1.2+cos(\x)}, {sin(\x)});
\node at (2.35,0.7) {feasible $q_1$};
\node at (1.2,-1.5) {Sign restrictions: $\beta\leq 0$ and $\alpha\geq 0$ };
\draw [green,ultra thick,domain=90:150][<->] plot ({1.2+cos(\x)}, {sin(\x)});
\node at (0.4,1) {feasible};
\node at (0.4,0.8) {$q_2$};
\end{tikzpicture}
\end{center}
{\scriptsize{\textit{Notes}: Set identification of the parameter $\beta$ (left panel) and joint set identification of $\alpha$ and $\beta$ (right panel). The identified set, under the sign restriction consistent with a demand curve, i.e. $\beta<0$, is represented by the red arc in both panels. In the right panel, the set identification  of $\alpha$ under the further sign restriction consistent with a supply curve, i.e. $\alpha\geq 0$, is represented by the green arc. In both cases, the standard assumption of positive diagonal terms on $A_0$ is considered 
($\sigma_{\varepsilon}>0$ and $\sigma_{\eta}>0$): in light red for the first equation and in light green for the second equation.\par}}
\end{figure}

\begin{figure}
\caption{Identification of $\alpha$ and $\beta$. Case II: $\omega_{pq}< 0$}
\label{fig:SetIdentCaseII}
\begin{center}
\begin{tikzpicture}[scale=2.5]
\filldraw[color=red!10](1.7,-0.87) arc(-60:120:1);
\draw [->] (0,0) -- (2.4,0);
\draw [->] (1.2,-1.2) -- (1.2,1.2);
\draw [black, ultra thick] [->] (1.2,0) -- (1.2,0.8);
\node [right] at (1.2,0.8) {$\omega_2$};
\draw [black, ultra thick][->] (1.2,0) -- ++(30:0.9);
\node [right] at (1.8,0.55) {$\omega_1$};
\draw [black, ultra thick][->] (1.2,0) -- ++(120:1);
\node [above left] at (0.8,0.9) {$\beta=-\infty$};
\node at (1.4,-0.5) {$q_1$};
\node [right] at (0.7,0.5) {$q_1$};
\draw [black, dashed][->] (1.2,0) -- ++(300:1);
\node at (2.1,-0.9) {$\beta=+\infty$};
\draw [black, ultra thick][->] (1.2,0) -- ++(0:1);
\node [above right] at (2.2,0) {$\beta=0$};
\node at (1.8,-0.1) {$q_1$};
\draw (1.2,0) circle [radius=1];
\draw [red,ultra thick,domain=0:120] plot ({1.2+cos(\x)}, {sin(\x)});
\node at (1.2,-1.5) {Sign restriction: $\beta\leq 0$};
\end{tikzpicture}
\hspace{0.5cm}
\begin{tikzpicture}[scale=2.5]
\filldraw[color=green!10](2.2,0) arc(0:180:1);
\draw [->] (0,0) -- (2.4,0);
\draw [->] (1.2,-1.2) -- (1.2,1.2);
\draw [black, ultra thick] [->] (1.2,0) -- (1.2,0.8);
\node [right] at (1.2,0.8) {$\omega_2$};
\draw [black, ultra thick][->] (1.2,0) -- ++(30:0.9);
\node [right] at (1.8,0.55) {$\omega_1$};
\draw [black, ultra thick][->] (1.2,0) -- ++(120:1);
\node [above left] at (0.9,0.9) {$\alpha=0$};
\node at (0.6,-0.1) {$q_2$};
\node [right] at (0.7,0.5) {$q_2$};
\draw [black, ultra thick][<->] (0.2,0) -- ++(0:2);
\node [above right] at (0.2,0) {$\alpha=+\infty$};
\node [above right] at (2.2,0) {$\alpha=-\infty$};
\node at (1.8,-0.1) {$q_2$};
\draw (1.2,0) circle [radius=1];
\draw [red,ultra thick,domain=30:90][<->] plot ({1.2+cos(\x)}, {sin(\x)});
\node at (1.9,1) {feasible $q_1$};
\node at (1.2,-1.5) {Sign restrictions: $\beta\leq 0$ and $\alpha\geq 0$ };
\draw [green,ultra thick,domain=120:180][<->] plot ({1.2+cos(\x)}, {sin(\x)});
\node at (0.2,0.8) {feasible};
\node at (0.2,0.6) {$q_2$};
\end{tikzpicture}
\end{center}
{\scriptsize{\textit{Notes}: Set identification of the parameter $\beta$ (left panel) and joint set identification of $\alpha$ and $\beta$ (right panel). The identified set, under the sign restriction consistent with a demand curve, i.e. $\beta<0$, is represented by the red arc in both panels. In the right panel, the set identification  of $\alpha$ under the further sign restriction consistent with a supply curve, i.e. $\alpha\geq 0$, is represented by the green arc. In both cases, the standard assumption of positive diagonal terms on $A_0$ is considered 
($\sigma_{\varepsilon}>0$ and $\sigma_{\eta}>0$): in light red for the first equation and in light green for the second equation.\par}}
\end{figure}

Although the formal proof is confined in the appendix, the intuition of the results can be obtained from the graphical representations reported in Figure \ref{fig:SetIdentCaseI} and Figure \ref{fig:SetIdentCaseII}, depending on the two potential values of $\omega_{pq}\geq 0$ (Case I) and $\omega_{pq}< 0$ (Case II), respectively, under the sign normalization restriction discussed in Assumption \ref{ass:SignNorm} and sign restrictions of Assumption \ref{ass:SignRestr}.

\begin{corol}[Point identification and OLS estimation of $\alpha$ and $\beta$]
\label{corol:IdSVAR}
Given the bivariate model in Eq. (\ref{eq:SVAR}), under Assumption \ref{ass:SignNorm}:\\
(Case I): if $\omega_{pq}\geq 0$ and we restrict $\beta=0$, then $\alpha = \frac{\omega_{pq}}{\omega_p^2}$\\
(Case II): if $\omega_{pq}< 0$ and we restrict $\alpha=0$, then $\beta = \frac{\omega_{pq}}{\omega_q^2}$.
\end{corol}

\medskip
 
\begin{proof} 
	See the Appendix \ref{sec:ProofsBivariate}.
\end{proof} 

The previous Corollary \ref{corol:IdSVAR} simply restates a standard result in econometrics. In fact, if we introduce a zero restriction on one of the two structural parameters, then the identified set, depending on the observed $\omega_{pq}$, reduces to a single point (point identification of the other parameter), which can be consistently obtained through the OLS estimator, as stated in the corollary.

\subsection{Point and set identification in bivariate HSVARs}
\label{sec:SetIdentSVARWB}

Consider the model in Eq.s (\ref{eq:SVAR})-(\ref{eq:SVARcomp}), but with a clear evidence of a shift in the variances of the observable variables. According to the HSVAR model introduced in Section \ref{sec:HSVAR}, such shift is simply due to a structural change involving the variances of the structural shocks, leaving unaffected the structural relationships among the variables, captured by the two parameters $\alpha$ and $\beta$, i.e.
\begin{equation}
	\label{eq:Sigmai}
	\Sigma_i\equiv \left(\begin{array}{cc} \sigma_{\varepsilon,i}^2 & 0 \\ 0& \sigma_{\eta,i}^2 \end{array}\right),\hspace{0.5cm}
	i=\left\{1,2\right\},\nonumber
\end{equation}
where $i=1$ denotes the first regime (before the break), $i=2$ indicates the second regime (after the break), while the $A$ matrix remains as defined in Eq. (\ref{eq:SVARcomp}). Similarly as before, we standardize with respect to the standard deviations of the structural shocks in the first regime, and define
\begin{equation}
	\label{eq:AzeroLambda}
\resizebox{0.9\hsize}{!}{$	A_0\equiv \Sigma_1^{-1/2} A =
	\left(\begin{array}{cc} 1/\sigma_{\varepsilon,1} & -\beta/\sigma_{\varepsilon,1} \\ & \\ -\alpha/\sigma_{\eta,1} & 1/\sigma_{\eta,1} 
	\end{array}\right)\hspace{0.4cm}\text{and}\hspace{0.4cm}
	\Lambda\equiv \Sigma_1^{-1}\Sigma_2 =
	\left(\begin{array}{cc} \sigma_{\varepsilon,2}^2/\sigma_{\varepsilon,1}^2 & 0 \\ & \\ 0 & \sigma_{\eta,2}^2/\sigma_{\eta,1}^2 
	\end{array}\right)$},\nonumber
\end{equation}
and, coherently with the definitions in Sections \ref{sec:HSVAR} and \ref{sec:identification}, let $C\equiv A_0^{-1}$.

From a different perspective, we consider that the change involves only the second moments of the distribution of the observable variables in $Y_t=(p_t\,,\,q_t)^\prime$, showing the two covariance matrices
\begin{equation}
	\label{eq:Omegai}
	\Omega_i\equiv \left(\begin{array}{cc} \omega_{p,i}^2 & \omega_{pq,i} \\ \omega_{pq,i} & \omega_{q,i}^2 \end{array}\right),\hspace{0.5cm}
	i=\left\{1,2\right\},
\end{equation}
that are connected to the structural parameters through the non-linear system of equations (\ref{eq.1st})-(\ref{eq.2nd}). When the solution of the system with respect to the structural parameters is unique, the identification problem is clearly solved. As before, we define the lower triangular Cholesky factorization of $\Omega_i$ as $\Omega_{tr,i}$, $i=\left\{1,2\right\}$, whose inverses are given by
\begin{equation}
	\label{eq:CholInvi}
	\Omega_{i,tr}^{-1} = \left(
	\begin{array}{cc} \frac{1}{\omega_{p,i}} & 0 \\ & \\ -\frac{\omega_{pq,i}}{\omega_{p,i}^2}\gamma_i & \gamma_i\end{array}
	\right)
	= \left(\omega_{1,i} \,,\,\omega_{2,i}\right),\hspace{0.5cm}i=\left\{1,2\right\},
\end{equation}
where the two $(2\times 1)$ vectors $\omega_{1,i}$ and $\omega_{2,i}$ are the two columns of $\Omega_{tr,i}^{-1}$, and where 
\begin{equation}
	\label{eq:gamma}
	\gamma_i=\left(\omega_{q,i}^2-\frac{\omega_{pq,i}^2}{\omega_{p,i}^2}\right)^{-1/2},\hspace{0.5cm}i=\left\{1,2\right\}.
\end{equation}

Based on the connections between the reduced-form and the structural-form parameters highlighted in Eq.s (\ref{eq.1st})-(\ref{eq.2nd}), the identification issue can be addressed by studying the solutions of the system of equations reported in Eq. (\ref{eigen decomposition}). Theorem \ref{theo:HSVAReigen} shows that it can be addressed as an eigen-decomposition problem that, under the condition of distinct eigenvalues, proves the structural parameters contained in $(C,\Lambda)$ (or, equivalently in $A$, $\Sigma_1$ and $\Sigma_2$) to be point identified, up to permutations and sign changes of the structural equations. However, as a matter of comparison, we first report the Rigobon's condition for identification in bivariate HSVARs.

\begin{theorem}[\cite{Rigobon03} condition for point identification in bivariate HSVAR]
\label{theo:Rigobon}
	Given the HSVAR model described in Eq. (\ref{eq:SVAR}) with the two covariance matrices reported in Eq. (\ref{eq:Omegai}), 
	under Assumption \ref{ass:SignNorm}, a necessary and sufficient condition for the uniqueness of the structural parameters 
	$(C,\Lambda)$ is that 
	\begin{equation}
		\label{eq:BivPointIdentNecSuff}
		\Omega_1\neq a\Omega_2
	\end{equation}
	for any scalar $a>0$.
\end{theorem}

\medskip
 
\begin{proof} 
	See the proof of Theorem 1 in \cite{Rigobon03} and the proof of our Theorem \ref{theo:PointIdentSVARWB} in the Appendix 
	\ref{sec:ProofsBivariate}.
\end{proof} 

The following theorem, instead, restates the \citeauthor{Rigobon03}'s condition (Proposition 1, page 780, or our Theorem \ref{theo:Rigobon}) as the solution of the eigen-decomposition problem based on the observed covariance matrices in Eq. (\ref{eq:Omegai}). Moreover, it extends the results to the case of lack of identification due to equivalent eigenvalues.

\begin{theorem}[Point and set identification in bivariate HSVAR]
\label{theo:PointIdentSVARWB}
	Given the HSVAR model described in Eq. (\ref{eq:SVAR}) with the two covariance matrices reported in Eq. (\ref{eq:Omegai}), 
	the structural parameters $(C,\Lambda)$ are obtained through the eigen-decomposition problem discussed in Theorem \ref{theo:HSVAReigen}. 
	In particular, the two variance shifts of the structural shocks contained in $\Lambda$ are given by the two eigenvalues of 
	$\Omega_{1,tr}^{-1}\Omega_2 \Omega_{1,tr}^{-1\prime}$, i.e. 
	\begin{equation}
	\label{eq:eigenvalues}
		\lambda_{1,2}=\frac{\opf\oqs+\ops\oqf-2\opqf\opqs\pm\Delta}{2\left(\opf\oqf-\omega_{pq,1}^2\right)}
	\end{equation}
	with 
	\begin{equation}
	\label{eq:delta}
		\Delta = \bigg[\left(\opf\oqs-\ops\oqf\right)^2+4\left(\opf\opqs-\ops\opqf\right)\left(\oqf\opqs-\oqs\opqf\right)\bigg]^{1/2}\nonumber.
	\end{equation}
	The associated unit eigenvectors $q_1$ and $q_2$ form the columns of the orthogonal matrix $Q=(q_1,q_2)$ such that $C=\Omega_{1,tr}Q$. Under Assumption 1, and for a fixed admissible ordering of the eigenvalues, the structural parameters $(C,\Lambda)$ are uniquely determined if and only if $\Delta>0$. Equivalently, $\Omega_2\ne a\Omega_1$ for every $a>0$. If $\Delta=0$, then $\Omega_2=\lambda\Omega_1$ for some $\lambda>0$, and heteroskedasticity provides no further identification of $C$ beyond that provided by $\Omega_1$ alone. If the additional restriction $Q\in SO(2)$ is imposed, the marginal identified sets for $\alpha$ and $\beta$ are those in Theorem \ref{theo:IdSVAR}, with $\Omega=\Omega_1$ and with Assumption \ref{ass:SignRestr} imposed where applicable.
\end{theorem}

\medskip
 
\begin{proof} 
	See the Appendix \ref{sec:ProofsBivariate}.
\end{proof} 

\medskip

Theorem \ref{theo:PointIdentSVARWB} provides analytical formula to calculate the structural parameters as a function of the eigenvalues and  eigenvectors of observable matrices, i.e. the covariance matrices of the reduced form in the two regimes, $\Omega_1$ and $\Omega_2$. Furthermore, in proving the theorem, in the Appendix, we also show that the necessary and sufficient condition in Eq. (\ref{eq:BivPointIdentNecSuff}), as expected, is equivalent to say that the two eigenvalues in Eq. (\ref{eq:eigenvalues}) must be distinct (as postulated in Theorem \ref{theo:HSVAReigen}) or, put differently, the shift in the variances of the structural shocks must be different. 

However, when the quantity $\Delta=0$, from Eq. (\ref{eq:eigenvalues}) we have that the two eigenvalues $\lambda_1$ and $\lambda_2$ coincide and, according to Theorem \ref{theo:HSVAReigen}, there will be infinite (not parallel) eigenvectors $q_1$ and $q_2$ and, as a consequence, the orthogonal $Q=(q_1\,,\,q_2)$ is not unique. The model, thus, is not point identified though we are in the presence of a structural break with distinct covariance matrices $\Omega_1$ and $\Omega_2$. Put differently, the information coming from the two different covariance matrices is not sufficient for point-identifying the structural parameters of the bivariate model.

Figure \ref{fig:IdentHSVAR_ellipse} provides a schematic illustration of the geometry underlying identification. Let $M=\Omega_{1,tr}^{-1}\Omega_2(\Omega_{1,tr}^{-1})'$. The ellipse $x'M^{-1}x=1$ has principal directions given by the columns of $Q$ and semi-axis lengths $\sqrt{\lambda_1}$ and $\sqrt{\lambda_2}$. When the relative variance shifts are distinct, these directions are uniquely determined up to signs and ordering. Within the admissible model, the maintained ordering and sign normalization therefore select a unique $Q$. When the shifts coincide, $M=\lambda I_2$ and the ellipse becomes a circle, whose principal directions are undetermined. This includes the homoskedastic case, where $\lambda=1$, as well as proportional changes in all structural variances. In these cases, the second covariance matrix provides no additional identification beyond that obtained from $\Omega_1$ alone.

\begin{figure}
\caption{Identification of a bivariate HSVAR}
\label{fig:IdentHSVAR_ellipse}
\centering
\begin{tikzpicture}[scale=2.5]
    \begin{scope}[rotate=20]
        \draw[black, very thick] (0,0) ellipse (1.5 and 0.7);

        \draw[blue, ultra thick, ->] (0,0) -- (1.5,0)
            node[right] {$\sqrt{\lambda_1}\,q_1$};

        \draw[blue, ultra thick, ->] (0,0) -- (0,0.7)
            node[above left] {$\sqrt{\lambda_2}\,q_2$};
    \end{scope}
\end{tikzpicture}

\par\smallskip
{\scriptsize\raggedright
\textit{Notes}: Schematic representation of the ellipse
$x'M^{-1}x=1$, where
$M=\Omega_{1,tr}^{-1}\Omega_2(\Omega_{1,tr}^{-1})'$
and $q_1,q_2$ are its orthonormal eigenvectors.
The arrows are the semi-axis vectors
$\sqrt{\lambda_1}q_1$ and $\sqrt{\lambda_2}q_2$,
with lengths $\sqrt{\lambda_1}$ and $\sqrt{\lambda_2}$.
For distinct eigenvalues, the principal directions are determined
up to signs and ordering. If $\lambda_1=\lambda_2$, the ellipse
becomes a circle and its principal directions are undetermined.
The homoskedastic case corresponds to $\lambda_1=\lambda_2=1$.
\par}
\end{figure}

\section{Proofs} 
\label{app:Proofs}

We first introduce some notation that will be used in the following proofs. For any reduced-form parameter $\phi\in \Phi$, let $\lambda^i$ be an eigenvalue of the eigenproblem as in Definition \ref{def:eigenspace} with algebraic multiplicity $g(\lambda^i)=m_i$, with associated eigenspace $\Qli$ as in Eq. (\ref{eq:eigenspace}), containing $\qja$, the column of $Q$ associated with the $\ja$-th structural shock (shock of interest).
We have defined the zero restrictions on the vectors $\big(q_1^i,\ldots,q_{m_i}^i\big)\in\Qli$ in terms of the matrix $F_j^i(\phi)$, with $\phi$-a.s. full row rank equal to $f_j^i$. Let $\Qlip$, instead, be the linear space in $\Re^n$, of dimension $(n-m_i)$, whose elements are orthogonal to $\Qli$. A basis for this linear space is given by $\big(v_1^i,\ldots,v_{(n-m_i)}\big)$. We define $\mathcal{F}_j^{i\perp}(\phi)$ the linear subspace of $\Re^n$ that is orthogonal to the row vectors of $F_j^i(\phi)$ and to $\Qlip$. 
We let $\mathcal{H}_j(\phi)$ be the half-space in $\Re^n$ defined by the sign-normalization constraint $\big\{x\in\Re^n\,\big| (\sigma^j)^\prime x > 0\big\}$, with $\sigma^i$ being the $j$-th column of $\Omega_{1,tr}^{-1}$. As before, $\mathcal{S}^{n-1}$ indicates the unit sphere in $\Re^n$. Finally, given $k$ linearly independent vectors in $\Re^n$, $V=\big(v_1,\ldots,v_k\big)\in \Re^{n\times k}$, let $\mathcal{P}(V)$ be the linear subspace in $\Re^n$, of dimension $(n-k)$ that is orthogonal to the column vectors of $V$. 

\vspace{0.5cm}

\begin{lemma}[Diagonalization of symmetric matrices]
\label{lemma:eigendecomposition}
	Let $\Omega$ be a symmetric matrix in $\Re^{n\times n}$, then it is diagonalizable, i.e. there exists an orthogonal matrix $Q\in \On$, made of the (unit) eigenvectors of $\Omega$, such that $\Omega Q=Q D$, or equivalently, $Q^\prime \Omega Q = D$, where $D$ is diagonal. Moreover, the matrix $D$ contains the (real) eigenvalues of $\Omega$, corresponding to the eigenvectors in $Q$. 
\end{lemma}

\begin{proof}[Proof of Lemma \ref{lemma:eigendecomposition}]
	See \cite{Magnus_Neudecker_2007}, Chapter 1, Theorem 13 (page 17).
\end{proof}

\begin{lemma}
\label{lemma:multiplicity}
	In real symmetric matrices the algebraic multiplicity does correspond to the geometric multiplicity.
\end{lemma}

\begin{proof}[Proof of Lemma \ref{lemma:multiplicity}] 
	Let $A$ be and $n\times n$ symmetric matrix whose \textit{i}-th eigenvalue is represented by $\lambda_i$, with algebraic multiplicity
	equal to $1<m_i\leq n$. Then, there exists some unit-length eigenvector $p_{i1}$. Let $B=\big(p_{i1}\:\:\:C\big)$ be an orthogonal matrix.
	Then we have
	\begin{equation}
	\label{eq:ProofLemmaBAB1}
		B^\prime A B = \left(\begin{array}{cc}\lambda_i & 0\\0 & C^\prime A C \end{array}\right).\nonumber
	\end{equation}
	As the algebraic multiplicity $m_i$ is greater than one, from the characteristic polynomial we have that 
	$\big|C^\prime A C-\lambda_i I_{n-1}\big|=0$, that implies there will be some non-null vector $q$ such that 
	$\big(C^\prime A C-\lambda_i I_{n-1}\big)q=0$.
	Let $p_{i2}=Cq$. It is easy to show that $p_{i2}$ is an eigenvector of $A$. In fact, from the previous relation 
	$\big(C^\prime A C-\lambda_i I_{n-1}\big)q=0$ we have $A\,C\,q=\lambda_i C\,q$, that implies $Ap_{i2}=\lambda_ip_{i2}$.
	Moreover, by construction, $p_{i1}$ will be orthogonal to $p_{i2}$. It will be possible, thus, to define a new $B$ of the form 
	$B=\big(p_{i1}\:\:\:p_{i2}\:\:\:C\big)$ such that
	\begin{equation}
	\label{eq:ProofLemmaBAB2}
		B^\prime A B = \left(\begin{array}{ccc}\lambda_i & 0 & 0\\0 & \lambda_i & 0\\
		0 & 0 & C^\prime A C \end{array}\right).\nonumber
	\end{equation}
	and proceed as before for all the algebraic multiplicity of $\lambda_i$. The matrix $E = \big(p_{i1},\:\ldots\:,p_{im_i}\big)$ will 
	be a basis for the eigenspace of $A$ associated with $\lambda_i$, and the dimension of such space will be clearly $m_i$, 
	being all the columns of $E$ orthogonal.	
\end{proof}

\vspace{1cm}

\begin{proof}[Proof of Theorem \ref{theo:SignPerm}] Let $(C,\Lambda)$ be a solution of the equation system different from $(C^{\ast},\Lambda^{\ast})$ with non singular $C$. Then, there exists a $n \times n$ matrix $A$ such that $C^{\ast}=CA$ holds. Note that $A$ has to be an orthogonal matrix $AA^{\prime}=I$ as otherwise $\Omega_1=C^{\ast}C^{\ast \prime} = CC^{\prime}$ violates. In order for (\ref{eq.2nd}) to hold for both $(C^{\ast},\Lambda^{\ast})$ and $(C,\Lambda)$, 
\begin{equation}
\Omega_2=C^{\ast}\Lambda^{\ast} C^{\ast \prime}=CA\Lambda^{\ast} A^{\prime} C^{\prime}
\end{equation} 
must hold and, hence, $\Lambda=A \Lambda^{\ast}A'$ holds. We therefore investigate the conditions on orthogonal matrix $A$ such that $A \Lambda^{\ast}A'$ yields a diagonal matrix with non negative entries. 

Let $(\lambda_1^{\ast},\dots,\lambda_n^{\ast})$ be the diagonal elements of $\Lambda^{\ast}$ and $a_{ij}$ be $(i,j)$-element of $A$. Note that non-singularity of $\Omega_2$ implies $\lambda_k^{\ast}>0$ for all $k=1,\dots,n$. Noting that the $(i,j)$-element of $A \Lambda^{\ast}A^{\prime}$ can be expressed as $\sum_{k=1}^n \lambda_k^{\ast}a_{ik}a_{jk}$, $A$ has to satisfy
\begin{equation}
\begin{cases}
\sum_{k=1}^{n}\lambda_k^{\ast}a_{ik}a_{jk} =0, \mspace{10mu} \forall i \neq j \notag \\
\sum_{k=1}^{n}\lambda_k^{\ast}a_{ik}a_{jk} \geq 0,\mspace{10mu} \forall i = j. \notag
\end{cases}
\end{equation}
The second set of conditions imposes no further restriction, since $\Lambda^{\ast}$ is positive definite. To characterize the first set, we use the assumption that the diagonal elements of $\Lambda^{\ast}$ are all distinct. From $\Lambda=A\Lambda^{\ast}A'$ and $A'A=I_n$, we obtain $\Lambda A=A\Lambda^{\ast}$. Writing $\Lambda=\operatorname{diag}(\lambda_1,\ldots,\lambda_n)$ and $\Lambda^{\ast}=\operatorname{diag}(\lambda_1^{\ast},\ldots,\lambda_n^{\ast})$, this implies
\[
(\lambda_i-\lambda_j^{\ast})a_{ij}=0,\qquad i,j=1,\ldots,n.
\]
Since $\Lambda$ and $\Lambda^{\ast}$ have the same distinct diagonal elements, for each $j$ there is exactly one $i$ such that $\lambda_i=\lambda_j^{\ast}$. Thus, each column of $A$ has exactly one nonzero element, which equals $+1$ or $-1$ by orthogonality. The nonzero elements lie in different rows, so $A=PS$ for some $P\in\mathcal P(n)$ and $S\in\mathcal D(n)$. Consequently, $C=C^{\ast}A'=C^{\ast}SP'$ and $\Lambda=P\Lambda^{\ast}P'$, as claimed.
\end{proof}

\vspace{1cm}

\begin{proof}[Proof of Theorem \ref{theo:SVD}] 
	The proof of the theorem is trivial and completely based on the proof of the Singular Value Decomposition for square matrices, see among 
	many others \cite{Magnus_Neudecker_2007}, pages 19-20. The first point to remark is that, if we call 
	$\Omega_{tr}=\Omega_{1,tr}^{-1}\Omega_{2,tr}$, then $\Lambda^{1/2}$ contains the positive square root of the eigenvalues of 
	$\Omega = \Omega_{1,tr}^{-1}\Omega_{2,tr}\Omega_{2,tr}^\prime\Omega_{1,tr}^{-1\prime}=\Omega_{1,tr}^{-1}\Omega_{2}\Omega_{1,tr}^{-1\prime}$,
	as described in Theorem \ref{theo:HSVAReigen}. However, for symmetric and non-singular real matrices like $\Omega$, the number of identical
	eigenvalues (real and different from zero) corresponds to the number of degenerate singular values in $\Omega_{tr}$. If the entries in $\Lambda$ are all distinct, the decomposition is unique up to matched sign changes and permutations of the columns of $Q_1$ and rows of $Q_2$, with the diagonal entries of $\Lambda$ reordered accordingly.
\end{proof}

\vspace{1cm}

\begin{proof}[Proof of Theorem \ref{theo:PointIdHSVAR}] 
	The proof of the theorem takes inspiration from \cite{RWZ10RES} (proof of their Theorem 7). When the two covariance matrices are perfectly
	proportional, or even equal, then the condition in the theorem collapses to the well known condition in \cite{RWZ10RES} and \cite{BKglob20},
	and the proof is thus immediate. On the other side, if all eigenvalues are distinct, i.e. $k=n$, then the results of Theorem 
	\ref{theo:HSVAReigen} apply. Similar results apply for all eigenvalues with algebraic multiplicity equal to one, i.e. $m_i=1$. 
	According to Lemma \ref{lemma:multiplicity}, the geometric multiplicity is equal to the algebraic one, and thus, if $m_i=1$ the eigenspace
	associated to such eigenvalues will generate spaces of dimension one, each. Imposing unit length and sign normalization allows to uniquely 
	identify such vectors. Moreover, given Lemma \ref{lemma:eigendecomposition}, such vectors will be mutually orthogonal. They will constitute the columns of $Q$ associated with the eigenvalues of multiplicity one.
	
    We therefore consider an eigenvalue $\lambda_i$ with multiplicity $m_i>1$. By Lemma 2, its eigenspace has dimension $m_i$. Let
    $V_i(\phi)=(v_1^i,\ldots,v_{m_i}^i)$ be an orthonormal basis for this eigenspace. Each vector associated with $\lambda_i$ can be
    written as $q_j^i=V_i(\phi)x_j^i$, where $x_j^i\in\mathbb{R}^{m_i}$. Because the columns of $V_i(\phi)$ are orthonormal, the vectors
    $x_j^i$ must have unit length and be mutually orthogonal.
    
    We first prove \textit{sufficiency}. The first vector is subject to $f_1^i=m_i-1$ zero restrictions. Substituting $q_1^i=V_i(\phi)x_1^i$ into $F_1^i(\phi)q_1^i=0$ gives 
    \begin{equation}
        F_1^i(\phi)V_i(\phi)x_1^i=0.
    \end{equation}
    By Definition 3, the matrix $F_1^i(\phi)V_i(\phi)$ has rank $m_i-1$. Its null space is therefore one-dimensional. If $\widetilde{x}_1^i$ is a unit vector spanning that null space, then $x_1^i=\alpha\widetilde{x}_1^i$. Unit length implies $\alpha^2=1$, leaving two opposite candidates for $q_1^i$. The sign normalization selects one of them.
    
    We next determine $q_2^i$. Besides satisfying its $f_2^i=m_i-2$ zero restrictions, it must be orthogonal to $q_1^i$. In the coordinates of the eigenspace, these conditions become
    \begin{equation*}
        \begin{pmatrix}
        F_2^i(\phi)V_i(\phi)\\
        (x_1^i)'
        \end{pmatrix}x_2^i=0.
    \end{equation*}
    According to Definition 3, this matrix has rank $f_2^i+1=m_i-1$. Its null space is again one-dimensional, so unit length and sign normalization uniquely determine $q_2^i$. Proceeding in the same way, once $q_1^i,\ldots,q_{j-1}^i$ have been determined, $x_j^i$ solves
    \begin{equation*}
        M_j^i x_j^i=0,
        \qquad
        M_j^i=
        \begin{pmatrix}
        F_j^i(\phi)V_i(\phi)\\
        (x_1^i)'\\
        \vdots\\
        (x_{j-1}^i)'
        \end{pmatrix}.
    \end{equation*}
    Since $f_j^i=m_i-j$ and the restrictions are non-redundant, Definition 3 gives $\operatorname{rank}(M_j^i)=m_i-1$. Thus $q_j^i$ is uniquely determined after unit length and sign normalization. Repeating the argument in each eigenspace uniquely determines $Q$ and hence the structural parameters.
    	
	\vspace{0.3cm}

    For \textit{necessity}, suppose that the structural parameters are exactly identified, so that the set of admissible orthogonal matrices is a singleton, say $Q$. Fix an eigenspace with multiplicity $m_i>1$. The assumptions preceding the theorem imply $f_j^i\leq m_i-j$, and hence, at the admissible $Q$,
    \begin{equation*}
    \operatorname{rank}(M_j^i)
    \leq f_j^i+j-1
    \leq m_i-1.
    \end{equation*}
    Suppose that $\operatorname{rank}(M_j^i)<m_i-1$ for at least one $j$, and let $j_0$ be the largest such index. The null space of $M_{j_0}^i$ then has dimension at least two. Holding $x_1^i,\ldots,x_{j_0-1}^i$ fixed, we can choose a nonconstant continuous path of unit vectors $x_{j_0}^i(t)$ in that null space, with $x_{j_0}^i(0)=x_{j_0}^i$.
    
    For every $\ell>j_0$, the matrix $M_\ell^i$ has rank $m_i-1$ at $t=0$, by the choice of $j_0$. Full row rank is preserved under sufficiently small perturbations. We can therefore determine $x_{j_0+1}^i(t),\ldots,x_{m_i}^i(t)$ successively as continuous unit vectors in the corresponding one-dimensional null spaces, choosing their signs to agree with the original vectors at $t=0$. This produces orthonormal bases satisfying all the zero restrictions and arbitrarily close to the original basis. Under the stated assumptions on sign restrictions, they
    cannot all be excluded while retaining $Q$ as the unique admissible matrix. This contradicts exact identification.
    
    Consequently, $\operatorname{rank}(M_j^i)=m_i-1$ for every $j$. Combined with $\operatorname{rank}(M_j^i)\leq f_j^i+j-1\leq m_i-1$, this
    implies $f_j^i=m_i-j$ and $\operatorname{rank}(M_j^i)=f_j^i+j-1$. The restrictions are therefore non-redundant according to Definition 3. For an
    eigenspace with $m_i=1$, the assumed bound gives $f_1^i=0=m_i-1$, completing the proof.
\end{proof}

\vspace{1cm}

\begin{proof}[Proof of Theorem \ref{theo:SetIdZero}] 
	Let $\ja$ the shock of interest, that is associated with $\qja\in\Qli$, the eigenspace related to $\lambda^i$, with multiplicity $m_i$. According 
	to Definition \ref{def:eigenspace}, the space $\Qli$ is orthogonal to the linear space generated by all the other eigenvectors of the 
	eigenproblem, that we denote by $\Qlip$. The dimension of $\Qlip$ is $(n-m_i)$ and let the vectors $(v_1^i,\ldots,v_{(n-m_i)}^i)$
	be a possible basis. 
    The following intersections define candidate sets for the successive columns. They are not assumed to be nonempty for every choice of the preceding columns:
	\begin{equation}
	\label{eq:qRecDef}
		\begin{array}{lcl}
		q_1^i& \in & V_1^i(\phi)\equiv \mathcal{F}_1^{i\perp}(\phi)\cap \mathcal{H}_1(\phi)\cap \mathcal{S}^{n-1}\\
		q_2^i& \in & V_2^i(\phi,q_1^i)\equiv \mathcal{F}_2^{i\perp}(\phi)\cap \mathcal{H}_2(\phi)\cap \mathcal{P}(q_1^i)\cap \mathcal{S}^{n-1}\\
		q_3^i& \in & V_3^i(\phi,Q_{1:2}^i)\equiv \mathcal{F}_3^{i\perp}(\phi)\cap \mathcal{H}_3(\phi)\cap \mathcal{P}(Q_{1:2}^i)\cap 
						\mathcal{S}^{n-1}\\
		& \vdots &\\
		q_{m_i}^i& \in & V_{m_i}^i(\phi,Q_{1:m_i-1}^i)\equiv \mathcal{F}_{m_i}^{i\perp}(\phi)\cap \mathcal{H}_{m_i}(\phi)\cap 
						\mathcal{P}(Q_{1:m_i-1}^i)\cap \mathcal{S}^{n-1}
		\end{array}
	\end{equation}
	where the generic $\mathcal{F}_j^{i\perp}(\phi)$ is the linear subspace of $\Re^n$ that is orthogonal to the row vectors of $F_j^i(\phi)$ 
	and to $\Qlip$. 
    
    Fix $\phi$ at a point where the maintained assumptions hold. By feasibility, there exists at least one matrix $Q\in\mathcal Q(\phi\mid F,\varnothing)$, so the impulse response identified set is nonempty. Moreover, since each admissible $q_{j^*}^i$ has unit norm,
    $$ 
    \left|c_{lh}(\phi)'q_{j^*}^i\right|
    \leq \|c_{lh}(\phi)\|<\infty.
    $$
    Thus the identified set is bounded at every finite horizon.

    Suppose now that one of conditions (1)--(3) holds. Define
    $$
    E_{j^*}^i(\phi)=
    \begin{cases}
    F_{j^*}^{i\perp}(\phi)\cap H_{j^*}(\phi)\cap S^{n-1},
    & \text{in cases (1) and (2)},\\[3pt]
    F_{j^*}^{i\perp}(\phi)\cap P(Q_{1:k}^i(\phi))
    \cap H_{j^*}(\phi)\cap S^{n-1},
    & \text{in case (3)}.
    \end{cases}
    $$
    Every admissible column $q_{j^*}^i$ belongs to this set. Conversely, by the completion condition, every vector in $E_{j^*}^i(\phi)$ is the column of an admissible matrix. Consequently,
    $$
    IS(\phi\mid F,\varnothing)
    =
    \left\{c_{lh}(\phi)'q:q\in E_{j^*}^i(\phi)\right\}.
    $$
    Feasibility also implies that $E_{j^*}^i(\phi)$ is nonempty.

    Take any $q_a,q_b\in E_{j^*}^i(\phi)$ and define
    $$
    v(t)=(1-t)q_a+tq_b,\qquad t\in[0,1].
    $$
    The strict normalization gives
    $$
    \sigma_{j^*}(\phi)'v(t)
    =
    (1-t)\sigma_{j^*}(\phi)'q_a
    +t\sigma_{j^*}(\phi)'q_b>0.
    $$
    Hence $v(t)\neq0$ for every $t\in[0,1]$, and
    $$
    q(t)=\frac{v(t)}{\|v(t)\|}
    $$
    defines a continuous path between $q_a$ and $q_b$. The path preserves eigenspace membership, the homogeneous zero restrictions, and the strict normalization. In case (3), it also preserves orthogonality to the first $k$ exactly identified columns. Therefore $q(t)\in E_{j^*}^i(\phi)$ for every $t$.

    It follows that $E_{j^*}^i(\phi)$ is path connected. Its continuous image under $q\mapsto c_{lh}(\phi)'q$ is an interval in $\mathbb R$, and is therefore convex. Since the assumptions hold $\phi$-almost surely on the maintained support, the conclusion follows.
\end{proof}

\vspace{1cm}

\begin{proof}[Proof of Theorem \ref{theo:SetIdZeroSign}] 
	The proof builds on Lemma A.2 in \cite{GK14}. Let first $\ja=1$ and $f_1^i<m_i-1$. According to the notation introduced in Eq. 
	(\ref{eq:qRecDef}), the set of admissible $q_1^i$ becomes $V_1^i(\phi)\cap \big\{x\in \Re^n: S_1^i(\phi)\,x\geq 0\big\}$. Moreover, let
	$\qt$ be another arbitrary unit length vector satisfying the zero, sign normalization and sign restrictions. 
    The intuition for proving the result consists in observing that any weighted average
	of the two admissible vectors, with positive weights summing to one, continues to belong to the set. Then, if we define
	\begin{equation}
		\label{eq:WAV1}
		q_1^i(\delta)=\frac{\delta q_1^i+(1-\delta)\qt}{\left\|\delta q_1^i+(1-\delta)\qt\right\|},
		\hspace{1cm}\delta\in[0\,,\,1]
	\end{equation}
	it represents a connected path in $V_1^i(\phi)\cap \big\{x\in \Re^n: S_1^i(\phi)\,x\geq 0\big\}$. The strict sign normalization excludes antipodal endpoints. Hence the denominator is strictly positive for every $\delta\in[0,1]$, and the normalized interpolation defines a continuous path preserving the zero, sign-normalization and sign restrictions. By the additional completion condition, every vector along this path is admissible. 
    Any generic couple of admissible vectors, thus, can be connected by a connected path. The convexity of the impulse response, thus, immediately follows.
	We now assume that condition (2) in Theorem \ref{theo:SetIdZero} holds. Now, let $\Eji$ be the set of admissible $\qja$ satisfying 
	zero, sign normalization and sign restrictions. Let $\qj$ and $\qtj$ be two arbitrary vectors belonging to $\Eji$. 
    As before, we consider a path between these two vectors as follows
	\begin{equation}
		\label{eq:WAVj}
		\qj(\delta)=\frac{\delta\qj+(1-\delta)\qtj}{\left\|\delta \qj+(1-\delta)\qtj\right\|},
		\hspace{0.5cm}\delta\in[0\,,\,1]
	\end{equation}
	which is a continuous path on the unit sphere. The strict sign normalization excludes antipodal endpoints. Hence the denominator is strictly positive for every $\delta\in[0,1]$, and the normalized interpolation defines a continuous path preserving the zero, sign-normalization and sign restrictions. Now, the path connectedness of $\Eji$ depends on whether it is possible to obtain an admissible set of vectors 
	$Q^i(\delta)=\big(q_1^i(\delta), \ldots ,q_{m_i}^i(\delta)\big)$ whose $\ja$-th element is represented by the $\qj(\delta)$ vector.
	Conditional on a basis $\big(v_1^i, \ldots ,v_{(n-m_i)}^i\big)$ for the space $\Qlip$, the first $k$ vectors in $Q^i(\delta)$,
	$k=1,\ldots ,\ja-1$ can be obtained through the solutions of the recursive system of equations
	\begin{equation}
		\label{eq:SysCon2}
		\left(\begin{array}{c}
		F_s^i(\phi)\\v_1^i\\\vdots\\ v_{n-m_i}\\q_1^i(\delta)\\\vdots\\q_{k-1}^i(\delta)\\\qj(\delta)
		\end{array}\right)\,q_k^i(\delta)=0,\hspace{0.5cm}\delta\in[0\,,\,1]
	\end{equation}
	satisfying the further sign normalization restriction. The rank bound guarantees that the system in Eq. (\ref{eq:SysCon2}) has a nonzero solution, but does not ensure that this solution satisfies the required normalization. By the additional completion condition in Theorem \ref{theo:SetIdZeroSign}, for every $\delta\in[0,1]$ the vector $q_{j^*}^i(\delta)$ admits a completion to an admissible $Q(\delta)$.
    
    As the rank of the matrix in the system is at most $n-m_i+k+f_k^i$, that is always
	less then $n$ because $f_k^i<m_i-k$, a solution always exists. The remaining vectors for $\ja+1,\ldots,m_i$ can be obtained recursively by
	extending the system in Eq. (\ref{eq:SysCon2}). The set $\Eji$, thus, is non empty and path connected. The convexity of the impulse 
	response identified set immediately follows.
	
	Concerning point (2) of the theorem, let the zero restriction satisfy condition (3) of Theorem \ref{theo:SetIdZero}, and let 
	$\big(q_1^i,\ldots,q_k^i\big)$ be the exactly identified vectors, common to all admissible $Q(\lambda)$ matrices. As before, we chose two 
	arbitrary vectors $\qj$ and $\qtj$, both satisfying the zero, sign normalization and sign restrictions, and obtain the further vector 
	$\qtj(\delta)$ as in Eq. (\ref{eq:WAVj}). We can thus construct the set $Q^i(\delta)$, whose first $k$ columns are given by 
	$\big(q_1^i,\ldots,q_k^i\big)$. Conditional on the choice of $\delta$, for $s=k+1,\ldots,\ja-1$, we can recursively derive $q_s^i(\delta)$
	by solving 
	\begin{equation}
		\label{eq:SysCon3}
		\left(\begin{array}{c}
		F_s^i(\phi)\\v_1^i\\\vdots\\ v_{n-m_i}\\q_1^i\\\vdots\\q_{k}^i\\q_{k+1}^i(\delta)\\\vdots\\q_{s-1}^i(\delta)\\\qj(\delta)
		\end{array}\right)\,q_s^i(\delta)=0,\hspace{0.5cm}\delta\in[0\,,\,1]
	\end{equation}
	where $q_s^i(\delta)$ satisfies the sign normalization restriction, and where $\big(v_1^i, \ldots ,v_{(n-m_i)}^i\big)$ is a basis
	for the space $\Qlip$. The rank bound guarantees that the system in Eq. (\ref{eq:SysCon3}) has a nonzero solution, but does not ensure that this solution satisfies the required normalization. By the additional completion condition in Theorem \ref{theo:SetIdZeroSign}, for every $\delta\in[0,1]$ the vector $q_{j^*}^i(\delta)$ admits a completion to an admissible $Q(\delta)$.

    The system always admits a solution, being the rank of the matrix less than $n$ by the assumption on the number of
	zero restrictions on $q_{k+1},\ldots,q_{j^\ast-1}^i$, being $f_s^i<m_i-s$, for $s=1,\ldots,\ja-1$. The remaining $q_{j^\ast+1}^i(\delta),
	\ldots,q_{m_i}^i(\delta)$ vectors can be recursively derived by extending the system in Eq. (\ref{eq:SysCon3}). Once proved on how
	to derive $Q^i(\delta)$ as a function of $\delta\in[0\,,\,1]$, the set $\Eji$ is path connected, and the associated impulse response
	identified set is convex for every variable at any horizon.
\end{proof}

\section{Proofs of Theorems on bivariate SVARs and HSVARs}
\label{sec:ProofsBivariate}

\begin{proof}[Proof of Theorem \ref{theo:IdSVAR}]
Let $u=\omega_p^2$, $v=\omega_q^2$, $c=\omega_{pq}$, and
$D=uv-c^2$. Positive definiteness of $\Omega$ implies
$u>0$, $v>0$, and $D>0$. Write
\[
A=
\begin{pmatrix}
1&-\beta\\
-\alpha&1
\end{pmatrix},
\qquad
\Sigma=\operatorname{diag}(\sigma_\epsilon^2,\sigma_\eta^2),
\qquad
A_0=\Sigma^{-1/2}A,
\]
with $\sigma_\epsilon,\sigma_\eta>0$.
Since $A_0=Q'\Omega_{1,tr}^{-1}$ and
$\det(\Omega_{1,tr})>0$, the additional restriction
$\det Q=1$ is equivalent to
\[
\det A=1-\alpha\beta>0.
\]

The off-diagonal entry of $A\Omega A'$ must vanish, giving
\[
c-\alpha u-\beta v+\alpha\beta c=0.
\]
If $\alpha c-v=0$, this equality would also require
$\alpha u-c=0$, contradicting $D>0$.
Consequently,
\[
\beta=f(\alpha):=\frac{\alpha u-c}{\alpha c-v},
\qquad
1-\alpha\beta
=\frac{u\alpha^2-2c\alpha+v}{v-\alpha c}.
\]
The numerator is strictly positive because
\[
u\alpha^2-2c\alpha+v
=u\left(\alpha-\frac{c}{u}\right)^2+\frac{D}{u}>0.
\]
Thus an admissible $\alpha$ must satisfy $v-\alpha c>0$.

Conversely, take any $\alpha$ satisfying $v-\alpha c>0$
and set $\beta=f(\alpha)$.
Then $A$ is nonsingular and $A\Omega A'$ is diagonal
and positive definite. Set $\Sigma=A\Omega A'$ and
$A_0=\Sigma^{-1/2}A$. The diagonal entries of $A_0$
are positive. Define
\[
Q=(A_0\Omega_{1,tr})'.
\]
Since
\[
Q'Q=A_0\Omega A_0'=I_2
\]
and $\det A>0$, this construction gives $\det Q=1$.
Hence every such $\alpha$ is admissible.

If $c\ne0$,
\[
f'(\alpha)=-\frac{D}{(\alpha c-v)^2}<0.
\]
For $c>0$, the admissible domain is
$(-\infty,v/c)$, on which $f$ decreases from the
unattained limit $u/c$ to $-\infty$.
Therefore $\beta\in(-\infty,u/c)$.
For $c<0$, the domain is $(v/c,+\infty)$,
on which $f$ decreases from $+\infty$ to the
unattained limit $u/c$.
Therefore $\beta\in(u/c,+\infty)$.
If $c=0$, the covariance restriction becomes
$\beta=-\alpha u/v$, and
$1-\alpha\beta=1+\alpha^2u/v>0$ for every
$\alpha\in\mathbb R$. Both marginal sets are then
$\mathbb R$. This proves the first part.

Under Assumption 2, impose $\alpha\ge0$ and $\beta\le0$.
On the admissible domain, $\alpha c-v<0$, so
$\beta\le0$ is equivalent to $\alpha u-c\ge0$.
If $c>0$, this gives
$\alpha\in[c/u,v/c)$, which is nonempty because
$D>0$ implies $c/u<v/c$.
Since $f(c/u)=0$, its image is $(-\infty,0]$.
If $c<0$, every $\alpha\ge0$ satisfies both
$v-\alpha c>0$ and $\alpha u-c>0$.
As $f(0)=c/v$ and
$\lim_{\alpha\to+\infty}f(\alpha)=u/c$,
the marginal set for $\beta$ is $(u/c,c/v]$.
Finally, if $c=0$, the relation
$\beta=-\alpha u/v$ yields
$\alpha\in[0,+\infty)$ and $\beta\in(-\infty,0]$.
This proves the remaining cases.
\end{proof}

\begin{proof}[Proof of Corollary \ref{corol:IdSVAR}] 
The covariance restriction established in the proof of Theorem~9 is
\[
\omega_{pq}-\alpha\omega_p^2-\beta\omega_q^2
+\alpha\beta\omega_{pq}=0.
\]
If $\omega_{pq}\ge0$ and $\beta=0$, this gives uniquely
$\alpha=\omega_{pq}/\omega_p^2$.
If $\omega_{pq}<0$ and $\alpha=0$, it gives uniquely
$\beta=\omega_{pq}/\omega_q^2$.
In both cases, $\det A=1$, so the additional restriction
$\det Q=1$ is automatically satisfied.
The two ratios are the population OLS coefficients in the
regressions of $q$ on $p$ and of $p$ on $q$, respectively.
\end{proof}

\begin{proof}[Proof of Theorem \ref{theo:PointIdentSVARWB}] 
Let $L=\Omega_{1,tr}$, so that $\Omega_1=LL'$, and define
\[
M=L^{-1}\Omega_2L^{-T}=(m_{ij})_{i,j=1}^2.
\]
Since $\Omega_1$ and $\Omega_2$ are positive definite,
$M$ is symmetric and positive definite.
The structural representation gives
$\Omega_1=CC'$ and $\Omega_2=C\Lambda C'$.
Hence $Q=L^{-1}C$ is orthogonal and
\[
M=Q\Lambda Q'.
\]
Conversely, any orthogonal eigendecomposition of $M$ yields
a covariance decomposition with $C=LQ$.

Let $D=\det(\Omega_1)=
\omega_{p,1}^2\omega_{q,1}^2-\omega_{pq,1}^2>0$.
Direct calculation gives
\begin{align*}
\operatorname{tr}(M)
&=\frac{\omega_{p,1}^2\omega_{q,2}^2
+\omega_{p,2}^2\omega_{q,1}^2
-2\omega_{pq,1}\omega_{pq,2}}{D},\\
\det(M)
&=\frac{\omega_{p,2}^2\omega_{q,2}^2-\omega_{pq,2}^2}{D}.
\end{align*}
The characteristic equation
$\lambda^2-\operatorname{tr}(M)\lambda+\det(M)=0$
therefore gives the eigenvalues stated in the theorem.
Its discriminant satisfies
\[
\operatorname{tr}(M)^2-4\det(M)
=\frac{\Delta^2}{D^2}
=(m_{11}-m_{22})^2+4m_{12}^2,
\]
where $\Delta$ is the quantity defined in the theorem.
Consequently,
\[
\Delta=0
\quad\Longleftrightarrow\quad
M=\lambda I_2
\quad\Longleftrightarrow\quad
\Omega_2=\lambda\Omega_1
\]
for some $\lambda>0$.

If $\Delta>0$, the eigenvalues are distinct and their unit
eigenvectors are unique up to signs and permutation.
Fixing an admissible ordering removes the permutation ambiguity,
and Assumption 1 fixes the signs through the positive diagonal
entries of $A_0=Q'L^{-1}$.
Thus $(C,\Lambda)$ is uniquely determined.

If $\Delta=0$, every orthogonal $Q$ diagonalizes
$M=\lambda I_2$. The common relative variance $\lambda$
is identified, but the covariance restrictions on $C$
reduce to $CC'=\Omega_1$.
The choice $Q=I_2$ satisfies Assumption \ref{ass:SignNorm}, and strict positivity
of the diagonal entries of $A_0$ is preserved by nearby rotations.
Thus a continuum of admissible rotations remains.
If $Q\in SO(2)$ is additionally imposed, the marginal
identified sets for $\alpha$ and $\beta$ therefore follow from
Theorem \ref{theo:IdSVAR} with $\Omega=\Omega_1$, adding Assumption \ref{ass:SignRestr}
when its economic sign restrictions are imposed.
\end{proof}

\newpage


\section{Simulation experiments - Further results}
\label{app:MC}

In this appendix we show some more results on the Monte Carlo experiments introduced in Section \ref{sec:MC}.

\subsection{Eigenvalue multiplicity and frequentist properties}
\label{app:MCcoverage}

We first report additional results that complement those obtained from the Monte Carlo experiment presented in Section \ref{sec:MCcoverage}. 

Figure \ref{fig:CoverageMed} compares the true identified-set bounds with Monte Carlo averages of the posterior mean bounds. The averages move closer to the true bounds as the sample size increases, indicating a reduction in average estimation error in this experiment.

\begin{figure}[H]

\caption{True and simulated identified sets}
  \includegraphics[width=1\linewidth]{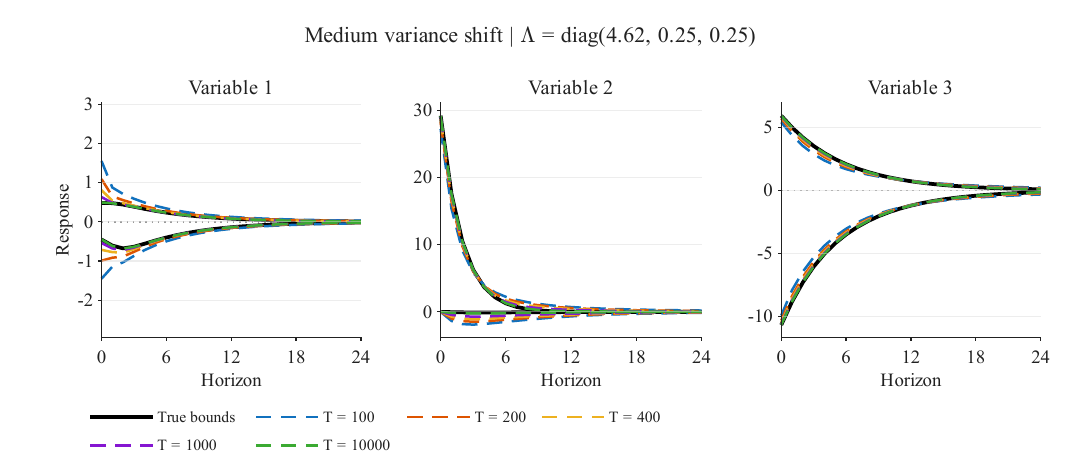}
\caption*{\scriptsize{\textit{Notes}: True identified-set bounds and Monte Carlo averages of the posterior mean bounds to a one standard deviation variation of the set identified $\varepsilon_{2t}$. Left panel: response of Variable 1; middle panel: response of Variable 2; right panel: response of Variable 3. DGP as in Eq. (\ref{eq:DGP}), with $\Lambda^{(1)} = diag(4.620,\: 0.250,\: 0.250)$. In all simulations, the number of Monte Carlo replications is $10,000$.}}

\label{fig:CoverageMed}

\end{figure} 

Here below, instead, we present two alternative Monte Carlo experiments to evaluate the coverage of our methodology. In both of them we use the same DGP used for the previous experiment, i.e.: 
\begin{equation}
    \label{eq:DGP}
    \left(\begin{array}{ccc} 0.185 & -0.005 & -0.014 \\ 0.001 & 0.036 & 0.006 \\ 0.002  & 0.050 & 0.168\end{array}\right)
    \left(\begin{array}{c}y_{1t}\\y_{2t}\\y_{3t}\end{array}\right) = 
    \left(\begin{array}{ccc} 0.075 & -0.003 & -0.003 \\ 0.003 & 0.022 & 0.005 \\ -0.011 & 0.032 & 0.144\end{array}\right)
    \left(\begin{array}{c}y_{1t-1}\\y_{2t-1}\\y_{3t-1}\end{array}\right) +
    \left(\begin{array}{c}\varepsilon_{1t}\\\varepsilon_{2t}\\\varepsilon_{3t}\end{array}\right).
\end{equation}
Moreover, we consider different values for the variance shifts: a) $\Lambda^{(2)}=diag(0.924,\: 0.250,\: 0.250)$ and b) $\Lambda^{(3)}=diag(23.110,\: 0.250,\: 0.250)$. With these two experiments we want to shed light on the extent to which the strength of the first identified shock can affect the results for the other two set identified shocks. As presented in Section \ref{sec:MCcoverage}, we repeat the experiments for $T = \{100,\:200,\:400,\:1000,\:10000\}$; the break date has been fixed to $T_B = 0.435\: T$, and the number of Monte Carlo replications is fixed to $10000$. The prior distributions are those described in Section \ref{sec:Estimation}, and in the implementation of Algorithm \ref{algo:PostBounds}, the sign normalization is obtained by imposing the diagonal elements of $A_0$ to be positive. No other sign restrictions have been included. 
The results are reported in Table \ref{tab:CoverageApp} and Figures \ref{fig:CoverageLow} and \ref{fig:CoverageHigh}. 

Table \ref{tab:CoverageApp} reports model-selection accuracy separately from coverage conditional on correct multiplicity-model selection and valid inference. Correct selection becomes nearly certain as the sample size increases, but coverage remains heterogeneous across variables and horizons. At $T=10,000$, second-variable coverage is $90\%$ at $h=1$ in all three configurations and exceeds $90\%$ at longer horizons, whereas several responses of the first and third variables have coverage around $87-88\%$. Thus, high selection accuracy does not eliminate coverage shortfalls. These finite-sample results do not, by themselves, establish a failure of asymptotic validity.

This evidence is even more clear when looking at Figures \ref{fig:CoverageLow} and \ref{fig:CoverageHigh}, where we report, for the two cases, the Monte Carlo mean of the identified set, for different sample sizes, and compare them with the true identified set. For case B, characterized by a moderate variance shift in $\varepsilon_{1t}$, Figure \ref{fig:CoverageLow} shows that, for the response of the first variable $y_{1,t}$, at least for short horizons, the mean of the Monte Carlo identified set is much larger than the true one. The situation, however, improves with larger samples. Results are very similar even in case C, as reported in Figure \ref{fig:CoverageHigh}. Overall, except in very small samples, the degree of identification of the point identified shock seems to play a moderate role on the coverage of the set identified shocks.  

\begin{table}[H]
\footnotesize

\centering
\caption{Model selection and coverage ratio: responses to set identified $\varepsilon_{2t}$ - no sign restrictions}
\begin{tabular}{lcccccccccc}

\hline\hline
\multicolumn{11}{l}{\rule{0pt}{3ex}Panel A: \hspace{0.5cm} $\Lambda^{(1)} = diag(4.622,\: 0.250,\: 0.250)$\rule[-1.4ex]{0pt}{0pt}}\\
\hline
& & \multicolumn{9}{c}{Coverage rate}\\
\cline{3-11}
& SDDR & \multicolumn{3}{c}{First variable} & \multicolumn{3}{c}{Second variable} & \multicolumn{3}{c}{Third variable}\\  
& $\%$ correct models & $h=1$ & $h=6$ & $h=12$ & $h=1$ & $h=6$ & $h=12$ & $h=1$ & $h=6$ & $h=12$ \\
\hline
T=100	&	99.8	&	1.00	&	0.99	&	0.94	&	0.73	&	1.00	&	1.00	&	0.59	&	0.76	&	0.81	\\
T=200	&	100.0	&	1.00	&	0.96	&	0.92	&	0.79	&	1.00	&	1.00	&	0.71	&	0.81	&	0.83	\\
T=400	&	100.0	&	0.99	&	0.91	&	0.88	&	0.84	&	1.00	&	1.00	&	0.79	&	0.85	&	0.85	\\
T=1000	&	100.0	&	0.94	&	0.89	&	0.87	&	0.88	&	0.98	&	1.00	&	0.82	&	0.86	&	0.87	\\
T=10000&	100.0	&	0.89	&	0.88	&	0.87	&	0.90	&	0.97	&	1.00	&	0.87	&	0.87	&	0.87	\\
\hline\hline
\multicolumn{11}{l}{\rule{0pt}{3ex}Panel B: \hspace{0.5cm} $\Lambda^{(2)} = diag(0.924, 0.250, 0.250)$\rule[-1.4ex]{0pt}{0pt}}\\
\hline
& & \multicolumn{9}{c}{Coverage rate}\\

\cline{3-11}
& SDDR & \multicolumn{3}{c}{First variable} & \multicolumn{3}{c}{Second variable} & \multicolumn{3}{c}{Third variable}\\  
& $\%$ correct models & $h=1$ & $h=6$ & $h=12$ & $h=1$ & $h=6$ & $h=12$ & $h=1$ & $h=6$ & $h=12$ \\
\hline
T=100	&	45.7	&	1.00	&	0.95	&	0.91	&	0.72	&	1.00	&	1.00	&	0.57	&	0.76	&	0.81	\\
T=200	&	93.4	&	1.00	&	0.92	&	0.89	&	0.78	&	1.00	&	1.00	&	0.70	&	0.81	&	0.83	\\
T=400	&	99.7	&	1.00	&	0.89	&	0.86	&	0.84	&	1.00	&	1.00	&	0.77	&	0.84	&	0.85	\\
T=1000	&	100.0	&	0.94	&	0.89	&	0.86	&	0.87	&	0.99	&	1.00	&	0.81	&	0.85	&	0.86	\\
T=10000&	100.0	&	0.89	&	0.88	&	0.88	&	0.90	&	0.96	&	1.00	&	0.86	&	0.87	&	0.87	\\
\hline\hline
\multicolumn{11}{l}{\rule{0pt}{3ex}Panel C: \hspace{0.5cm} $\Lambda^{(3)} = diag(23.110, 0.250, 0.250)$\rule[-1.4ex]{0pt}{0pt}}\\
\hline
& & \multicolumn{9}{c}{Coverage rate}\\

\cline{3-11}
& SDDR & \multicolumn{3}{c}{First variable} & \multicolumn{3}{c}{Second variable} & \multicolumn{3}{c}{Third variable}\\  
& $\%$ correct models & $h=1$ & $h=6$ & $h=12$ & $h=1$ & $h=6$ & $h=12$ & $h=1$ & $h=6$ & $h=12$ \\
\hline
T=100	&	100.0	&	1.00	&	0.99	&	0.97	&	0.68	&	0.99	&	1.00	&	0.52	&	0.74	&	0.80	\\
T=200	&	100.0	&	1.00	&	0.97	&	0.95	&	0.72	&	0.99	&	1.00	&	0.62	&	0.78	&	0.82	\\
T=400	&	100.0	&	0.99	&	0.92	&	0.90	&	0.83	&	0.99	&	1.00	&	0.77	&	0.85	&	0.86	\\
T=1000	&	100.0	&	0.95	&	0.89	&	0.87	&	0.87	&	0.97	&	1.00	&	0.83	&	0.86	&	0.87	\\
T=10000&	100.0	&	0.89	&	0.88	&	0.87	&	0.90	&	0.97	&	1.00	&	0.88	&	0.88	&	0.88	\\
\hline\hline

\end{tabular}
\label{tab:CoverageApp}
\caption*{\scriptsize{\textit{Notes}: In the three panels the first column reports the percentage of correct model selection by Kass--Raftery rule of thumb: $\text{log--SDDR}>-5$ $\approx$ strong evidence. The following columns report the coverage rate of the robust credible interval for the response of all the three variables to the set identified shock $\varepsilon_{2,t}$, at different horizons ($h=1$, $h=6$, $h=12$). 
Coverage rates are computed over replications in which the correct multiplicity model is selected and valid intervals are available. They therefore do not measure unconditional coverage of the complete model-selection procedure.
In all simulations, the number of Monte Carlo replications is $10,000$, and the break date is $T_B=0.435\,T$.
Panel (A) reports the results for $\Lambda^{(1)} = diag(4.622, 0.250, 0.250)$;
Panel (B) reports the results for $\Lambda^{(2)} = diag(0.924, 0.250, 0.250)$;
Panel (C) reports the results for $\Lambda^{(3)} = diag(23.110, 0.250, 0.250)$.}}
\end{table}

\begin{figure}[H]
    \centering

\caption{True and simulated identified sets}
  \includegraphics[width=1\linewidth]{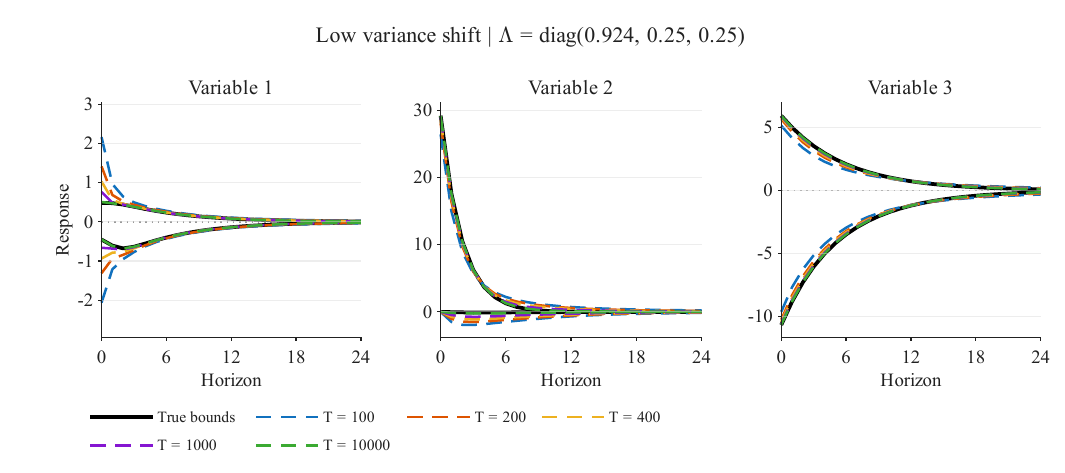}
\caption*{\scriptsize{\textit{Notes}: True identified-set bounds and Monte Carlo averages of the posterior mean bounds to a one standard deviation variation of the set identified $\varepsilon_{2t}$. Left panel: response of Variable 1; middle panel: response of Variable 2; right panel: response of Variable 3. DGP as in Eq. (\ref{eq:DGP}), with $\Lambda^{(2)} = diag(0.924,\: 0.250,\: 0.250)$. In all simulations, the number of Monte Carlo replications is $10,000$.}}

\label{fig:CoverageLow}

\end{figure} 
\begin{figure}[H]
    \centering

\caption{True and simulated identified sets}
  \includegraphics[width=1\linewidth]{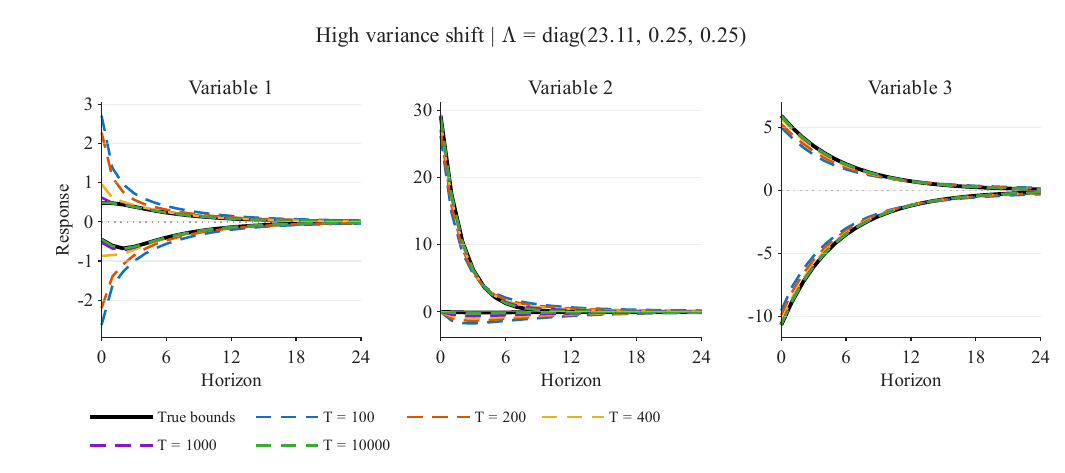}
\caption*{\scriptsize{\textit{Notes}: True identified-set bounds and Monte Carlo averages of the posterior mean bounds to a one standard deviation variation of the set identified $\varepsilon_{2t}$. Left panel: response of Variable 1; middle panel: response of Variable 2; right panel: response of Variable 3. DGP as in Eq. (\ref{eq:DGP}), with $\Lambda^{(3)} = diag(23.110,\: 0.250,\: 0.250)$. In all simulations, the number of Monte Carlo replications is $10,000$.}}

\label{fig:CoverageHigh}

\end{figure}

\subsection{Miss-specification of the true model}
\label{app:MCmiss}

In this appendix we report further results about the simulation experiments presented in Section \ref{sec:MCmiss}. 

In particular, we complement the results shown in Panel A of Table \ref{tab:Miss}, by considering different forms of miss-specification. 
For this experiment, we retain the reduced-form dynamics of the first design and use $a_{12} = 0.005$ and $a_{23} = -0.006$, with the remaining contemporaneous coefficients unchanged. The response of the second variable is cumulated.
However, there are differences in the shifts in the variances: $\Lambda^{(4)} = diag(4.622,\: 0.350,\: 0.260)$ and $\Lambda^{(6)} = diag(4.622,\: 0.700,\: 0.260)$, denoting a low and high degree of miss-specification, respectively.
For each replication, we discard $1000$ burn-in iterations and generate $5,000$ subsequent MCMC draws. Posterior summaries use up to $1000$ admissible draws, with the actual retained count recorded for each replication. Identified-set bounds are computed analytically, and robust credible regions are constructed using a $1000$-point grid for their centers. Both correctly and miss-specified models are imposed in advance, and coverage is calculated over replications yielding valid intervals.

Overall, the results presented in Panel A and C are not substantially different than those for the medium degree of miss-specification. Across the three configurations, imposing equality produces high coverage of the true point IRF.
The width ratio nevertheless increases with the separation between the true variance shifts and with the sample size. The precision cost of the equality restriction is particularly large for the second and third variables.

\begin{table}[H]
\centering
\caption{Monte Carlo coverage under misspecified equality of variance shifts}
\label{tab:Miss}
\begingroup
\fontsize{8}{9.3}\selectfont
\setlength{\tabcolsep}{2pt}
\renewcommand{\arraystretch}{1.0}
\begin{tabular*}{\linewidth}{@{\extracolsep{\fill}}rl*{9}{r}@{}}
\toprule
$T$ & Statistic & \multicolumn{3}{c}{First variable} & \multicolumn{3}{c}{Second variable} & \multicolumn{3}{c}{Third variable} \\
\cmidrule(lr){3-5}\cmidrule(lr){6-8}\cmidrule(lr){9-11}
 & & $h=1$ & $h=6$ & $h=12$ & $h=1$ & $h=6$ & $h=12$ & $h=1$ & $h=6$ & $h=12$ \\
\midrule
\multicolumn{2}{l}{True IRF} & $-0.9434$ & $-0.2782$ & $-0.0924$ & $42.3384$ & $63.6095$ & $64.9906$ & $-6.2899$ & $-2.3636$ & $-0.8083$ \\
\midrule
\multicolumn{11}{l}{\textbf{Panel A: Low misspecification}, $\Lambda=\operatorname{diag}(4.620,0.350,0.260)$} \\
$400$ & Correct (\%) & 91.6 & 90.8 & 88.0 & 80.2 & 81.8 & 83.2 & 95.3 & 91.2 & 85.1 \\
 & Equality imposed (\%) & 98.4 & 99.4 & 98.2 & 94.2 & 93.5 & 93.8 & 99.4 & 98.4 & 94.8 \\
 & $R_W$ & 1.77 & 1.83 & 1.76 & 2.13 & 2.16 & 2.16 & 2.20 & 2.06 & 1.99 \\
\addlinespace[2pt]
$1\,000$ & Correct (\%) & 92.5 & 90.7 & 88.9 & 86.6 & 88.1 & 88.5 & 93.3 & 92.9 & 89.1 \\
 & Equality imposed (\%) & 98.6 & 98.1 & 96.9 & 99.1 & 97.5 & 97.3 & 100.0 & 99.9 & 98.0 \\
 & $R_W$ & 2.28 & 2.35 & 2.13 & 3.55 & 3.26 & 3.16 & 3.23 & 2.87 & 2.59 \\
\addlinespace[2pt]
$10\,000$ & Correct (\%) & 89.2 & 89.6 & 89.5 & 89.2 & 89.2 & 89.5 & 89.3 & 89.7 & 89.8 \\
 & Equality imposed (\%) & 100.0 & 100.0 & 99.9 & 100.0 & 100.0 & 100.0 & 100.0 & 100.0 & 100.0 \\
 & $R_W$ & 6.58 & 6.67 & 5.34 & 24.67 & 15.87 & 13.58 & 12.43 & 10.15 & 7.80 \\
\midrule
\multicolumn{11}{l}{\textbf{Panel B: Medium misspecification}, $\Lambda=\operatorname{diag}(4.620,0.480,0.260)$} \\
$400$ & Correct (\%) & 90.8 & 90.2 & 88.8 & 84.0 & 85.2 & 85.9 & 91.9 & 87.6 & 83.9 \\
 & Equality imposed (\%) & 98.7 & 99.5 & 98.4 & 99.1 & 97.5 & 97.2 & 99.7 & 98.9 & 95.1 \\
 & $R_W$ & 2.02 & 2.20 & 1.97 & 4.40 & 3.59 & 3.34 & 4.09 & 3.14 & 2.62 \\
\addlinespace[2pt]
$1\,000$ & Correct (\%) & 90.0 & 90.1 & 89.2 & 87.3 & 87.5 & 87.8 & 90.5 & 88.6 & 86.7 \\
 & Equality imposed (\%) & 98.9 & 98.4 & 97.2 & 100.0 & 99.6 & 99.4 & 100.0 & 99.9 & 98.2 \\
 & $R_W$ & 2.77 & 2.99 & 2.46 & 8.84 & 6.05 & 5.34 & 7.27 & 5.00 & 3.64 \\
\addlinespace[2pt]
$10\,000$ & Correct (\%) & 89.2 & 89.8 & 89.9 & 89.5 & 89.4 & 89.2 & 89.8 & 90.1 & 90.1 \\
 & Equality imposed (\%) & 100.0 & 100.0 & 99.9 & 100.0 & 100.0 & 100.0 & 100.0 & 100.0 & 100.0 \\
 & $R_W$ & 7.07 & 7.37 & 5.61 & 28.68 & 17.32 & 14.52 & 23.73 & 14.39 & 9.00 \\
\midrule
\multicolumn{11}{l}{\textbf{Panel C: High misspecification}, $\Lambda=\operatorname{diag}(4.620,0.700,0.260)$} \\
$400$ & Correct (\%) & 89.7 & 90.0 & 89.2 & 85.2 & 85.0 & 85.6 & 88.8 & 85.4 & 84.2 \\
 & Equality imposed (\%) & 99.0 & 99.5 & 98.4 & 99.8 & 99.4 & 98.9 & 99.8 & 99.2 & 95.5 \\
 & $R_W$ & 2.14 & 2.38 & 2.05 & 6.49 & 4.59 & 4.10 & 6.48 & 3.97 & 2.92 \\
\addlinespace[2pt]
$1\,000$ & Correct (\%) & 89.6 & 89.8 & 89.5 & 88.1 & 87.9 & 87.9 & 89.7 & 87.9 & 86.6 \\
 & Equality imposed (\%) & 99.3 & 98.7 & 97.6 & 100.0 & 100.0 & 99.9 & 100.0 & 99.9 & 98.5 \\
 & $R_W$ & 2.88 & 3.11 & 2.51 & 10.15 & 6.65 & 5.76 & 10.41 & 5.79 & 3.80 \\
\addlinespace[2pt]
$10\,000$ & Correct (\%) & 89.4 & 89.9 & 89.9 & 89.6 & 89.5 & 89.2 & 89.8 & 90.1 & 90.1 \\
 & Equality imposed (\%) & 100.0 & 100.0 & 100.0 & 100.0 & 100.0 & 100.0 & 100.0 & 100.0 & 100.0 \\
 & $R_W$ & 7.31 & 7.62 & 5.68 & 30.80 & 18.48 & 15.31 & 32.24 & 16.11 & 9.24 \\
\bottomrule
\end{tabular*}
\par\smallskip
{\raggedright\noindent
\textit{Notes}: Correct denotes inference with distinct shifts; Equality imposed denotes robust inference under the erroneous restriction $\lambda_2=\lambda_3$.
Coverage entries are percentages of valid intervals containing the true point IRF; posterior credibility is 90\%.
$R_W$ is the ratio of mean interval widths (Equality imposed/Correct), computed on replications valid under both specifications. True IRFs are common to all panels and sample sizes.
$N=10\,000$ replications are used for each DGP and sample size.
All reported intervals are valid.
Responses are to shock 2, normalized to one first-regime standard deviation; $y_2$ is cumulated. No economic sign restrictions are imposed.
\par}
\endgroup
\end{table}

\clearpage


\section{Empirical application -- Further details and results}
\label{app:Empirics}
\subsection{Data}
The data entering the VAR model in Section \ref{sec:empirics} are the following:
\begin{itemize}
\item $\Delta prod_t$ is the percent change in global crude oil production.
\item The index of real economic activity, $rea_t$, is the logarithm of a detrended index based on dry cargo ocean shipping rates. It is used to proxy monthly changes in the world demand for industrial commodities, including crude oil.
\item The real price of crude oil, $rpo_t$, is the logarithm of the U.S. refiner acquisition cost of imported crude oil, deflated by the U.S. consumer price index.
\end{itemize}

The data are taken from \citetalias{LMNS20} and are constructed to reproduce the three variables of the oil market model of \citet{kilian200AER}. Data are monthly and span from February 1973 to December 2006 ($T = 407$). The time series included in the VAR and the plotted residuals are shown in Figures \ref{fig:DataPlot} and \ref{fig:ResidPlot}.

\begin{figure}[ht!]
 \caption{Data used in the SVAR model for the global market of crude oil (February 1973-December 2006)}
  \includegraphics[width=1\linewidth]{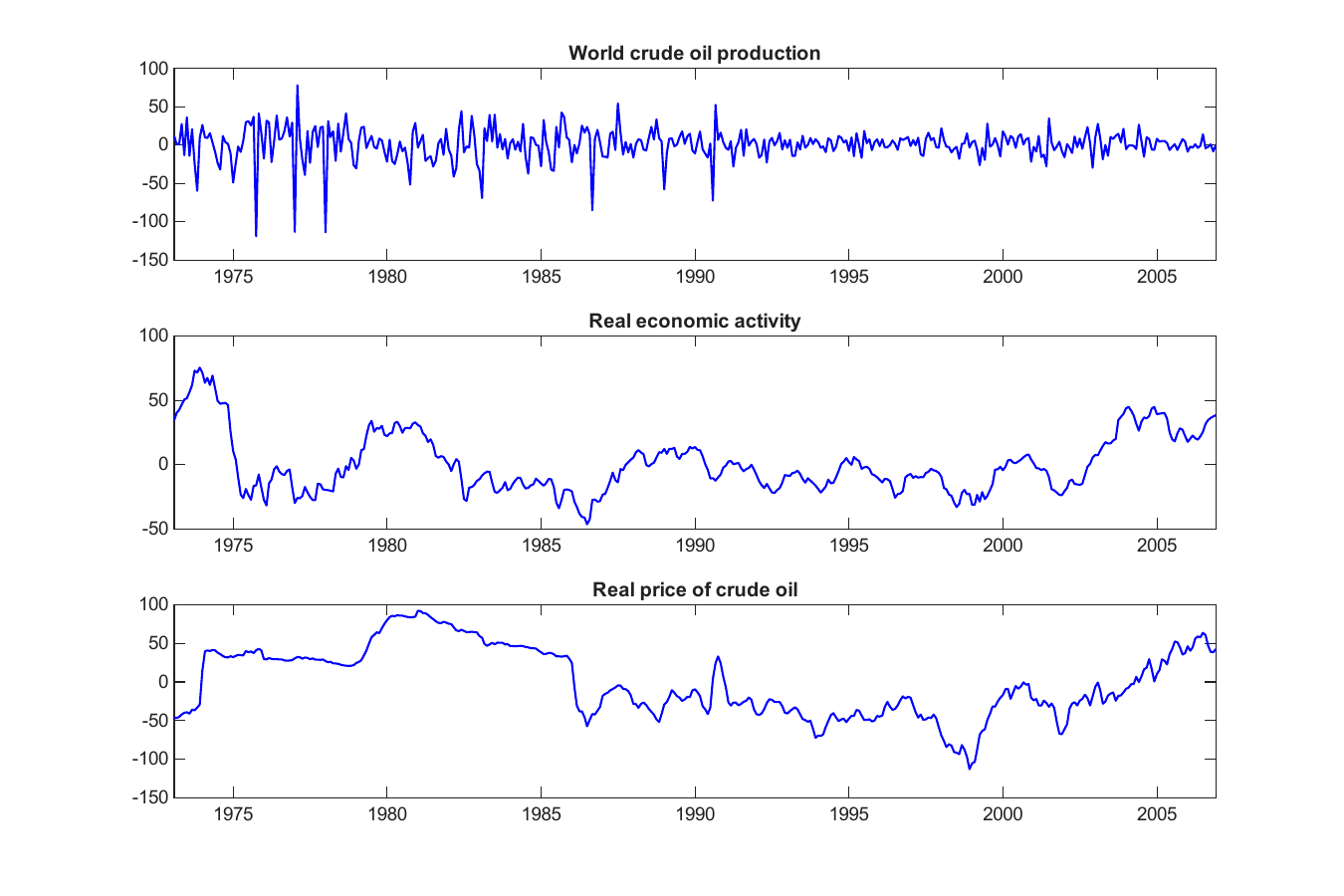}%
\label{fig:DataPlot}
\end{figure}

\begin{figure}[ht!]
\caption{Reduced form residuals and break date}
  \includegraphics[width=1\linewidth]{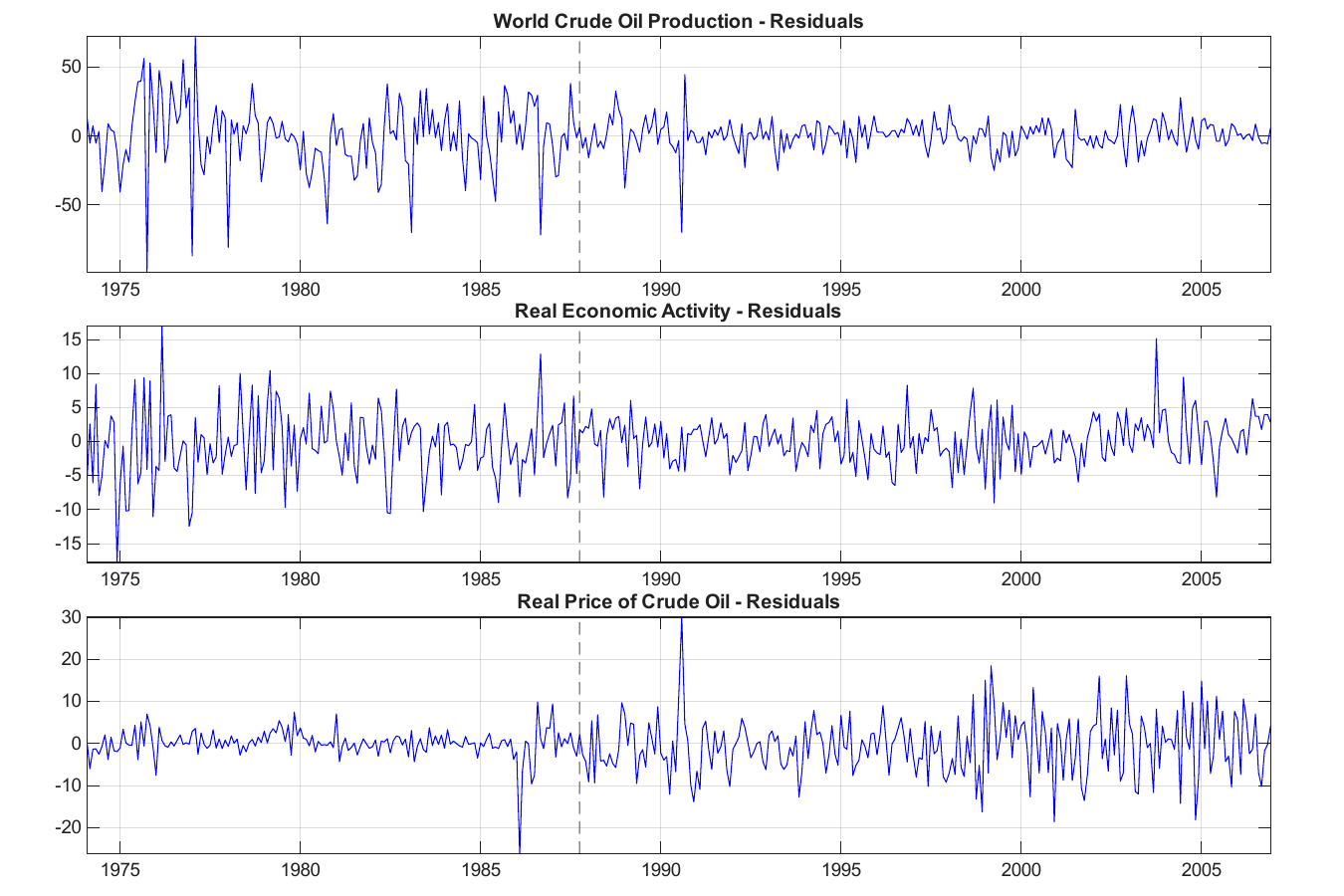}%
\caption*{\scriptsize\textit{Notes}: Reduced form residuals and time of the break (1st October 1987). Monthly data.}
\label{fig:ResidPlot}
\end{figure}

\subsection{Model and identification schemes}
We write the relationship between structural innovations, $\varepsilon_t$ and reduced-form errors, $u_t = C \varepsilon_t$ as follows:
 \begin{equation}
	\label{eq:KilianOil}
	\left(\begin{array}{l}
	u_t^{\Delta prod}\\u_t^{rea}\\u_t^{rpo}\end{array}\right)=\left(\begin{array}{ccc}
	c_{11} & c_{12} & c_{13}\\%
	c_{21} & c_{22} & c_{23}\\%
	c_{31} & c_{32} & c_{33}\end{array}\right)%
\left(\begin{array}{l}
	\e_{t}^{oil\ supply\ shock}\\\e_t^{aggregate \ demand \ shock}\\\e_t^{oil-specific \ demand \ shock}\end{array}\right)
	\end{equation}
where $C$ corresponds to the impact response matrix ($A_0^{-1}$). We consider models based on different identifying restrictions:
\begin{itemize}
    \item $\mathcal{M}_0$, a recursively identified SVAR model with $c_{12} = c_{13} = c_{23} = 0$;
    \item $\mathcal{M}_1$, a standard HSVAR model identified through heteroskedasticity, exploiting changes in volatility and assuming distinct eigenvalues;
    \item $\mathcal{M}_2$, an HSVAR model that imposes eigenvalue multiplicity and exploits static and dynamic sign restrictions;
    \item $\mathcal{M}_3$, an HSVAR model that allows for eigenvalue multiplicity and imposes one exclusion restriction on $\mathbf{A}_0$ ($a_{12}=0$).
\end{itemize}

\noindent Model $\mathcal{M}_0$ is the recursively identified SVAR model of \cite{kilian200AER} used as a benchmark against which we compare results from HSVAR models. Three exclusion restrictions ($c_{12}=c_{13}=c_{23}=0$) allow to exactly identify an oil supply shock and two demand shocks (i.e. aggregate and oil-specific demand shocks). Oil supply shocks represent innovations to the current physical availability of crude oil. Aggregate demand shocks capture unexpected changes of the demand for all industrial commodities driven by fluctuations in the global business cycle, while oil-specific demand shocks represent shifts in the precautionary demand for crude oil triggered by concerns about the future availability of supplies.

The HSVAR model $\mathcal{M}_1$ exploits changes in volatility to identify structural shocks, without imposing additional exclusion or sign restrictions. We assume the eigenvalues to be all distinct, although the log-SDDR provide evidence in favor of the existence of eigenvalue multiplicity.

Model $\mathcal{M}_2$ is the one used in Section \ref{sec:empirics} of the paper. While specification $\mathcal{M}_2$ allows to point identify a single structural shock, because of eigenvalue multiplicity the HSVAR model remains set identified. In this case, Theorem \ref{theo:PointIdHSVAR} shows that exact identification of the HSVAR model can be achieved with a single exclusion restriction. 

In Model $\mathcal{M}_3$ we add one exclusion restriction on the matrix of contemporaneous correlation $\mathbf{A}_0$ on $a_{12}=0$. This is in line with the oil supply equation in \citet{baumeister2019structural}, and implies that real economic activity does not enter the contemporaneous oil-supply equation directly.
Compared to the recursively identified model, $\mathcal{M}_0$, when information coming from a volatility shift is exploited the identification scheme is less demanding. In fact, in model $\mathcal{M}_3$ exact identification is achieved -- at least in statistical sense -- combining the shift in the volatility of one of the structural shocks with a single zero restriction.

As for the lag order of the VAR, we select 12 lags. We note that \citet{LN2014JaE} and \citet{LukAdvEcnm} rely on a VAR of order 3, while \citet{kilian200AER} stresses that a VAR of order 24 is necessary to capture the long price cycles of crude oil and to accurately estimate impulse responses in global oil market models. The selected lag order can thus be seen as a compromise between these two approaches.

\subsection{Additional results}
\noindent\textit{Alternative identification schemes.} Figures \ref{fig:IRF_M0} and \ref{fig:IRF_M1} report the structural impulse response functions (IRFs) for models $\mathcal{M}_0$ and $\mathcal{M}_1$, respectively.

Model $\mathcal{M}_1$ builds on the standard identification-through-heteroskedasticity approach, implicitly assuming that all eigenvalues are distinct. However, as the log SDDR results suggest, changes in volatility alone may not contain sufficient information to exactly identify all structural shocks. Note that we addressed this issue in model $\mathcal{M}_2$ by imposing additional identifying restrictions. See the Empirical section of the paper.

Figure \ref{fig:IRF_M0} displays the impulse responses for the recursively identified model of \citet{kilian200AER}, denoted $\mathcal{M}_0$, along with 68\% highest posterior density (HPD) credible regions.\footnote{Estimation relies on a noninformative improper Jeffreys prior, allowing reduced-form parameters to be drawn from a Normal-inverse-Wishart posterior.} The results are fully consistent with the original findings in \citet{kilian200AER}. An oil supply shock causes an immediate and persistent decline in global oil production, a fall in real economic activity, and a transitory rise in the real price of crude oil that peaks about six months after the shock. The 68\% HPD region for the price response excludes zero only for the first eight months. 

An aggregate demand shock initially raises global oil production, but the posterior mean response subsequently becomes negative, reaching approximately $-4$ at $h=24$. The $68\%$ HPD region lies below zero at longer horizons. Real economic activity and the real price of oil exhibit positive posterior mean responses. However, the lower HPD boundary for activity approaches zero at the final horizon, while the price HPD region touches or includes zero at some early horizons.

Model $\mathcal{M}_1$ assumes distinct relative variance shifts. The first, second, and third columns of Figure \ref{fig:IRF_M1} report responses to oil supply, aggregate demand, and oil-specific demand shocks, respectively. The responses to the latter resemble those in $\mathcal{M}_0$. Following an oil supply disruption, the $68\%$ HPD interval for the real oil-price response includes zero at all reported horizons. Real economic activity exhibits a short-run contraction: its posterior mean reaches approximately $-1.8$ at $h=3$, where the $68\%$ HPD interval excludes zero. 

The structural impulse responses implied by Model $\mathcal{M}_3$ are shown in Figure \ref{fig:IRF_M3}. The model is exactly identified by adding one exclusion restriction on $\mathbf{A}_0$ ($a_{12}=0$) to the information provided by a distinct eigenvalue. The results are broadly consistent with those of the previous specifications, although in this case the HPD interval for the response of oil production to an aggregate demand shock always includes zero.

\begin{figure}[ht!]
\caption{Impulse response functions $\mathcal{M}_{0}$}
  \includegraphics[width=.33\linewidth]{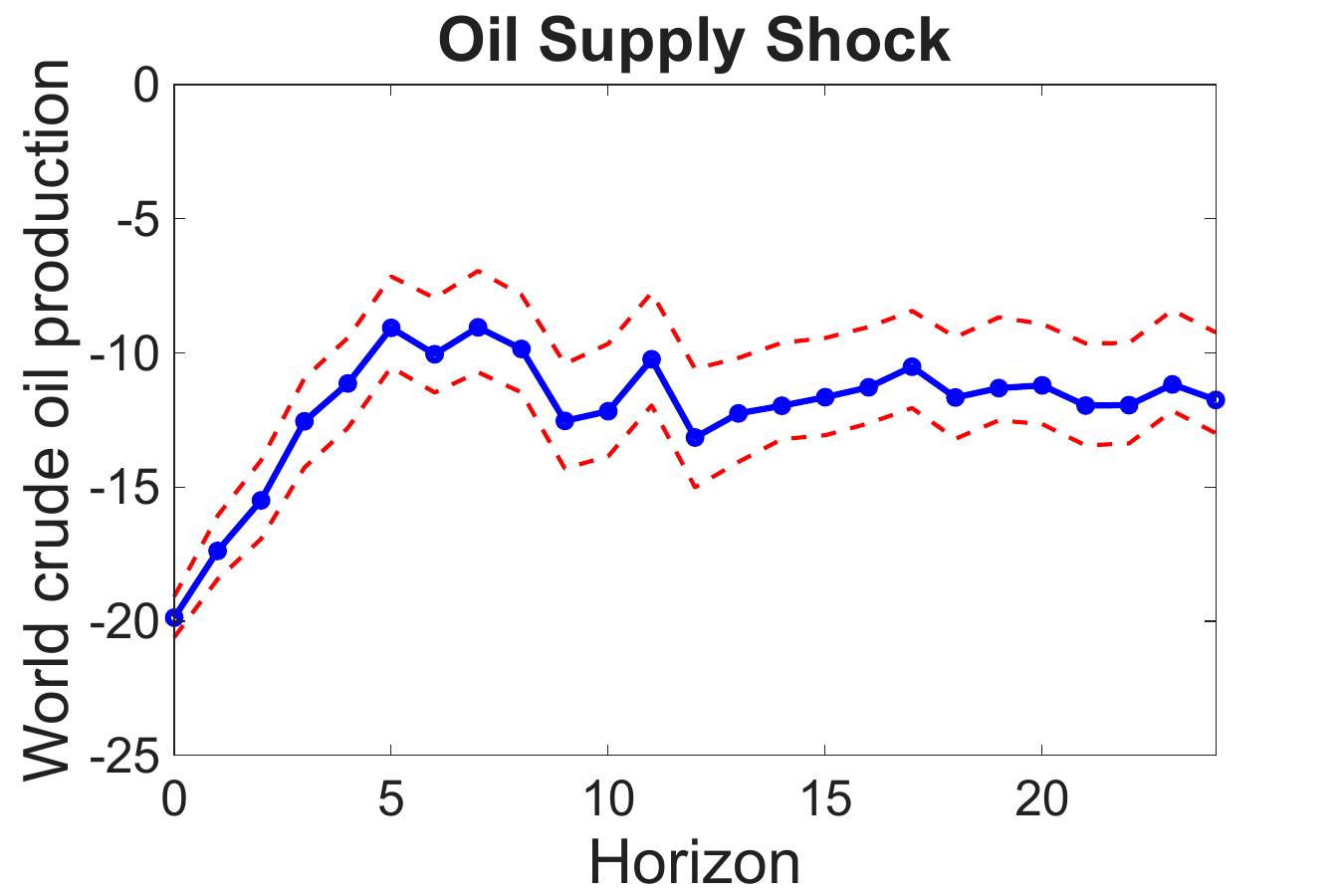}%
  \includegraphics[width=.33\linewidth]{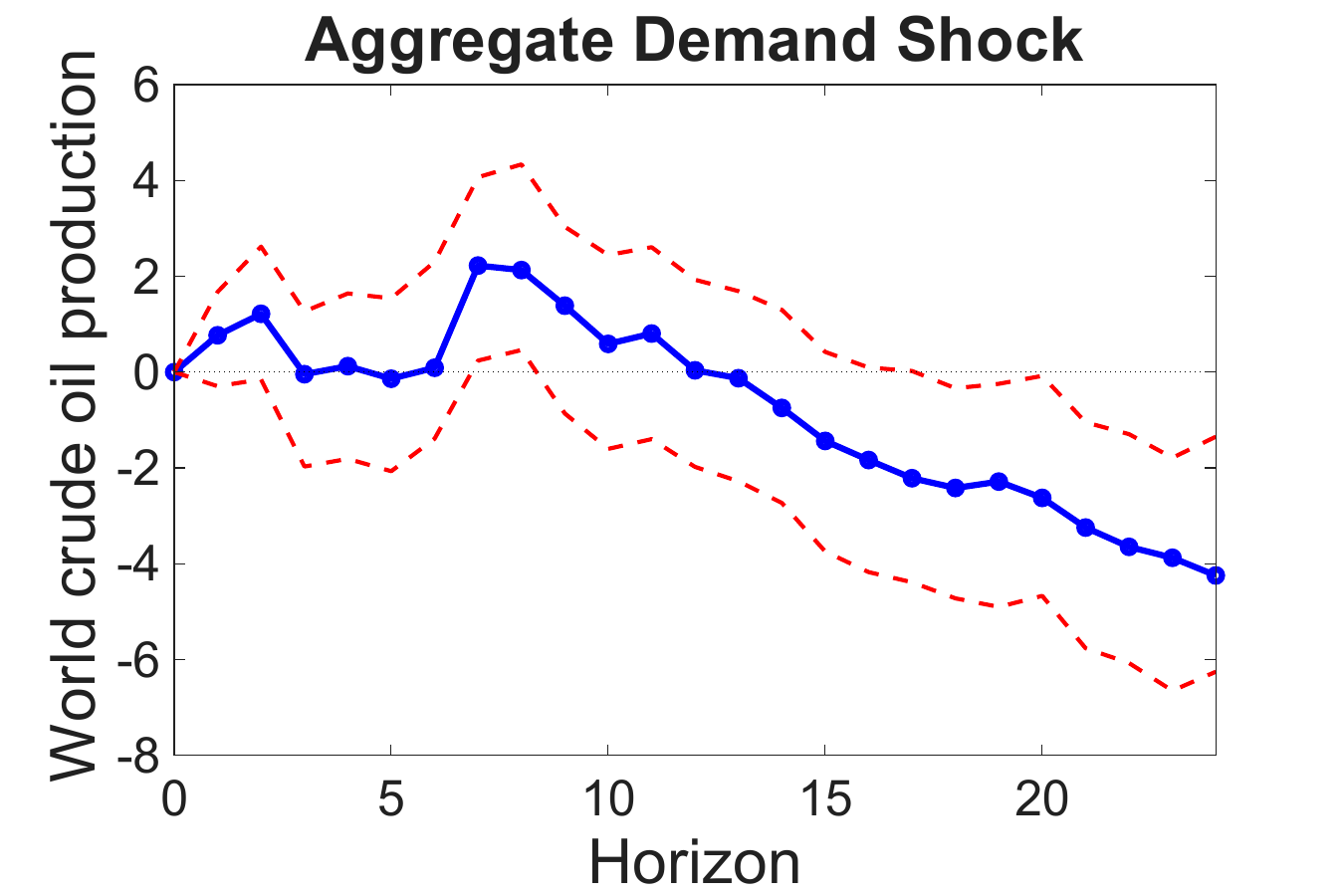}
  \includegraphics[width=.33\linewidth]{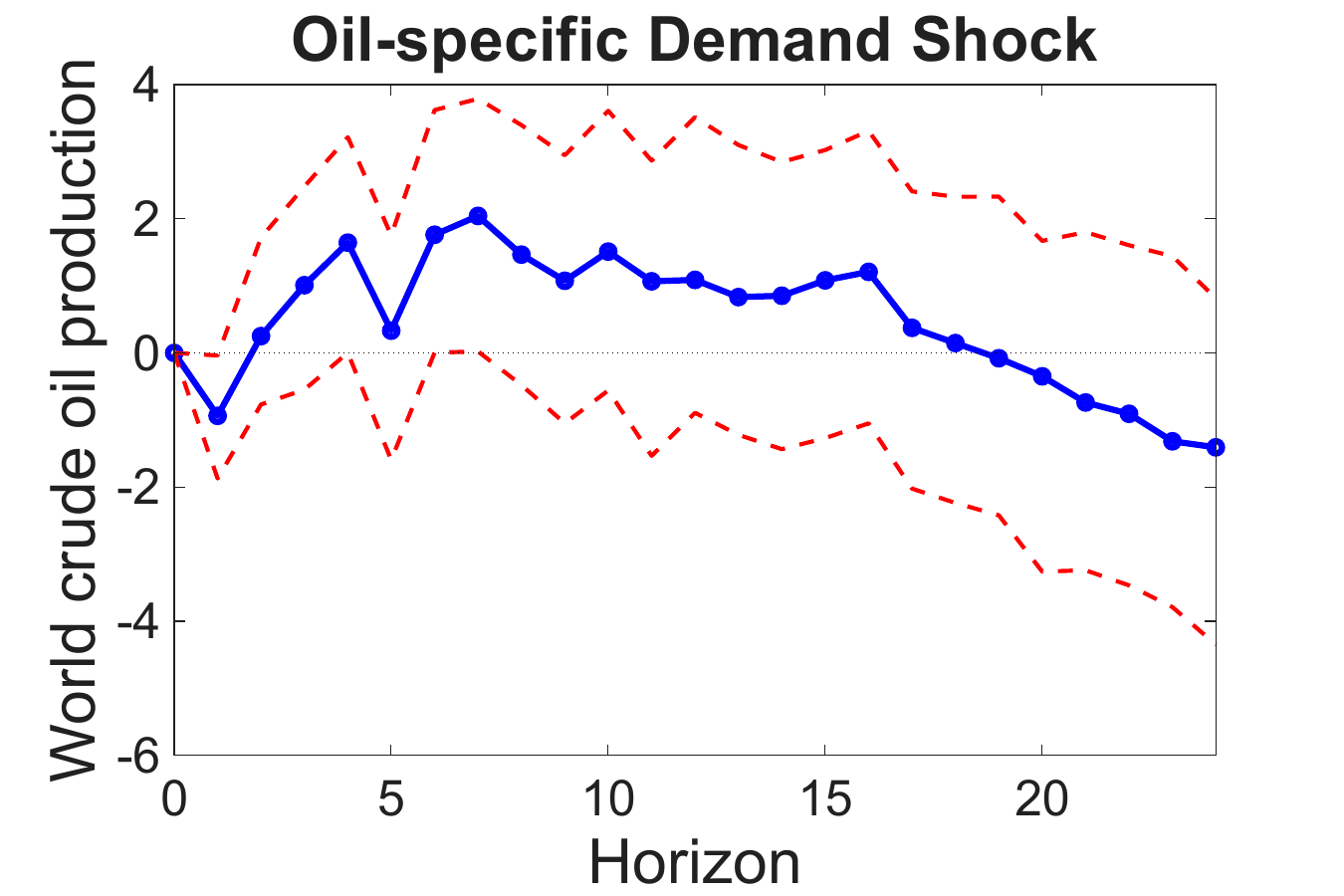}\\
  \includegraphics[width=.33\linewidth]{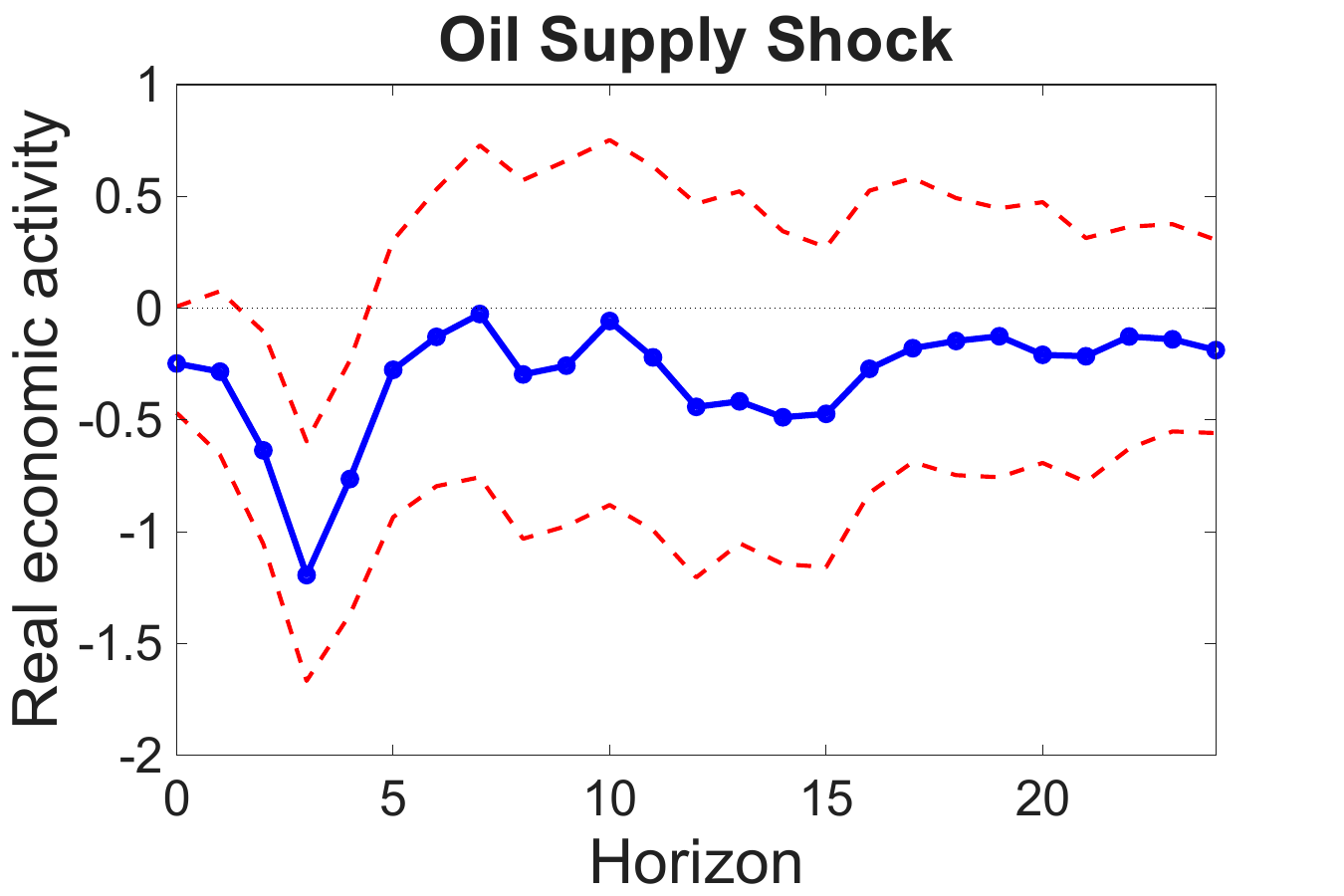}%
  \includegraphics[width=.33\linewidth]{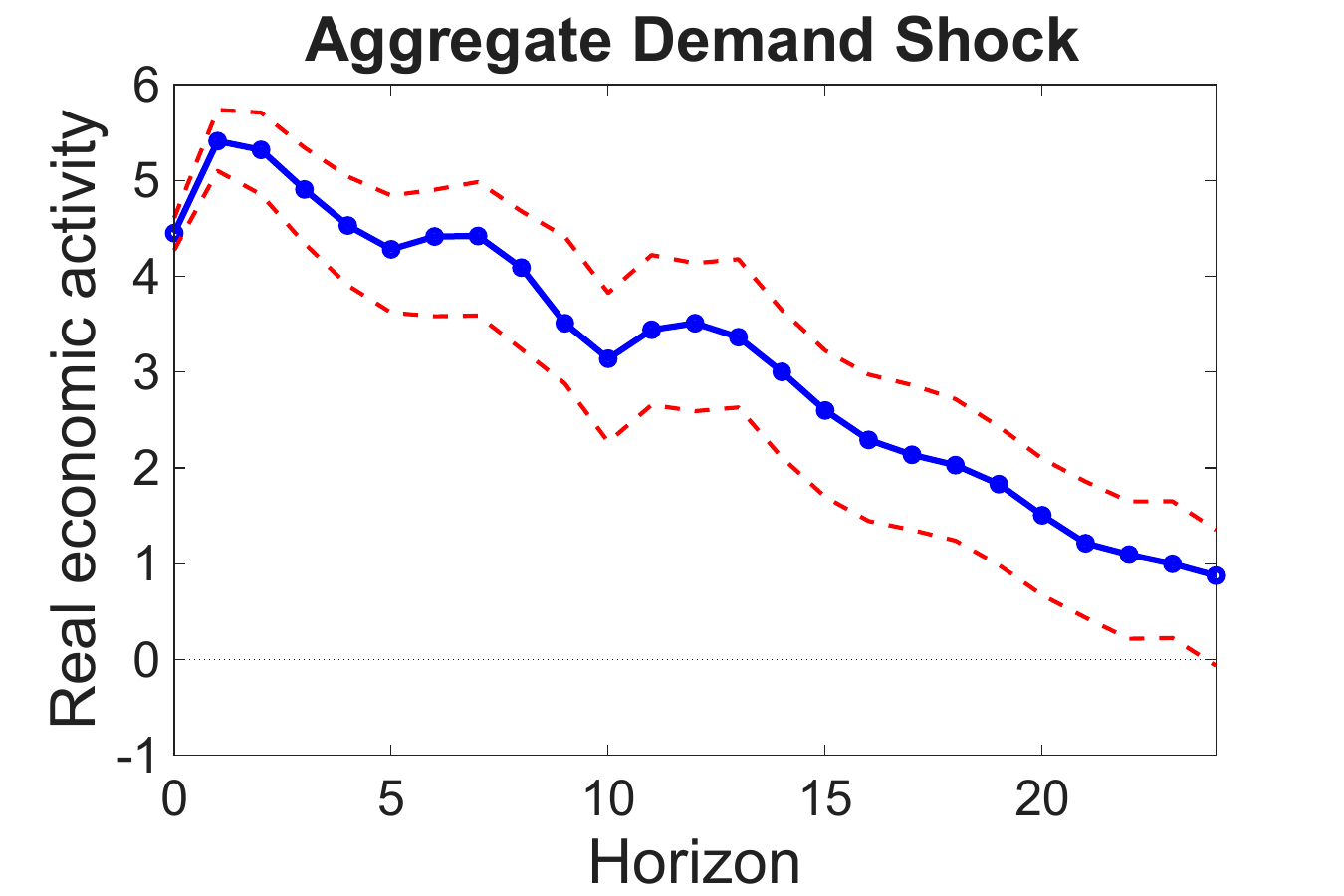}%
  \includegraphics[width=.33\linewidth]{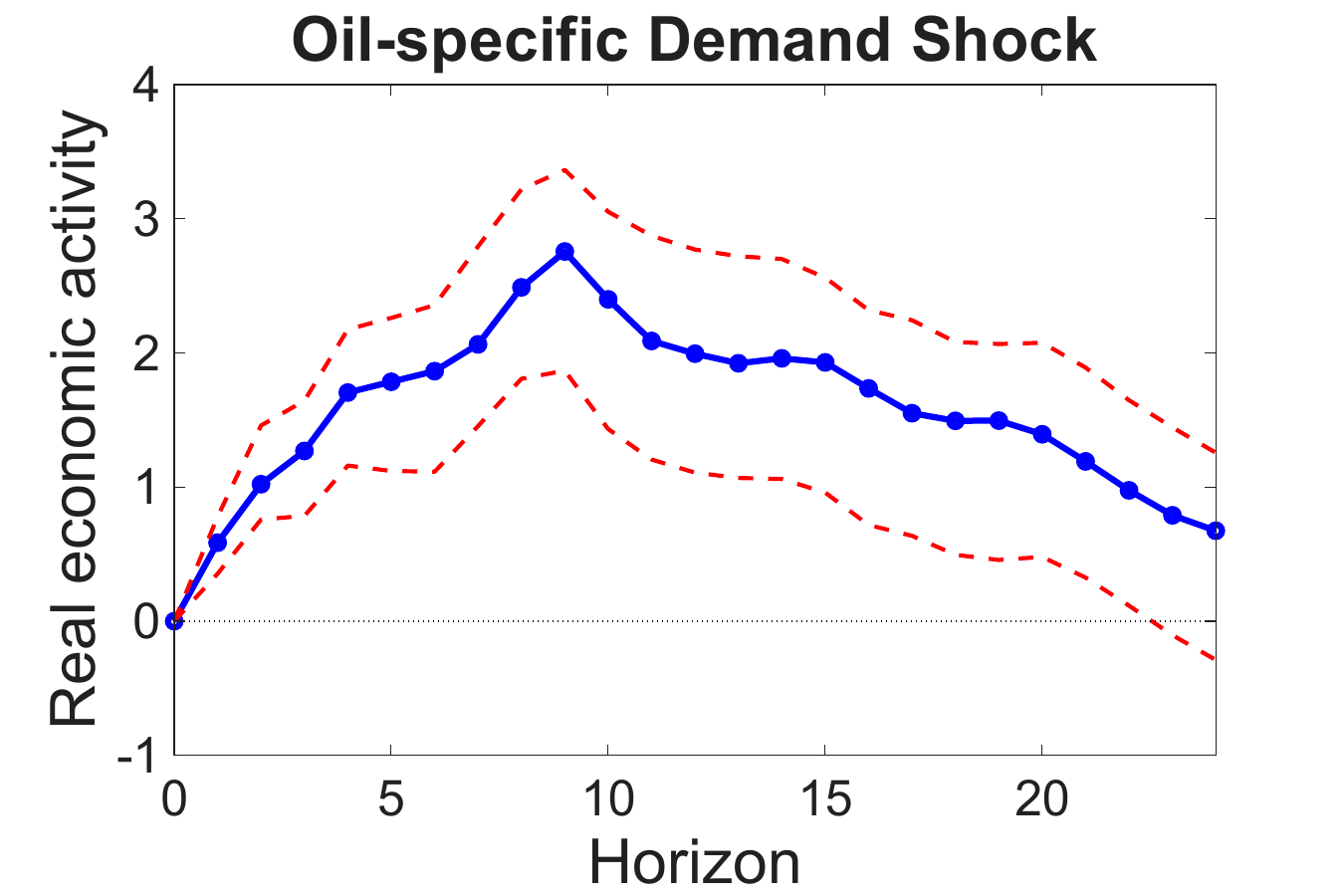}\\
  \includegraphics[width=.33\linewidth]{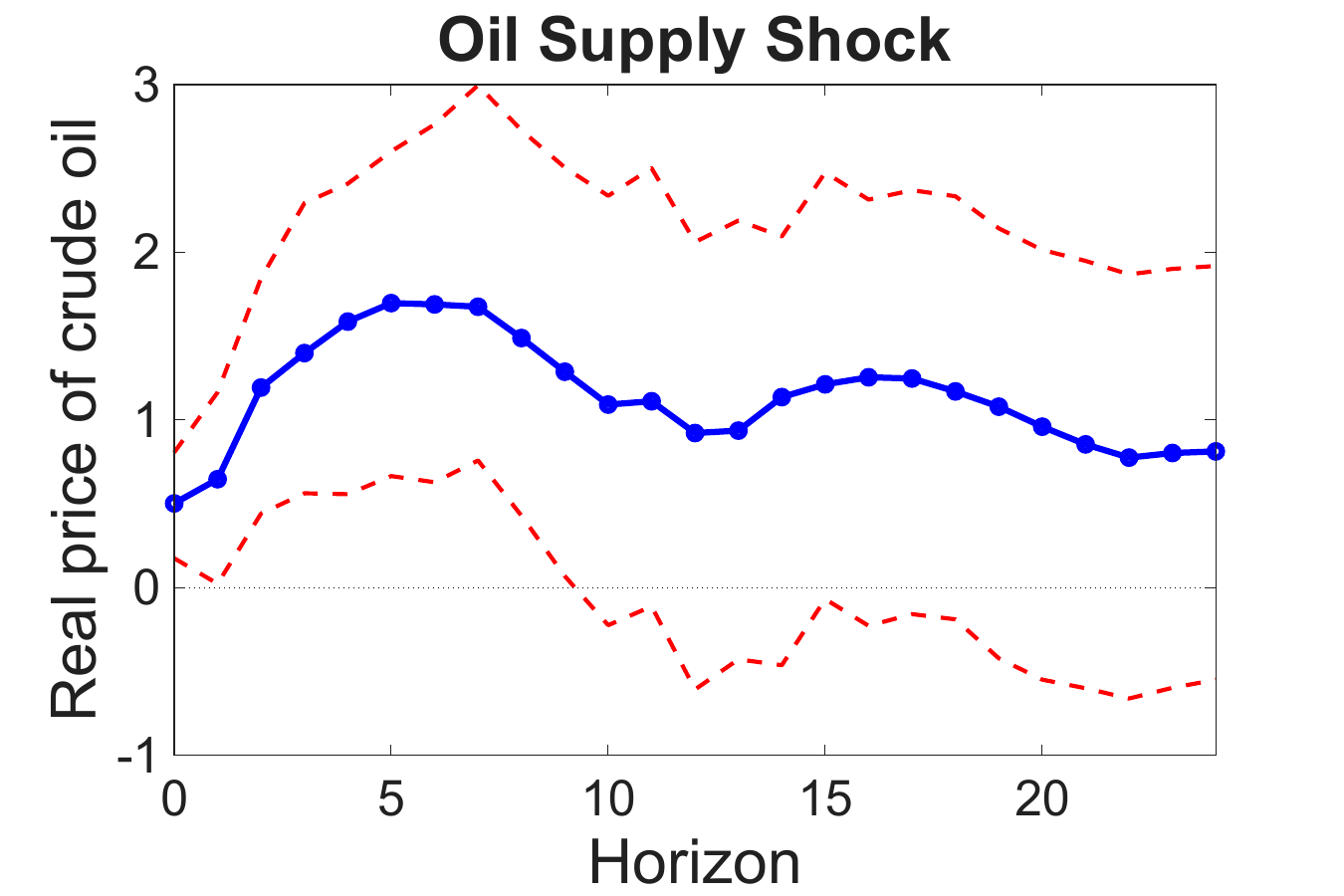}%
  \includegraphics[width=.33\linewidth]{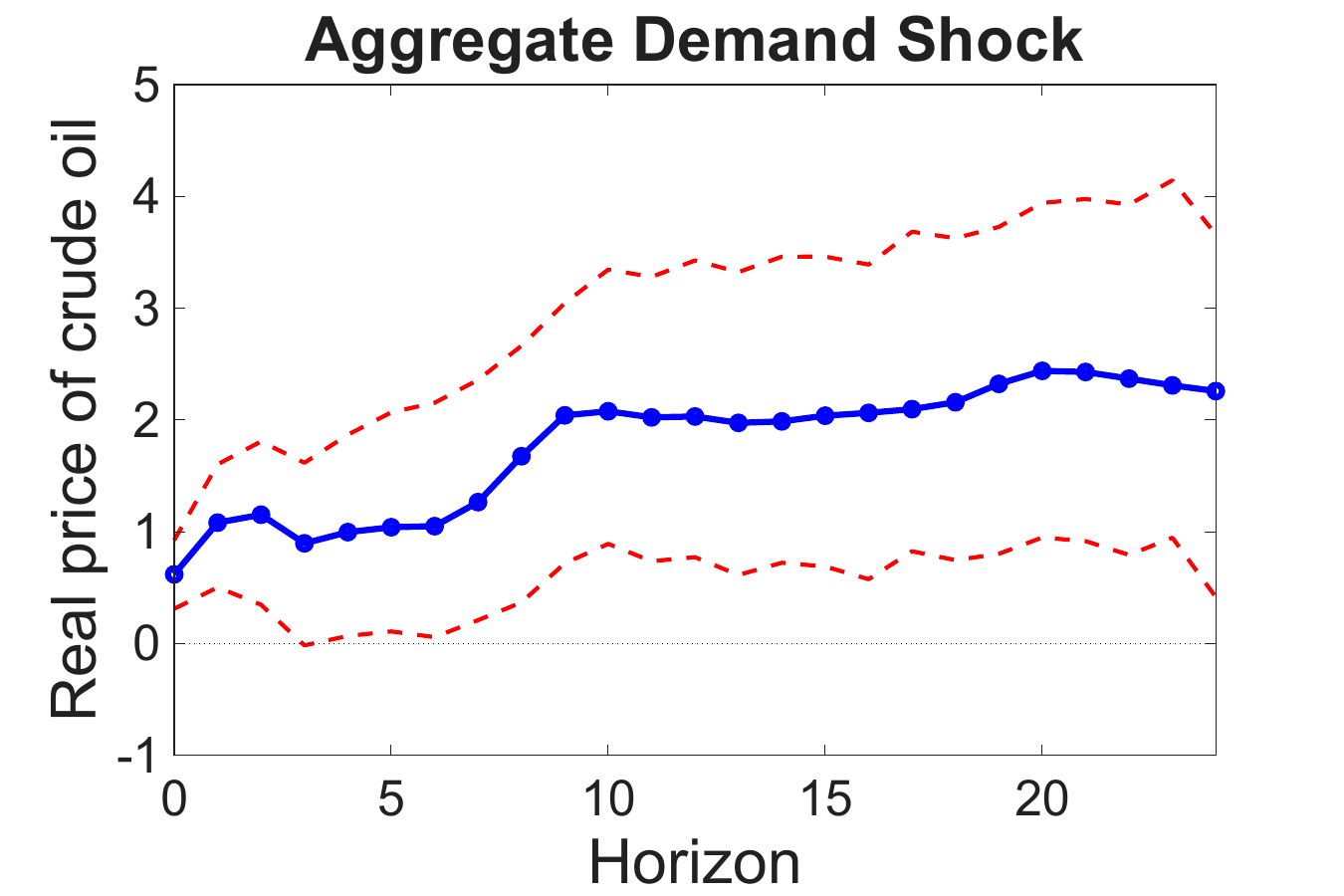}%
  \includegraphics[width=.33\linewidth]{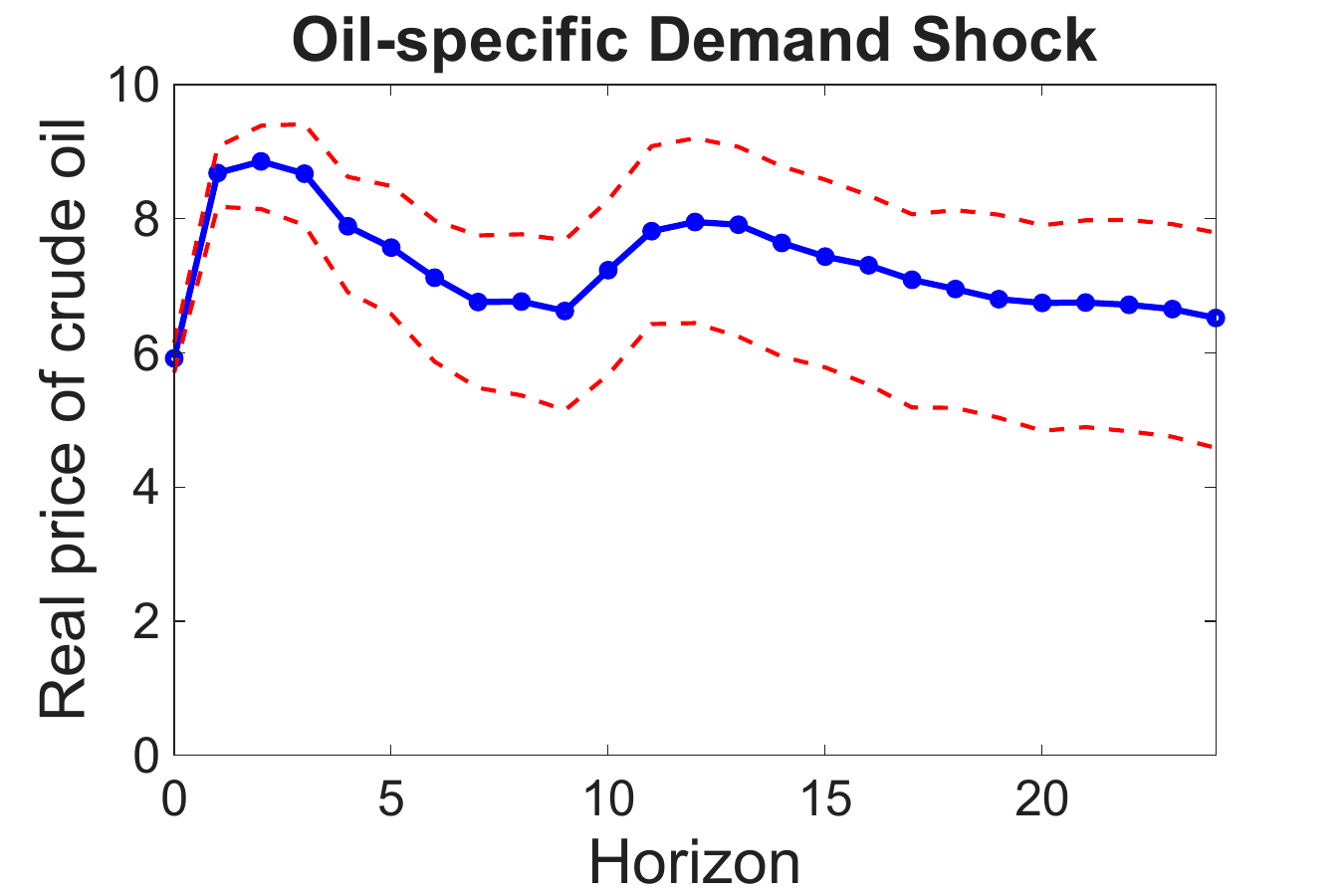}
\caption*{\scriptsize\textit{Notes}: the blue line with dots represents the posterior mean response, the dashed red lines identify upper and lower bounds of the highest posterior density region with credibility 68\%. Recursive identification imposing $c_{12}=c_{13}=c_{23}=0$. The model is point-identified}
\label{fig:IRF_M0}
\end{figure}

\begin{figure}[ht!]
\caption{Impulse response functions $\mathcal{M}_{1}$}

  \includegraphics[width=.33\linewidth]{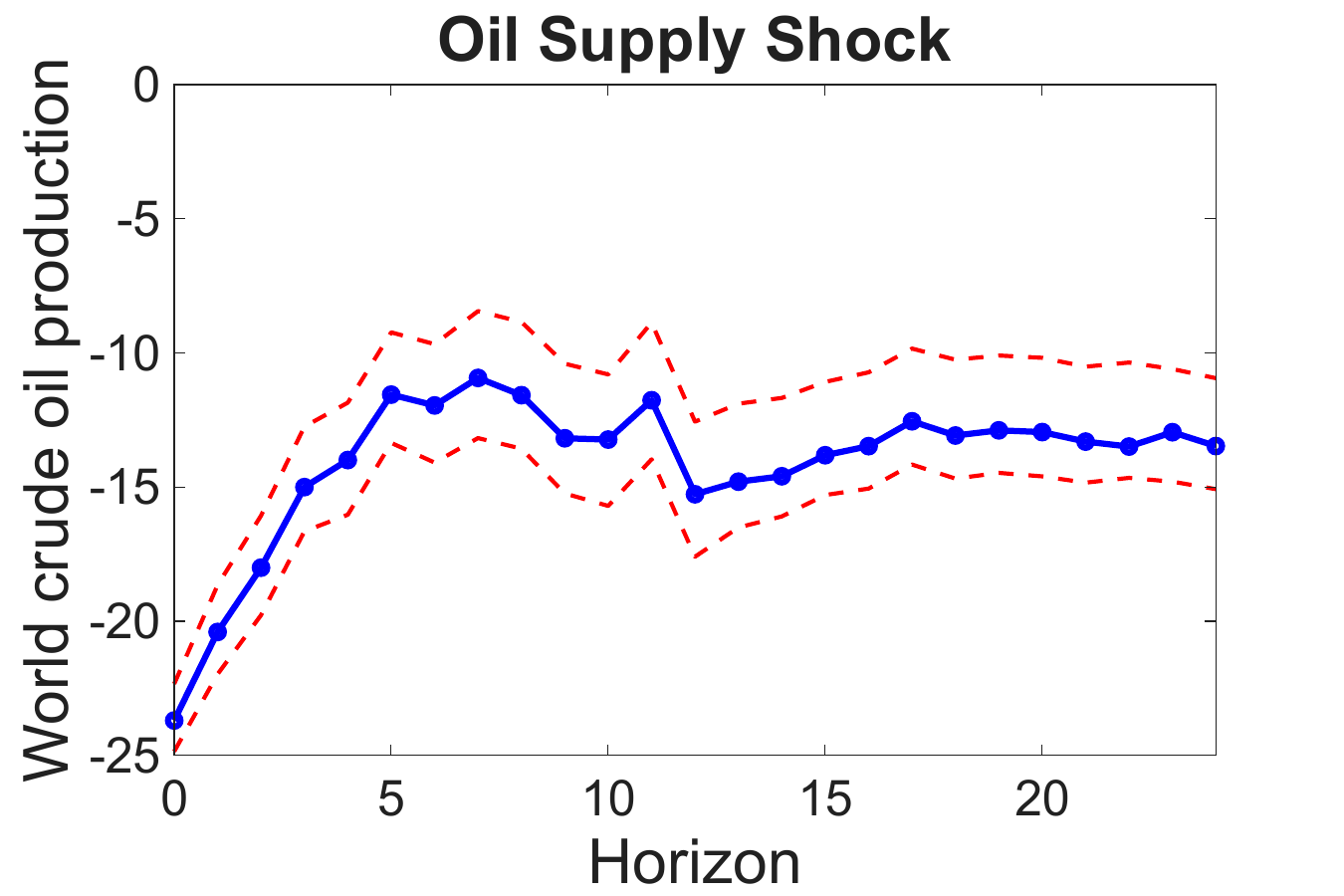}%
  \includegraphics[width=.33\linewidth]{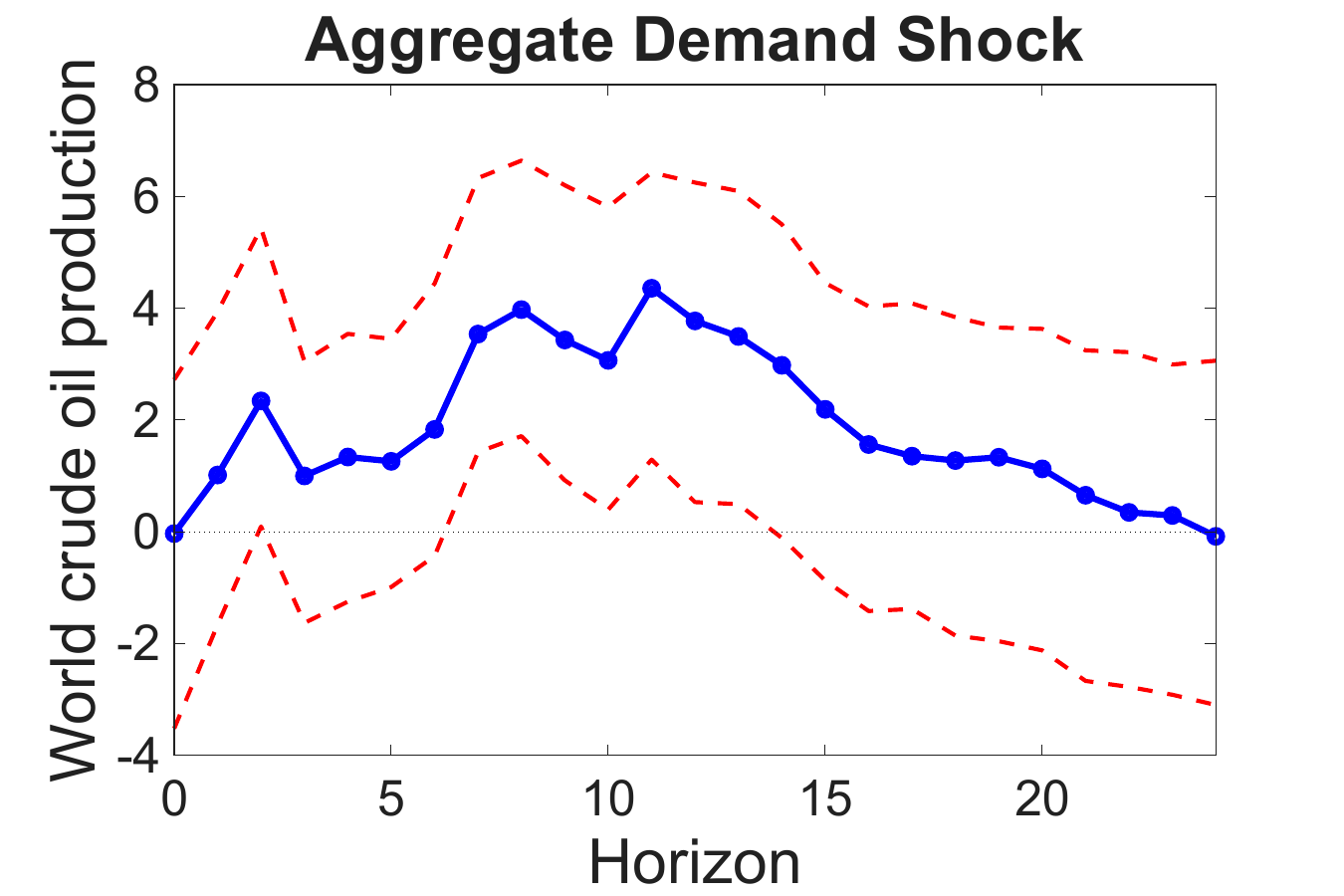}
  \includegraphics[width=.33\linewidth]{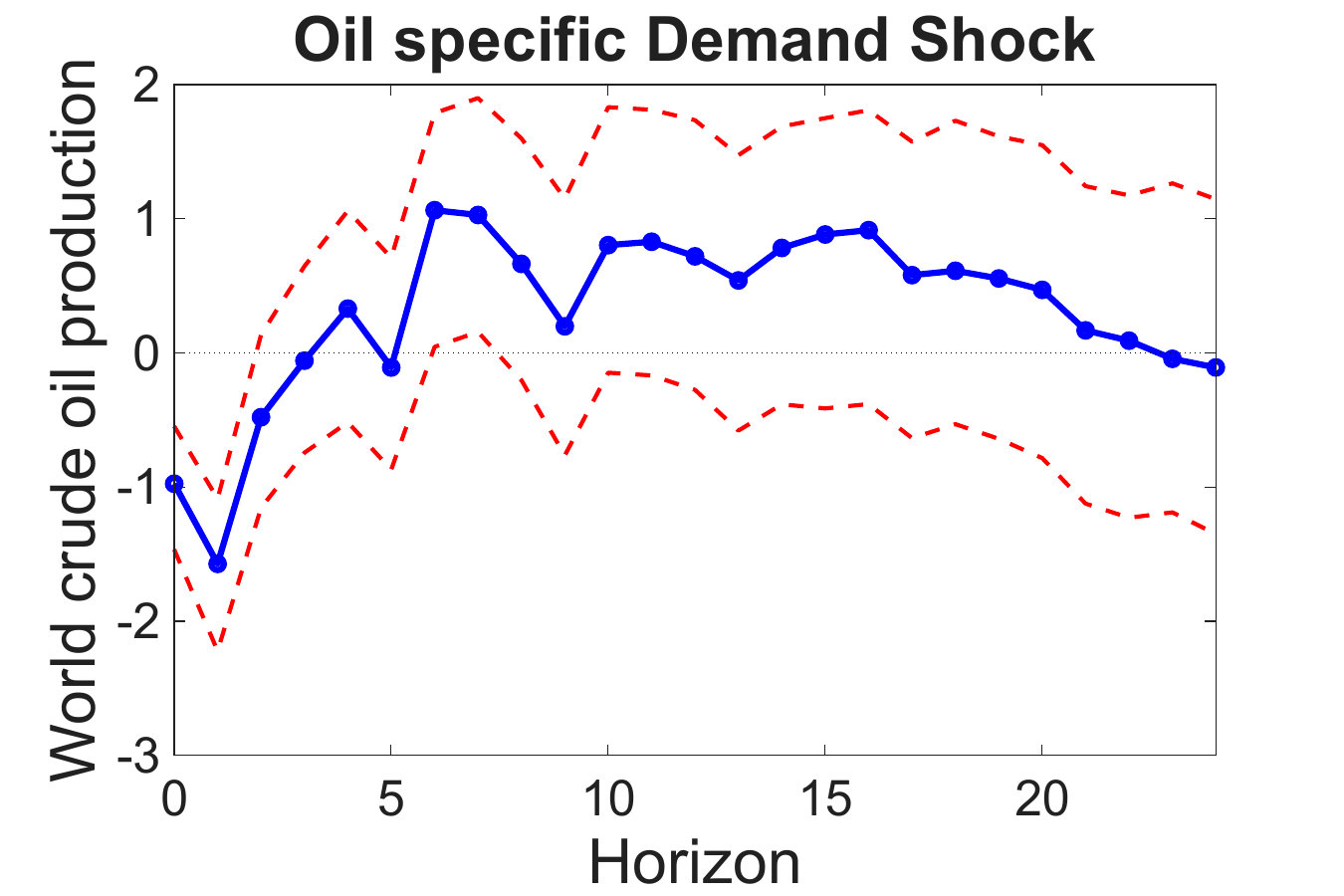}\\
  \includegraphics[width=.33\linewidth]{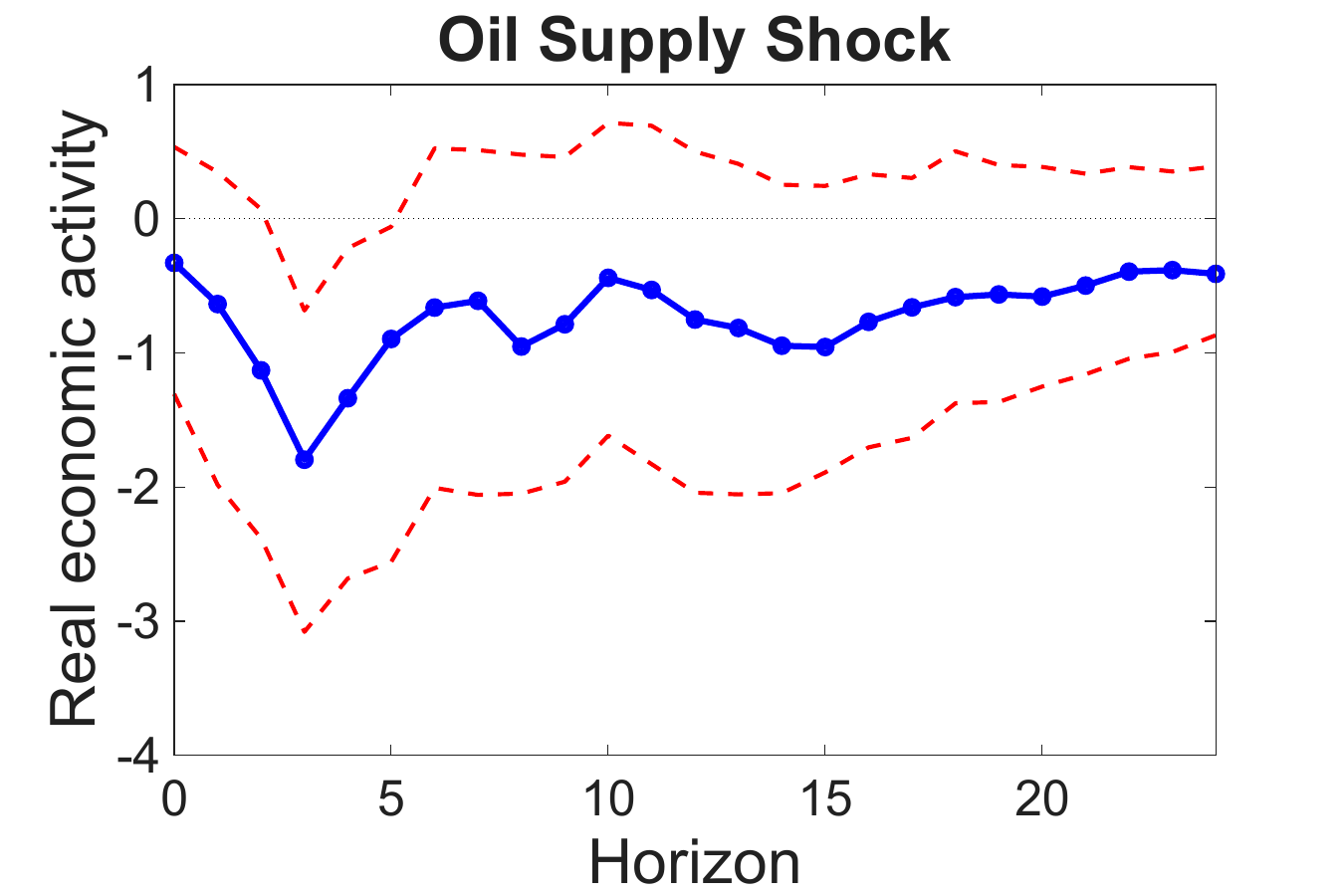}%
  \includegraphics[width=.33\linewidth]{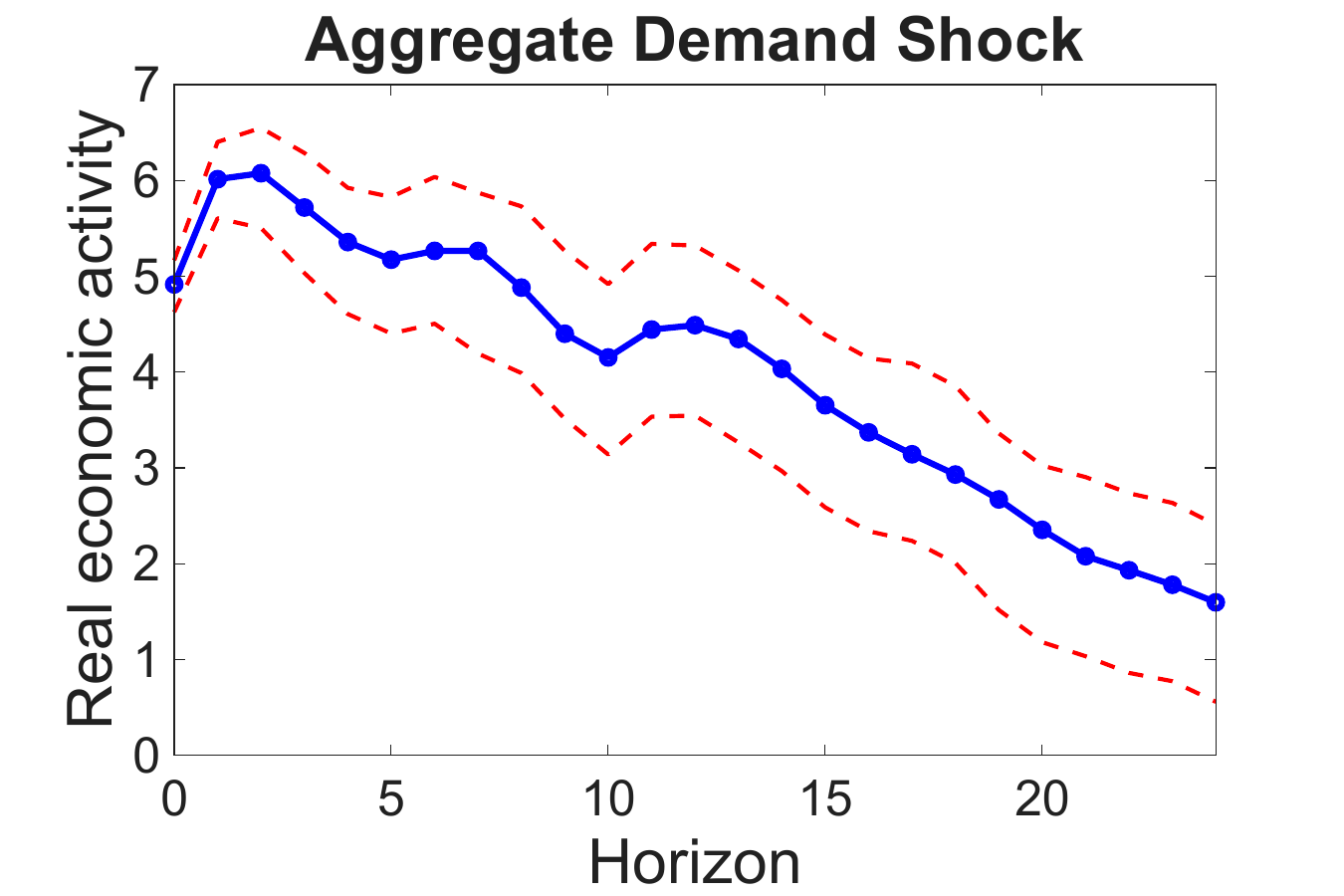}%
  \includegraphics[width=.33\linewidth]{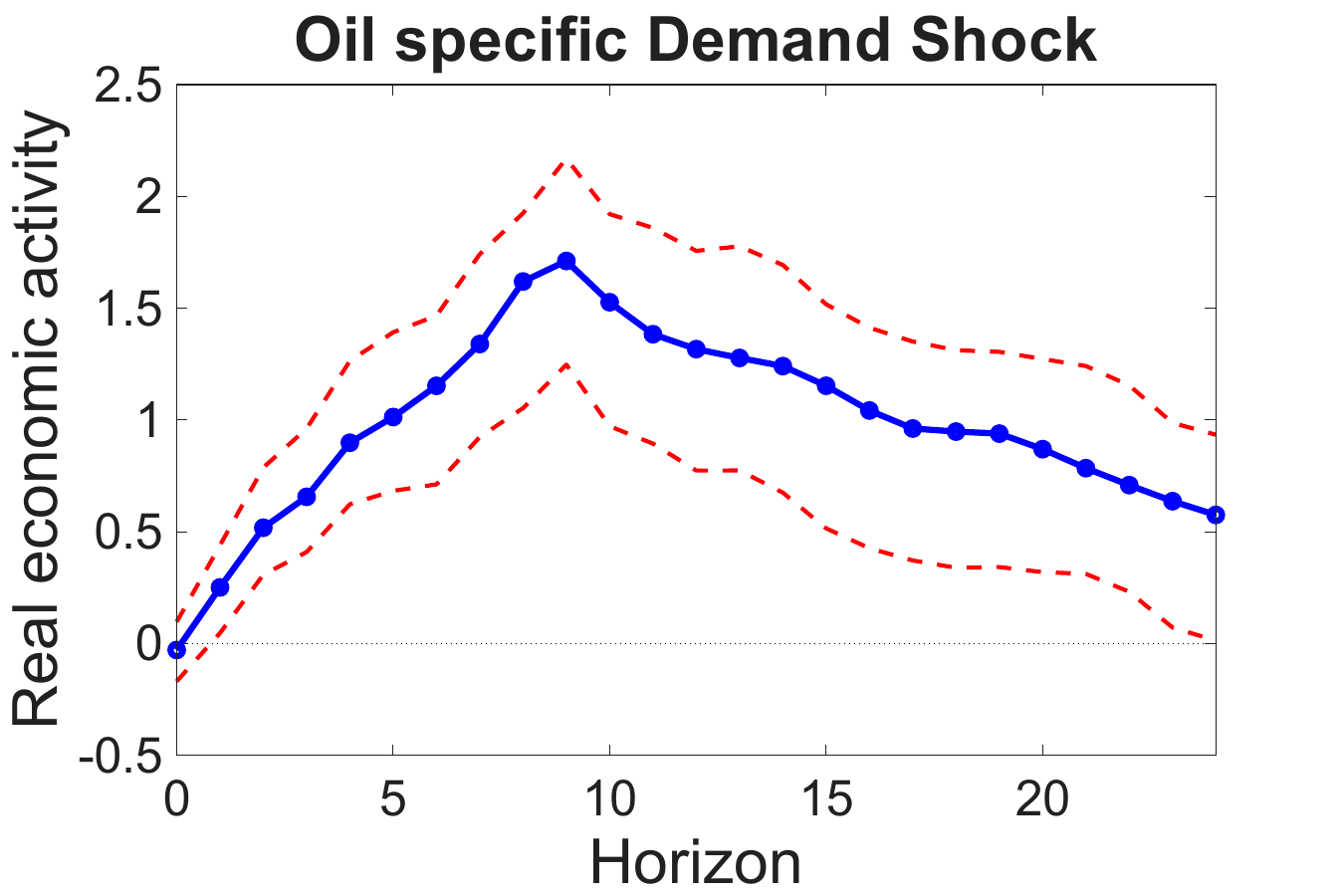}\\
  \includegraphics[width=.33\linewidth]{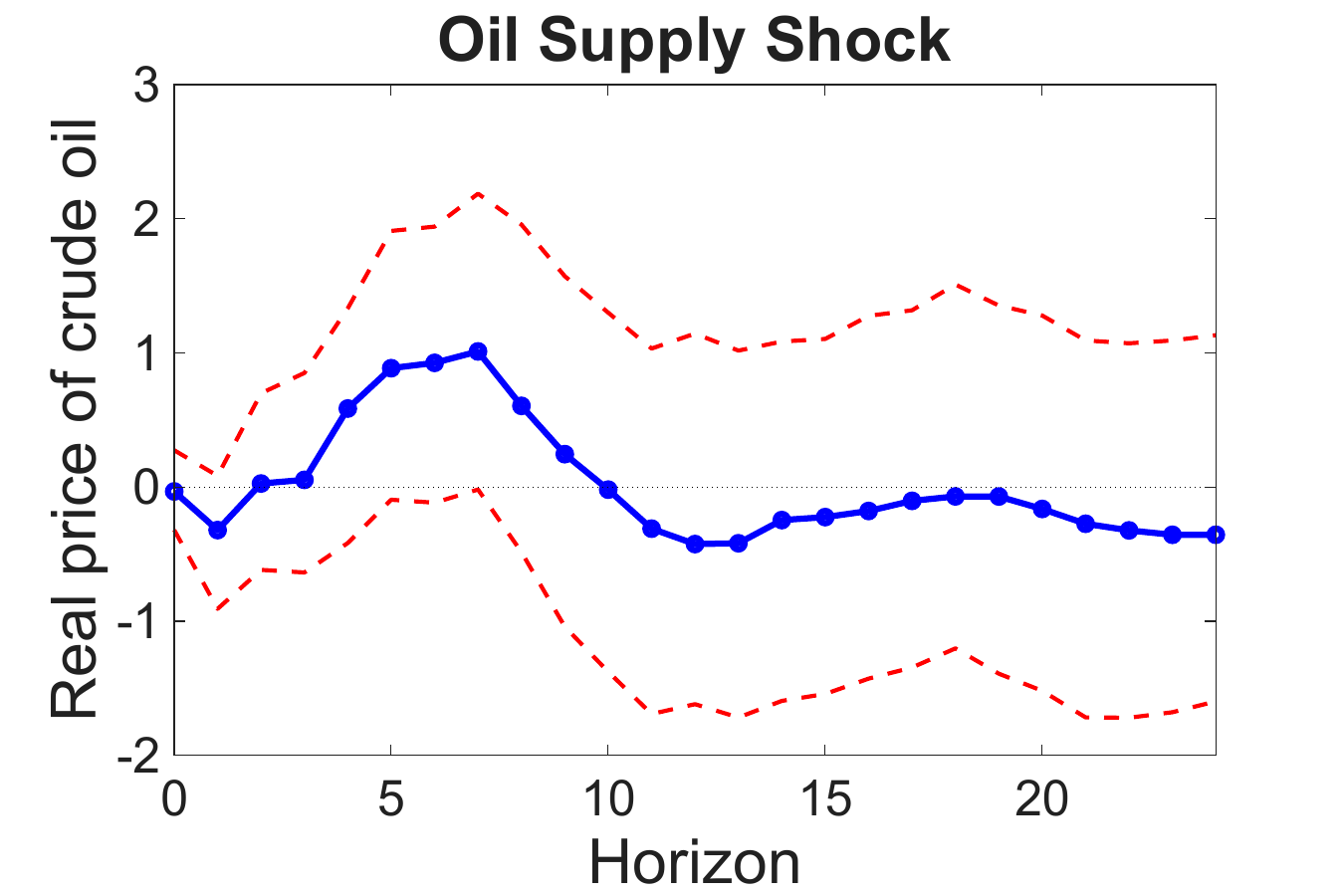}%
  \includegraphics[width=.33\linewidth]{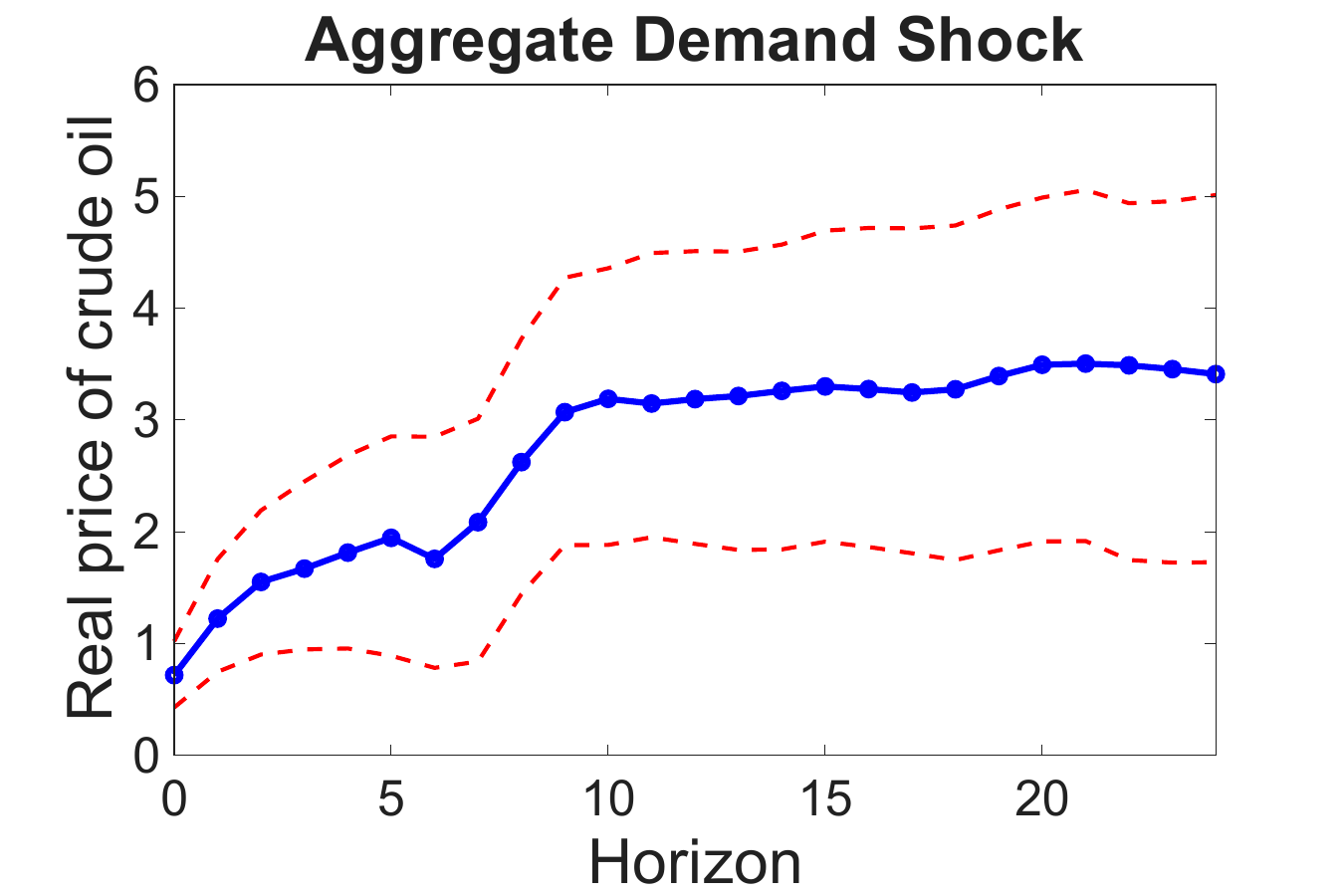}%
  \includegraphics[width=.33\linewidth]{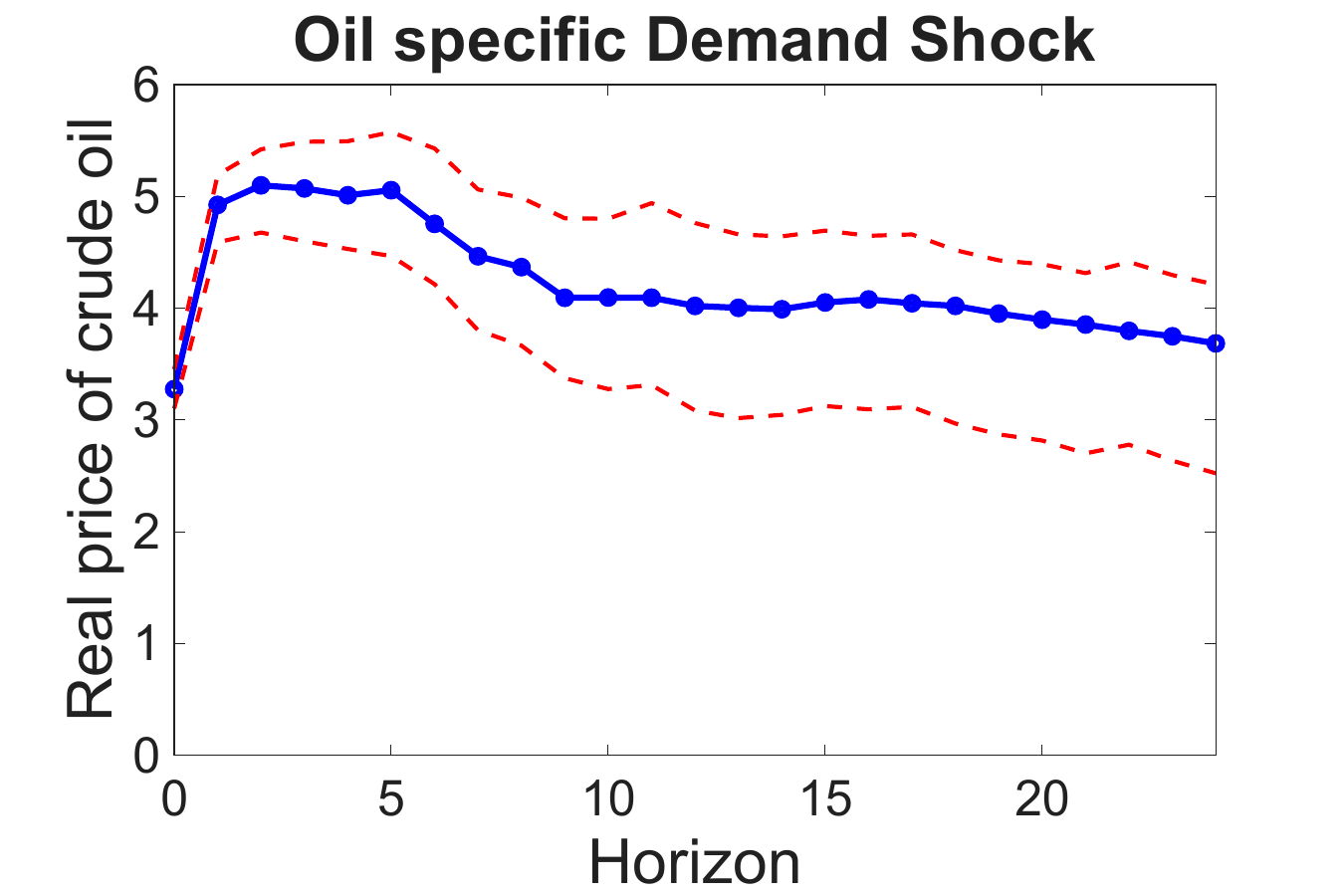}
\caption*{\scriptsize\textit{Notes}: the blue line with dots represents the standard Bayesian posterior mean response, the dashed red lines identify upper and lower bounds of the highest posterior density region with credibility 68\%. 
Identification is obtained via heteroskedasticity assuming distinct eigenvalues}
\label{fig:IRF_M1}
\end{figure}

\begin{figure}[ht!]
\caption{Impulse response functions $\mathcal{M}_{3}$}

  \includegraphics[width=.33\linewidth]{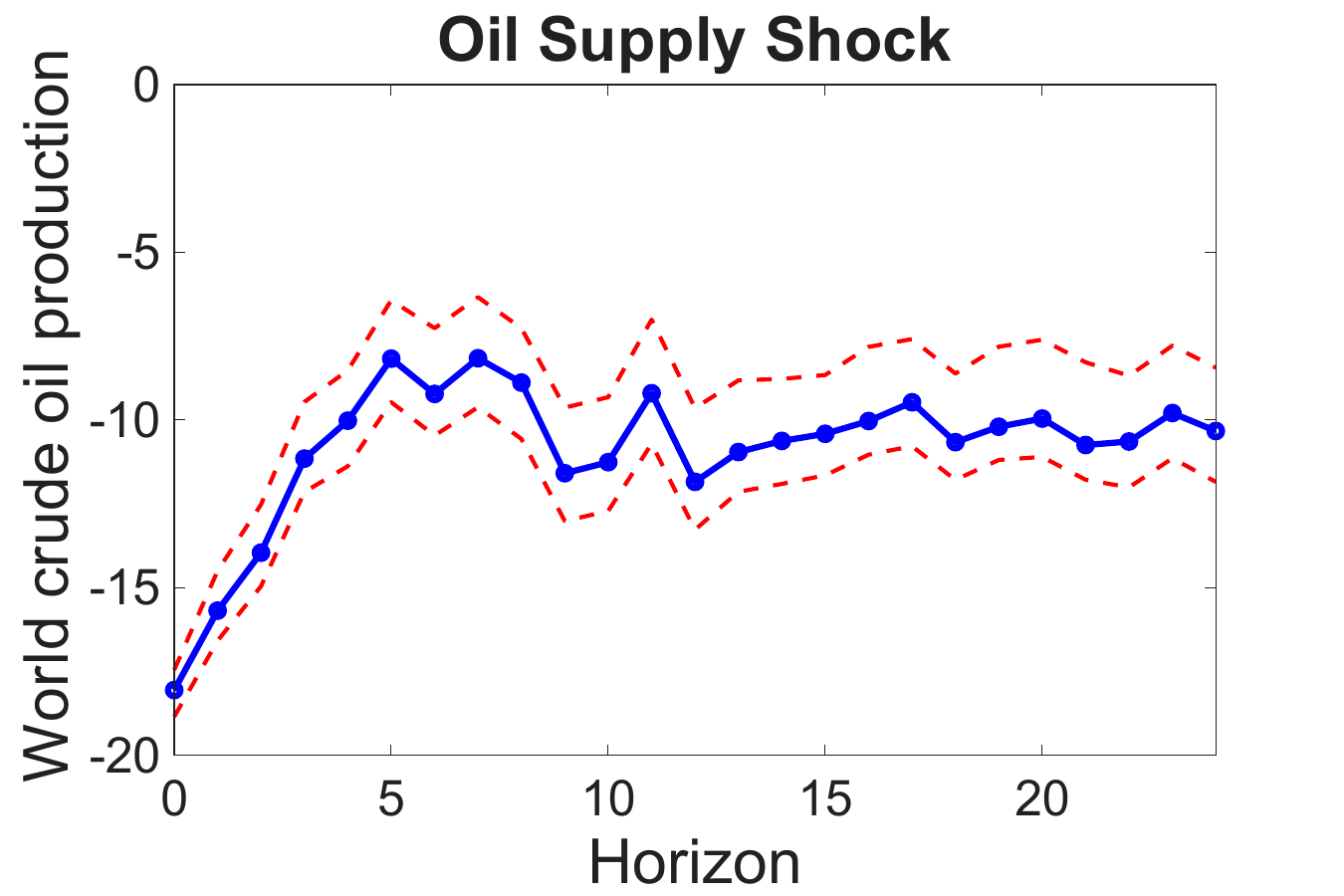}%
  \includegraphics[width=.33\linewidth]{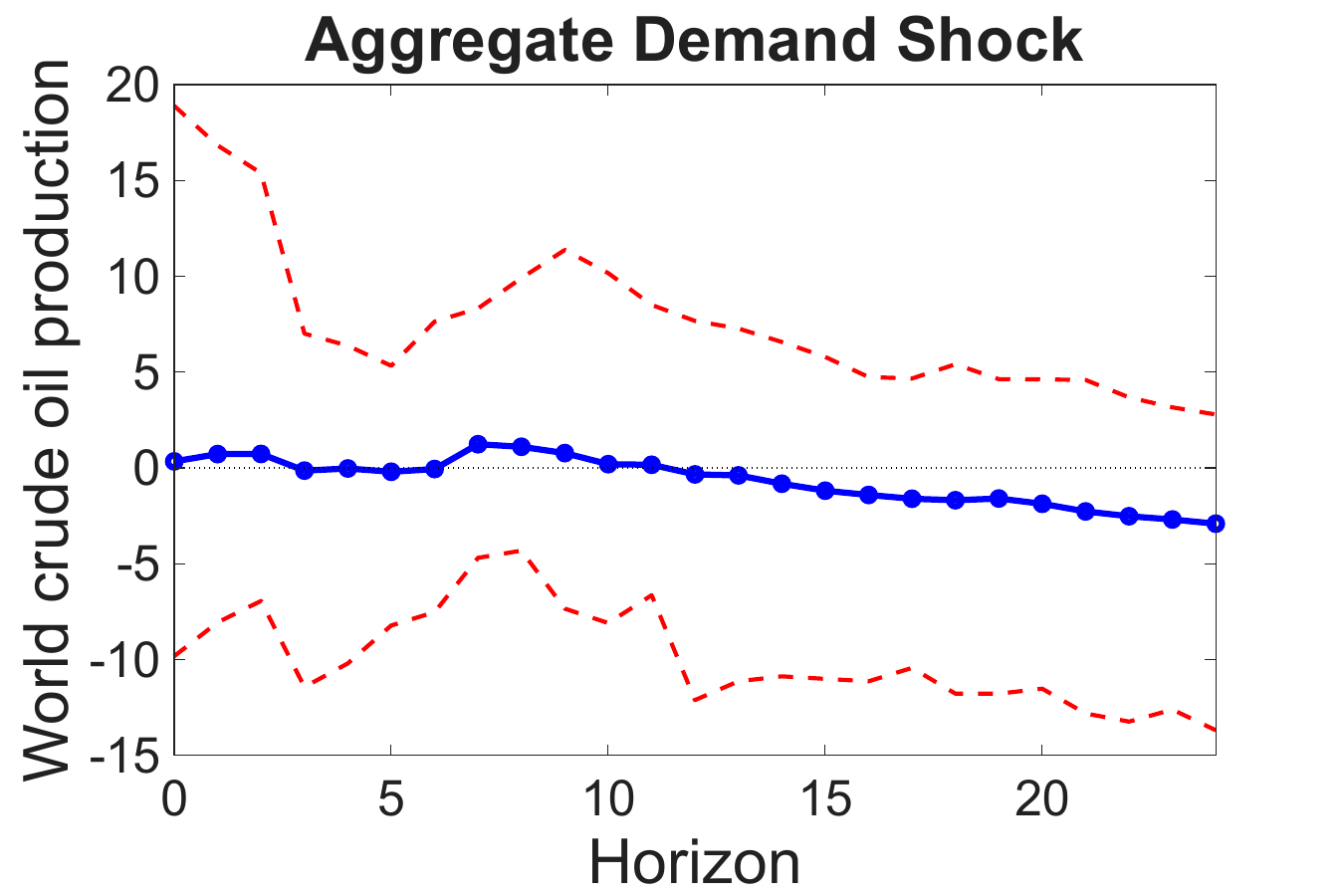}
  \includegraphics[width=.33\linewidth]{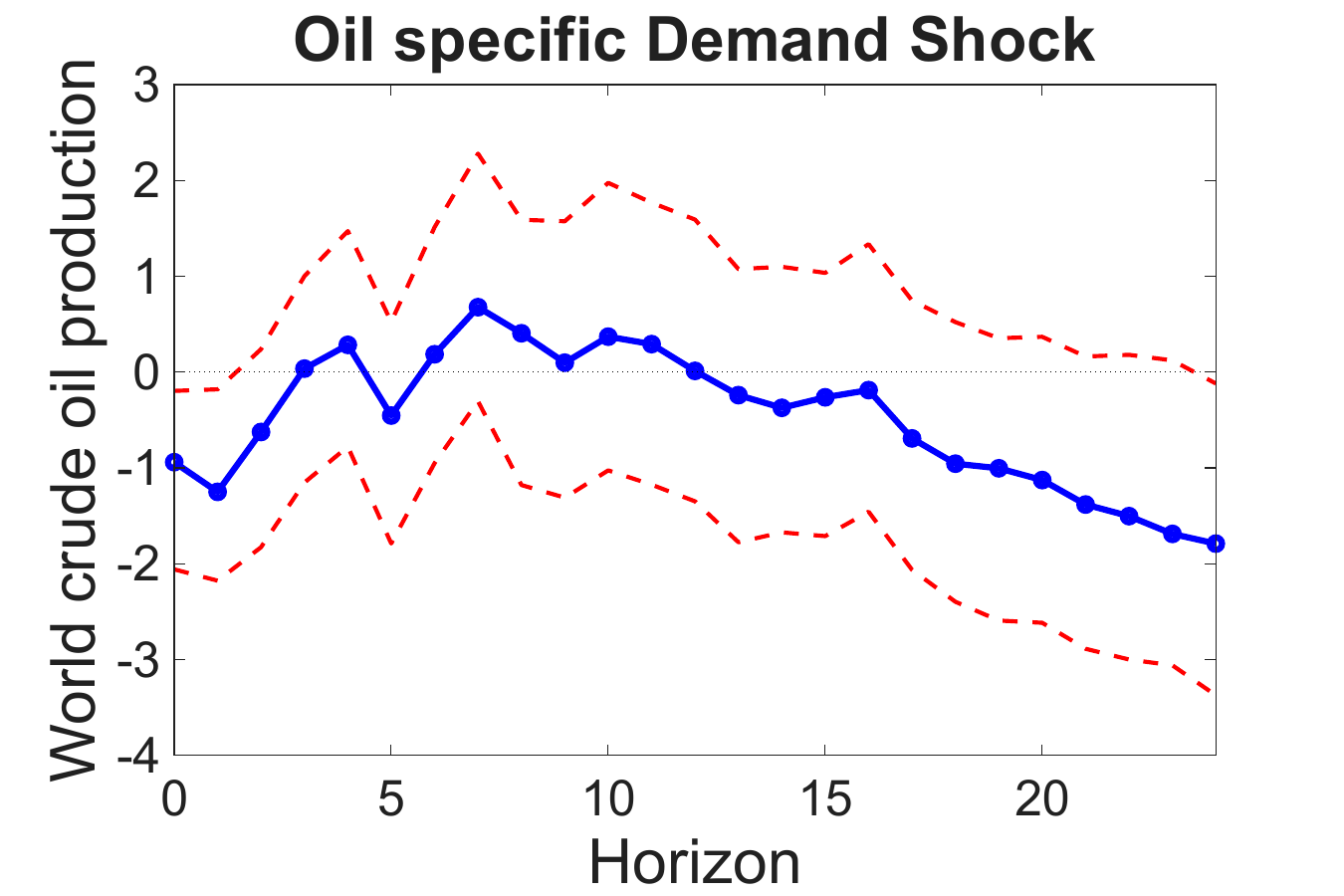}\\
  \includegraphics[width=.33\linewidth]{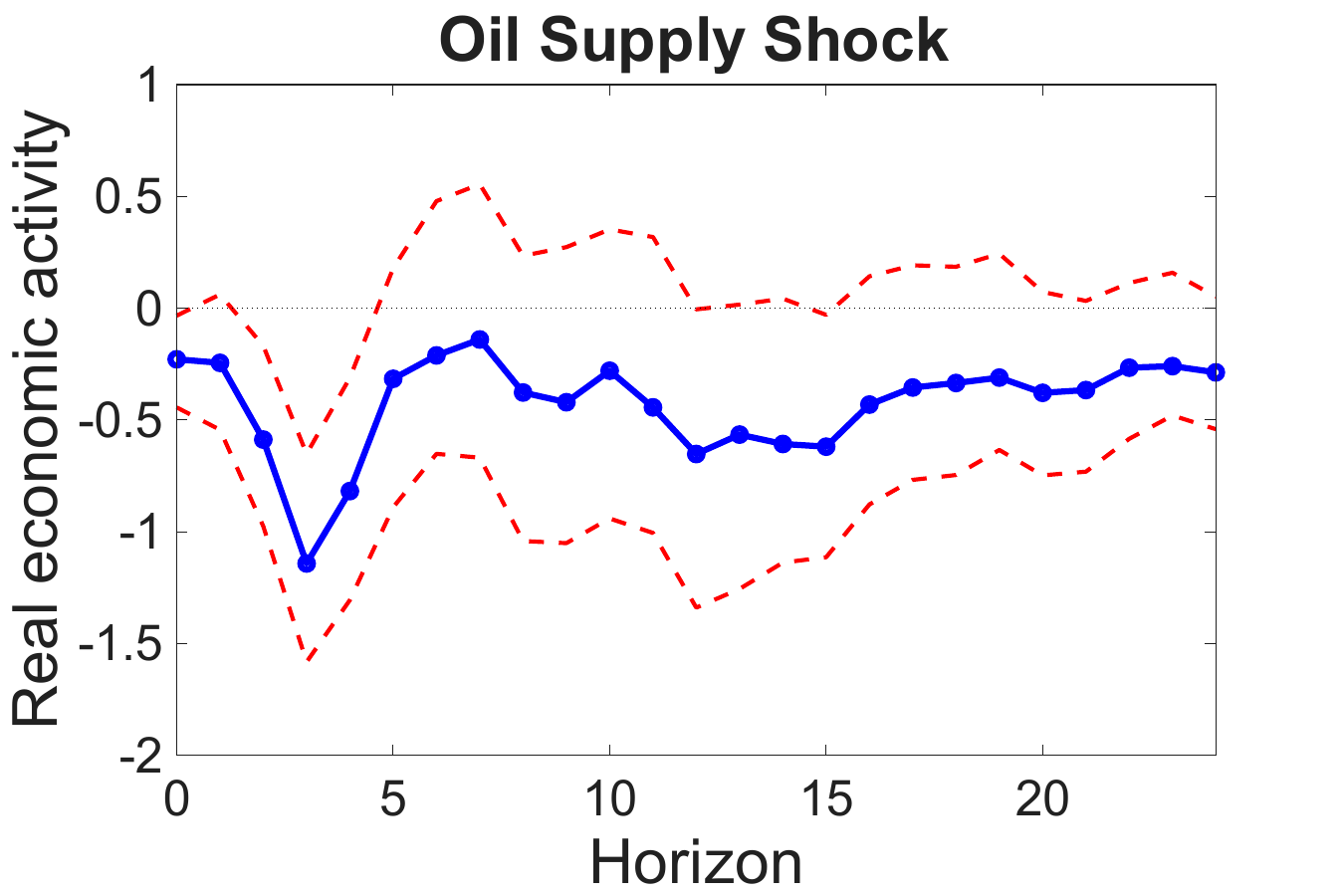}%
  \includegraphics[width=.33\linewidth]{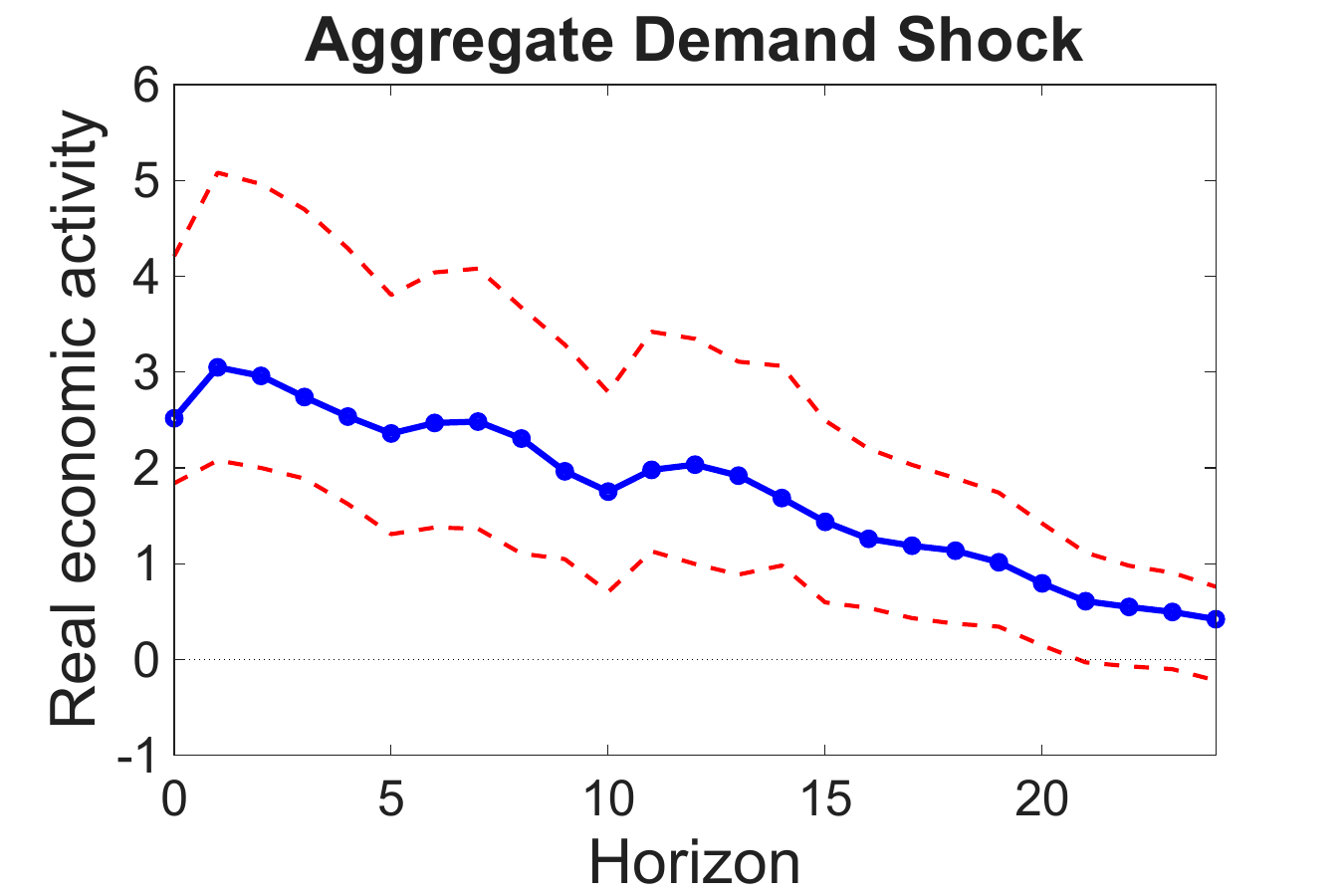}%
  \includegraphics[width=.33\linewidth]{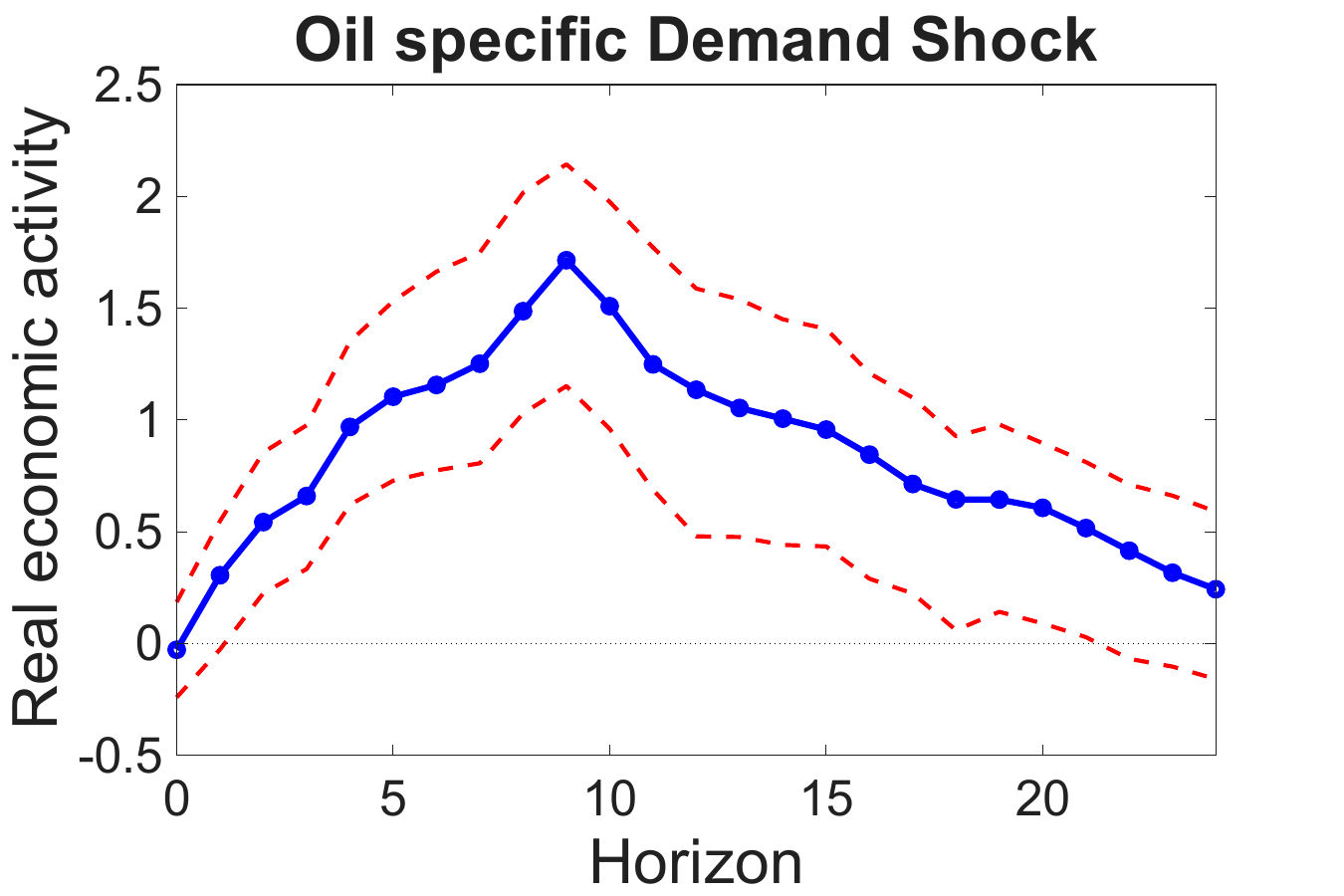}\\
  \includegraphics[width=.33\linewidth]{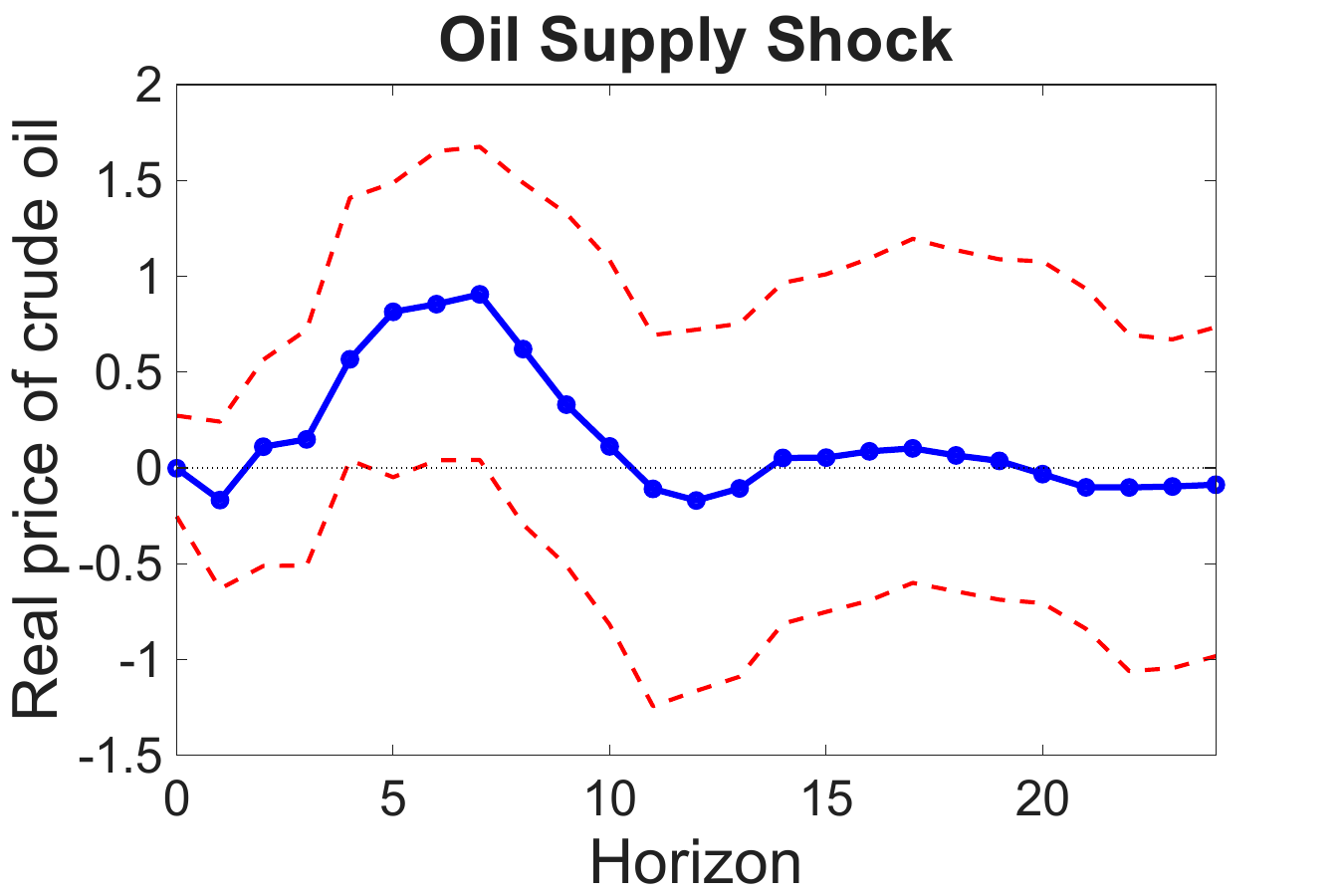}%
  \includegraphics[width=.33\linewidth]{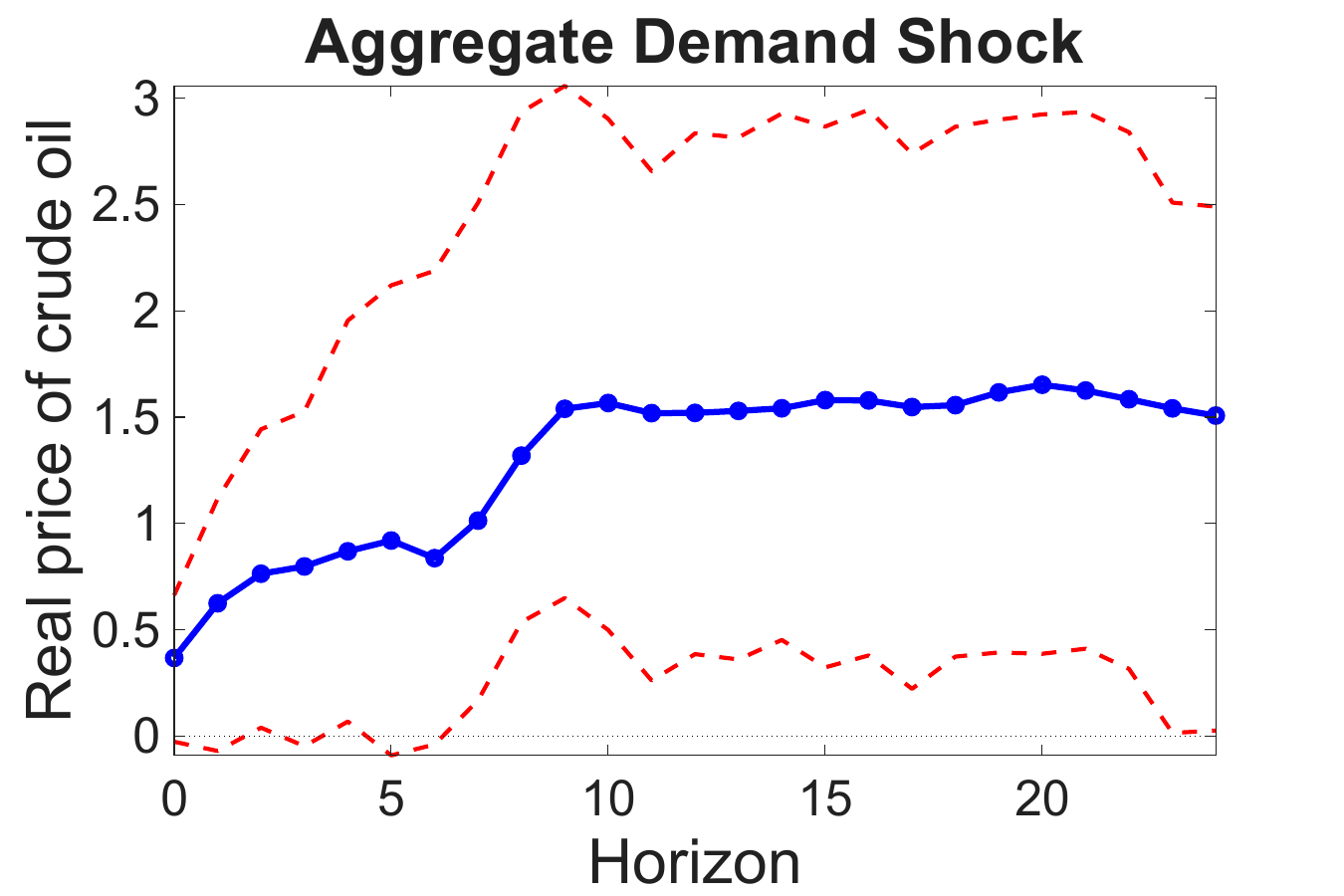}%
  \includegraphics[width=.33\linewidth]{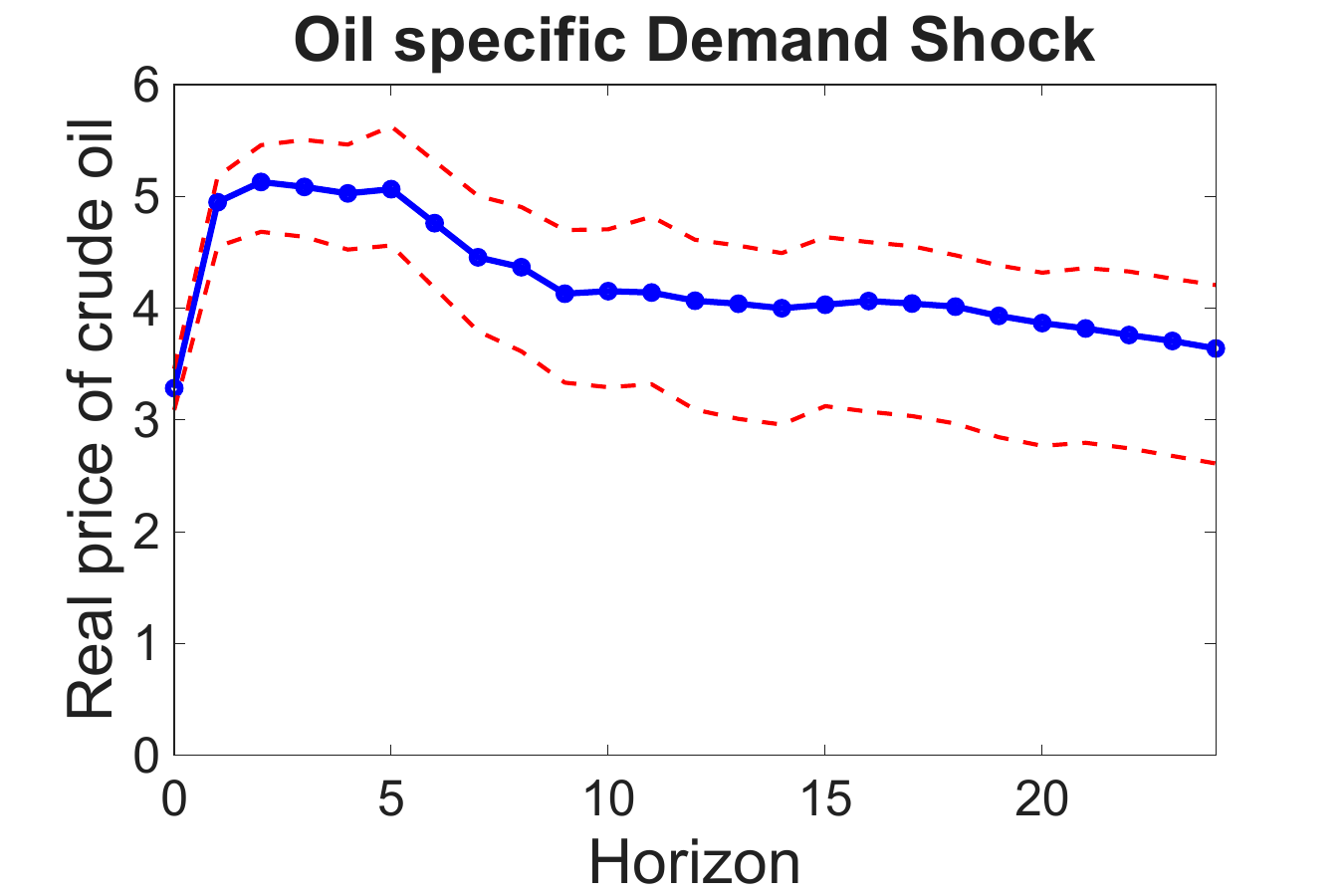}
\caption*{\scriptsize\textit{Notes}: the blue line with dots represents the posterior mean response, the dashed red lines identify upper and lower bounds of the highest posterior density region with credibility 68\%. Identification through heteroskedasticity exploiting the fact that one eigenvalue is distinct from the others; moreover, we impose $a_{12}=0$.}
\label{fig:IRF_M3}
\end{figure}

\begin{figure}[ht!]
\caption{HSVAR with Eigenvalue Multiplicity Identified with Sign Restrictions -- Alternative implementation of Algorithm 1.}

  \includegraphics[width=.33\linewidth]{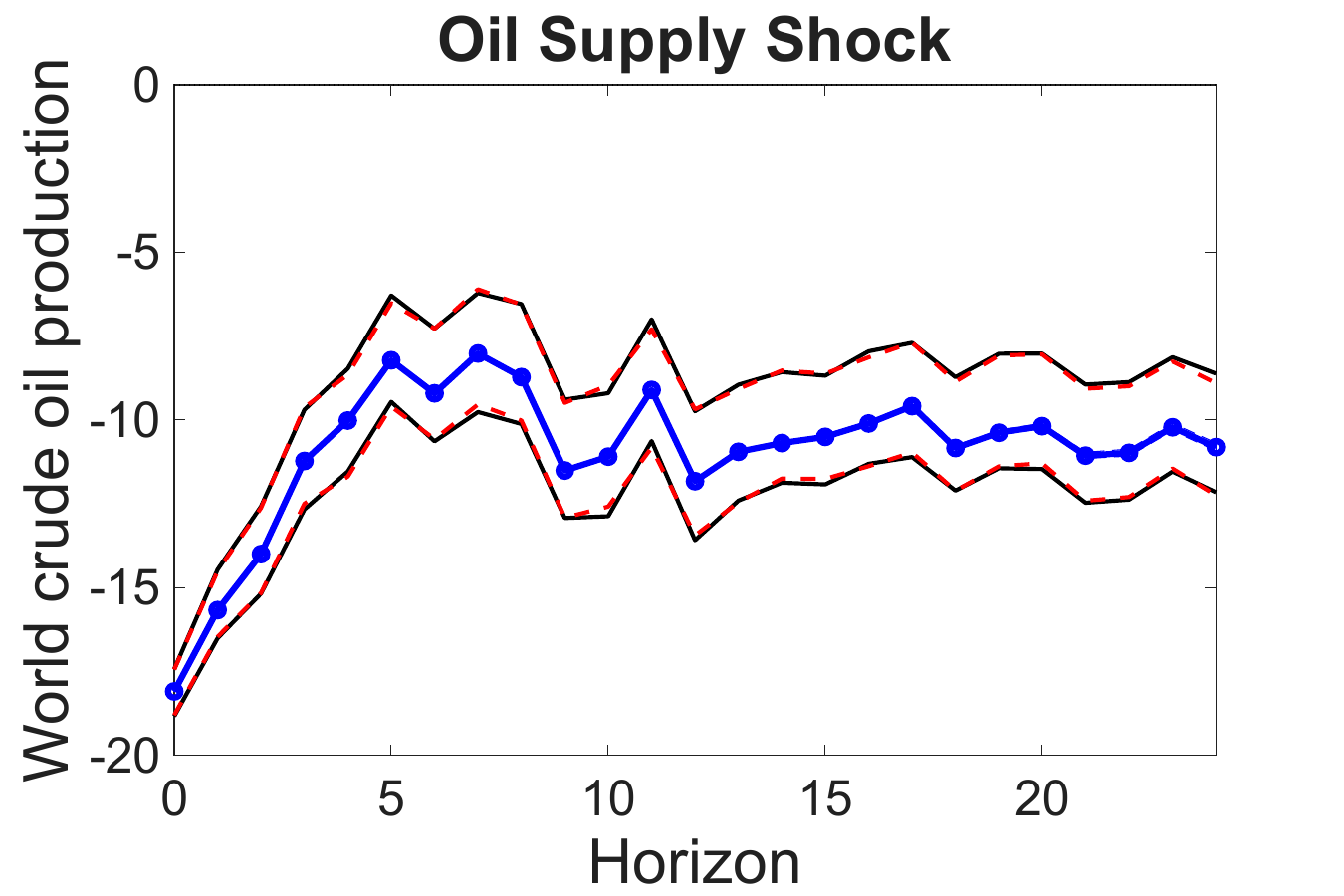}%
  \includegraphics[width=.33\linewidth]{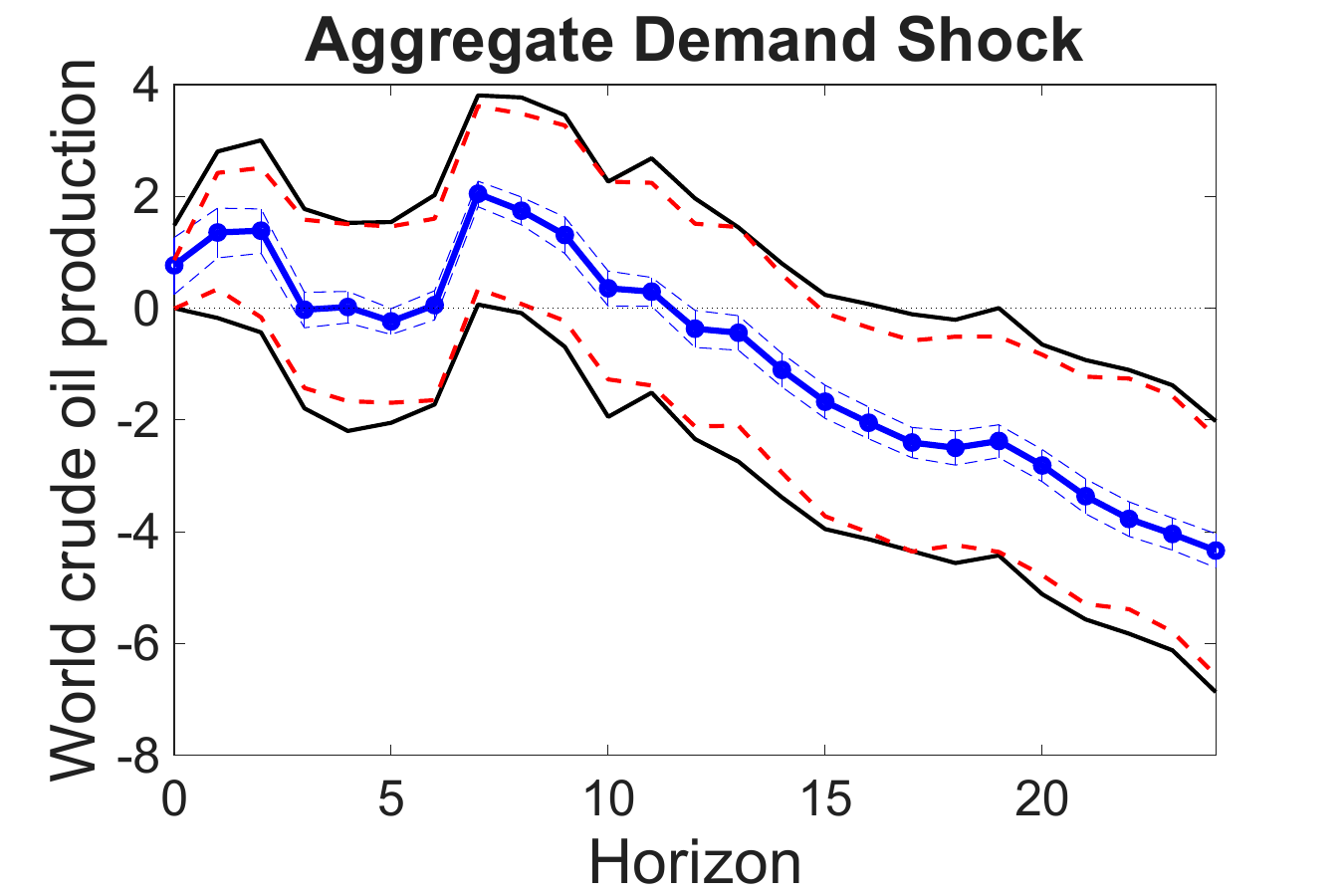}
  \includegraphics[width=.33\linewidth]{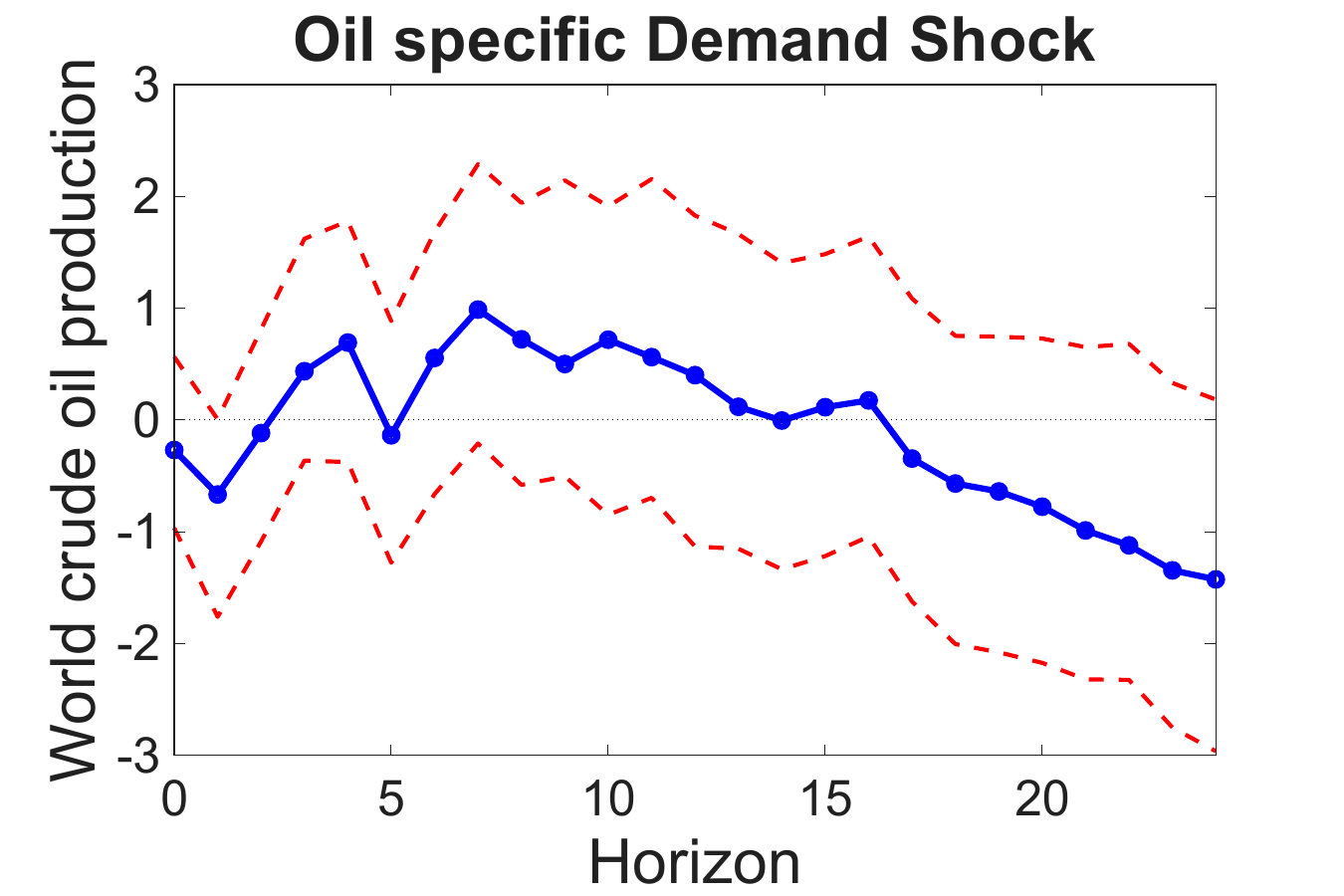}\\
  \includegraphics[width=.33\linewidth]{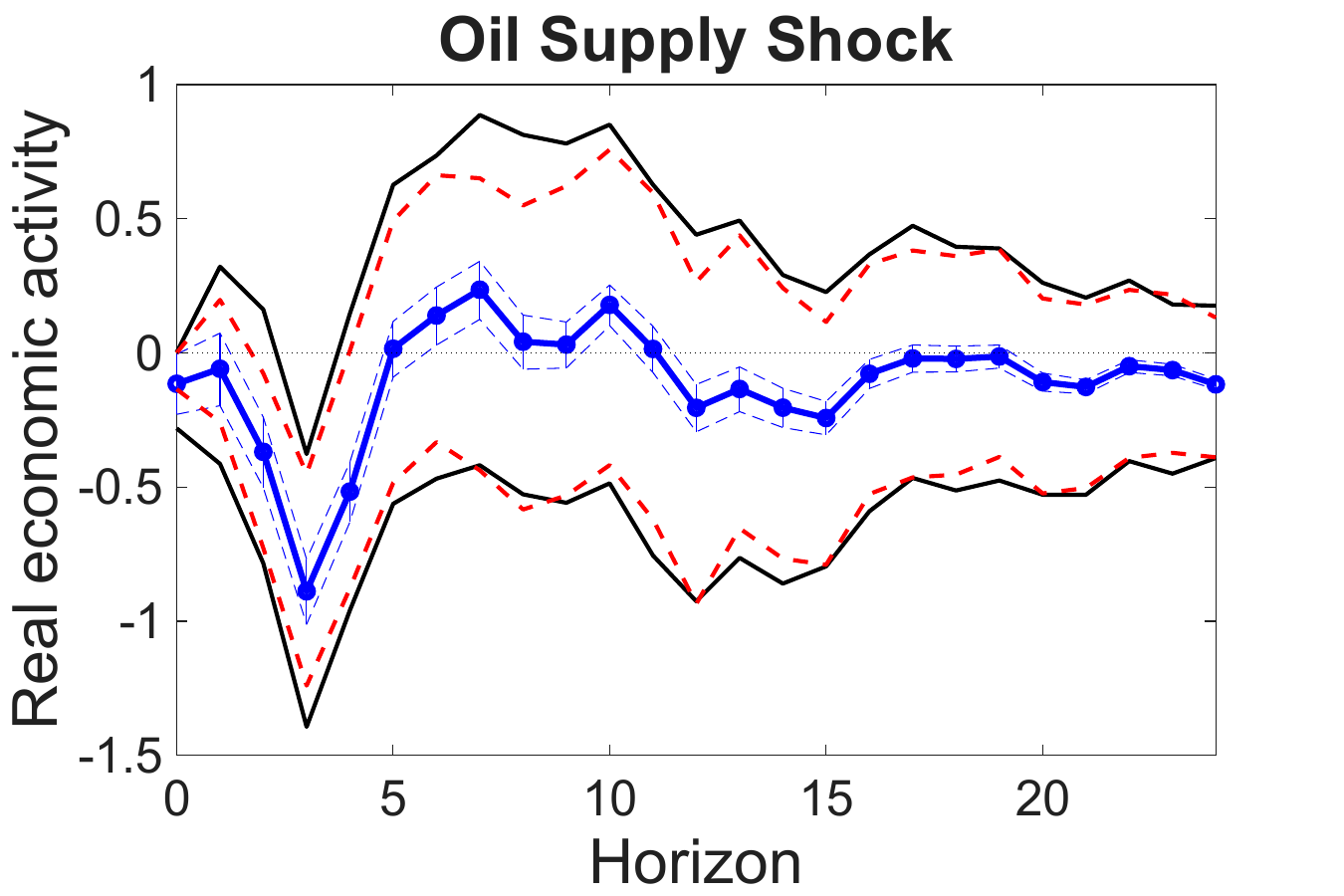}%
  \includegraphics[width=.33\linewidth]{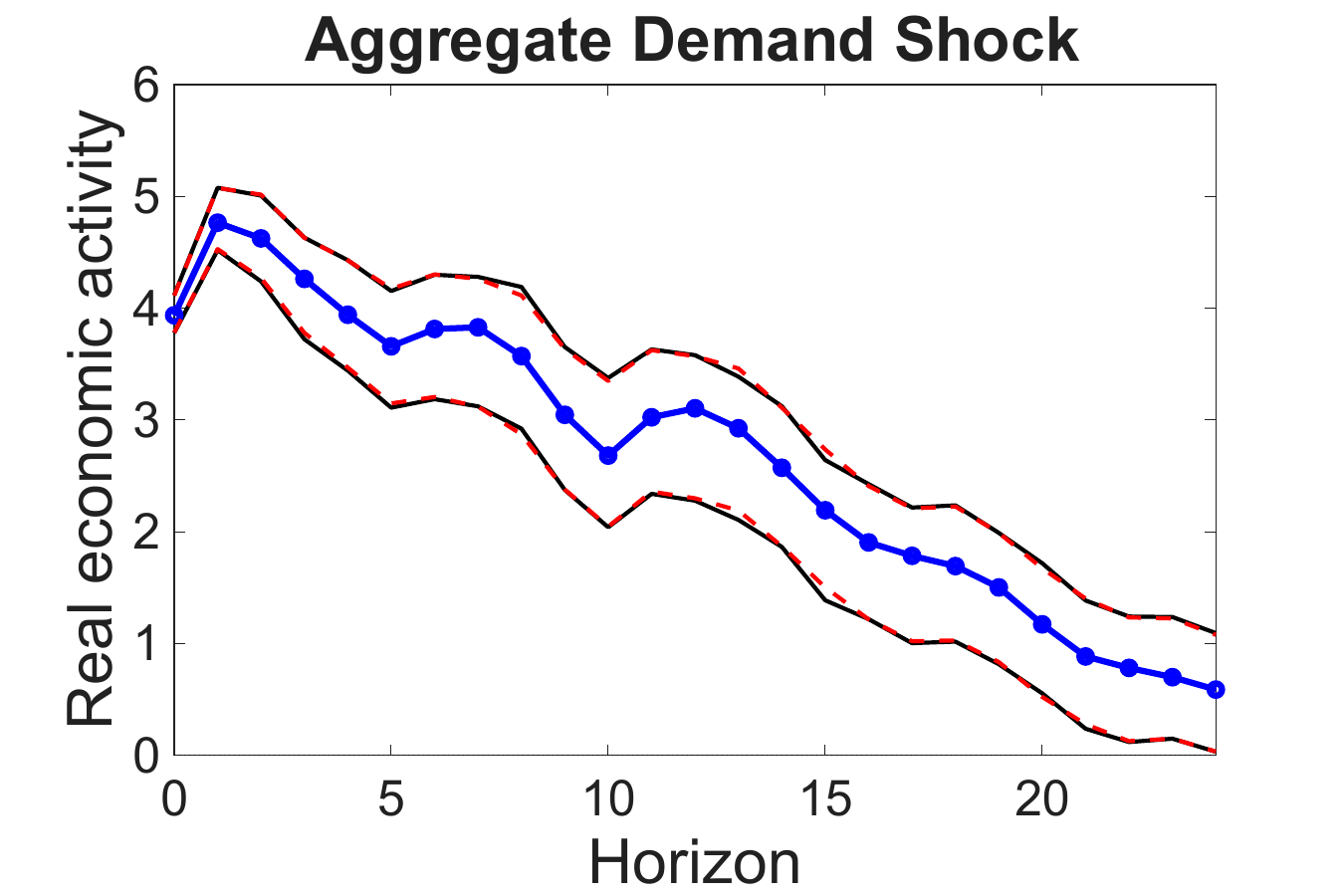}%
  \includegraphics[width=.33\linewidth]{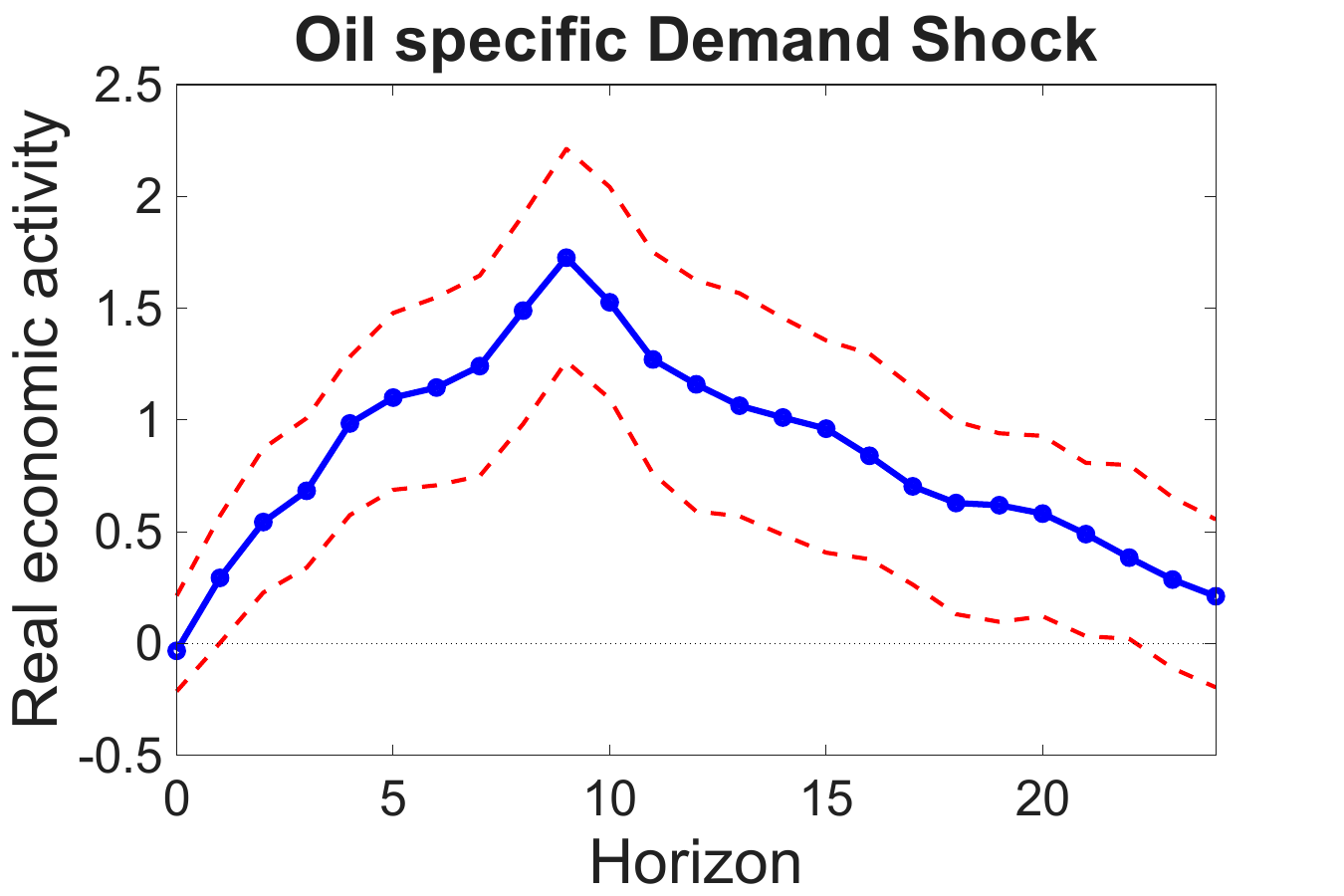}\\
  \includegraphics[width=.33\linewidth]{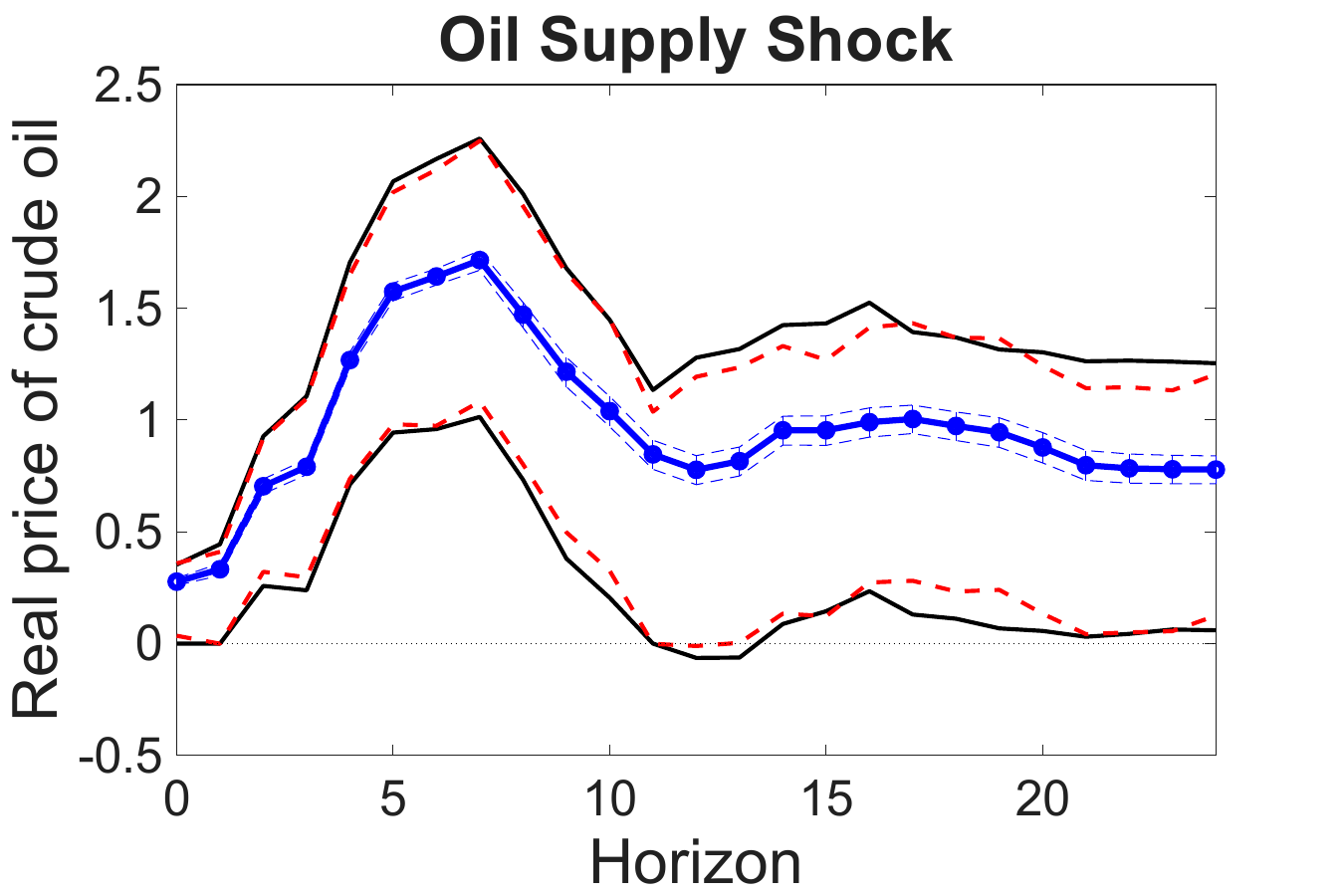}%
  \includegraphics[width=.33\linewidth]{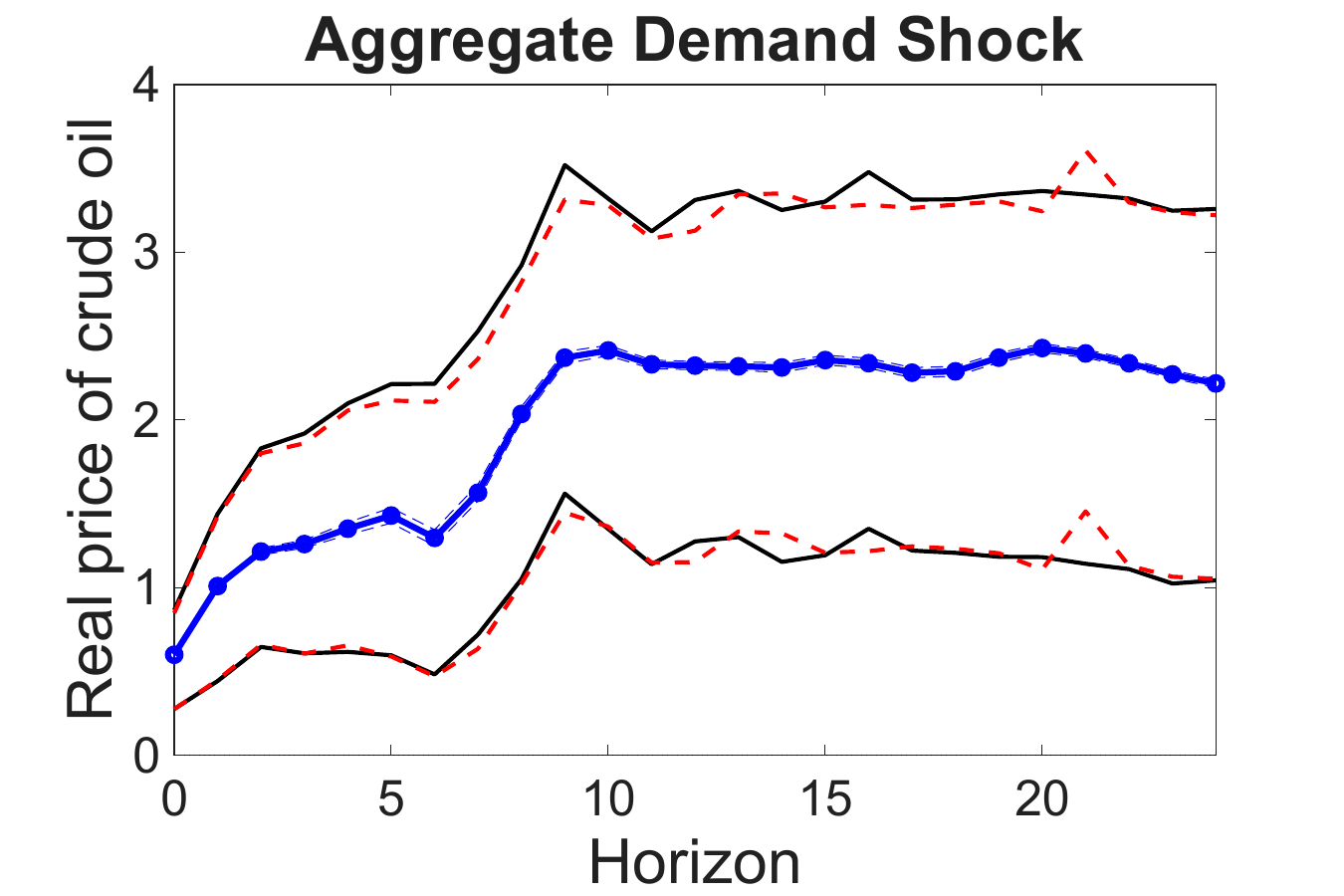}%
  \includegraphics[width=.33\linewidth]{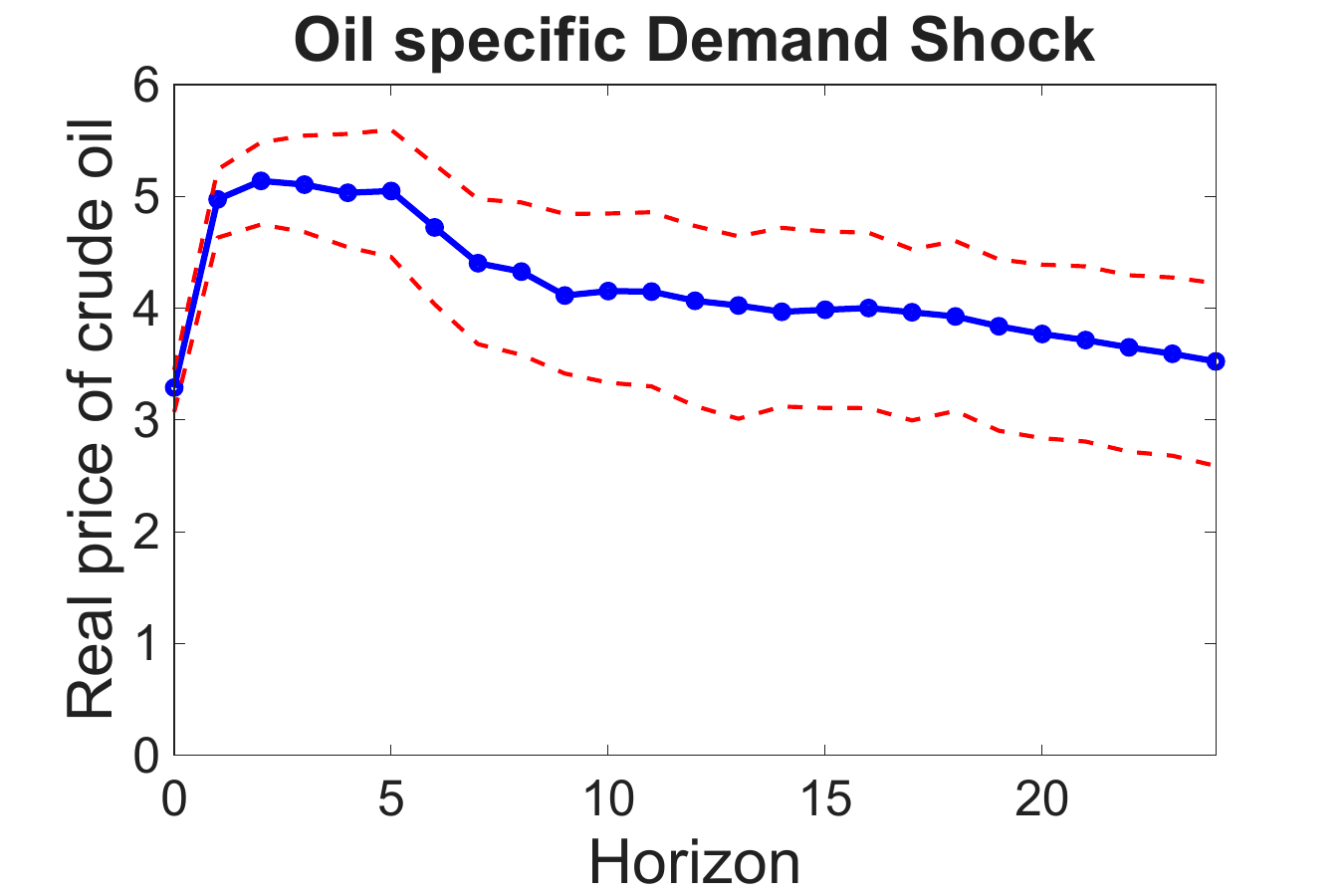}
\caption*{\scriptsize\textit{Notes}: the blue line with dots represents the standard Bayesian posterior mean response, the dashed red lines identify upper and lower bounds of the highest posterior density region with credibility 68\%. Plots in first and second columns of the figure also report the set of posterior means (blue vertical bars) and the bounds of the robust credible region with credibility 68\% (solid black curves). Identification via heteroskedasticity with multiple eigenvalues (i.e. only one shock is point identified), static and dynamic sign restrictions.}
\label{fig:IRF_M2_dynsignnew_numerical}
\end{figure}
\bigskip

\noindent\textit{Alternative implementation of Algorithm \ref{algo:PostBounds}.} In Section \ref{sec:empirics}, Step 4 of Algorithm \ref{algo:PostBounds} is implemented using a simulation-based approach similar to that in \citet{GK18}. Specifically, we repeat Steps 3.1-3.3 one thousand times and approximate the interval $\big[\ell(\phi_m),: u(\phi_m)\big]$ by the minimum and maximum values obtained across these iterations. Results based on the exact solution of the constrained nonlinear optimization problem described in Step 4, reported in Figure \ref{fig:IRF_M2_dynsignnew_numerical}, are qualitatively indistinguishable from those presented in Section \ref{sec:empirics}. This comparison concerns the computation of bounds conditional on the retained reduced-form draws; it does not assess possible false rejections of non-empty admissible sets arising from the limit of ($L$) rotation draws in Step 3.3.

\bigskip

\noindent\textit{Robustness to Fuzzy Break Dates.} To assess the robustness of our results to uncertainty about the exact timing of the volatility break, we conducted two complementary exercises. First, we re-estimated the model while varying the break date within a $\pm 12$-month window around October 1987 and recomputed the \textit{log--SDDR} for each specification. The results, reported in Figure \ref{fig:BFfuzzy}, confirm that the main conclusions regarding the proportionality of variances remain qualitatively unchanged. In particular, for the $(\lambda_2,\lambda_3)$ pair, the  \textit{log--SDDR} consistently remains above minus five, providing no strong evidence against proportionality.

Second, we re-estimated the HSVAR with static and dynamic sign restrictions assuming alternative break dates -- October 1986 and October 1988 -- and obtained impulse responses virtually identical to those reported in Section \ref{sec:empirics}, where the break date was set to October 1987 (see Figures \ref{fig:IRF_M2_1986} and \ref{fig:IRF_M2_1988}).

\begin{figure}[ht]
    \centering
       \caption{\textit{$|\text{log--SDDR}|$} for the Proportionality Restriction on Variance Changes Across Regimes: Robustness to Changes in the Break Date}
    \includegraphics[width=1\linewidth]{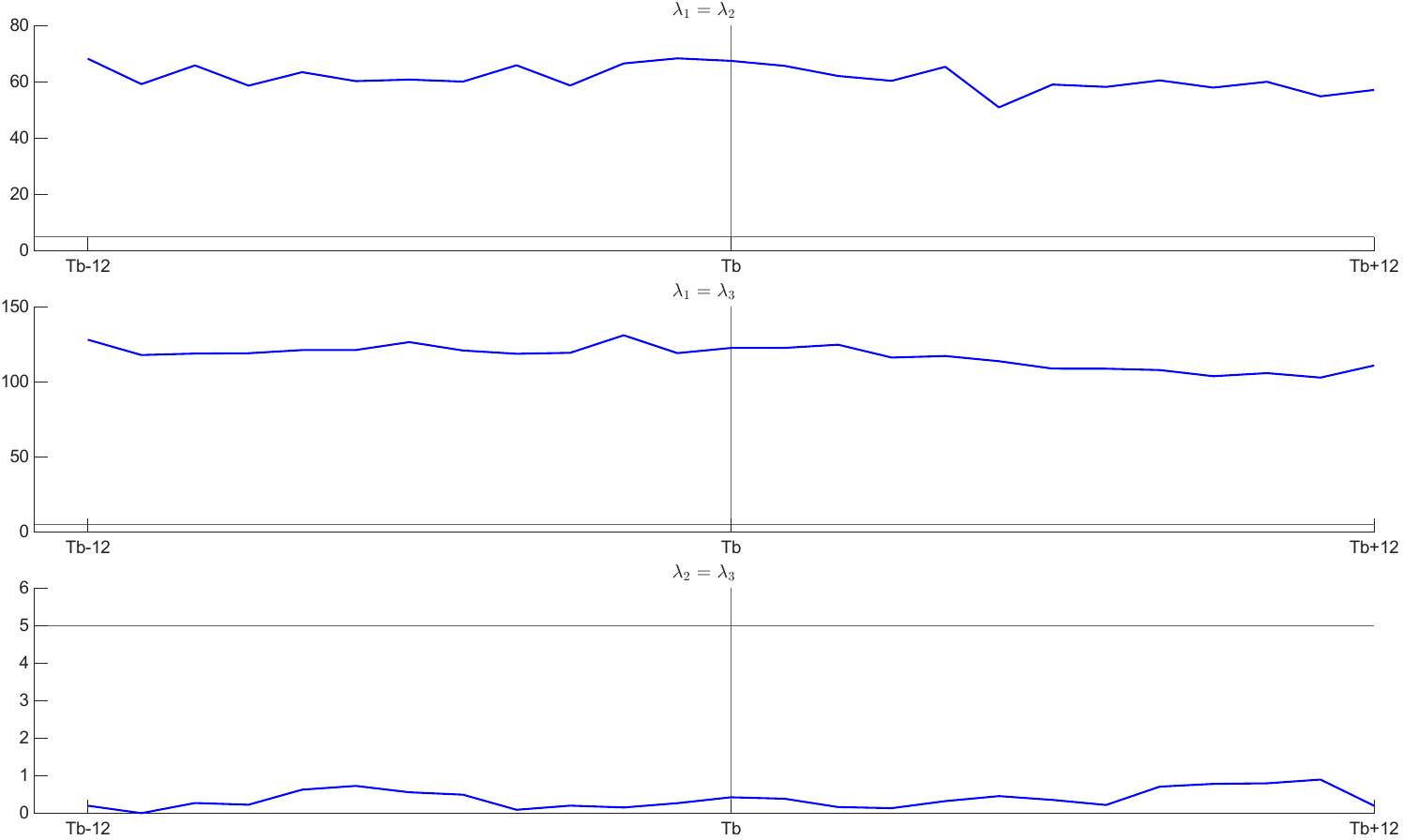}
    \label{fig:BFfuzzy}
    \caption*{\scriptsize\textit{Notes}: The figure reports $|\text{log--SDDR}|$ (absolute log Bayes factors) for the proportionality restriction on variance changes across regimes, for break dates ranging from October 1986 to October 1988.  For each break date, the log--SDDRs are computed from 50{,}000 posterior draws after a burn-in of 5{,}000, using the same procedure as in Table~1. The horizontal line marks the Kass--Raftery rule of thumb used in Table~1: $|\text{log--SDDR}|>5$ $\approx$ strong evidence. For the pairs $(\lambda_1,\lambda_2)$ and $(\lambda_1,\lambda_3)$, the log--SDDRs are negative at all break dates, so values above the line indicate strong evidence against proportionality. For the pair $(\lambda_2,\lambda_3)$, $|\text{log--SDDR}|$ remains below 5 throughout, indicating no strong evidence against proportionality.}
\end{figure}

\begin{figure}[ht!]
\caption{HSVAR with Eigenvalue Multiplicity Identified with Sign Restrictions - Break date October 1986}

  \includegraphics[width=.33\linewidth]{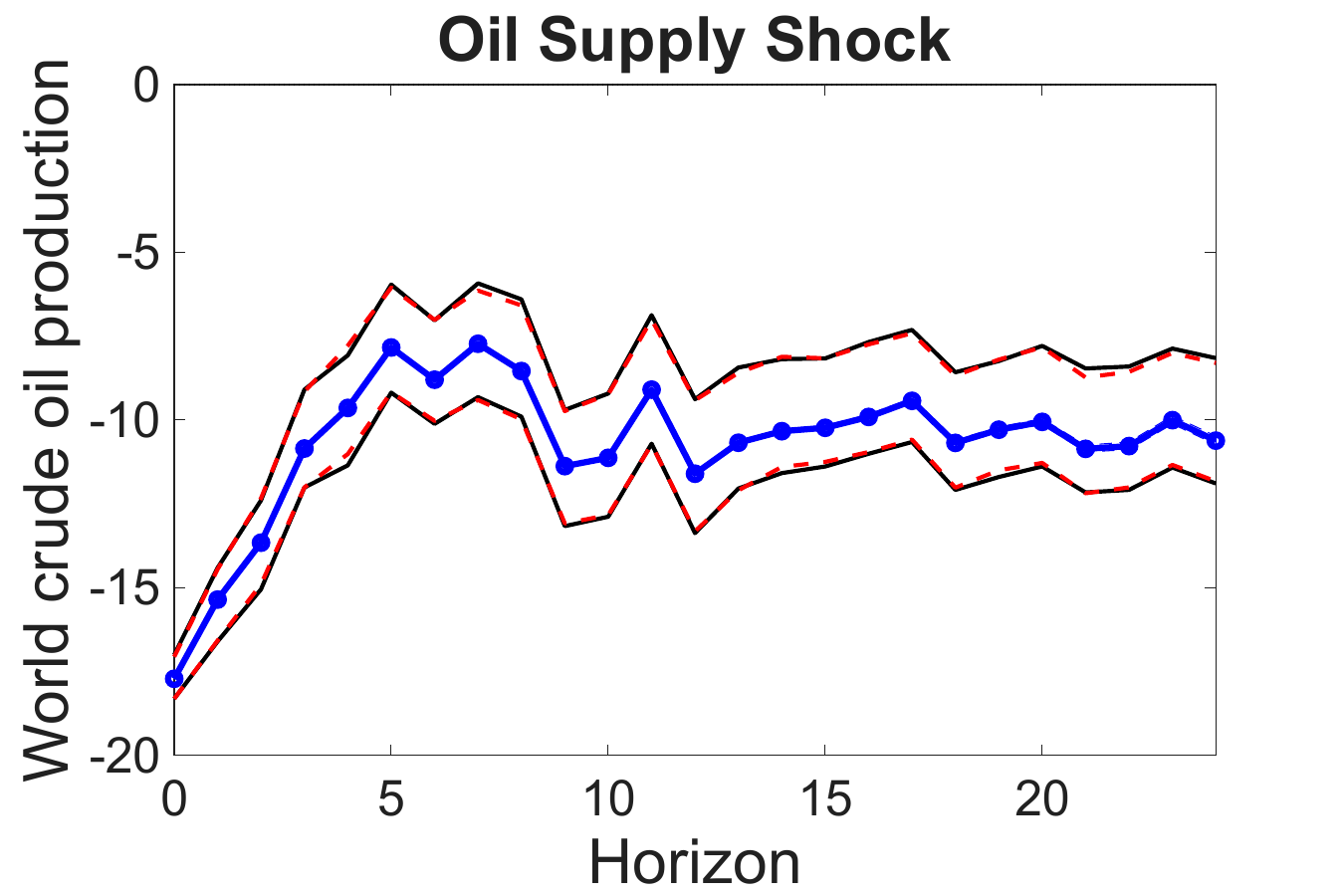}%
  \includegraphics[width=.33\linewidth]{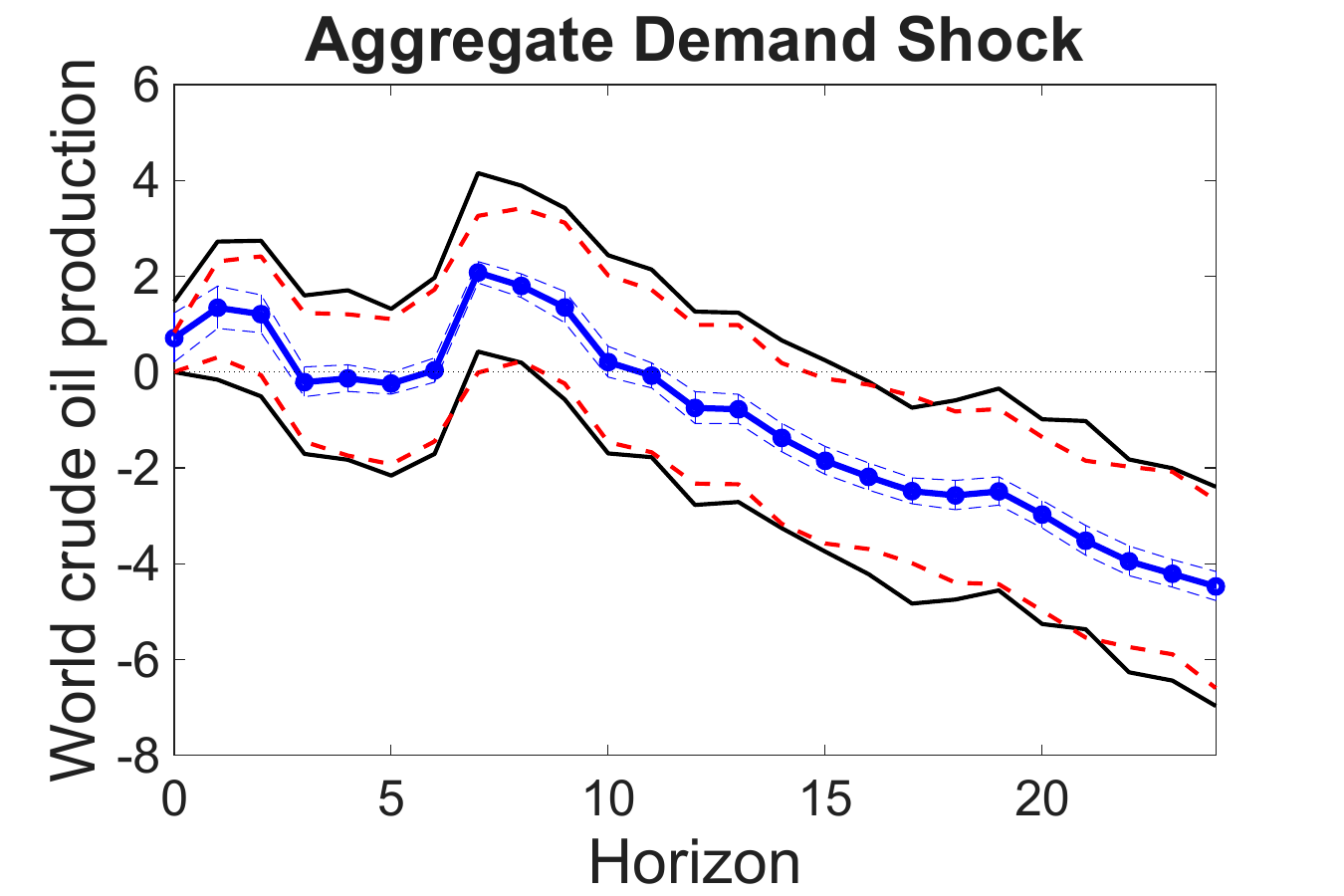}
  \includegraphics[width=.33\linewidth]{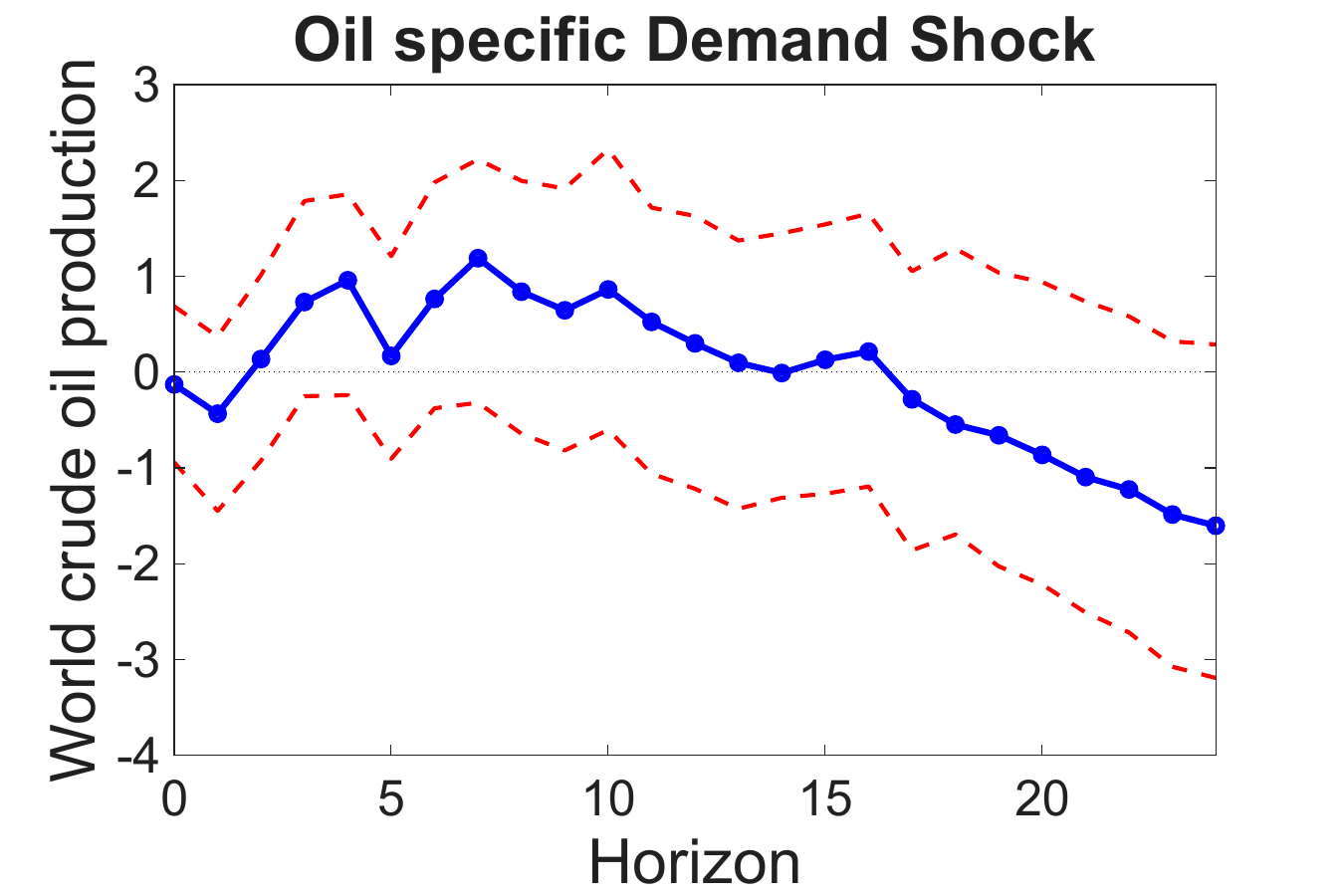}\\
  \includegraphics[width=.33\linewidth]{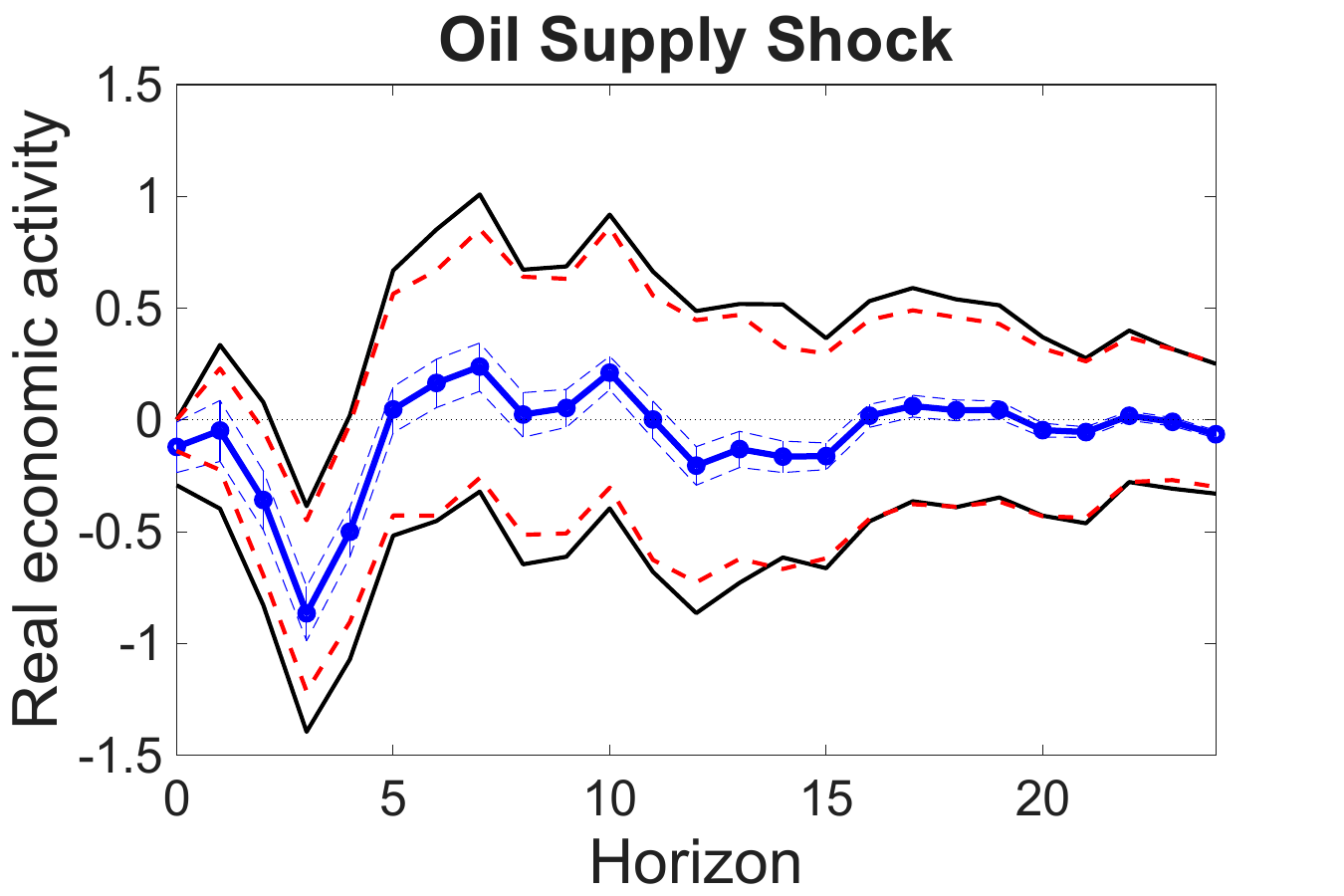}%
  \includegraphics[width=.33\linewidth]{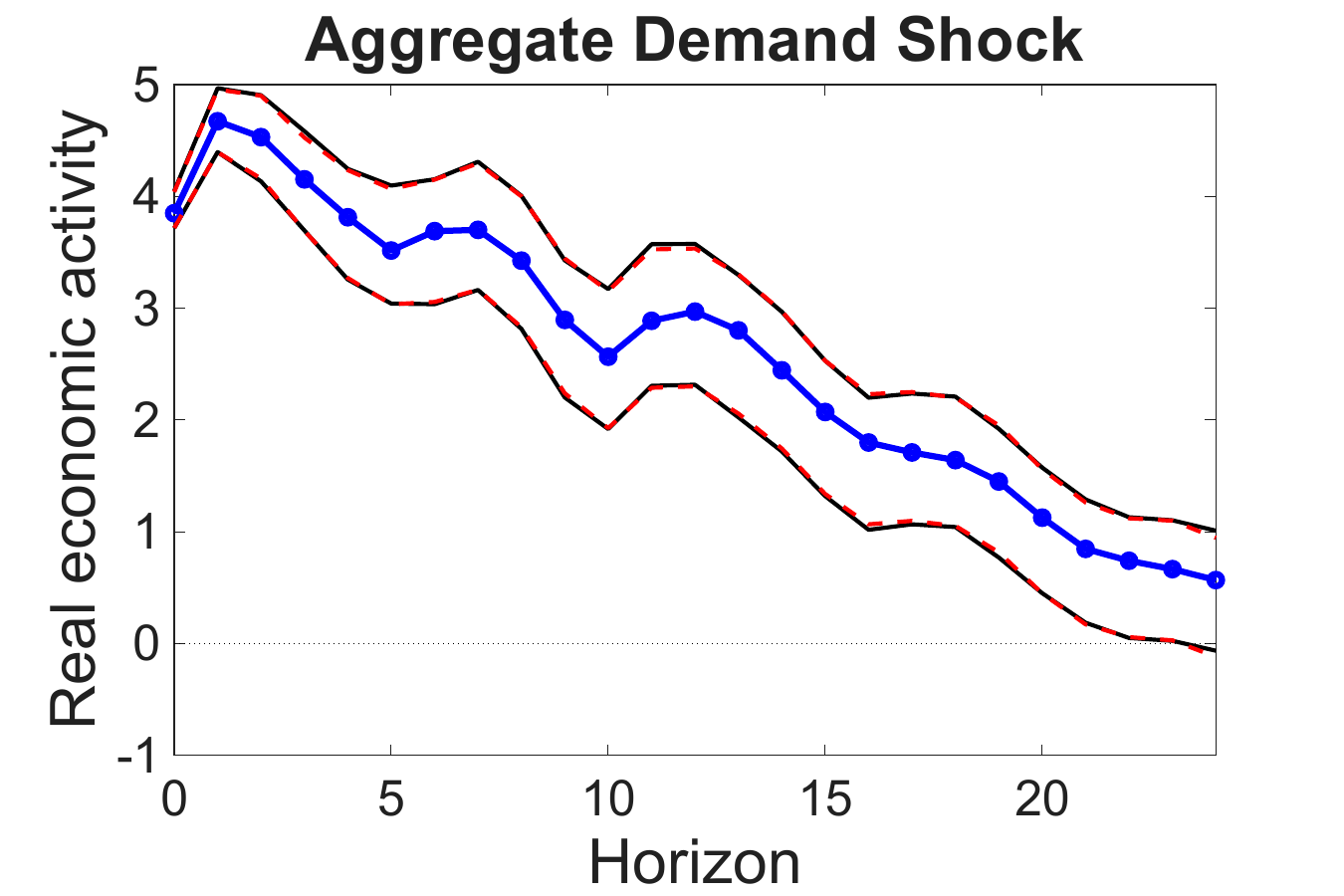}%
  \includegraphics[width=.33\linewidth]{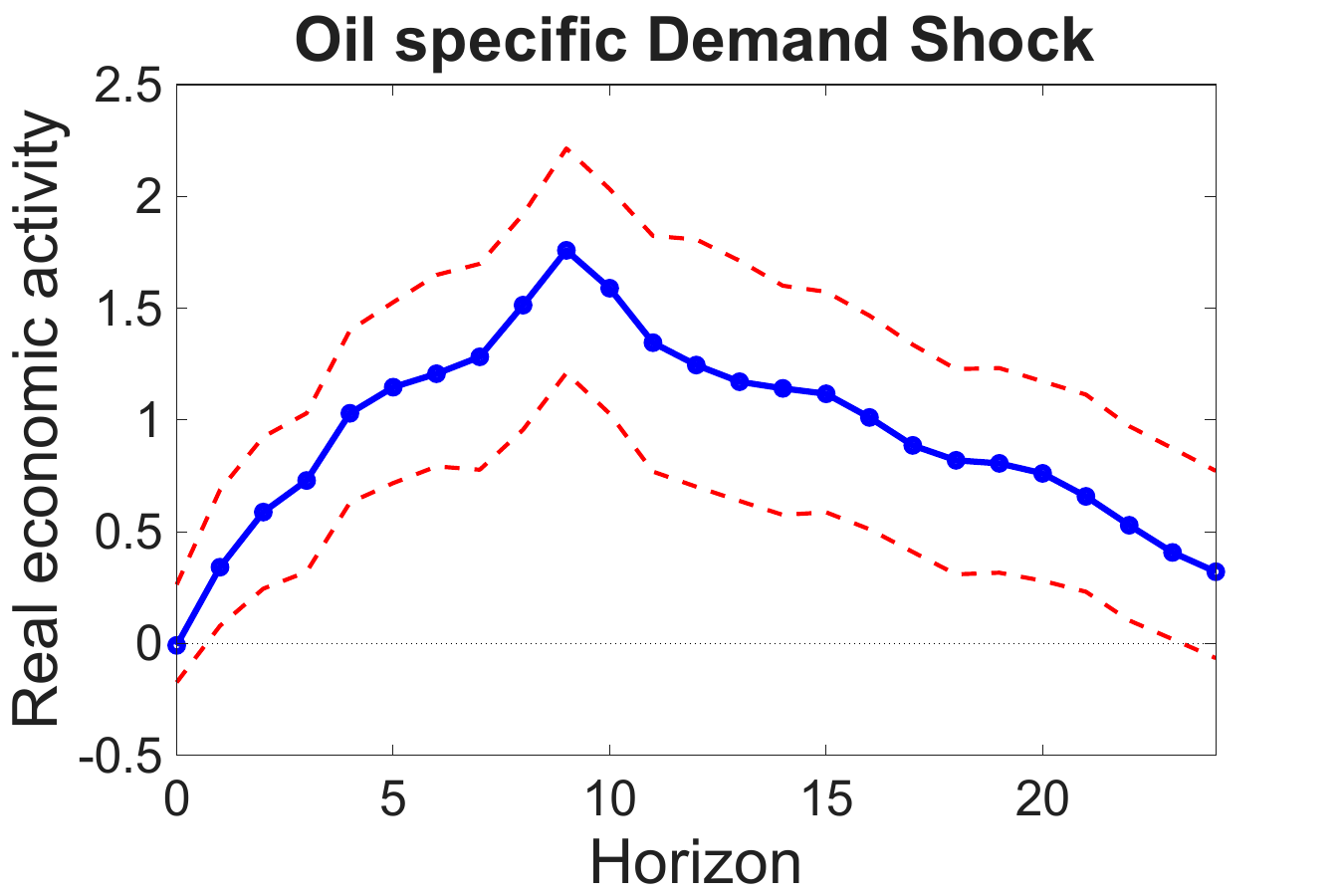}\\
  \includegraphics[width=.33\linewidth]{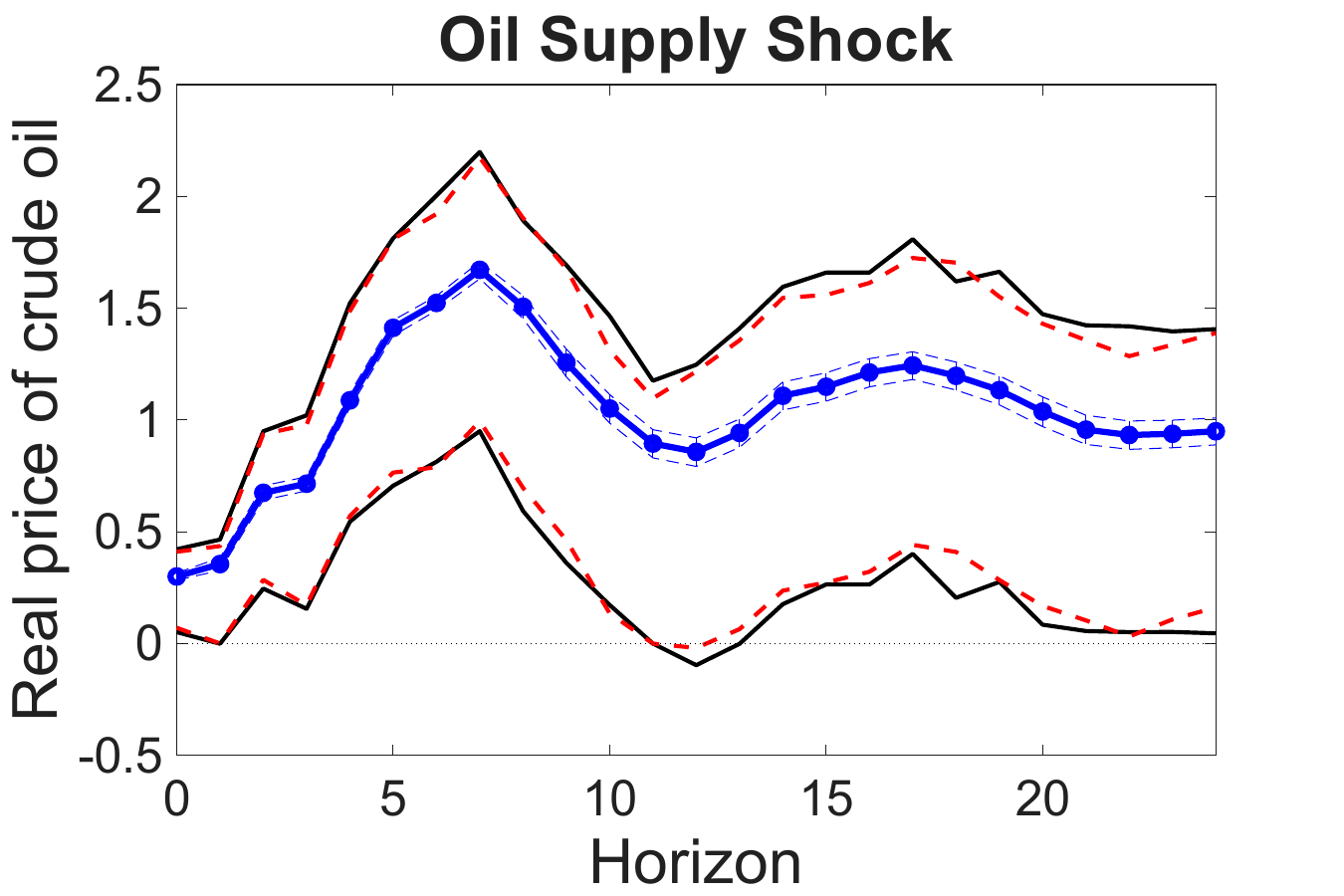}%
  \includegraphics[width=.33\linewidth]{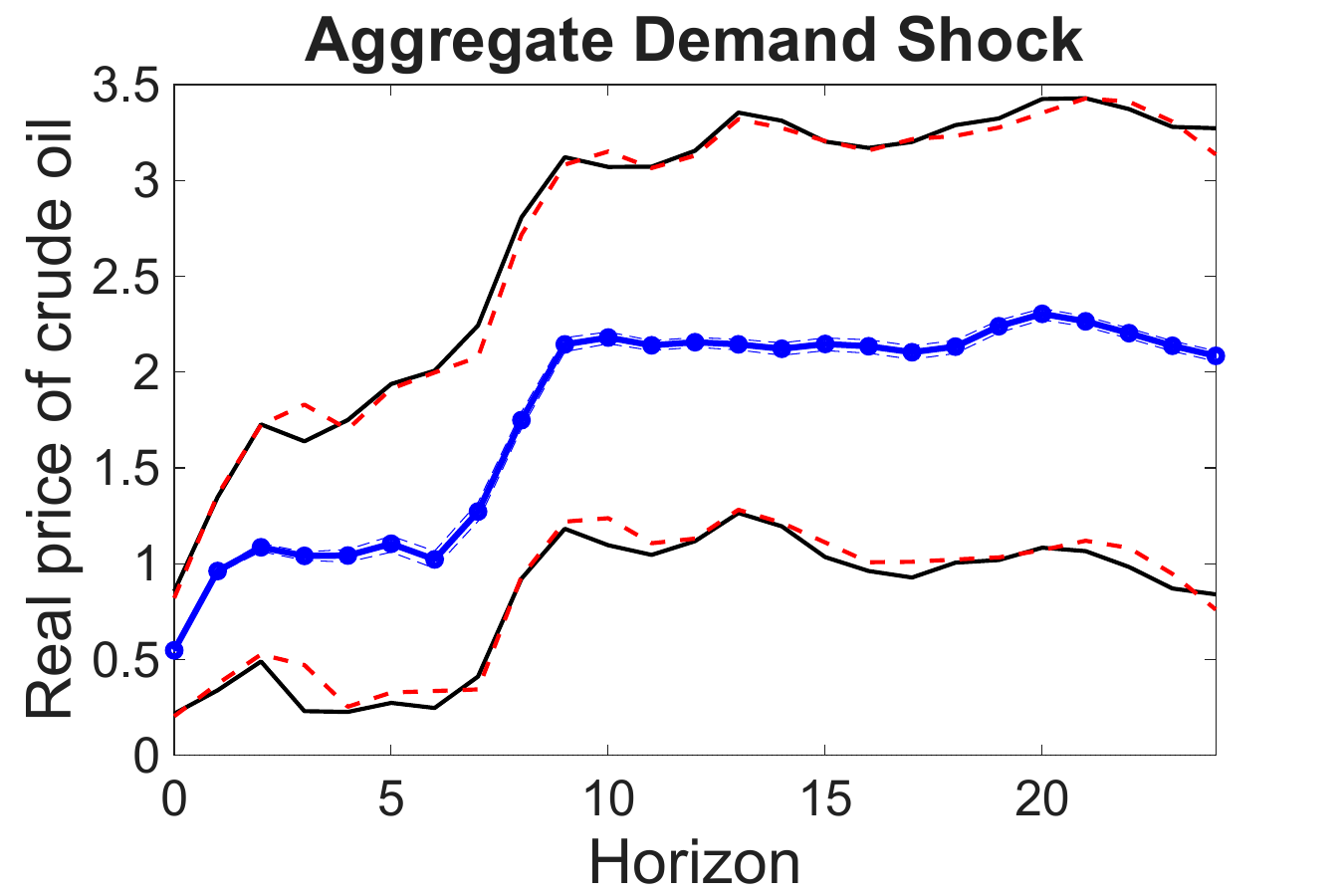}%
  \includegraphics[width=.33\linewidth]{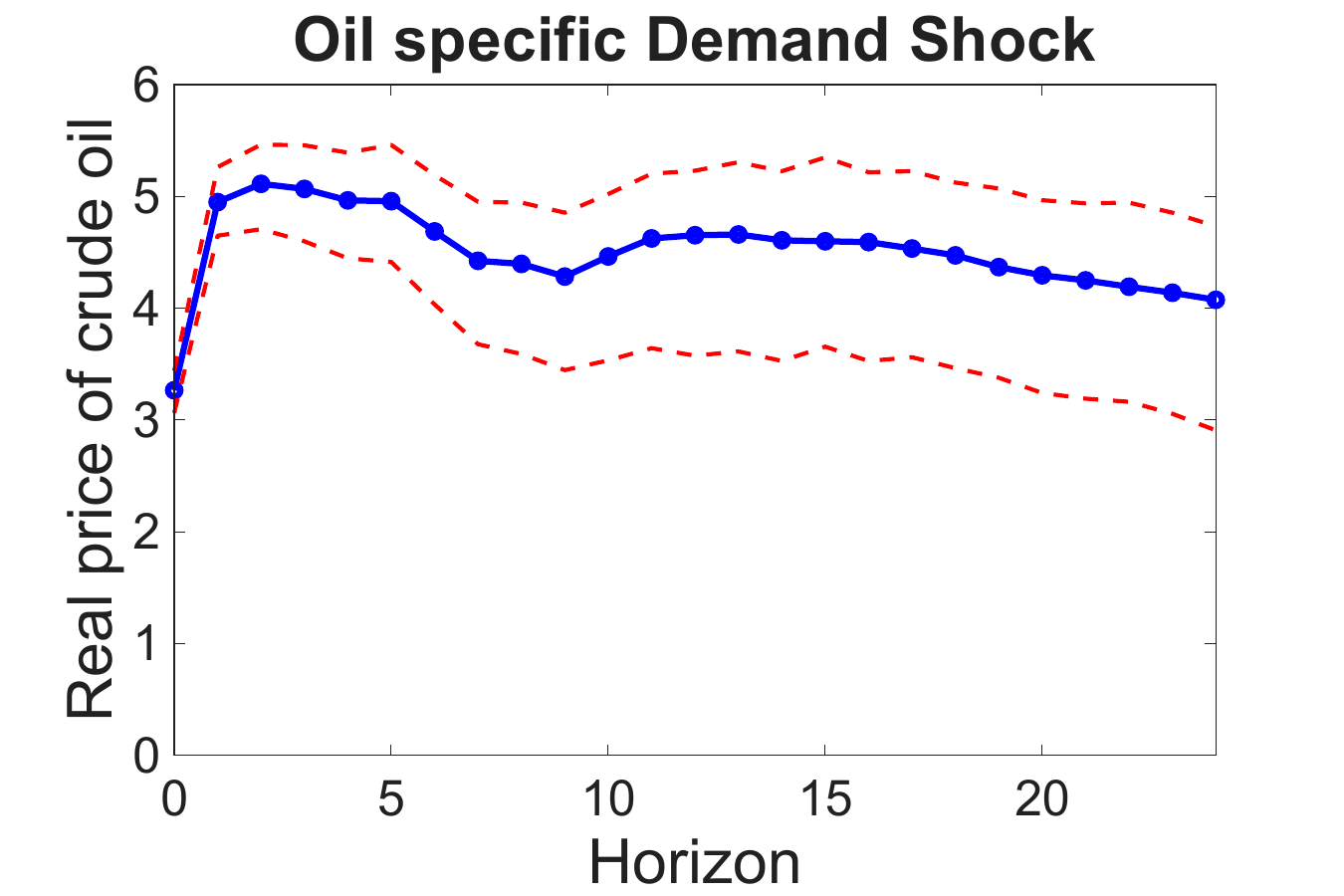}
\caption*{\scriptsize\textit{Notes}: the blue line with dots represents the standard Bayesian posterior mean response, the dashed red lines identify upper and lower bounds of the highest posterior density region with credibility 68\%. Plots in first and second columns of the figure also report the set of posterior means (blue vertical bars) and the bounds of the robust credible region with credibility 68\% (solid black curves). Identification via heteroskedasticity with multiple eigenvalues (i.e. only one shock is point identified), static and dynamic sign restrictions.}
\label{fig:IRF_M2_1986}
\end{figure}
\begin{figure}[ht!]
\caption{HSVAR with Eigenvalue Multiplicity Identified with Sign Restrictions - Break date October 1988}

  \includegraphics[width=.33\linewidth]{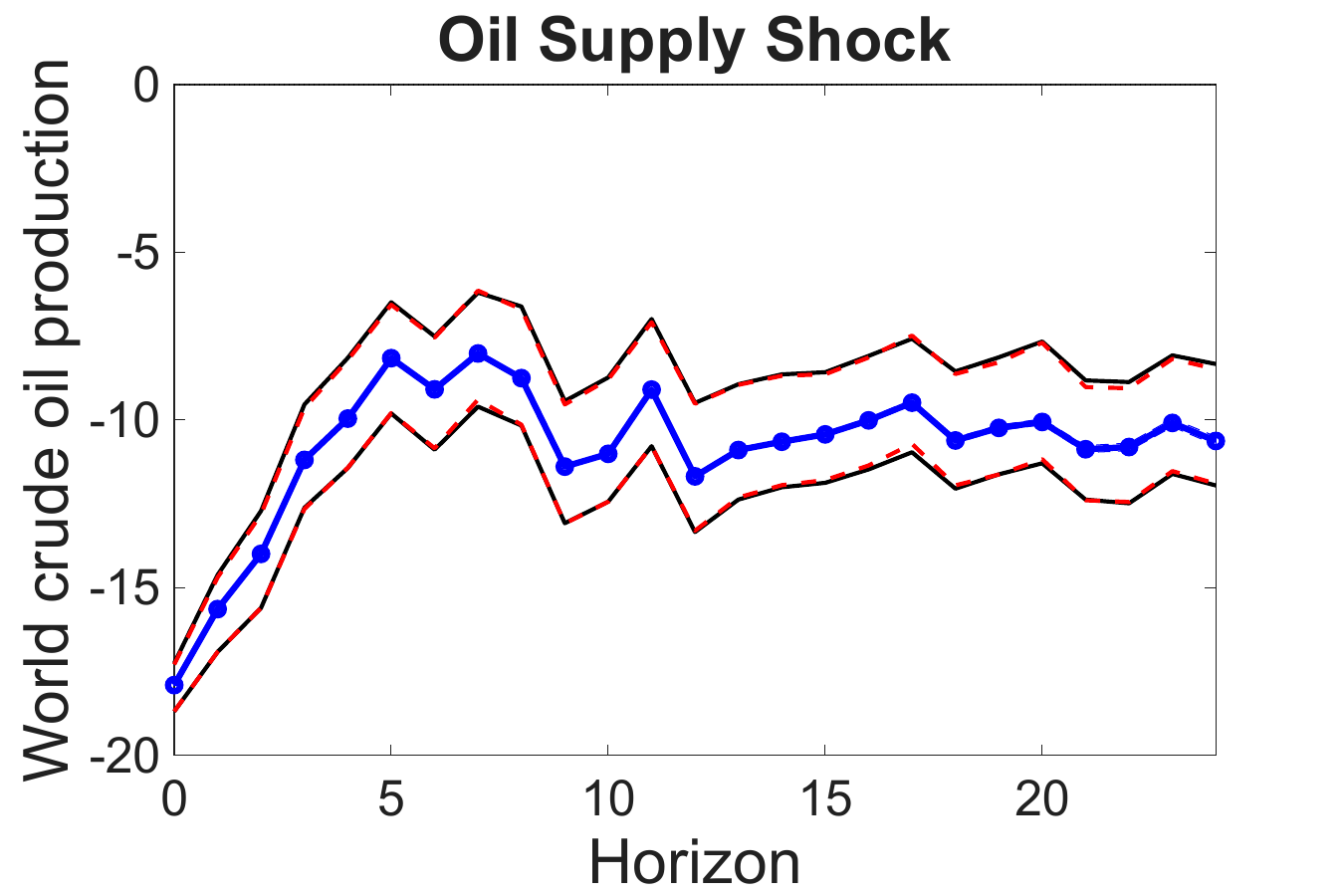}%
  \includegraphics[width=.33\linewidth]{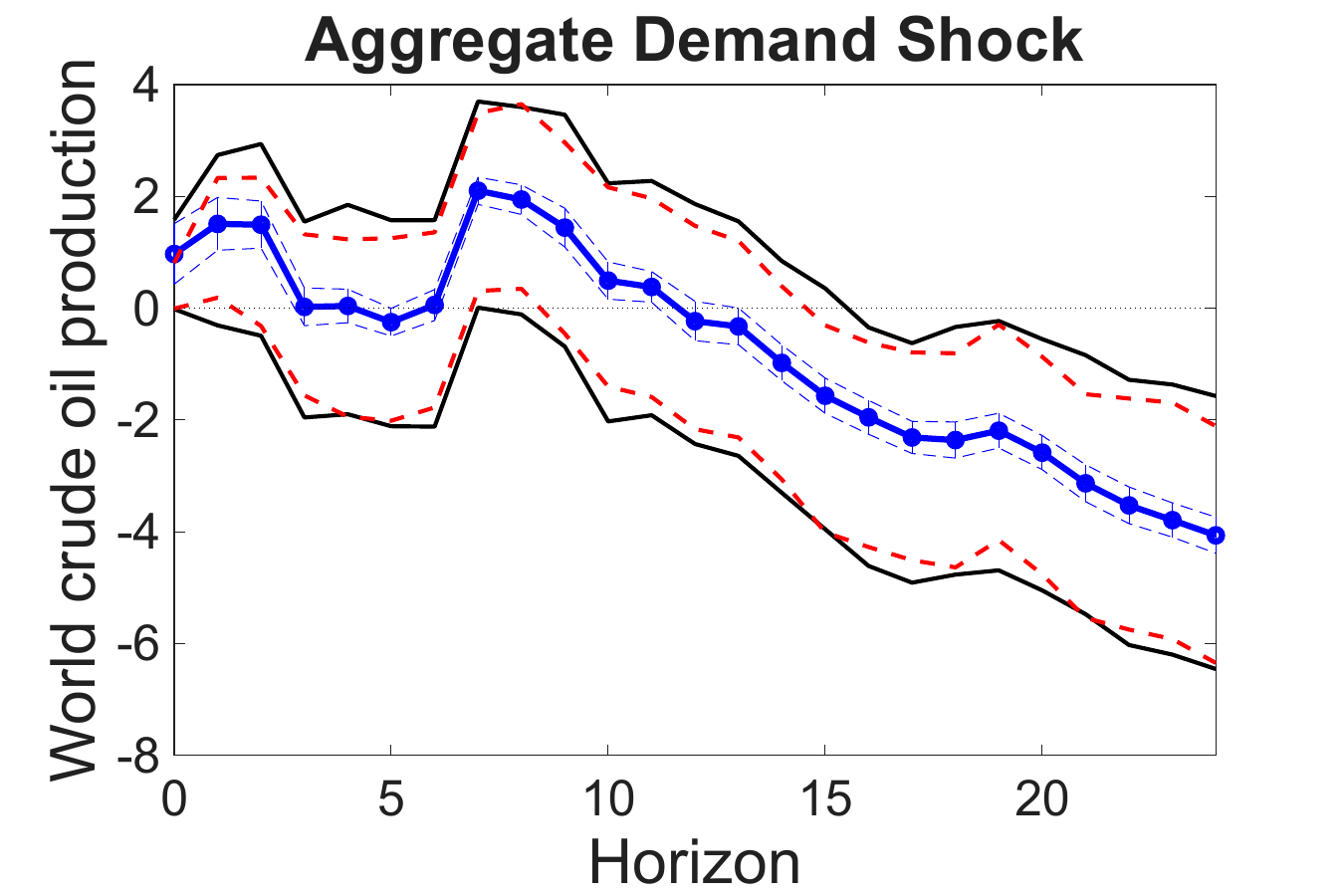}
  \includegraphics[width=.33\linewidth]{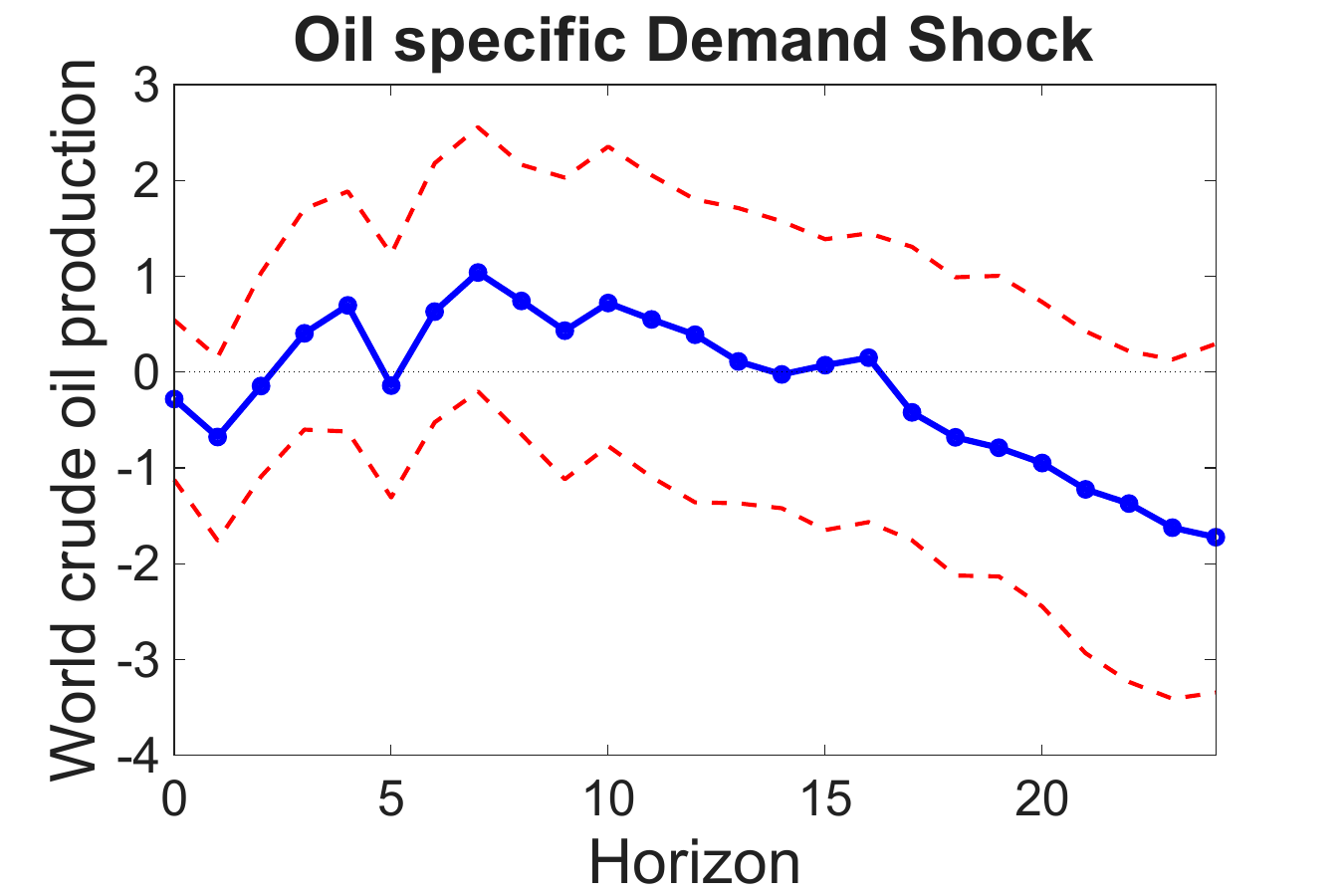}\\
  \includegraphics[width=.33\linewidth]{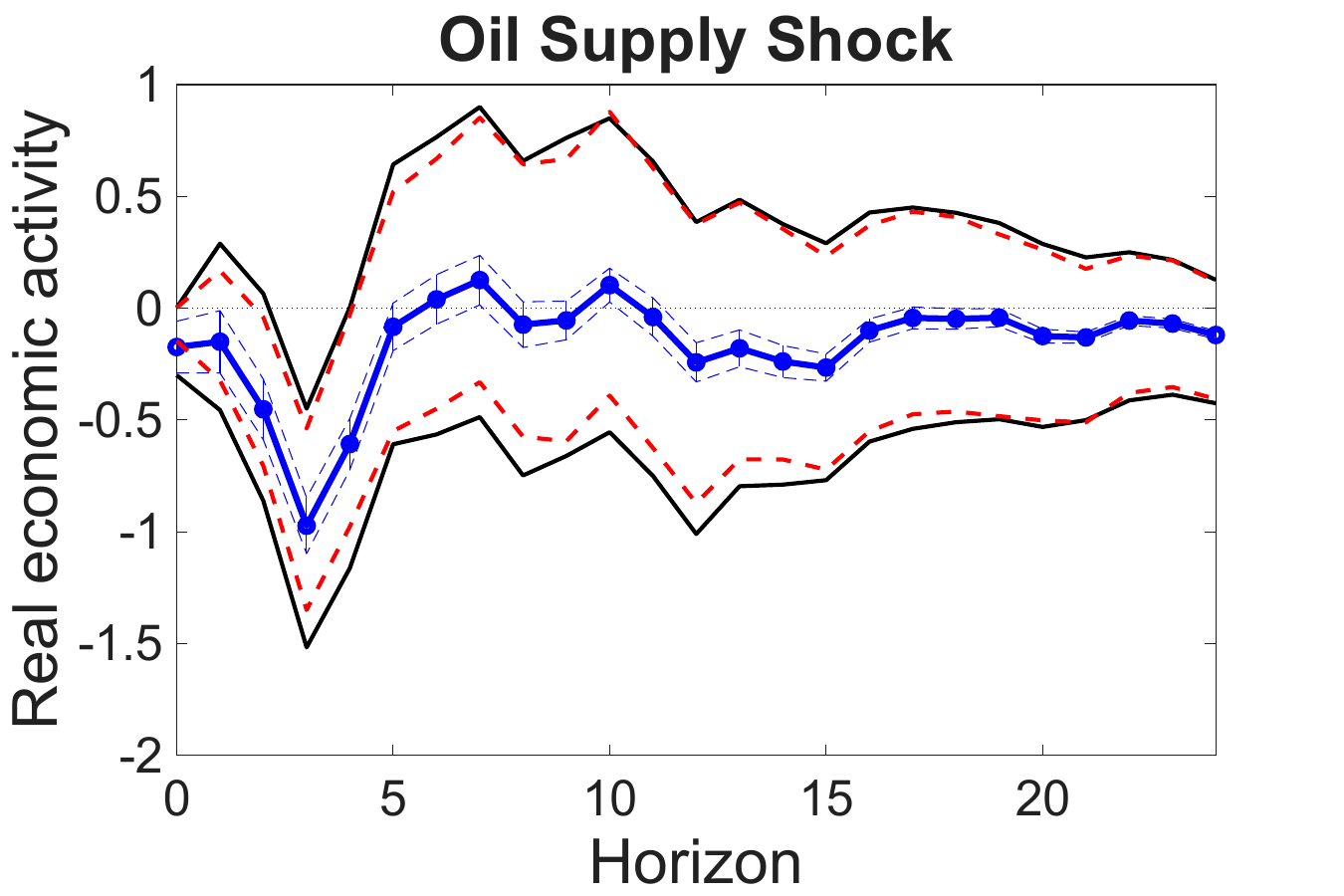}%
  \includegraphics[width=.33\linewidth]{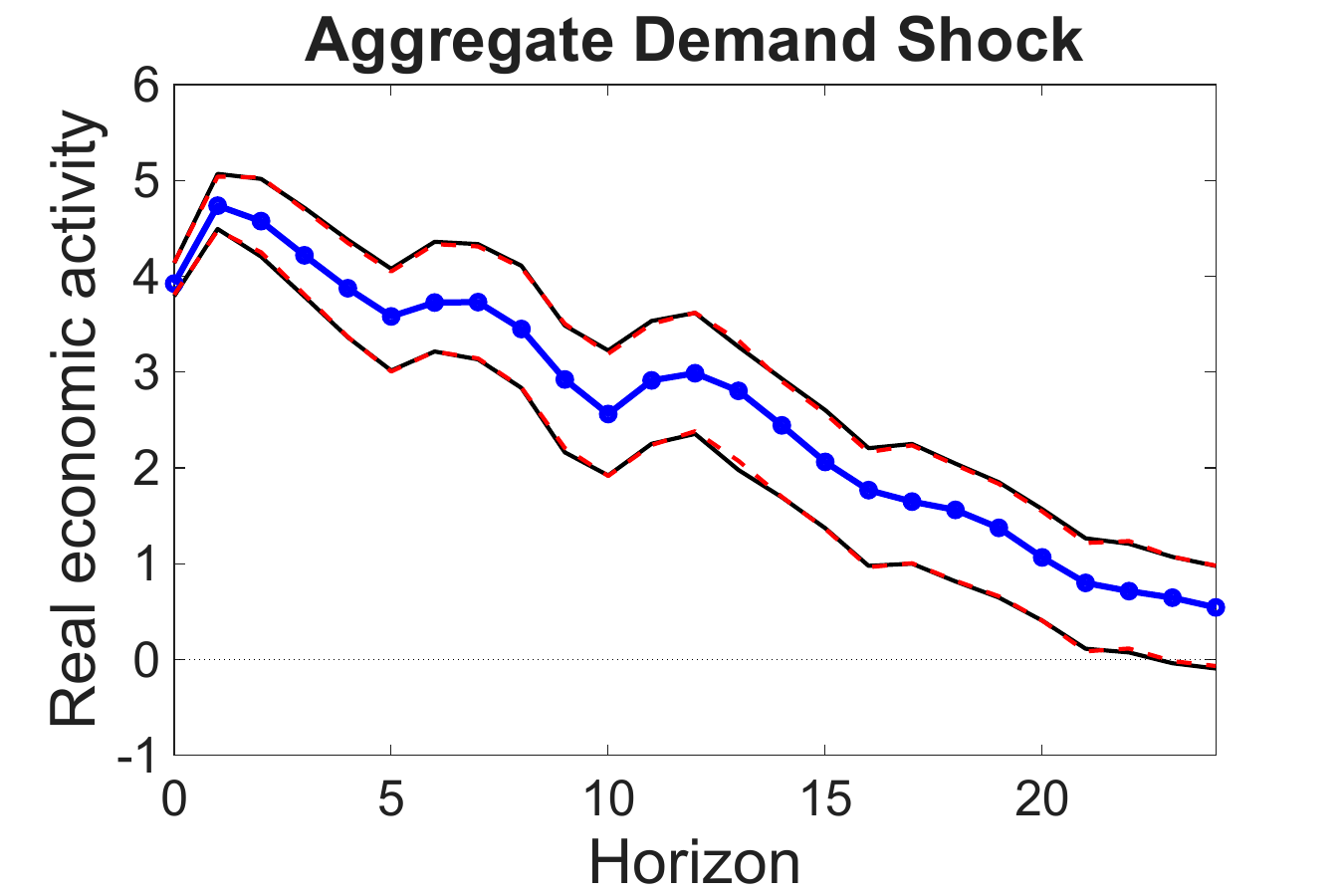}%
  \includegraphics[width=.33\linewidth]{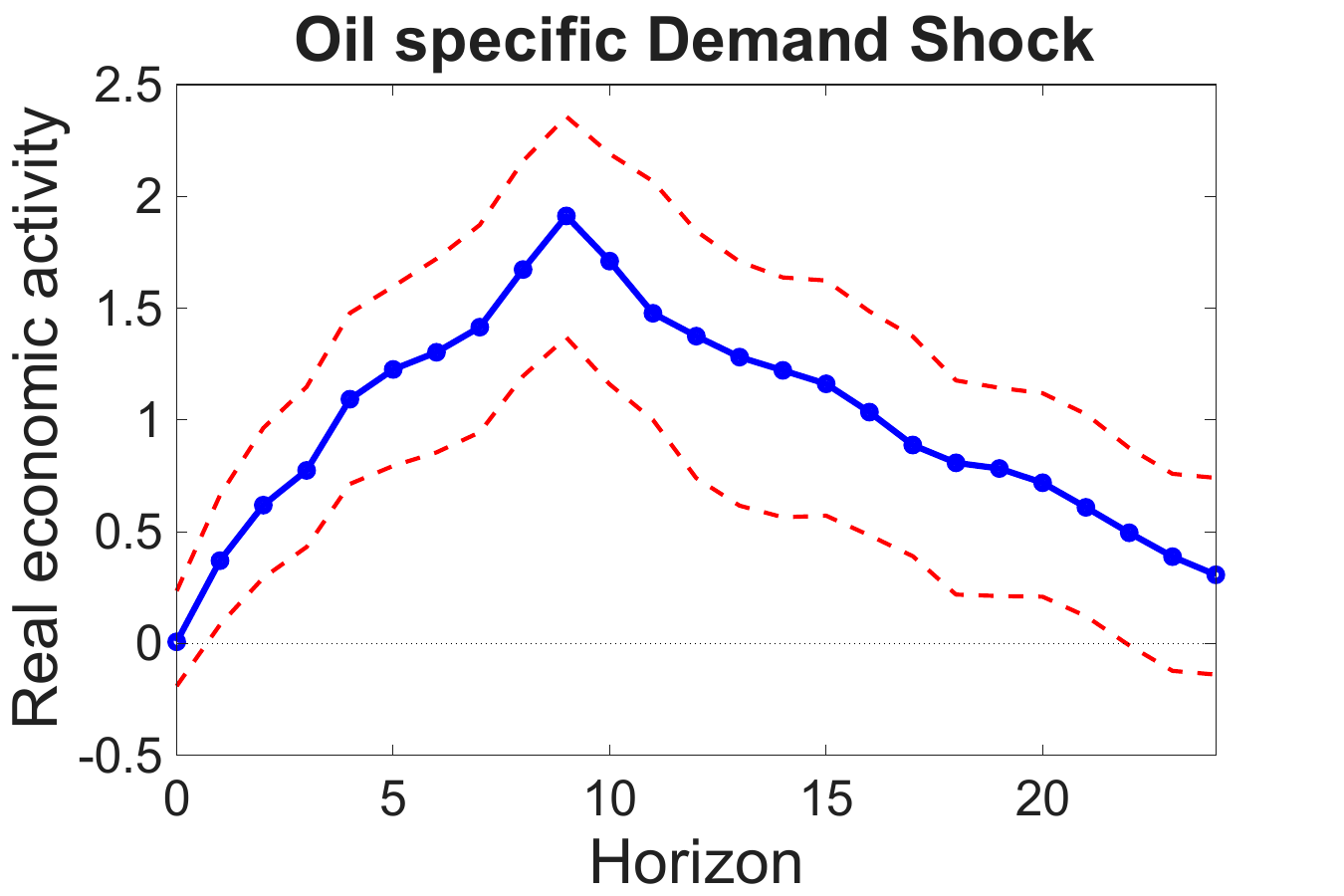}\\
  \includegraphics[width=.33\linewidth]{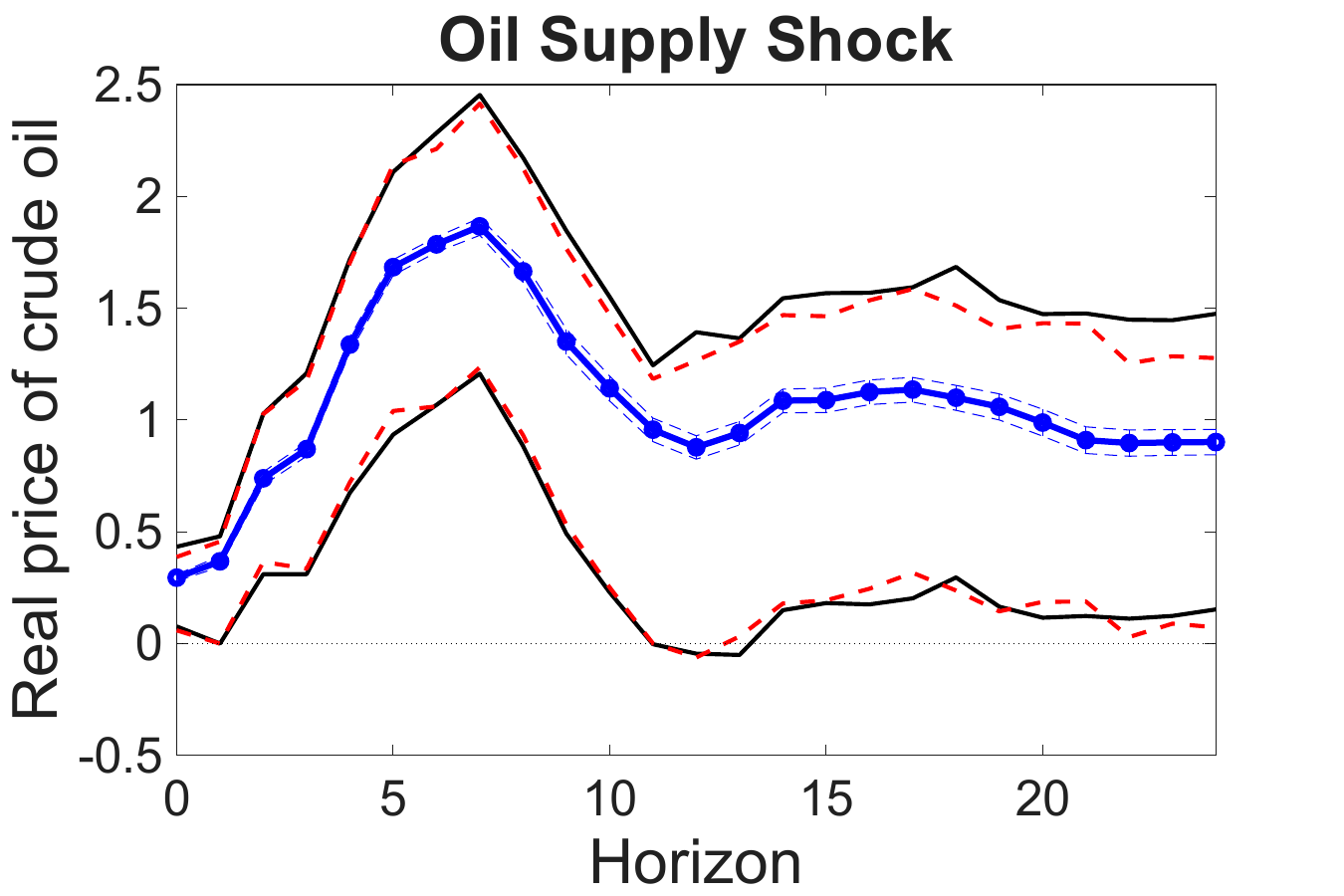}%
  \includegraphics[width=.33\linewidth]{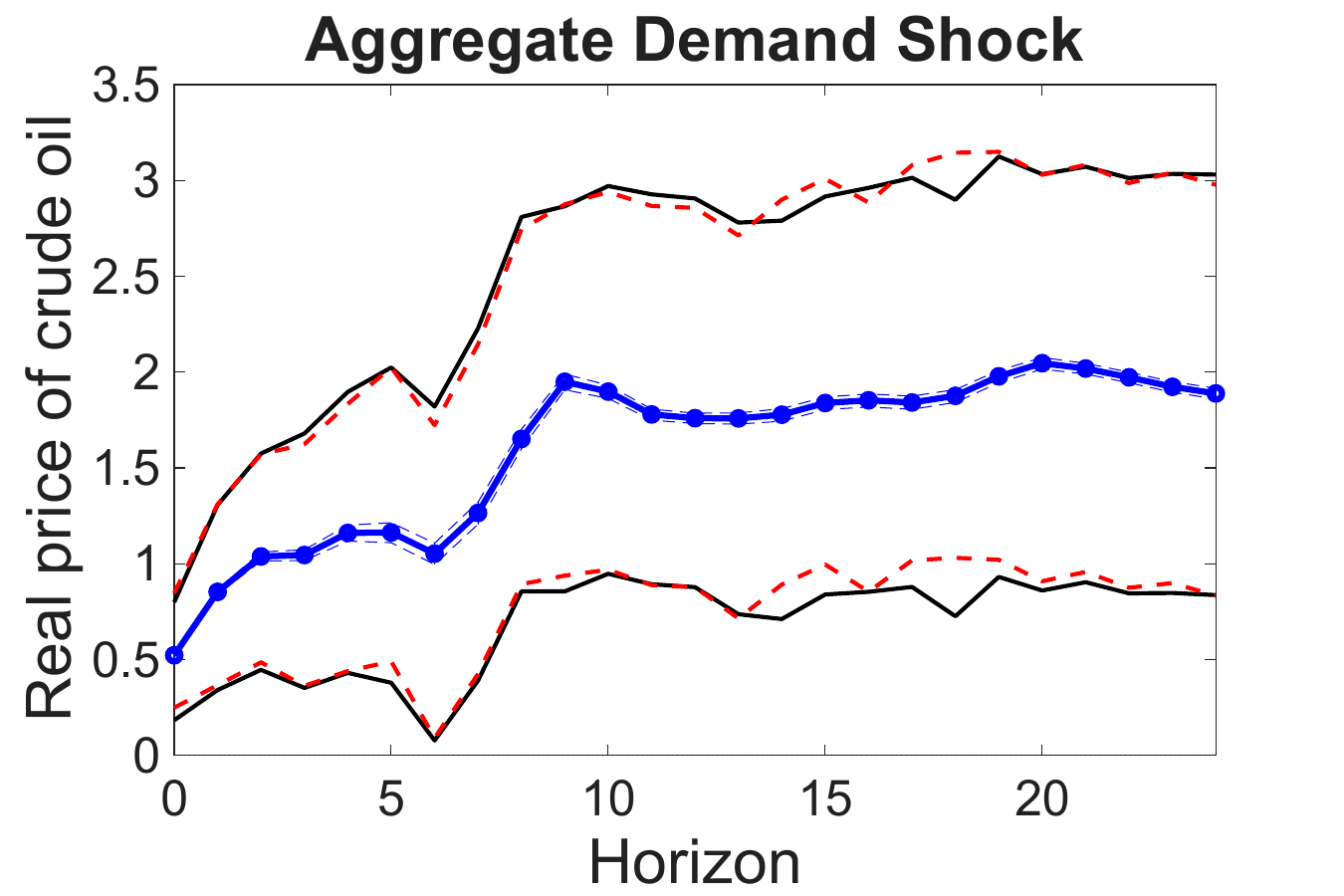}%
  \includegraphics[width=.33\linewidth]{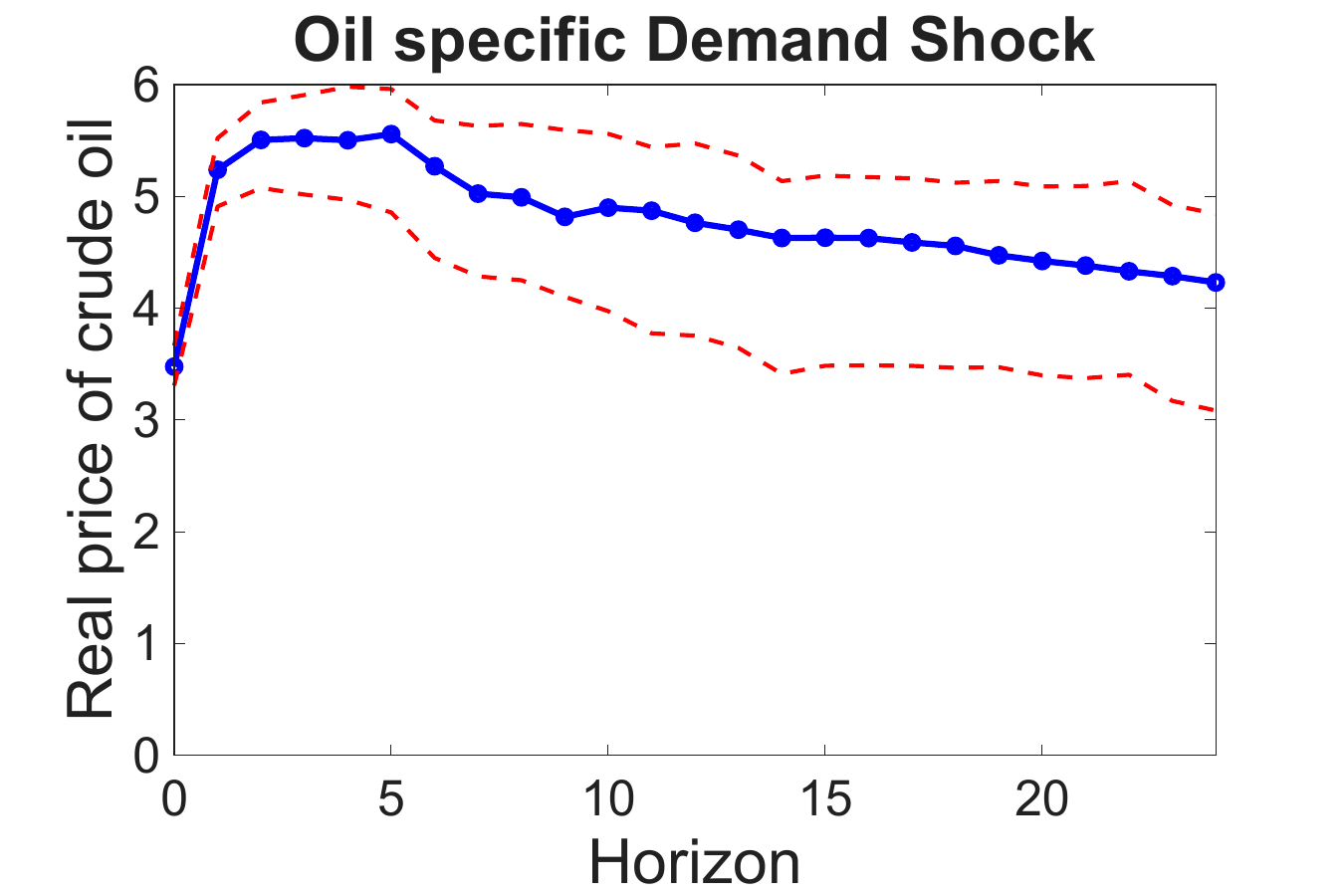}
\caption*{\scriptsize\textit{Notes}: the blue line with dots represents the standard Bayesian posterior mean response, the dashed red lines identify upper and lower bounds of the highest posterior density region with credibility 68\%. Plots in first and second columns of the figure also report the set of posterior means (blue vertical bars) and the bounds of the robust credible region with credibility 68\% (solid black curves). Identification via heteroskedasticity with multiple eigenvalues (i.e. only one shock is point identified), static and dynamic sign restrictions.}
\label{fig:IRF_M2_1988}
\end{figure}

\end{appendices}

\end{document}